\documentclass[11pt,final]{article}
\pdfoutput=1

\usepackage[utf8]{inputenc}
\usepackage{amsmath, amsthm, amssymb, mathtools}
\usepackage{amsfonts}
\usepackage{bbm}
\usepackage{booktabs}
\usepackage[parskip=10pt]{caption}
\usepackage{color}
\usepackage{enumitem}
\usepackage[left=1in,top=1in,right=1in,bottom=1in]{geometry} 
\usepackage{inconsolata}
\usepackage{mleftright}\mleftright
\usepackage{graphicx}
\usepackage{caption}
\usepackage{tikz,tikz-cd}
\usepackage{microtype}
\usepackage{hyperref}
\usepackage{thmtools}
\usepackage{thm-restate}
\usepackage{url}
\usepackage{verbatim}
\usepackage[colorinlistoftodos,prependcaption,textsize=small,color=green!30,obeyDraft]{todonotes}
\usepackage[linesnumbered,ruled,procnumbered]{algorithm2e}
\usepackage[nameinlink,capitalize]{cleveref}

\hypersetup{
    colorlinks=true,
    citecolor=blue,
    linkcolor=blue,
    filecolor=cyan,
    urlcolor=magenta,
    draft=false
}
\declaretheorem[numberwithin=section]{theorem}
\declaretheorem[sibling=theorem]{lemma}
\declaretheorem[sibling=theorem]{corollary}

\theoremstyle{definition}
\declaretheorem[sibling=theorem,name=Problem]{prob}
\declaretheorem[sibling=theorem]{definition}
\declaretheorem[sibling=theorem]{remark}

\newtheoremstyle{named}{}{}{\itshape}{}{}{.}{.5em}{\thmnote{#3}}
\theoremstyle{named}
\newtheorem*{namedtheorem}{Theorem}

\newcommand{\norm}[1]{\left\lVert#1\right\rVert}
\newcommand{\nrm}[1]{\norm{#1}}

\DeclareMathOperator{\finv}{inv}
\DeclareMathOperator{\fsqrt}{sqrt}
\DeclareMathOperator{\nnz}{nnz}
\DeclareMathOperator{\poly}{poly}
\DeclareMathOperator{\polylog}{polylog}
\DeclareMathOperator{\sinc}{sinc}
\DeclareMathOperator{\rank}{rank}
\DeclareMathOperator{\Tr}{Tr}
\DeclareMathOperator{\Var}{Var}
\DeclareMathOperator{\vvec}{vec}
\DeclareMathOperator{\sq}{SQ}
\DeclareMathOperator{\q}{Q}
\DeclareMathOperator{\sqcb}{\mathbf{sq}}
\DeclareMathOperator{\scb}{\mathbf{s}}
\DeclareMathOperator{\qcb}{\mathbf{q}}
\DeclareMathOperator{\ncb}{\mathbf{n}}
\DeclareMathOperator{\pcb}{\mathbf{q}}
\DeclareMathOperator{\sqrun}{\widetilde{\mathbf{sq}}}

\DeclarePairedDelimiter{\abs}{\lvert}{\rvert}

\newcommand{\bbr}{\mathbb{R}}
\newcommand{\bbR}{\mathbb{R}}
\newcommand{\bbc}{\mathbb{C}}
\newcommand{\bbC}{\mathbb{C}}
\newcommand{\E}{\mathbb{E}}
\newcommand{\N}{\mathbb{N}}
\newcommand{\fr}{{\mathrm{F}}}
\newcommand{\op}{{\mathrm{Op}}}
\newcommand{\eps}{\varepsilon}
\newcommand{\range}[1]{[#1]}
\newcommand{\bigO}[1]{\mathcal{O}\left( #1 \right)}
\newcommand{\bOt}[1]{\widetilde{\mathcal O}\left(#1\right)}
\newcommand{\evt}{ {(\mathrm{EV})} }
\newcommand{\qsvt}{ {(\mathrm{QV})} }
\newcommand{\svt}{ {(\mathrm{SV})} }

\SetAlgoProcName{Procedure}{Procedure}
\crefname{claim}{claim}{claims}
\crefname{claim}{Claim}{Claims}

\renewcommand{\emptyset}{\varnothing}
\renewcommand{\ln}{\log}

\usepackage{ifdraft}
\ifdraft{\newcommand{\authnote}[3]{{\color{#3}{\text{ (\bf #1:}} #2)}}}{\newcommand{\authnote}[3]{}}

\begin{document}

\title{Sampling-based sublinear low-rank matrix arithmetic framework for dequantizing quantum machine learning}

\author{Nai-Hui Chia\thanks{Department of Computer Science, University of Texas at Austin. Research supported by Scott Aaronson's Vannevar Bush Faculty Fellowship from the US Department of Defense. Email: \tt{$\{$nai,linhh,chunhao$\}$@cs.utexas.edu}} 
\qquad Andr\'{a}s Gily\'{e}n\thanks{Alfréd Rényi Institute of Mathematics. Formerly at the Institute for Quantum Information and Matter, California Institute of Technology. Funding provided by Samsung Electronics Co., Ltd., for the project ``The Computational Power of Sampling on Quantum Computers'', and by the Institute for Quantum Information and Matter, an NSF Physics Frontiers Center (NSF Grant PHY-1733907), as well as by the EU's Horizon 2020 Marie Skłodowska-Curie program 891889-QuantOrder. Email: \tt{gilyen@renyi.hu}} 
\qquad Tongyang Li\thanks{Department of Computer Science, Institute for Advanced Computer Studies, and Joint Center for Quantum Information and Computer Science, University of Maryland. Research supported by IBM PhD Fellowship, QISE-NET Triplet Award (NSF DMR-1747426), and the U.S. Department of Energy, Office of Science, Office of Advanced Scientific Computing Research, Quantum Algorithms Teams program. Email: \tt{tongyang@cs.umd.edu}} \\
\,\, Han-Hsuan Lin$^{*}$
\qquad Ewin Tang\thanks{ University of Washington. This material is based upon work supported by the National Science Foundation Graduate Research Fellowship Program under Grant No. DGE-1762114. Email: \tt{ewint@cs.washington.edu}} 
\qquad Chunhao Wang$^{*}$}

\date{\empty}
\maketitle

\begin{abstract}
We present an algorithmic framework for quantum-inspired classical algorithms on close-to-low-rank matrices, generalizing the series of results started by Tang's breakthrough quantum-inspired algorithm for recommendation systems [STOC'19]. Motivated by quantum linear algebra algorithms and the quantum singular value transformation (SVT) framework of Gily{\'e}n, Su, Low, and Wiebe~[STOC'19], we develop classical algorithms for SVT that run in time independent of input dimension, under suitable quantum-inspired sampling assumptions. Our results give compelling evidence that in the corresponding QRAM data structure input model, quantum SVT does not yield exponential quantum speedups. Since the quantum SVT framework generalizes essentially all known techniques for quantum linear algebra, our results, combined with sampling lemmas from previous work, suffice to generalize all prior results about dequantizing quantum machine learning algorithms. In particular, our classical SVT framework recovers and often improves the dequantization results on recommendation systems, principal component analysis, supervised clustering, support vector machines, low-rank regression, and semidefinite program solving. We also give additional dequantization results on low-rank Hamiltonian simulation and discriminant analysis. Our improvements come from identifying the key feature of the quantum-inspired input model that is at the core of all prior quantum-inspired results: $\ell^2$-norm sampling can approximate matrix products in time independent of their dimension. We reduce all our main results to this fact, making our exposition concise, self-contained, and intuitive.
\end{abstract}

\newpage

\tableofcontents

\ifdraft{\listoftodos}

\newpage

\section{Introduction}

\subsection{Motivation}
Quantum machine learning (QML) is a field of study with a rapidly growing number of proposals for how quantum computers could significantly speed up machine learning tasks \cite{Dunjko2020nonreviewofquantum,Ciliberto2018}.
If any of these proposals yield substantial practical speedups, it could be the killer application motivating the development of scalable quantum computers~\cite{preskill2018QuantCompNISQEra}.
At first glance, many applications of QML seem to admit exponential speedups.
However, these exponential speedups are less likely to manifest in practice compared to, say, Shor's algorithm for factoring~\cite{shor1994Factoring}, because unlike their classical counterparts, QML algorithms must make strong input assumptions and learn relatively little from their output~\cite{aaronson2015caveat}.
These caveats arise because both loading input data into a quantum computer and extracting amplitude data from an output quantum state are hard in their most generic forms.

A recent line of research analyzes the speedups of QML algorithms by developing classical counterparts that carefully exploit these restrictive input and output assumptions.
This began with a breakthrough 2018 paper by Tang~\cite{tang2018QuantumInspiredRecommSys} showing that the quantum recommendation systems algorithm~\cite{kerenidis2016QRecSys}, previously believed to be one of the strongest candidates for a practical exponential speedup in QML, does not give an exponential speedup.
Specifically, Tang describes a ``dequantized'' algorithm that solves the same problem as the quantum algorithm and only suffers a polynomial slowdown.
Tang's algorithm crucially exploits the input data structure assumed by the quantum algorithm, which is used for efficiently preparing states.
Subsequent work relies on similar techniques to dequantize a wide range of QML algorithms, including those for principal component analysis and supervised clustering~\cite{tang2018QInspiredClassAlgPCA}, low-rank linear system solving~\cite{chia2020QuantInsLinEqSolving}, low-rank semidefinite program solving~\cite{chia2019QInspiredSubLinLowRankSDPSolver}, support vector machines~\cite{ding2019SVM}, nonnegative matrix factorization~\cite{chen2019NMF}, and minimal conical hull~\cite{du2019conical}.
These results show that the advertised exponential speedups of many QML algorithms disappear when compared to classical algorithms with input assumptions analogous to the state preparation assumptions of the quantum algorithms, drastically changing our understanding of the landscape of potential QML algorithm speedups.

A recent line of work in quantum algorithms has worked to unify many quantum algorithms ranging from quantum walks to QML, under a quantum linear algebra framework called quantum singular value transformation (QSVT)~\cite{low2016HamSimQSignProc,chakraborty2018BlockMatrixPowers,gilyen2018QSingValTransf}.
Since this framework captures essentially all known linear algebraic QML techniques~\cite{martyn2021GrandUnificationQAlgs}, including all prior dequantized QML algorithms (up to minor technical details), a natural question is whether this framework can be dequantized.
One cannot hope to dequantize \emph{all} of QSVT, because with sparse block-encodings of input data, QSVT can simulate Harrow, Hassidim, and Lloyd's pioneering poly-logarithmic time algorithm (HHL) for the BQP-complete problem of sampling from the solution of a sparse system of linear equations~\cite{harrow2009QLinSysSolver}.
However, one could hope to dequantize QSVT provided that the input data comes in the state preparation data structure used commonly for quantum linear algebra.
This data structure allows for efficient QML when the input is close-to-low-rank, but as dequantized algorithms show, it also gives significant power to classical algorithms.
Prior work \cite{tang2018QInspiredClassAlgPCA,chia2020QuantInsLinEqSolving} has made similar speculations that these techniques could feasibly dequantize wide swathes of quantum linear algebra.
In this work, we give evidence for these hopes by presenting a classical analogue of the QSVT framework and applying it to dequantize QML algorithms.

\subsection{Results} \label{subsec:mr}
We describe a \emph{simple} framework for quantum-inspired classical algorithms with \emph{wide applicability}, grasping the capabilities and limitations of these techniques.
We use this framework to dequantize many quantum linear algebra algorithms.
We also prove QSVT-like extensibility properties of our framework, giving evidence that it can dequantize any QSVT algorithms in the QRAM input model.

\paragraph{Input model: Oversampling and query access.}
Our framework assumes a specific input model called \emph{oversampling and query access}, which can be thought of as a classical analogue to quantum state preparation assumptions, i.e., the ability to prepare a quantum state $|v\rangle$ proportional to some input vector~$v$.
Our conceptual contribution is to define this generalization of sampling and query access, because it has better closure properties.

We have \emph{sampling and query access} to a vector $v \in \bbc^n$, denoted $\sq(v)$, if we can efficiently make the following kinds of queries (\cref{defn:sq-access}): (1) given an index $i \in [n]$, output the corresponding entry $v(i)$; (2) sample an index $j \in [n]$ with probability $|v(j)|^2/\|v\|^2$; and (3) output the vector's $\ell^2$-norm $\|v\|$.
We have sampling and query access to a matrix $A\in\bbc^{m\times n}$, denoted $\sq(A)$, if we have $\sq(A(i,\cdot))$ for all rows of $A$, $A(i,\cdot)$, and also $\sq(a)$ for $a$ the vector of row norms (i.e., $a(i) \coloneqq \|A(i,\cdot)\|$).
We have \emph{$\phi$-oversampling and query access} to a vector $v$, denoted $\sq_\phi(v)$, if (1) we can query for entries of $v$ and (2) we have sampling and query access to an ``entry-wise upper bound'' vector $\tilde{v}$ satisfying $\|\tilde{v}\|^2 = \phi\|v\|^2$ and $|\tilde{v}(i)| \geq |v(i)|$ for all indices $i$; the definition for a matrix is analogous (\cref{defn:phi-sq-access}).

The parameter $\phi$ should be seen as a form of overhead that comes out in the runtime of algorithms: through rejection sampling, $\sq_\phi(v)$ can do approximate versions of all the queries of $\sq(v)$ with a factor $\phi$ of overhead (\cref{lem:b-sq-approx}).
In this paper, we most often think of $\phi$ as being independent of input size.

To motivate this definition, we make the following observations about this input model.
First, as far as we know, if input data is given \emph{classically},\!\footnote{This assumption is important. When input data is quantum (say, it is coming directly from a quantum system), a classical computer has little hope of performing linear algebra on it efficiently, see for example~\cite{aharononv2021QAlgorithmicMeasurement,huang2021InfThBoundsOnQAdvantageML}.} classical algorithms in the sampling and query model can be run whenever the corresponding algorithms in the quantum model can (\cref{rmk:when-sq-access}).
For example, if input is loaded in the QRAM data structure, as commonly assumed in QML in order to satisfy state preparation assumptions \cite{prakash2014QLinAlgAndMLThesis,Ciliberto2018}, then we have log-time sampling and query access to it.
So, a fast classical algorithm for a problem in this classical model implies lack of quantum speedup for the problem, at least in the usual settings explored in the QML literature.
In particular, a polynomial-time classical algorithm in this model implies lack of exponential quantum speedup.
Second, oversampling and query access has many similarities to the notion of quantum block-encodings in quantum singular value transformation~\cite{gilyen2018QSingValTransf}.
The commonly used data-structures that enable oversampling and query access to a matrix $A$ also enable implementing an efficient quantum circuit whose unitary is a block-encoding of $A$.
Further, in both input models one can perform efficient matrix arithmetic.

\paragraph{Matrix arithmetic.} 
The thrust of our main results is to demonstrate that \emph{oversampling and query access is approximately closed under arithmetic operations}.
We argue that the essential power of quantum-inspired algorithms lies in their ability to leverage sampling and query access of the input matrices to provide oversampling and query access to complex arithmetic expressions as output (possibly with some approximation error), without paying the (at least) linear time necessary to compute such expressions in conventional ways.
While the proven closure properties are important from a complexity theoretic point of view, some of them don't explicitly come into play in the demonstrated applications of our framework for dequantizing QML, since they often come with undesirable polynomial overhead. 
This is in contrast with quantum block-encodings, which generally compose with minimal overhead.

We now list the closure properties that we show, along with the corresponding closure properties proven for block-encodings in \cite{gilyen2018QSingValTransf}.
For all of these, the query time for access to the output is just polynomial in the query times for access to the input, so in particular, these procedures run in time independent of input dimension.
More specifically, we will compare what we call (sub)normalization ``overhead'' between the two, which is the value $\phi$ in the classical setting and what \cite{gilyen2018QSingValTransf} denotes as $\alpha$ in the quantum setting.
The two quantities are analogous, and roughly correspond to overheads in rejection sampling and post-selection, inducing a multiplicative factor in sampling times for both.
\begin{itemize}
    \item Given access to a constant number of vectors $v_1,\ldots,v_\tau$, we have access to linear combinations $\sum_{t=1}^\tau \lambda_t v_t$, and analogously with linear combinations of matrices (\cref{lemma:sample-Mv,lem:weighted-oversampling}).
    This is a classical analogue to the ``linear combinations of unitaries'' technique for block-encodings \cite[Lemma 52]{gilyen2018QSingValTransf}.
    In the quantum setting, there is less overhead,\footnote{In fact, already $\sqrt{\phi}\geq \alpha/\nrm{\sum_{t=1}^\tau \lambda_t v_t}$ which can be seen using that the root mean-square is at least the average.} allowing for efficient block encodings for linear combinations of arbitrarily many matrices in certain settings.
    \item Given access to two matrices $A, B$ with Frobenius norm at most one, we have access to a matrix $Z$ $\eps$-close to the product $A^\dagger B$ in Frobenius norm (\cref{prop:appr-mms,rmk:appr-mms-sq}).
    In the quantum setting, closure of block-encodings under products is almost immediate \cite[Lemma 53]{gilyen2018QSingValTransf} and is not approximate.
    In both cases the individual input overheads of $A$ and $B$ are multiplied. With the same overheads one can also form Kronecker products $A\otimes B$ exactly---this is immediate both in the classical and quantum case \cite{camps2020ApxQCircSynthesisBlockEncoding}.  In particular, given access to two vectors $u$ and $v$, we have access to their outer product $uv^\dagger$ (\cref{lem:outersampling}).
    \item Given access to a matrix $A$ with Frobenius norm at most one and a Lipschitz function $f$, we have access to a matrix $Z$ $\eps$-close to $f(A^\dagger A)$ in Frobenius norm (\cref{thm:evenSing})\footnote{For a Hermitian matrix $H$ and a function $f\colon\bbR\mapsto \bbC$, $f(H)$ denotes applying $f$ to the eigenvalues of $H$. That is, $f(H) \coloneqq \sum_{i=1}^n f(\lambda_i) v_iv_i^\dagger$, for $\lambda_i$ and $v_i$ the eigenvalues and eigenvectors of $H$.}.
    In the quantum setting, block-encodings are closed under even and odd polynomial singular value transformations \cite[Lemmas 8, 10]{gilyen2018QSingValTransf} without approximation, provided the polynomial is low-degree and bounded.
    This block-encoding closure property can be viewed as a corollary of the above two properties, but can also be achieved directly with some more efficient technical machinery.

    An even polynomial singular value transformation of $A$ is precisely $f(A^\dagger A)$ for $f$ a low-degree polynomial (which means $f$ is Lipschitz), and odd polynomials can be decomposed into a product of an even polynomial with $A$, so our closure property is as strong as the quantum one.
    The full details are derived in \cref{subsec:dequant}.
\end{itemize}
To summarize, every arithmetic operation of matrices with block-encodings in \cite{gilyen2018QSingValTransf} (with the possible exception of long linear combinations) can be mimicked by matrices with oversampling and query access, up to Frobenius norm error, provided that an input matrix in a block-encoding corresponds to having\footnote{We take some care here to distinguish whether we have oversampling and query access to $A$ or $A^\dagger$.
    We don't need to: we show that having either one of them implies having the other, up to approximation (\cref{rmk:a-to-a-dagger}).
    However, the accesses assumed in our closure properties are in some sense the most natural choices and require the least overhead.} $\sq(A)$ and $\sq(A^\dagger)$.
The ``linear combinations of vectors'' and ``outer products'' classical closure properties have been used in prior work~\cite{tang2018QuantumInspiredRecommSys,chia2020QuantInsLinEqSolving}.
However, without our new definition of oversampling and query access, it was not clear that these algorithms could be chained indefinitely as we show with these closure properties.

\paragraph{Implications for quantum singular value transformation.}
Our results give compelling evidence that there is indeed no exponential speedup for QRAM-based QSVT, and show that oversampling and query access can be thought of as a classical analogue to block-encodings in the bounded Frobenius norm regime.
Nevertheless, we do not rule out the possibility for large polynomial speedups, as the classical runtimes tend to have impractically large polynomial exponents.
To elaborate more on this connection, we now recall the QSVT framework in more detail.

The QSVT framework of Gily\'{e}n, Su, Low, and Wiebe~\cite{gilyen2018QSingValTransf} assumes that the input matrix $A$ is given by a \emph{block-encoding}, which is a quantum circuit implementing a unitary transformation whose top-left block contains (up to scaling) $A$ itself~\cite{low2016HamSimQSignProc}.
Given a block-encoding of $A$, one can apply it to a quantum state or form block-encodings of other expressions using the closure properties mentioned above.
One can get a block-encoding of an input matrix $A$ through various methods.
If $A$ is $s$-sparse with efficiently computable elements and $\|A\| \leq 1$, then one can directly get a block-encoding of $A/s$~\cite[Lemma~48]{gilyen2018QSingValTransf}.
If $A$ is in the QRAM data structure (used for efficient state preparation for QML algorithms \cite{prakash2014QLinAlgAndMLThesis}), one can directly get a block-encoding of $A/\|A\|_\fr$~\cite[Lemma~50]{gilyen2018QSingValTransf}.
We will use the term \emph{QRAM-based QSVT} to refer to the family of quantum algorithms possible in the QSVT framework when all input matrices \& vectors are given in the QRAM data structure.

The normalization in QRAM-based QSVT means that it has an implicit dependence on the Frobenius norm $\|A\|_\fr$.
Since $\|A\|_\fr$ is also the key parameter in the complexity of our corresponding classical algorithms, this suggests that QRAM-based QSVT does not give inherent exponential quantum speedups (though, if input preparation/output analysis protocols have no classical analogues, they can act as a subroutine in an algorithm that does give an exponential quantum speedup).
Our closure results confirm this: if input matrices and vectors are given in QRAM data structure, then on one (quantum) hand we can construct block-encodings of these matrices normalized to have Frobenius norm one and on the other (classical) hand we have sampling and query access to the input. The conclusion is that, up to some controllable approximation error, an algorithm using the block-encoding framework has a classical analogue in the oversampling and query access model.
The classical algorithm's runtime is only polynomially slower than the corresponding quantum algorithm, except in the $\eps$ parameter.\!\footnote{The QML algorithms we discuss generally only incur $\polylog(\frac1\eps)$ terms, but need to eventually pay $\poly(1/\eps)$ to extract information from output quantum states. So, we believe this exponential speedup is artificial. See the open questions section for more discussion of this error parameter.}
One can argue similarly that there should be no exponential speedup for QSVT for block-encodings derived from (purifications of) density operators \cite[Lemma~45]{gilyen2018QSingValTransf} that come from some well-structured classical data (see \cref{subsec:dequant}).
This stands in contrast to, for example, block-encodings that come from sparsity assumptions \cite[Lemma~48]{gilyen2018QSingValTransf}, where the matrix in the block-encoding can have Frobenius norm as large as $\sqrt{sn}$ (where we take $A$ to be $n\times n$), and so the classical techniques cannot be applied without incurring dependence on $n$ in the runtime.

\subsection{Technical overview}

We now illustrate the flavor of the algorithmic ideas underlying our main results, by showing why the ``oversampling'' input model is closed under approximate matrix products.
Suppose we are given sampling and query access to two matrices $A \in \bbc^{m\times n}$ and $B \in \bbc^{m\times p}$, and desire (over)sampling and query access to $A^\dagger B$.
$A^\dagger B$ is a sum of outer products of rows of $A$ with rows of $B$ (that is, $A^\dagger B = \sum_{i=1}^m A(i,\cdot)^\dagger B(i,\cdot)$), so a natural idea is to use the outer product closure property to get access to each outer product individually, and then use the linear combination closure property to get access to their sum, which is $A^\dagger B$ as desired.
However, there are $m$ terms in the sum, which is too large: we can't even compute entries of $A^\dagger B$ in time independent of $m$.
So, we use sampling to approximate this sum of $m$ terms by a linear combination over far fewer terms, allowing us to get access to $Z$ for $Z \approx A^\dagger B$.
This type of matrix product approximation is well-known in the classical literature~\cite{DKM06}.
Given $\sq(A)$, we can pull samples $i_1,\ldots,i_s$ according to the row norms of $A$, a distribution we will denote $p$ (so $p(i) = \|A(i,\cdot)\|^2/\|A\|_\fr^2$).
Consider $Z \coloneqq \frac{1}{s}\sum_{k=1}^s \frac{1}{p(i_k)}A(i_k,\cdot)^\dagger B(i_k,\cdot)$.
$Z$ is an unbiased estimator of $A^\dagger B$: $\E[Z]=\frac1s\sum_{k=1}^s\sum_{\ell=1}^m p(\ell)\frac{A(\ell,\cdot)^\dagger B(\ell,\cdot)}{p(\ell)}=\sum_{\ell=1}^m A(\ell,\cdot)^\dagger B(\ell,\cdot)=A^\dagger B$.
Further, the variance of this estimator is small.
In the following computation, we consider $s = 1$, because the variance for general $s$ decreases as $1/s$.
\begin{multline*}
    \E[\|A^\dagger B - Z\|_\fr^2] \leq \sum_{i,j} \E[\abs{Z(i,j)}^2]
    = \sum_{i,j}\sum_\ell p(\ell) \frac{1}{p(\ell)^2}\abs{A(\ell,i)}^2\abs{B(\ell,j)}^2 \\
    = \sum_\ell \frac{1}{p(\ell)} \|A(\ell,\cdot)\|^2\|B(\ell,\cdot)\|^2
    = \sum_\ell \|A\|_\fr^2\|B(\ell,\cdot)\|^2
    = \|A\|_\fr^2\|B\|_\fr^2.
\end{multline*}
By Chebyshev's inequality, we can choose $s = \bigO{\frac{1}{\eps^2}}$ to get that $\|Z - A^\dagger B\|_\fr < \eps\|A\|_\fr\|B\|_\fr$ with probability $0.99$.
Since $Z$ is a linear combination of $s$ outer products, this gives us oversampling and query access to $Z$ as desired.
In our applications we would keep $Z$ as an outer product $A'^\dagger B'$ for convenience.
We have just sketched the proof of our key lemma: an approximate matrix product protocol.

\begin{namedtheorem}[\textbf{Key lemma} \cite{DKM06} (informal version of \cref{prop:appr-mms})]
Suppose we are given $\sq(X) \in \bbc^{m\times n}$ and $\sq(Y) \in \bbc^{m\times p}$.
Then we can find normalized submatrices of $X$ and $Y$, $X' \in \bbc^{s\times n}$ and $Y' \in \bbc^{s\times p}$, in $\bigO{s}$ time for $s = \Theta(\frac{1}{\eps^2}\log\frac{1}{\delta})$, such that
\[
  \Pr\Big[\|X'^\dagger Y' - X^\dagger Y\|_\fr \leq \eps\|X\|_\fr\|Y\|_\fr\Big] > 1-\delta.
\]
We subsequently have $\bigO{s}$-time $\sq(X'), \sq(X'^\dagger), \sq(Y'), \sq(Y'^\dagger)$.
\end{namedtheorem}

Prior quantum-inspired algorithms \cite{tang2018QuantumInspiredRecommSys,tang2018QInspiredClassAlgPCA,chia2018QInspiredSubLinLowRankLinEqSolver,chia2019QInspiredSubLinLowRankSDPSolver} indirectly used this lemma by using \cite{frieze2004FastMonteCarloLowRankApx}, which finds a low-rank approximation to the input matrix in the form of an approximate low-rank SVD and relies heavily on this lemma in the analysis.

One of our main results, mentioned earlier as our singular value transformation closure property, is that, given $\sq(A) \in \bbc^{m\times n}$, in time independent of $m$ and $n$, we can access an approximation of $f(A^\dagger A)$ for a Lipschitz-function $f$ that, without loss of generality, satisfies $f(0) = 0$ (\cref{thm:evenSing}).
One could use \cite{frieze2004FastMonteCarloLowRankApx} to give a classical algorithm for SVT, but a more efficient approach is to directly apply the key lemma twice to get an approximate decomposition of $f(A^\dagger A)$:
\begin{align*}
    f(A^\dagger A) &\approx f(R^\dagger R) \tag*{by key lemma, with $R \in \bbc^{r \times n}$ normalized rows of $A$}\\
    &= R^\dagger \bar{f}(RR^\dagger) R \tag*{by computation, where $\bar{f}(x) \coloneqq f(x)/x$)}\\
    &\approx R^\dagger \bar{f}(CC^\dagger) R \tag*{by key lemma, with $C \in \bbc^{r\times c}$ normalized columns of $R$}
\end{align*}
We call $R^\dagger \bar{f}(CC^\dagger) R$ an \emph{RUR decomposition} because $R \in \bbc^{r\times n}$ is a subset of rows of the input matrix and $U$ is a matrix with size independent of input dimension ($R$ corresponds to the `R' of the RUR decomposition, and $\bar{f}(CC^\dagger) \in \bbc^{r\times r}$ corresponds to the `U').
In other words, an RUR decomposition expresses a desired matrix as a linear combination of $r^2$ outer products of rows of the input matrix ($\sum_{i,j} [\bar{f}(CC^\dagger)](i,j)R(i,\cdot)^\dagger R(j, \cdot)$, for example).\!\footnote{This is the relevant variant of the notion of a \emph{CUR decomposition} from the randomized numerical linear algebra and theoretical computer science communities \cite{Drineas2008}.}
We want our output in the form of an RUR decomposition, since we can describe such a decomposition implicitly just as a list of row indices and some additional coefficients, which avoids picking up a dependence on $m$ or $n$ in our runtimes.
Further, having $\sq(A)$ gives us $\sq_\phi(R^\dagger  U R)$ via closure properties, enabling efficient access to matrix-vector expressions like $R^\dagger URb$.

More general results follow as corollaries of our main result on even SVT (\cref{thm:gen-svt,thm:eig-svt}).
However, using only our main theorem about even SVT, we can directly recover most existing quantum-inspired machine learning algorithms without using these general results, yielding faster dequantization for QML algorithms.
We now outline our results recovering such applications.

\subsection{Applications: dequantizing QML \& more} \label{subsec:intro-apps}

\renewcommand{\arraystretch}{2.5}
\begin{figure}
\begin{center}\begin{tabular}{r | c c c }
    & quantum algorithm & prior work & this work \\
    \hline
    \begingroup 
    \renewcommand{\arraystretch}{1.2}
    \begin{tabular}{@{}r@{}}
        simple QSVT \\
        \cite{gilyen2018QSingValTransf}, \S\ref{subsec:dequant}
    \end{tabular}
    \endgroup
    & $\displaystyle \frac{d\|A\|_\fr\|b\|}{\|p^\qsvt(A)b\|}$
    &
    & $\displaystyle \frac{d^{22}\|A\|_\fr^6\|b\|^6}{\eps^6\|p^\qsvt(A)b\|^6}$ \\

    \begingroup 
    \renewcommand{\arraystretch}{1.2}
    \begin{tabular}{@{}r@{}}
        recommendation systems  \\
        \cite{chakraborty2018BlockMatrixPowers}, \cite{tang2018QuantumInspiredRecommSys}, \S\ref{sec:recommendation-systems}
    \end{tabular}
    \endgroup
    & $\displaystyle \frac{\|A\|_\fr}{\sigma}$
    & $\displaystyle \frac{\|A\|_\fr^{24}}{\sigma^{24}\eps^{12}}$
    & $\displaystyle \frac{\|A\|_\fr^6\|A\|^{10}}{\sigma^{16}\eps^6}$ \\

    \begingroup 
    \renewcommand{\arraystretch}{1.2}
    \begin{tabular}{@{}r@{}}
        supervised clustering \\
        \cite{lloyd2013Clustering}, \cite{tang2018QInspiredClassAlgPCA}, \S\ref{sec:supervised-clustering}
    \end{tabular}
    \endgroup
    & $\displaystyle \frac{\|M\|_\fr^2\|w\|^2}{\eps}$\makebox[0pt]{\quad${\vphantom{\Big)}}^{(\clubsuit)}$} & $\displaystyle \frac{\|M\|_\fr^4\|w\|^4}{\eps^2}$
    & $\displaystyle \frac{\|M\|_\fr^4\|w\|^4}{\eps^2}$ \\

    \begingroup 
    \renewcommand{\arraystretch}{1.2}
    \begin{tabular}{@{}r@{}}
        \kern-3mm principal component analysis \\
        \cite{chakraborty2018BlockMatrixPowers}, \cite{tang2018QInspiredClassAlgPCA}, \S\ref{sec:PCA}
    \end{tabular}
    \endgroup
    & $\displaystyle \frac{\|X\|_\fr\|X\|}{\lambda_k \eps}$
    & $\displaystyle \frac{\|X\|_\fr^{36}}{\|X\|^{12}\lambda_k^{12}\eta^6\eps^{12}}$
    & $\displaystyle \frac{\|X\|_\fr^{6}}{\|X\|^2\lambda_k^2\eta^6\eps^6}$ \\

    \begingroup 
    \renewcommand{\arraystretch}{1.2}
    \begin{tabular}{@{}r@{}}
        matrix inversion \\
        \cite{gilyen2018QSingValTransf}, \cite{gilyen2018QInsLowRankHHL}, \S\ref{sec:matrix-inversion}
    \end{tabular}
    \endgroup
    & $\displaystyle \frac{\|A\|_\fr}{\sigma}$
    & $\displaystyle \frac{\|A\|_\fr^6\|A\|^{16}k^6}{\sigma^{22}\eps^6}$\makebox[0pt]{\quad${\vphantom{\Big)}}^{(\diamondsuit)}$}
    & $\displaystyle \frac{\|A\|_\fr^6\|A\|^{22}}{\sigma^{28}\eps^6}$ \\

    \begingroup 
    \renewcommand{\arraystretch}{1.2}
    \begin{tabular}{@{}r@{}}
        support vector machines \\
        \cite{rebentrost2014QSVM}, \cite{ding2019SVM}, \S\ref{sec:SVM}
    \end{tabular}
    \endgroup
    & $\displaystyle \frac{1}{\lambda^3\eps^3}$\makebox[0pt]{\quad${\vphantom{\Big)}}^{(\diamondsuit)}$}
    & $\displaystyle \poly\Big(\frac1\lambda, \frac1\eps\Big)$\makebox[0pt]{\quad${\vphantom{\Big)}}^{(\clubsuit)}$}
    & $\displaystyle \frac{1}{\lambda^{28}\eps^{6}}$ \\

    \begingroup 
    \renewcommand{\arraystretch}{1.2}
    \begin{tabular}{@{}r@{}}
        Hamiltonian simulation \\
        \cite{gilyen2018QSingValTransf}, \S\ref{sec:Hamiltonian-simulation}
    \end{tabular}
    \endgroup
    & $\displaystyle \|H\|_\fr$
    &
    & $\displaystyle \frac{\|H\|_\fr^6\|H\|^{16}}{\max(1,\sigma^{16})\eps^6}$ \\


    \begingroup 
    \renewcommand{\arraystretch}{1.2}
    \begin{tabular}{@{}r@{}}
        \kern-4mm semidefinite program solving \\
        \cite{apeldoorn2018ImprovedQSDPSolving}, \cite{chia2019QInspiredSubLinLowRankSDPSolver}, \S\ref{sec:SDP}
    \end{tabular}
    \endgroup
    & $\displaystyle \frac{\|A^{(\cdot)}\|_\fr^7}{\eps^{7.5}} + \frac{\sqrt{m}\|A^{(\cdot)}\|_\fr^2}{\eps^4}$\kern-1mm
    & $\displaystyle \frac{mk^{57}}{\eps^{92}}$\makebox[0pt]{\quad${\vphantom{\Big)}}^{(\diamondsuit)}$}
    & $\displaystyle \frac{\|A^{(\cdot)}\|_\fr^{22}}{\eps^{46}}+ \frac{m\|A^{(\cdot)}\|_\fr^{14}}{\eps^{28}}$ \\

    \begingroup 
    \renewcommand{\arraystretch}{1.2}
    \begin{tabular}{@{}r@{}}
        discriminant analysis \\
        \cite{cong2016quantum}, \S\ref{sec:discriminant-analysis}
    \end{tabular}
    \endgroup
    & $\displaystyle \frac{\|B\|_\fr^7}{\eps^3\sigma^7} + \frac{\|W\|_\fr^7}{\eps^3\sigma^7}$\makebox[0pt]{\quad${\vphantom{\Big)}}^{(\diamondsuit)}$}
    &
    &  \kern-4mm $\displaystyle \frac{\|B\|_\fr^6\|B\|^4}{\eps^6\sigma^{10}} + \frac{\|W\|_\fr^6\|W\|^{10}}{\eps^6\sigma^{16}}$ \kern-4mm
\end{tabular}\end{center}
\caption[Summary of time complexities for our dequantizations]{
The time complexity for our algorithms, the quantum algorithms they are based on, and prior quantum-inspired algorithms (where they exist).
We assume our sampling and query accesses to the input takes $\bigO{1}$ time.
There are data structures that can support such queries (\cref{rmk:when-sq-access}), and if the input is in QRAM, the runtime only increases by at most a factor of log of input size.

We list the runtime of the algorithm, not including the time it takes to access the output (denoted with $\sqrun$).
The runtimes as listed ignore polylog terms, particularly those in error parameters ($\eps$ and $\delta$) and dimension parameters ($m$ and $n$).
The matrices and vectors referenced in these runtimes are always the input, $\sigma$ refers to a singular value threshold of the input matrices, $\lambda$ refers to an eigenvalue threshold (which can be thought of here as $\sigma^2$), and $\eta > \eps$ is a (dimensionless) gap parameter.

$(\clubsuit)$ indicates that the error analyses of the corresponding results are incomplete; we list the runtime they achieve for completeness.

$(\diamondsuit)$ indicates that the corresponding results only hold in the restricted setting where the input matrices are strictly rank $k$.
For the quantum algorithms with this tag, they allow for general matrices, but only have an informal error analysis arguing that singular values outside the range considered don't affect the final result.}
\label{fig:results}
\end{figure}

We use the results above to recover existing quantum-inspired algorithms for recommendation systems \cite{tang2018QuantumInspiredRecommSys}, principal component analysis \cite{tang2018QInspiredClassAlgPCA}, supervised clustering \cite{tang2018QInspiredClassAlgPCA}, support vector machines \cite{ding2019SVM}, low-rank matrix inversion \cite{chia2020QuantInsLinEqSolving}, and semidefinite program solving \cite{chia2019QInspiredSubLinLowRankSDPSolver}.
We also propose new quantum-inspired algorithms for low-rank Hamiltonian simulation and discriminant analysis (dequantizing the quantum algorithm of Cong \& Duan~\cite{cong2016quantum}).

\cref{fig:results} has a summary of our results, along with a comparison of runtimes to the corresponding quantum algorithms and prior quantum-inspired work, where it exists.
All our results match or improve on prior dequantized algorithms apart from that for matrix inversion, where prior work gives an incomparable runtime that only holds for strictly low-rank matrices of rank $k$.
Our results for matrix inversion and semidefinite program solving solve the problem in greater generality than prior work, without the restriction that the input matrices are strictly rank-$k$.\footnote{For semidefinite program solving, $\|A^{(\cdot)}\|_\fr \leq \sqrt{k}$, which makes the runtimes comparable.}

We do \emph{not} claim any meaningful breakthroughs for these problems in the classical literature: the problems that these QML algorithms solve differ substantially from their usual classical counterparts.
For example, the quantum recommendation systems algorithm of Kerenidis and Prakash~\cite{kerenidis2016QRecSys} performs sampling from a low-rank approximation of the input instead of low-rank matrix completion, which is the typical formalization of the recommendation systems problem~\cite{tang2018QuantumInspiredRecommSys}.
Evaluating these quantum algorithms' justifications for their versions of problems is outside the scope of this work: instead, we argue that these algorithms would likely not give exponential speedups when implemented, regardless of whether such implementations would be useful.
The goal of our framework is to demonstrate what can be done classically and establish a classical frontier for quantum algorithms to push past.

The proofs for these dequantization results follow the same general structure: consider the quantum algorithm and formulate the problem that this algorithm solves, and in particular, the linear algebra expression that the quantum algorithm computes.
From there, repeatedly use the SVT result and key lemma to approximate this expression by something like an RUR decomposition.
Finally, use closure properties to gain oversampling and query access to that output decomposition.
This procedure is relatively straightfoward and flexible.
Also, unlike previous work \cite{chia2020QuantInsLinEqSolving,chia2019QInspiredSubLinLowRankSDPSolver}, our results need not assume that the input is strictly low-rank.
Instead, following~\cite{tang2018QuantumInspiredRecommSys,gilyen2018QSingValTransf}, our algorithms work on close-to-low-rank matrices by doing SVTs that smoothly threshold to effectively only operate on large-enough singular values.

\subsection{Related work}

\paragraph{Quantum-inspired algorithms.}
Our approach and analysis is much simpler than that of Frieze, Kannan, and Vempala~\cite{frieze2004FastMonteCarloLowRankApx}, while it also gives improved results in our applications, and has several other advantages. For example, the reduction to \cite{frieze2004FastMonteCarloLowRankApx} first given by Tang to get an SVT-based low-rank approximation bound from the standard notion of low-rank approximation \cite[Theorem~4.7]{tang2018QuantumInspiredRecommSys} induces a quadratic loss in precision, which appears to be only an artifact of the analysis.
Also, \cite{frieze2004FastMonteCarloLowRankApx} gives Frobenius norm error bounds, though for applications we often only need spectral norm bounds; our main theorem can get improved runtimes by taking advantage of the weaker spectral norm bounds.
Finally, we take a reduced number of rows compared to columns, whereas \cite{frieze2004FastMonteCarloLowRankApx} approximates the input by taking the same number of rows and columns.

\paragraph{Randomized numerical linear algebra.}
All of the results presented here are more or less randomized linear algebra algorithms \cite{mahoney2011randomized,w14}.
The kind of sampling we get from sampling and query access is called \emph{importance sampling} or \emph{length-square sampling} in that body of work: see the survey by Kannan and Vempala~\cite{kannan2017RandAlgNumLinAlg} for more on importance sampling.
Importance sampling, and specifically, its approximate matrix product property, is the core primitive of this work.
In addition to the low-rank approximation algorithms~\cite{frieze2004FastMonteCarloLowRankApx} used in the quantum-inspired literature, others have used importance sampling for, e.g., orthogonal tensor decomposition~\cite{drineas2007randomized,mahoney2008tensor,swz16} (generalizing low-rank approximation \cite{frieze2004FastMonteCarloLowRankApx}) and support vector machines \cite{hks11}.

The fundamental difference between quantum-inspired algorithms and traditional sketching algorithms is that we assume ``we can perform quantum measurements'' of states corresponding to input in time independent of input dimension (that is, we have efficient sampling and query access to input), and in exchange want algorithms that run in time independent of dimension and provide only (over)sampling and query access to the output.
This quantum-inspired model is weaker than the standard sketching algorithm model (\cref{rmk:when-sq-access}): an algorithm taking $T$ time in the quantum-inspired model for an input matrix $A$ can be converted to a standard algorithm that runs in time $\bigO{\nnz(A) + T}$, where $\nnz(A)$ is the number of nonzero entries of $A$.
So, we can also think about an $\bigO{T}$-time quantum-inspired algorithm as an $\bigO{\nnz(A) + T}$-time sketching algorithm, where the $\nnz(A)$ portion of the runtime can \emph{only} be used to facilitate importance sampling.\footnote{The same holds for quantum algorithms using the QRAM data structure input model: the data structure itself can be built during an $\bigO{\nnz(A)}$-time (classical) preprocessing phase.}
This restriction makes for algorithms that may perform worse in generic sketching settings, but work in more settings, and so demonstrate lack of exponential quantum speedup for a wider range of problems.

A natural question is whether more modern sketching techniques can be used in our model.
After all, importance sampling is only one of many sketching techniques studied in the large literature on sketching algorithms.
Notably, though, \emph{other types of sketches seem to fail in the input regimes where quantum machine learning succeeds}: assuming sampling and query access to input, importance sampling takes time independent of dimension, whereas other randomized linear algebra methods such as Count-Sketch and Johnson-Lindenstrauss still take time linear in input-sparsity.

Subsequent work by Chepurko, Clarkson, Horesh, Lin, and Woodruff~\cite{chepurko2020quantum} notes that importance sampling oversamples leverage score sampling, so usual analyses for leverage score sampling also hold for importance sampling, up to some small overhead.
It is reasonable to suspect that this connection could lead to significant improvements over the results presented here.
However, exploiting this connection for improved runtimes seems nontrivial, since most approaches using leverage score sampling requires performing $O(\nnz(A))$-time pre-processing operations, even if one has $\sq(A)$.
As a simple example, for low-rank approximation of an input matrix $A$, if we wish to adapt the algorithm of Clarkson and Woodruff~\cite{cw17,w14}, importance sampling of $A$ can replace CountSketch for one of the sketches \cite[Lemma 4.2]{w14}, but importance sampling of $SA$ cannot replace the other \cite[Theorem 4.3]{w14}.
Versions of importance sampling may work for the second sketch (for example, sampling from a low-rank approximation of $SA$), but we are aware of none that can be obtained easily from $\sq(A)$.
This may be a manifestation of the difficulty of achieving relative-error estimates to quantities like leverage scores and residuals in this model.

An alternative approach is to use a projection-cost preserving sketch (PCP) like \emph{ridge} leverage score sampling to sketch $A$ on both sides~\cite{cmm17}.
The importance sampling from $\sq(A)$ $\lambda/2$-oversamples ridge leverage score sampling, where $\lambda \coloneqq \|A\|_\fr^2/\|A - A_k\|_\fr^2 \geq \|A\|_\fr^2/\sigma_{k+1}^2$, so using importance sampling in place of ridge leverage score sampling can give algorithms.
This is how \cite{chepurko2020quantum} gets their algorithms for low-rank sampling (recommendation systems) and quantum-inspired linear regression.
This does appear to be a promising approach, with possibility to extend to dequantizing all of QSVT, but to the authors' knowledge, it is still an open question how to improve the algorithms presented in this work, with the exception of \cite{chepurko2020quantum} improving over our recommendation systems algorithm.
Their use of PCPs significantly improves the runtime for this low-rank sampling task down below $\frac{\|A\|_\fr^6}{\sigma^6\eps^6}$, but getting a similarly good runtime for all functions seems nontrivial.
For example, their linear regression algorithm requires that the input matrix is strictly rank-$k$ (or is regularized).

The quantum-like closure properties of importance sampling shown here may be useful in the context of classical sketching algorithms.
This insight unlocks surprising power in importance sampling.
For example, it reveals that Frieze, Kannan, and Vempala's low-rank approximation algorithm~ \cite{frieze2004FastMonteCarloLowRankApx}, which, as stated, requires $\bigO{kmn}$ time to output the desired matrix, actually can produce useful results (samples and entries) in time independent of input dimension.
To use the language in \cite{frieze2004FastMonteCarloLowRankApx}, if Assumptions 1 and 2 hold for the input matrix, they also hold for the output matrix!

\paragraph{Classical algorithms for quantum problems.}
We are aware of two important prior results from before Tang's first paper~\cite{tang2018QuantumInspiredRecommSys} that connect quantum algorithms to randomized numerical linear algebra.
The first is Van den Nest's work on using probabilistic methods for quantum simulation \cite{vanDenNest2011SimulatingQCompwProbMeth}, which defines a notion of ``computationally tractable'' (CT) state equivalent to our notion of sampling and query access and then uses it to simulate restricted classes of quantum circuits.
We share some essential ideas with this work, such as the simple sampling lemmas \cref{lemma:sample-Mv,lemma:inner-prod}, but also differ greatly since we focus on low-rank matrices relevant for QML, whereas \cite{vanDenNest2011SimulatingQCompwProbMeth} focuses on simulating potentially large quantum circuits that correspond to high-rank matrices.
The second is a paper by Rudi, Wossnig, Ciliberto, Rocchetto, Pontil, and Severini~\cite{rudi2018nystrom} that uses the Nystr{\"o}m method to simulate a sparse Hamiltonian $H$ on a sparse input state in time poly-logarithmic in dimension and polynomial in $\|H\|_\fr$, assuming sampling and query access to $H$.
Our Hamiltonian simulation results do not require a sparsity assumption and still achieve a dimension-independent runtime, but get slightly larger exponents in exchange.

\paragraph{Practical implementation.}
A work by Arrazola, Delgado, Bardhan, and Lloyd~\cite{arrazola2019QInspiredInPractice} implements and benchmarks quantum-inspired algorithms for regression and recommendation systems.
The aforementioned paper of Chepurko, Clarkson, Horesh, Lin, and Woodruff~\cite{chepurko2020quantum} does the same for the quantum-inspired algorithms they introduce.
The former work makes various conclusions, including that the $\eps^2$ scaling in the number of rows/columns taken in our recommendation systems algorithm is inherent and that the quantum-inspired algorithms performed slower and worse than direct computation for practical datasets.
The latter work finds that their algorithms perform faster than direct algorithms, with an accompanying increase in error comparable to that of other sketching algorithms~\cite{DKW18}.
This improvement appears to come from both a better-performing implementation as well as an algorithm with better asymptotic runtime.
Nevertheless, it is difficult to draw definitive conclusions about the practicality of quantum-inspired algorithms as a whole from these experimental results.
Since quantum-inspired algorithms are a restricted, weaker form of computation than classical randomized numerical linear algebra algorithms (see the comparison made above), it seems possible that they perform worse than standard sketching algorithms, despite seemingly having exponentially improved runtime in theory.

Modern sketching algorithms use similar techniques to quantum-inspired algorithms, but are more natural to run on a classical computer and are likely to be faster.
For example, Dahiya, Konomis, and Woodruff~\cite{DKW18} conducted an empirical study of sketching algorithms for low-rank approximation on both synthetic datasets and the movielens dataset, reporting that their implementation ``finds a solution with cost at most 10 times the optimal one \ldots but does so 10 times faster.''
Sketching algorithms like those in \cite{DKW18} may become a relevant point of reference for benchmarking quantum linear algebra, when the implementation of these quantum algorithms on actual quantum hardware becomes possible.
In a sense, our work shows using asymptotic runtime bounds that in many scenarios sketching and sampling techniques give similar computational power to quantum linear algebra, which is a counterintuitive point since the former typically leads to linear runtimes and the latter leads to poly-logarithmic ones.

\paragraph{Quantum machine learning.}
Our work has major implications for the landscape of quantum machine learning.
Since we have presented many dequantized versions of QML algorithms, the question remains of what QML algorithms don't have such versions.
In other words, what algorithms still have the potential to give exponential speedups?

There are two general paradigms for employing quantum linear algebra techniques in quantum machine learning: the low-rank approach and the high-rank approach.
In both, we need to turn classical input data vectors to quantum states (via what's called an amplitude encoding), perform linear algebra operations on those vectors, and extract information about the output via sampling~\cite{aaronson2015caveat}.
Since preparing generic quantum states require a number of quantum gates that is proportional to the dimension of the vectors, we need some state preparation assumptions (like having QRAM with an appropriate data structure, etc.,~cf.~\cref{rmk:when-sq-access}) in order to achieve sublinear runtimes.
The main difference between the two paradigms is how the input matrices are given: in the low-rank approach, they are also given in QRAM, and so must be rescaled to have Frobenius norm one.
QSVT with block-encodings coming from the QRAM data structure\footnote{The QRAM data structure used here has alternatives which in some sense vary the ``norm'' in which one stores the input matrix \cite[Theorem~IV.4]{kerenidis2017QGradDesc}, \cite[Lemma~25]{chakraborty2018BlockMatrixPowers}. We do not expect that these alternatives would give sampling and query access to the matrix they store or could be otherwise dequantized, since these datastructures generalize and strengthen the sparse-access input model, which is known to be BQP-complete~\cite{harrow2009QLinSysSolver}.} or density operators is an example framework in this vein.
The restrictions of this block-encoding method means that the rank needs to be small, but assuming the hardware needed for QRAM can be realized, this is still a flexible setting, one that QML researchers find interesting.
In the high-rank approach, which is used in the HHL algorithm~\cite{harrow2009QLinSysSolver} and its derivatives, the matrix needs to be represented by a concise quantum circuit and have a small (poly-logarithmic in input dimension) condition number in order to gain an exponential speedup over classical algorithms.
This doesn't happen in typical datasets.
The collection of these demanding requirements hamstrings most attempts to find applications of HHL~\cite{harrow2009QLinSysSolver} with the potential for practical super-polynomial speedups.

Our results give evidence for the lack of exponential speedup for the former, low-rank approach.
It is important to note however, that our results do not rule out the possibility of large polynomial quantum speedups.
In order to assess the potential usefulness of QML algorithms in this regime it is important to improve on classical upper and lower bounds for these problems, which we leave as an open question.
On the other hand, high-rank block-encodings, such as those coming from sparsity assumptions in the original HHL algorithm~\cite{harrow2009QLinSysSolver}, remain impervious to our techniques.
This suggests that the most promising way to get exponential quantum speedups for QML algorithms is by assuming sparse matrices as input, or utilizing other efficiently implementable high-rank quantum operation such as the Quantum Fourier Transform.

Works from Zhao, Fitzsimons, and Fitzsimons on Gaussian process regression~\cite{zhao2015QAssisstedGaussProcRegr}; from Lloyd, Garnerone, and Zanardi on topological data analysis~\cite{lloyd2016topological,gcd20}; and from Yamasaki, Subramanian, Sonoda, and Koashi~\cite{yamasaki2020OptRandFeat} on learning random features attempt to address these issues to get a super-polynomial quantum speedup. 
Though these works avoid the dequantization barrier to large quantum speedups, it remains to be seen how broad will be their impact on QML and whether these speedups manifest for data seen in practice.

\paragraph{Related independent work.}
Independently from our work, Jethwani, Le Gall, and Singh simultaneously derived similar results~\cite{jethwani2019QInsClassAlgSVT}.
They implicitly derive a version of our even SVT result, and use it to achieve generic SVT (approximate $\sq(b^\dagger f^\svt(A))$ for a vector $b$) by writing $f^\svt(A) = A g(A^\dagger A)$ for $g(x) = f(\sqrt{x})/\sqrt{x}$ and then using sampling subroutines to get the solution from the resulting expression $b^\dagger A R^\dagger U R$.
It is difficult to directly compare the main SVT results, because the parameters that appear in their runtime bounds are somewhat non-standard, but one can see that for typical choices of~$f$, their results require a strictly low-rank~$A$.
In comparison our results apply to general~$A$, and we also demonstrate how to apply them to (re)derive dequantized algorithms.

\subsection{Open questions}
Our framework recovers recent dequantization results, and we hope that it will be used for dequantizing more quantum algorithms. In the meantime, our work leaves several natural open questions:
\begin{enumerate}[label=(\alph*)]
    \item Is there an approach to QML that does not go through HHL (whose demanding assumptions make exponential speedups difficult to demonstrate even in theory) or a low-rank assumption (which, as we demonstrate, makes the tasks ``easy'' for classical computers) and yields a provable superpolynomial speedup for a practically relevant ML problem?
    \item Our algorithms still have significant slowdown as compared to their quantum counterparts.
    Can we shave condition number factors to get runtimes of the form $\bOt{\frac{\|A\|_\fr^6}{\sigma^6\eps^6}\log^3\frac{1}{\delta}}$ (for the recommendation systems application, for instance), without introducing additional assumptions?
    Can we get even better runtimes by somehow avoiding SVD computation?
    \item Do the matrix arithmetic closure properties we showed for $\ell_2$-norm importance sampling hold for other kinds of sampling and sketching distributions, like leverage score or $\ell_p$-norm sampling?
    \item In the quantum setting, linear algebra algorithms~\cite{gilyen2018QSingValTransf} can achieve logarithmic dependence on the precision $\eps$. Can classical algorithms also achieve such exponentially improved dependence, when the goal is restricted to sampling from the output (i.e., without the requirement to query elements of the output)? If not, is there a mildly stronger classical model that can achieve this? Can one prove that this exponential advantage for sampling problems cannot be conferred to estimation/decision problems?
\end{enumerate}

\subsection{Organization}

The paper proceeds as follows.
\cref{subsec:sq-oracles} introduces the notion of (over)sampling and query access and some of its closure properties.
\cref{subsec:sketch} gives the fundamental idea of using sampling and query access to sketch matrices used for the approximation results in \cref{subsec:tool_box} and singular value transformation results in \cref{subsec:main-svt}.
These results form the framework that is used to dequantize QSVT in \cref{subsec:dequant} and recover all the quantum-inspired results in \cref{sec:applications}.
These applications of our framework contain various tricks and patterns that we consider to be ``best practice'' for coercing problems into our framework, since they have given us the best complexities and generality.
More general results of SVT are shown in \cref{sec:cur}.


\section{Preliminaries}\label{sec:prelim}
To begin with, we define notation to be used throughout this paper.
For $n \in \N$, $[n] \coloneqq \{1,\ldots,n\}$.
For $z \in \bbc$, its absolute value is $\abs{z} = \sqrt{z^*z}$, where $z^*$ is the complex conjugate of $z$.
$f \lesssim g$ denotes the ordering $f = \bigO{g}$ (and respectively for $\gtrsim$ and $\eqsim$).
$\bOt{g}$ is shorthand for $\bigO{g\poly(\log g)}$.
$\log$ refers to the natural logarithm.
Finally, we assume that arithmetic operations (e.g., addition and multiplication of real numbers) and function evaluation oracles (computing $f(x)$ from $x$) take unit time, and that queries to oracles (like the queries to input discussed in \cref{subsec:sq-oracles}) are at least unit time cost.

\subsection{Linear algebra}

In this paper, we consider complex matrices $A\in \mathbb{C}^{m\times n}$ for $m,n \in \N$.
For $i \in [m], j \in [n]$, we let $A(i,\cdot)$ denote the $i$-th row of $A$, $A(\cdot,j)$ denote the $j$-th column of $A$, and $A(i,j)$ denote the $(i,j)$-th element of $A$.
$(A \mid B)$ denotes the concatenation of matrices $A$ and $B$ and $\vvec(A) \in \mathbb{C}^{mn}$ denotes the vector formed by concatenating the rows of $A$.
For vectors $v \in \mathbb{C}^n$, $\|v\|$ denotes standard Euclidean norm (so $\|v\| \coloneqq (\sum_{i=1}^n \abs{v(i)}^2)^{1/2}$).
For a matrix $A \in \mathbb{C}^{m\times n}$, the {\em Frobenius norm} of $A$ is $\|A\|_\fr \coloneqq \nrm{\vvec(A)}= (\sum_{i=1}^m\sum_{j=1}^{n} \abs{A(i,j)}^2)^{1/2}$ and the {\em spectral norm} of $A$ is $\|A\| \coloneqq \|A\|_\op \coloneqq \sup_{x\in \bbc^{n}, \|x\|=1} \|Ax\|$.
We say that $U$ is an isometry if $\|Ux\| = \|x\|$ for all $x$, or equivalently, if $U$ is a subset of columns of a unitary.

A \emph{singular value decomposition} (SVD) of $A$ is a representation $A = UDV^{\dag}$, where for $N \coloneqq \min(m,n)$, $U\in \bbc^{m\times N}$ and $V\in \bbc^{n\times N}$ are isometries and $D\in \bbr^{N\times N}$ is diagonal with $ \sigma_i\coloneqq D(i,i)$ and $\sigma_1 \geq \sigma_2 \geq \cdots \geq \sigma_N \geq 0$.
We can also write this decomposition as $A = \sum_{i=1}^{N} \sigma_i u_i v_i^\dagger$, where $u_i \coloneqq U(\cdot,i)$ and $v_i \coloneqq V(\cdot, i)$.
For Hermitian $A$, an \emph{(unitary) eigendecomposition} of $A$ is a singular value decomposition where $U = V$, except the entries of $D$ are allowed to be negative.

Using SVD, we can define the rank-$k$ approximation of $A$ to be $A_k \coloneqq \sum_{i=1}^{k} \sigma_i u_iv_i^\dagger$ and the pseudoinverse of $A$ to be $A^+ \coloneqq \sum_{i=1}^{\rank(A)} \frac{1}{\sigma_i}  v_iu_i^\dagger$.
We now formally define singular value transformation:

\begin{definition} \label{def:svt}
For a function $f\colon[0,\infty)\to \bbc$ such that $f(0)=0$ and a matrix $A \in \bbc^{m\times n}$, we define the \emph{singular value transform} of $A$ via a singular value decomposition $A=\sum_{i = 1}^{\min(m,n)} \sigma_i u_i v_i^\dagger$:
\begin{equation}\label{eqn:svt}
  f^\svt(A)\coloneqq \sum_{i = 1}^{\min(m,n)} f(\sigma_i) u_i v_i^\dagger.
\end{equation}
\end{definition}

\noindent The requirement that $f(0)=0$ ensures that the definition is independent of the (not necessarily unique) choice of SVD.

\begin{definition}
For a function $f: \bbr \to \bbc$ and a Hermitian matrix $A \in \bbc^{n\times n}$, we define the eigenvalue transform of $A$ via a \emph{unitary eigendecomposition} $A = \sum_{i=1}^n \lambda_i v_iv_i^\dagger$:
\begin{equation}\label{eqn:evt}
  f^\evt(A)\coloneqq \sum_{i=1}^n f(\lambda_i) v_iv_i^\dagger.
\end{equation}
\end{definition}
Since we only consider eigenvalue transformations of Hermitian matrices, where singular vectors/values and eigenvectors/values (roughly) coincide, the key difference between singular value transformation and eigenvalue transformation is that the latter can distinguish eigenvalue sign.
As eigenvalue transformation is the standard notion of a matrix function, we will usually drop the superscript in notation: $f(A) \coloneqq f^\evt(A)$.

We will use the following standard definition of a Lipschitz function.
\begin{definition}
  We say $f\colon \bbr \to \bbc$ is \emph{$L$-Lipschitz} on $\mathfrak{F}\subseteq\bbr$ if for all $x, y \in \mathfrak{F}$, $\abs{f(x) - f(y)} \leq L\abs{x - y}$.
\end{definition}
\noindent We define approximate isometry as follows:\footnote{This is the notion of approximate orthonormality as given by the first arXiv version of \cite{tang2018QuantumInspiredRecommSys}.}
\begin{definition}
  \label{defn:approx-isometry}
  Let $m,n\in \mathbb{N}$ and $m\geq n$. A matrix $V \in \bbc^{m \times n}$ is an \emph{$\alpha$-approximate isometry} if $\norm{V^{\dag}V - I} \leq \alpha$.
  It is an \emph{$\alpha$-approximate projective isometry} if $\|V^\dagger V - \Pi\| \leq \alpha$ for $\Pi$ an orthogonal projector.
\end{definition}

If $V$ is an $\alpha$-approximate isometry, among other things, it implies that $\abs{\|V\|^2 - 1} \leq \alpha$ and that there exists an isometry $U \in \bbc^{m \times n}$ with $\operatorname{im}(U) = \operatorname{im}(V)$ such that $\norm{U - V} \leq \alpha$.
We show this and other basic facts in the following lemma, whose proof is deferred to \cref{apx:proofs}.

\begin{restatable}{lemma}{apporthfacts} \label{lem:apporth-facts}
  If $\hat{X} \in \bbc^{m\times n}$ is an $\alpha$-approximate isometry, then there is an exact isometry $X \in \bbc^{m\times n}$ with the same columnspace as $\hat{X}$ such that $\|\hat{X} - X\| \leq \alpha$.
  Furthermore, for any matrix $Y \in \bbc^{n\times n}$,
  \begin{align*}
    \|\hat{X}Y\hat{X}^\dagger - XYX^\dagger\| \leq (2\alpha + \alpha^2)\|Y\|.
  \end{align*}
  If $\alpha < 1$, then $\|\hat{X}^+\| \leq (1-\alpha)^{-1}$ and
  \begin{align*}
    \|\hat{X}Y\hat{X}^\dagger - XYX^\dagger\| \leq \alpha\frac{2 - \alpha}{(1-\alpha)^2}\|\hat{X}Y\hat{X}^\dagger\|.
  \end{align*}
\end{restatable}


\section{Sampling and query access oracles} \label{subsec:sq-oracles}

Since we want our algorithms to run in time sublinear in input size, we must carefully define our access model.
The sampling and query oracle we present below is unconventional, being designed as a reasonable classical analogue for the input model of some quantum algorithms.
It will also be used heavily to move between intermediate steps of these quantum-inspired algorithms.
First, as a warmup, we define a simple query oracle:

\begin{definition}[Query access]
  \label{defn:q-access}
  For a vector $v \in \bbc^n$, we have $\q(v)$, \emph{query access} to $v$, if for all $i \in [n]$, we can query for $v(i)$.
  Likewise, for a matrix $A \in \bbc^{m\times n}$, we have $\q(A)$ if for all $(i,j) \in [m] \times [n]$, we can query for $A(i, j)$.
  Let $\qcb(v)$ (respectively $\qcb(A)$) denote the (time) cost of such a query.
\end{definition}

For example, in the typical RAM access model, we are given our input $v \in \bbc^n$ as $\q(v)$ with $\qcb(v)=1$.
For brevity, we will sometimes abuse this notation (and other access notations) and, for example, abbreviate ``$\q(A)$ for $A \in \bbc^{m\times n}$'' as ``$\q(A) \in \bbc^{m\times n}$''.
We will also sometimes abuse complexity notation like $\qcb$ to refer to known bounds on the complexity, instead of the complexity itself.

\begin{definition}[Sampling and query access to a vector]
  \label{defn:sq-access}
  For a vector $v \in \bbc^n$, we have $\sq(v)$, \emph{sampling and query access} to $v$, if we can:
  \begin{enumerate}
    \item query for entries of $v$ as in $\q(v)$;
    \item obtain independent samples $i \in [n]$ following the distribution $\mathcal{D}_v \in \bbr^n$, where $\mathcal{D}_v(i) \coloneqq |v(i)|^2/\norm{v}^2$;
    \item query for $\|v\|$.
  \end{enumerate}
  Let $\qcb(v)$, $\scb(v)$, and $\ncb(v)$ denote the cost of querying entries, sampling indices, and querying the norm respectively.
  Further define $\sqcb(v)\coloneqq \max(\qcb(v), \scb(v), \ncb(v))$.
\end{definition}

We will refer to these samples as \emph{importance samples from $v$}, though one can view them as measurements of the quantum state $|v\rangle \coloneqq \frac{1}{\|v\|}\sum v_i|i\rangle$ in the computational basis.

Quantum-inspired algorithms typically don't give exact sampling and query access to the output vector.
Instead, we get a more general version of sampling and query access, which assumes we can only access a sampling distribution that \emph{oversamples} the correct distribution.\footnote{Oversampling turns out to be the ``natural'' form of approximation in this setting; other forms of error do not propagate through quantum-inspired algorithms well.}

\begin{definition}
For $p, q \in \bbr_{\geq 0}^n$ that are distributions, meaning $\sum_i p(i) = \sum_i q(i) = 1$, we say that $p$ \emph{$\phi$-oversamples} $q$ if, for all $i \in [n]$, $p(i) \geq q(i) / \phi$.
\end{definition}

The motivation for this definition is the following: if $p$ $\phi$-oversamples $q$, then we can convert a sample from $p$ to a sample from $q$ with probability $1/\phi$ using rejection sampling: sample an $i$ distributed as $p$, then accept the sample with probability $q(i)/(\phi p(i))$ (which is $\leq 1$ by definition).

\begin{definition}[Oversampling and query access]
  \label{defn:phi-sq-access}
  For $v \in \bbc^n$ and $\phi \geq 1$, we have $\sq_\phi(v)$, \emph{$\phi$-oversampling and query access} to $v$, if we have $\q(v)$ and $\sq(\tilde{v})$ for $\tilde{v} \in \bbc^n$ a vector satisfying $\|\tilde{v}\|^2 = \phi\|v\|^2$ and $\abs{\tilde{v}(i)}^2 \geq \abs{v(i)}^2$ for all $i \in [n]$.
  Denote $\scb_\phi(v) \coloneqq \scb(\tilde{v})$, $\pcb_\phi(v) \coloneqq \qcb(\tilde{v})$, $\ncb_\phi(v) \coloneqq \ncb(\tilde{v})$, and $\sqcb_\phi(v) \coloneqq \max(\scb_\phi(v), \pcb_\phi(v), \qcb(v), \ncb_\phi(v))$.
\end{definition}

\noindent The distribution $\mathcal{D}_{\tilde{v}}$ $\phi$-oversamples $\mathcal{D}_{v}$, since for all $i \in [n]$,
\begin{align*}
  \mathcal{D}_{\tilde{v}}(i) = \frac{\abs{\tilde{v}_i}^2}{\|\tilde{v}\|^2} = \frac{\abs{\tilde{v}_i}^2}{\phi\|v\|^2} \geq \frac{\abs{v_i}^2}{\phi\|v\|^2} = \frac1\phi \mathcal{D}_v(i).
\end{align*}
For this reason, we call $\mathcal{D}_{\tilde{v}}$ a \emph{$\phi$-oversampled importance sampling distribution} of $v$.
$\sq(v)$ is the same as $\sq_1(v)$, by taking $\tilde{v} = v$.
Note that we do not assume knowledge of $\phi$ (though it can be estimated, (though it can be estimated as shown in \cref{lem:b-sq-approx}).
However, we do need to know $\|\tilde{v}\|$ (even if $\nrm{v}$ is known), as it cannot be deduced from a small number of queries, samples, or probability computations.
So, we will be choosing $\tilde{v}$ (and, correspondingly, $\phi$) such that $\|\tilde{v}\|^2$ remains computable, even if potentially some $c\tilde{v}$ satisfies all our other requirements for some $c < 1$ (giving a smaller value of $\phi$).

Intuitively speaking, estimators that use $\mathcal{D}_v$ can also use $\mathcal{D}_{\tilde{v}}$ via rejection sampling at the expense of a factor $\phi$ increase in the number of utilized samples.
From this observation we can prove that oversampling access implies an approximate version of the usual sampling access:

\begin{restatable}{lemma}{oversampling} \label{lem:b-sq-approx}
  Suppose we are given $\sq_\phi(v)$ and some $\delta \in (0,1]$.
  Denote $\sqrun(v) \coloneqq \phi\sqcb_\phi(v)\log\frac{1}{\delta}$.
  We can sample from $\mathcal{D}_v$ with probability $\geq 1-\delta$ in $\bigO{\sqrun(v)}$ time.
  We can also estimate $\|v\|$ to $\nu$ multiplicative error for $\nu \in (0,1]$ with probability $\geq 1-\delta$ in $\bigO{\frac{1}{\nu^2}\sqrun(v)}$ time.
\end{restatable}
\begin{proof}
  Consider the following rejection sampling algorithm to generate samples: sample an index $i$ from $\tilde{v}$, and output it as the desired sample with probability $r(i) \coloneqq \frac{\abs{v(i)}^2}{\abs{\tilde{v}(i)}^2}$.
  Otherwise, restart.
  We can perform this: we can compute $r(i)$ in $\bigO{\sqcb_\phi(v)}$ time and $r(i) \leq 1$ since $\tilde{v}$ bounds $v$.

  The probability of accepting a sample in a round is $\sum_i \mathcal{D}_{\tilde{v}}(i)r(i) = \|v\|^2/\|\tilde{v}\|^2 = \phi^{-1}$ and, conditioned on a sample being accepted, the probability of it being $i$ is $\abs{v(i)}^2/\|v\|^2$, so the output distribution is $\mathcal{D}_v$ as desired.
  So, to get a sample with $\geq 1-\delta$ probability, run rejection sampling for at most $2\phi\log\frac{1}{\delta}$ rounds.

  To estimate $\|v\|^2$, notice that we know $\|\tilde{v}\|^2$, so it suffices to estimate $\|v\|^2/\|\tilde{v}\|^2$ which is $\phi^{-1}$.
  The probability of accepting the rejection sampling routine is $\phi^{-1}$, so we run $3\nu^{-2}\phi\log\frac2\delta$ rounds of it for estimating $ \phi^{-1}$.
  Let $Z$ denote the fraction of them which end in acceptance.
  Then, by a Chernoff bound we have
  \begin{align*}
    \Pr[\abs{Z - \phi^{-1}} \geq \nu\phi^{-1}] \leq 2\exp\Big(-\frac{\nu^2z\phi^{-1}}{2+\nu} \Big) \leq \delta,
  \end{align*}
  so $Z\|\tilde{v}\|^2$ is a good multiplicative approximation to $\|v\|^2$ with probability $\geq 1-\delta$.
\end{proof}

Generally, compared to a quantum algorithm that can output (and measure) a desired vector $|v\rangle$, our algorithms will output $\sq_\phi(u)$ such that $\|u - v\|$ is small.
So, $\sqrun(u)$ is the relevant complexity measure that we will analyze and bound: if we wish to mimic samples from the output of the quantum algorithm we dequantize, we will pay a one-time cost to run our quantum-inspired algorithm for ``obtaining'' $\sq_\phi(u)$, and then pay $\sqrun(u)$ cost per additional measurement.
As for error, bounds on $\|u - v\|$ imply that measurements from $u$ and $v$ follow distributions that are close in total variation distance {\cite[Lemma~4.1]{tang2018QuantumInspiredRecommSys}}.
Now, we show that oversampling and query access of vectors is closed under taking small linear combinations.

\begin{restatable}[Linear combinations, Proposition~4.3 of \cite{tang2018QuantumInspiredRecommSys}]{lemma}{thinmatvec} \label{lemma:sample-Mv}
  Given $\sq_{\varphi_t}(v_t) \in \bbc^n$ and $\lambda_t \in \bbc$ for all $t \in [\tau]$, we have $\sq_\phi(\sum_{t=1}^{\tau} \lambda_tv_t)$ for $\phi = \tau\frac{\sum \varphi_t\|\lambda_tv_t\|^2}{\|\sum \lambda_tv_t\|^2}$ and $\sqcb_\phi(\sum \lambda_tv_t) = {\displaystyle\max_{t \in [\tau]}}\scb_{\varphi_t}(v_t) +  \sum_{t=1}^\tau \qcb(v_t)$ (after paying $\bigO{\sum_{t=1}^\tau \ncb_{\varphi_t}(v_t)}$ one-time pre-processing cost to query for norms).
\end{restatable}
\begin{proof}
  Denote $u \coloneqq \sum \lambda_tv_t$.
  To compute $u(s)$ for some $s \in [n]$, we just need to query $v_t(s)$ for all $t \in [\tau]$, paying $\bigO{\sum \qcb(v_t)}$ cost.
  So, it suffices to get $\sq(\tilde{u})$ for an appropriate bound $\tilde{u}$.
  We choose
  \[
    \textstyle\tilde{u}(s) = \sqrt{\tau\sum_{t=1}^\tau |\lambda_t\tilde{v}_t(s)|^2},
  \]
  so that $|\tilde{u}(s)| \geq |u(s)|$ by Cauchy--Schwarz, and $\|\tilde{u}\|^2 = \tau\sum_{t=1}^\tau \|\lambda_t\tilde{v}_t\|^2 = \tau\sum_{t=1}^\tau \varphi_t\|\lambda_tv_t\|^2$, giving the desired value of $\phi$.

  We have $\sq(\tilde{u})$: we can compute $\|\tilde{u}\|^2$ by querying for all norms $\|\tilde{v}_t\|$, compute $\tilde{u}(s)$ by querying $\tilde{v}_t(s)$ for all $t \in [\tau]$.
  We can sample from $\tilde{u}$ by first sampling $t \in [\tau]$ with probability $\frac{\|\lambda_t\tilde{v}_t\|^2}{\sum_\ell \|\lambda_\ell \tilde{v}_\ell\|^2}$, and then taking our sample to be $j \in [n]$ from $\tilde{v}_t$.
  The probability of sampling $j \in [n]$ is correct:
  \begin{equation*}
    \sum_{t=1}^\tau \frac{\|\lambda_t\tilde{v}_t\|^2}{\sum_\ell \|\lambda_\ell \tilde{v}_\ell\|^2} \frac{\abs{\tilde{v}_t(j)}^2}{\|\tilde{v}_t\|^2}
    = \frac{\sum_{t=1}^\tau \abs{\lambda_t\tilde{v}_t(j)}^2}{\sum_{\ell=1}^\tau \|\lambda_\ell \tilde{v}_\ell\|^2}
    = \frac{|\tilde{u}(j)|^2}{\|\tilde{u}\|^2}.
  \end{equation*}
  If we pre-process by querying all the norms $\|\tilde{v}_\ell\|$ in advance, we can sample from the distribution over $i$'s in $\bigO{1}$ time, using an alias sampling data structure for the distribution (\cref{rmk:when-sq-access}), and we can sample from $\tilde{v}_t$ using our assumed access to it, $\sq_{\varphi_t}(v_t)$.
\end{proof}

So, our general goal will be to express our output vector as a linear combination of a small number of input vectors that we have sampling and query access to.
Then, we can get an approximate $\sq$ access to our output using \cref{lem:b-sq-approx}, where we pay an additional ``cancellation constant'' factor of $\phi = \tau\frac{\sum \varphi_t\|\lambda_tv_t\|^2}{\|\sum \lambda_tv_t\|^2}$.
This factor is only large when the linear combination has significantly smaller norm than the components $v_t$ in the sum suggest.
Usually, in our applications, we can intuitively think about this overhead being small when the desired output vector mostly lies in a subspace spanned by singular vectors with large singular values in our low-rank input.
Quantum algorithms also have the same kind of overhead.
Namely, the QSVT framework encodes this in the subnormalization constant $\alpha$ of block-encodings, and the overhead from the subnormalization appears during post-selection~\cite{gilyen2018QSingValTransf}.
When this cancellation is not too large, the resulting overhead typically does not affect too badly the runtime of our applications.

We also define oversampling and query access for a matrix.
The same model (under an alternative definition) is also discussed in prior work \cite{frieze2004FastMonteCarloLowRankApx,DKR02} and is the right notion for the sampling procedures we will use.
\begin{definition}[Oversampling and query access to a matrix]
  \label{defn:sampling-A}
For a matrix $A \in \bbc^{m\times n}$, we have $\sq(A)$ if we have $\sq(A(i,\cdot))$ for all $i \in [m]$ and $\sq(a)$ for $a \in \bbr^m$ the vector of row norms ($a(i)\!\coloneqq \!\|A(i,\cdot)\|$).

We have $\sq_\phi(A)$ if we have $\q(A)$ and $\sq(\tilde{A})$ for $\tilde{A} \in \bbc^{m\times n}$ satisfying $\|\tilde{A}\|_\fr^2 = \phi\|A\|_\fr^2$ and $\abs{\tilde{A}(i,j)}^2 \geq \abs{A(i,j)}^2$ for all $(i,j) \in [m]\times[n]$.

The complexity of (over)sampling and querying from the matrix $A$ is denoted by $\scb_\phi(A) \coloneqq \max(\scb(\tilde{A}(i,\cdot)), \scb(\tilde{a}))$, $\pcb_\phi(A) \coloneqq \max(\pcb(\tilde{A}(i,\cdot)),\pcb(\tilde{a}))$, $\qcb(A) \coloneqq \max(\pcb(A(i,\cdot)))$, and $\ncb_\phi(A)\coloneqq \ncb(\tilde{a})$ respectively.
We also denote $\sqcb_\phi(A) \coloneqq \max(\scb_\phi(A), \pcb_\phi(A),\qcb(A), \ncb_\phi(A))$.
We omit subscripts if $\phi = 1$.
\end{definition}
Observe that access to a matrix, $\sq_\phi(A)$, implies access to its vectorized version, $\sq_\phi(\vvec(A))$: we can take $\widetilde{\vvec(A)} = \vvec(\tilde{A})$, and the distribution for $\vvec(\tilde{A})$ is sampled by sampling $i$ from $\mathcal{D}_{\tilde{a}}$, and then sampling $j$ from $\mathcal{D}_{\tilde{A}(i,\cdot)}$.
This gives the output $(i,j)$ with probability $\abs{\tilde{A}(i,j)}^2/\|\tilde{A}\|_\fr^2$.
Therefore, one can think of $\sq_\phi(A)$ as $\sq_\phi(\vvec(A))$, with the addition of having access to samples $(i,j)$ from $\vvec(A)$, conditioned on fixing a particular row $i$ and also knowing the probabilities of these conditional samples.

Now we prove that oversampling and query access is closed under taking outer products.
The same idea also extends to taking Kronecker products of matrices.

\begin{restatable}{lemma}{outersampling} \label{lem:outersampling}
  Given vectors $\sq_{\varphi_u}(u)\in\bbC^m$ and $\sq_{\varphi_v}(v)\in\bbC^n$, we have $\sq_\phi(A)$ for their outer product $A \coloneqq u v^\dagger$ with $\phi = \varphi_u \varphi_v$ and $\scb_\phi(A) = \scb_{\varphi_u}(u)+ \scb_{\varphi_v}(v)$, $\pcb_\phi(A) = \pcb_{\varphi_u}(u)+ \pcb_{\varphi_v}(v)$, $\qcb(A) = \qcb(u)+ \qcb(v)$, and $\ncb_\phi(A) = \ncb_{\varphi_u}(u)+ \ncb_{\varphi_v}(v)$,
\end{restatable}
\begin{proof}
  We can query an entry $A(i,j) = u(i)v(j)^\dagger$ by querying once from $u$ and $v$.
  Our choice of upper bound is $\tilde{A} = \tilde{u}\tilde{v}^\dagger$.
  Clearly, this is an upper bound on $uv^\dagger$ and $\|\tilde{A}\|_\fr^2 = \|\tilde{u}\|^2\|\tilde{v}\|^2 = \varphi_u\varphi_v\|A\|_\fr^2$.
  We have $\sq(\tilde{A})$ in the following manner: $\tilde{A}(i,\cdot) = \tilde{u}(i)\tilde{v}^\dagger$, so we have $\sq(\tilde{A}(i,\cdot))$ from $\sq(\tilde{v})$ after querying for $\tilde{u}(i)$, and $\tilde{a} = \|\tilde{v}\|^2\tilde{u}$, so we have $\sq(\tilde{a})$ from $\sq(\tilde{u})$ after querying for $\|\tilde{v}\|$.
\end{proof}

Using the same ideas as in \cref{lemma:sample-Mv}, we can extend sampling and query access of input matrices to linear combinations of those matrices.
\begin{restatable}{lemma}{wosampling} \label{lem:weighted-oversampling}
  Given $\sq_{\varphi^{(t)}}(A^{(t)}) \in \bbc^{m\times n}$ and $\lambda_t \in \bbc$ for all $t \in [\tau]$, we have $\sq_\phi(A) \in \bbc^{m\times n}$ for $A \coloneqq \sum_{t=1}^\tau \lambda_t A^{(t)}$ with $\phi = \tau\frac{\sum_{t=1}^\tau {\varphi^{(t)}}\|\lambda_t A^{(t)}\|_\fr^2}{\|A\|_\fr^2}$ and $\scb_\phi(A) = {\displaystyle\max_{t \in [\tau]}}\scb_{\varphi^{(t)}}(A^{(t)})+\sum_{t=1}^\tau \pcb_{\varphi^{(t)}}(A^{(t)})$, $\pcb_\phi(A) = \sum_{t=1}^\tau \pcb_{\varphi^{(t)}}(A^{(t)})$, $\qcb(A) = \sum_{t=1}^\tau \qcb(A^{(t)})$, and $\ncb_\phi(A)=1$ (after paying $\bigO{\sum_{t=1}^\tau \ncb_{\varphi^{(t)}}(A^{(t)})}$ one-time pre-processing cost).
\end{restatable}
\begin{proof}
  To compute $A(i,j) = \sum_{t=1}^\tau \lambda_tA^{(t)}(i,j)$ for $(i,j) \in [m]\times[n]$, we just need to query $A^{(t)}(i,j)$ for all $t \in [\tau]$, paying $\bigO{\sum_t \qcb(A^{(t)})}$ cost.
  So, it suffices to get $\sq(\tilde{A})$ for an appropriate bound $\tilde{A}$.
  We choose
  \[
    \textstyle\tilde{A}(i,j) = \sqrt{\tau\sum_{t=1}^\tau |\lambda_t\tilde{A}^{(t)}(i,j)|^2}.
  \]
  That $|\tilde{A}(i,j)| \geq |A(i,j)|$ follows from Cauchy--Schwarz, and we get the desired value of $\phi$:
  \[
    \|\tilde{A}\|_\fr^2 = \tau\sum_{t=1}^\tau \|\lambda_i\tilde{A}^{(t)}\|_\fr^2 = \tau\sum_{t=1}^\tau \varphi^{(t)}\|\lambda_iA^{(t)}\|_\fr^2.
  \]

  We have $\sq(\tilde{A})$: we can compute $\|\tilde{A}\|_\fr$ by querying for all norms $\|\tilde{A}^{(t)}\|_\fr$, compute $\tilde{a}(i) = \|\tilde{A}(i,\cdot)\| = \sqrt{\tau\sum_{t=1}^\tau \|\lambda_t\tilde{A}^{(t)}(i,\cdot)\|^2}$ by querying $\tilde{a}^{(t)}(i)$ for all $t \in [\tau]$, and compute $\tilde{A}(i,j)$ by querying $\tilde{A}^{(t)}(i,j)$ for all $t \in [\tau]$.
  Analogously to \cref{lemma:sample-Mv}, we can sample from $\tilde{a}$ by first sampling $s \in [\tau]$ with probability $\frac{\|\lambda_s\tilde{A}^{(s)}\|_\fr^2}{\sum_t \|\lambda_t\tilde{A}^{(t)}\|_\fr^2}$, then taking our sample to be $i \in [m]$ from $\mathcal{D}_{\tilde{a}^{(s)}}$.
  If we pre-process by querying all the Frobenius norms $\|\tilde{A}^{(t)}\|_\fr$ in advance, we can sample from $\tilde{a}$ in $\bigO{\max_{t \in [\tau]}\scb_{\varphi^{(t)}}(A^{(t)})}$ time.
  We can sample from $\tilde{A}(i,\cdot)$ by first sampling $s \in [\tau]$ with probability $\frac{\|\lambda_s\tilde{A}^{(s)}(i,\cdot)\|^2}{\sum_t \|\lambda_t\tilde{A}^{(t)}(i,\cdot)\|^2}$, then taking our sample to be $j \in [n]$ from $\mathcal{D}_{\tilde{A}^{(s)}(i,\cdot)}$.
  This takes $\bigO{\sum_{t=1}^\tau \pcb_{\varphi^{(t)}}(A^{(t)})+\max_{t \in [\tau]}\scb_{\varphi^{(t)}}(A^{(t)})}$ time.
\end{proof}

\begin{remark}
With the lemmas we've introduced, we can already get oversampling and query access to some modest expressions.
For example, consider RUR decompositions, which show up frequently in our results: suppose we have $\sq(A)$ for $A \in \bbc^{m\times n}$, $R \in \bbc^{r\times n}$ a (possibly normalized) subset of rows of $A$, and a matrix $U \in \bbc^{r\times r}$.
Then
\[
  R^\dagger U R = \sum_{i=1}^r \sum_{j=1}^r U(i,j) R(i,\cdot)^\dagger R(j,\cdot),
\]
which is a linear combination of $r^2$ outer products involving rows of $A$.
So, by \cref{lem:outersampling} and \cref{lem:weighted-oversampling}, we have $\sq_\phi(R^\dagger U R)$.
\end{remark}

\noindent For us, the most interesting scenario is when our sampling and query oracles take poly-logarithmic time, since this corresponds to the scenarios where quantum state preparation procedures can run in time $\polylog(n)$.
In these scenarios, quantum machine learning have the potential to achieve exponential speedups.
We can provide such  classical access in various ways.

\definecolor{fig2}{HTML}{0074d9}
\definecolor{fig1}{HTML}{d9003a}
\begin{figure*}[ht]
 \centering
 \begin{scriptsize}
 \[ \begin{tikzcd}[column sep=-10mm]
 & & & & & & & {\color{fig1} \nrm{a}^2 = \nrm{A}_F^2} \ar[fig1,dllll] \ar[fig1,drrrr] & & & & & & & \\[-3mm]
 & & & {\color{fig1}\abs{a_1}^2} = {\color{fig2}\nrm{A(1,\cdot)}^2} \ar[fig2,dll] \ar[fig2,drr] & & & & & & & & {\color{fig1}\abs{a_2}^2} = {\color{fig2}\nrm{A(2,\cdot)}^2} \ar[fig2,dll] \ar[fig2,drr] & & & \\
 & {\color{fig2}\abs{A(1,1)}^2+\abs{A(1,2)}^2} \ar[fig2,dl] \ar[fig2,dr] & & & & {\color{fig2}\abs{A(1,3)}^2 +\abs{A(1,4)}^2} \ar[fig2,dl] \ar[fig2,dr] & & & & {\color{fig2}\abs{A(2,1)}^2+\abs{A(2,2)}^2} \ar[fig2,dl] \ar[fig2,dr] & & & & {\color{fig2}\abs{A(2,3)}^2+\abs{A(2,4)}^2} \ar[fig2,dl] \ar[fig2,dr] &  \\
 {\color{fig2}\kern-2mm\abs{A(1,1)}^2} \ar[fig2,d] & & {\color{fig2}\abs{A(1,2)}^2}\ar[fig2,d] & & {\color{fig2}\abs{A(1,3)}^2} \ar[fig2,d] & & {\color{fig2}\abs{A(1,4)}^2} \ar[fig2,d] & &  {\color{fig2}\abs{A(2,1)}^2} \ar[fig2,d] & & {\color{fig2}\abs{A(2,2)}^2}\ar[fig2,d] & & {\color{fig2}\abs{A(2,3)}^2} \ar[fig2,d] & & {\color{fig2}\abs{A(2,4)}^2\kern-2.5mm} \ar[fig2,d]\\
 {\color{fig2} \frac{A(1,1)}{\abs{A(1,1)}}} & & {\color{fig2} \frac{A(1,2)}{\abs{A(1,2)}}} & & {\color{fig2} \frac{A(1,3)}{\abs{A(1,3)}}} & & {\color{fig2} \frac{A(1,4)}{\abs{A(1,4)}}} & & {\color{fig2} \frac{A(2,1)}{\abs{A(2,1)}}} & & {\color{fig2} \frac{A(2,2)}{\abs{A(2,2)}}} & & {\color{fig2} \frac{A(2,3)}{\abs{A(2,3)}}} & & {\color{fig2} \frac{A(2,4)}{\abs{A(2,4)}}}
 \end{tikzcd} \]
 \end{scriptsize}
 \caption{Dynamic data structure for a matrix $A \in \bbC^{2\times 4}$ discussed in \cref{rmk:when-sq-access} part (b). We compose {\color{fig1} the data structure for $a$} with {\color{fig2} the data structure for $A$'s rows}.}
 \label{fig:ds}
\end{figure*}

\begin{remark} \label{rmk:when-sq-access}
Below, we list settings where we have sampling and query access to input matrices and vectors, and whenever relevant, we compare the resulting runtimes to the time to prepare analogous quantum states.
Note that because we do not analyze classical algorithms in the bit model, i.e., we do not count each operation bitwise, their runtimes may be missing log factors that should be counted for a fair comparison between classical and quantum.
\begin{enumerate}[label=(\alph*)]
  \item (Data structure)
  Given $v \in \bbc^n$ in the standard RAM model, the alias method \cite{Vose1991} takes $\Theta(n)$ pre-processing time to output a data structure that uses $\Theta(n)$ space and can sample from $v$ in $\Theta(1)$ time.
  In other words, we can get $\sq(v)$ with $\sqcb(v) = \Theta(1)$ in $\bigO{n}$ time, and by extension, for a matrix $A \in \bbc^{m\times n}$, $\sq(A)$ with $\sqcb(A) = \Theta(1)$ in $\bigO{mn}$ time.

  If the input vector (resp.\ matrix) is given as a list of $\nnz(v)$ (resp.\ $\nnz(A)$) of its non-zero entries, then the pre-processing time is linear in that number of entries.
  Therefore, the quantum-inspired setting can be directly translated to a basic randomized numerical linear algebra algorithm.
  More precisely, with this data structure, a fast quantum-inspired algorithm (say, one running in time $\bigO{T\sqcb(A)}$ for $T$ independent of input size) implies an algorithm in the standard computational model (running in $\bigO{\nnz(A) + T}$ time).

  \item (Dynamic data structure)
  QML algorithms often assume that their input is in a data structure with a certain kind of quantum access~\cite{prakash2014QLinAlgAndMLThesis,kerenidis2017QGradDesc,giovannetti2007QuantumRAM,wossnig2018QLinSysAlgForDensMat,rebentrost2016QGradDesc,chakraborty2018BlockMatrixPowers}.
  They argue that, since this data structure allows for circuits preparing input states with linear gate count but polylog \emph{depth}, hardware called QRAM might be able to parallelize these circuits enough so that they run in effectively polylog \emph{time}.
  In the interest of considering the best of all possible worlds for QML, we will treat circuit depth as runtime for QRAM and ignore technicalities.

  This data structure (see \cref{fig:ds}) admits sampling and query access to the data it stores with just-as-good runtimes: specifically, for a matrix $A \in \bbc^{m\times n}$, we get $\sq(A)$ with $\qcb(A) = \bigO{1}$, $\scb(A) = \bigO{\log mn}$, and $\ncb(A) = \bigO{1}$.
  So, quantum-inspired algorithms can be used whenever QML algorithms assume this form of input.

  Further, unlike the alias method stated above, this data structure supports updating entries in $\bigO{\log mn}$ time, which is used in applications of QML where data accumulates over time \cite{kerenidis2016QRecSys}.

  \item (Integrability assumption)
  For $v \in \bbc^n$, suppose we can compute entries $v(i)$ and sums $\sum_{i \in I(b)} \abs{v(i)}^2$ in time $T$, where $I(b) \subset [n]$ is the set of indices whose binary representation begins with the bitstring~$b$.
  Then we have $\sq(v)$ where $\qcb(v) = \bigO{T}, \scb(v) = \bigO{T\log n}$, and $\ncb(v) = \bigO{T}$.
  Analogously, the quantum state that encodes $v$ in its amplitudes, $|v\rangle = \sum_i \frac{v_i}{\|v\|}|i\rangle$, can be prepared in time $\bigO{T\log n}$ via Grover-Rudolph state preparation \cite{grover2002SuperposEffIntegrProbDistr}.
  (One can think about the QRAM data structure as pre-computing all the necessary sums for this protocol.)

  \item (Uniformity assumption)
  Given $\bigO{1}$-time $\q(v) \in \bbc^n$ and a $\beta$ such that $\max \abs{v(i)}^2 \leq \beta/n$, we have $\sq_{\phi}(v)$ with $\phi = \beta/\|v\|^2$ and $\sqcb_{\phi}(v) = \bigO{1}$, by using the vector whose entries are all $\sqrt{\beta/n}$ as the upper bound $\tilde{v}$.
  Assuming the ability to query entries of $v$ in superposition, a quantum state corresponding to $v$ can be prepared in time $\mathcal{O}\big(\sqrt{\phi}\log n\big)$.

  \item (Sparsity assumption)
  If $A \in \bbc^{m\times n}$ has at most $s$ non-zero entries per row (with efficiently computable locations) and the matrix elements are $\abs{A(i, j)} \leq c$ (and efficiently computable), then we have $\sq_{\phi}(A)$ for $\phi = c^2\frac{sm}{\|A\|_\fr^2}$, simply by using the uniform distribution over non-zero entries for the oversampling and query oracles.
  For example, for $\sq(\tilde{a})$ we can set $\tilde{a}(i) \coloneqq c\sqrt{s}$, and for $\tilde{A}(i,\cdot)$ we use the vector with entries $c$ at the non-zeros of $A(i,\cdot)$ (potentially adding some ``dummy'' zero locations to have exactly $s$ non-zeroes).

  Note that similar sparse-access assumptions are often seen in the QML and Hamiltonian simulation literature~\cite{harrow2009QLinSysSolver}.
  Also, if $A$ is not much smaller than we expect, then $\phi$ can be independent of dimension.
  For example, if $A$ has exactly $s$ non-zero entries per row and $\abs{A(i,j)} \geq c'$ for non-zero entries, then $\phi \leq (c/c')^2$.

  \item (CT states) In 2009, Van den Nest defined the notion of a ``computationally tractable'' (CT) state \cite{vanDenNest2011SimulatingQCompwProbMeth}.
  Using our notation, $|\psi\rangle \in \bbc^n$ is a CT state if we have $\sq(\psi)$ with $\sqcb(\psi) = \polylog(n)$.
  Van den Nest's paper identifies several classes of CT states, including product states, quantum Fourier transforms of product states, matrix product states of polynomial bond dimension, stabilizer states, and states from matchgate circuits.
  For more details on how can one get efficient sampling and query access to such vectors we direct the reader to~\cite{vanDenNest2011SimulatingQCompwProbMeth}.
\end{enumerate}
\end{remark}


\section{Matrix sketches} \label{subsec:sketch}

We now introduce the workhorse of our algorithms: the matrix sketch.
Using sampling and query access, we can generate these sketches efficiently, and these allow one to reduce the dimensionality of a problem, up to some approximation.
Most of the results presented in this section are known in the classical sketching literature: we present them here for completeness, and to restate them in the context of sampling and query access.

\begin{definition}
For a distribution $p \in \bbr^m$, we say that a matrix $S \in \bbr^{s\times m}$ is \emph{sampled according to $p$} if each row of $S$ is independently chosen to be $e_i / \sqrt{s\cdot p(i)}$ with probability $p(i)$, where $e_i$ is the vector that is one in the $i$th position and zero elsewhere.
If $p$ is an $\ell_2$-norm sampling distribution $\mathcal{D}_v$ as defined in \cref{defn:sq-access}, then we also say $S$ is \emph{sampled according to $v$}.

We call $S$ an \emph{importance sampling sketch} for $A \in \bbc^{m\times n}$ if it is sampled according to $A$'s row norms $a$, and we call $S$ a \emph{$\phi$-oversampled importance sampling sketch} if it is sampled according to the bounding row norms from $\sq_\phi(A)$, $\tilde{a}$ (or, more generally, from a $\phi$-oversampled importance sampling distribution of $a$).
\end{definition}
\noindent One should think of $S$ as a description of how to sketch $A$ down to $SA$.
The following lemma shows that $\|SA\|_\fr$ approximates $\|A\|_\fr$, giving a simple example of the phenomenon that $SA$ approximates $A$ in certain senses: it shows that $\|SA\|_\fr = \Theta(\|A\|_\fr)$ with probability $\geq 0.9$ when $S$ has $\Omega(\frac{1}{\phi^2})$ rows.
We show later (\cref{lem:AMP-spectral}) that a similar statement holds for spectral norm: $\|SA\| = \Theta(\|A\|)$ with probability $\geq 0.9$ when $S$ has $\tilde{\Omega}(\phi^2\|A\|_\fr^2/\|A\|^2)$ rows.

\begin{lemma}[Frobenius norm bounds for matrix sketches] \label{lem:sa-frob-norm}
  Let $S \in \bbc^{r\times m}$ be a $\phi$-oversampled importance sampling sketch of $A \in \bbc^{m\times n}$.
  Then $\|[SA](i,\cdot)\| \leq \sqrt{\phi/r}\|A\|_\fr$ for all $i \in [r]$, so $\|SA\|_\fr^2 \leq \phi\|A\|_\fr^2$ (unconditionally).
  Equality holds when $\phi = 1$.
  Further,
  \begin{equation*}
    \Pr\Big[\abs{\|SA\|_\fr^2 - \|A\|_\fr^2} \geq \sqrt{\frac{\phi^2\ln(2/\delta)}{2r}}\|A\|_\fr^2\Big] \leq \delta. \qedhere
  \end{equation*}
\end{lemma}
\begin{proof}
  Let $p$ be the distribution used to create $S$, and let $s_i$ be the sample from $p$ used for row $i$ of $S$.
  Then $\|SA\|_\fr^2$ is the sum of the row norms $\|[SA](i,\cdot)\|^2$ over all $i \in [r]$, and
  \begin{gather*}
    \|[SA](i,\cdot)\|^2 = \frac{\|A(s_i,\cdot)\|^2}{r\cdot p(s_i)} \leq \frac\phi r\|A\|_\fr^2 \\
    \E[\|[SA](i,\cdot)\|^2] = \sum_{s=1}^m p(s)\frac{\|A(s,\cdot)\|^2}{r\cdot p(s)} = \frac1r \|A\|_\fr^2
  \end{gather*}
  The first equation shows the unconditional bounds on $\|SA\|_\fr$.
  When $\phi = 1$, $p(i) = \|A(i,\cdot)\|^2/\|A\|_\fr^2$ so the inequality becomes an equality.
  By the second equation, $\|SA\|_\fr^2 - \|A\|_\fr^2$ has expected value zero and is the sum of independent random variables bounded in $[-\|A\|_\fr^2, (\phi - 1)\|A\|_\fr^2]$, so the probabilistic bound follows immediately from Hoeffding's inequality.
\end{proof}

In the standard algorithm setting, computing an importance sampling sketch requires reading all of $A$, since we need to sample from $\mathcal{D}_a$.
If we have $\sq_\phi(A)$, though, we can efficiently create a $\phi$-oversampling sketch $S$ in $\bigO{s(\scb_\phi(A) + \pcb_\phi(A)) + \ncb_\phi(A)}$ time: for each row of $S$, we pull a sample from $p$, and then compute $\sqrt{p(i)}$.
After finding this sketch $S$, we have an implicit description of $SA$: it is a normalized multiset of rows of $A$, so we can describe it with the row indices and corresponding normalization, $(i_1,c_1),\ldots,(i_s,c_s)$.

$SA$ can be used to approximate matrix expressions involving $A$.
Further, we can chain sketches using the lemma below, which shows that from $\sq_\phi(A)$, we have $\sq_{\leq 2\phi}((SA)^\dagger)$, under a mild assumption on the size of the sketch $S$.
This can be used to find a sketch $T^\dagger$ of $(SA)^\dagger$.
The resulting expression $SAT$ is small enough that we can compute functions of it in time independent of dimension, and so will be used extensively.
When we discuss sketching $A$ down to $SAT$, we are referring to the below lemma for the method of sampling $T$.

\begin{restatable}{lemma}{sqsketching} \label{lem:sq-sketching}
  Consider $\sq_\varphi(A) \in \bbc^{m\times n}$ and $S \in \bbr^{r\times m}$ sampled according to $\tilde{a}$, described as pairs $(i_1,c_1),\ldots,(i_r,c_r)$.
  If $r \geq 2\varphi^2\ln\frac{2}{\delta}$, then with probability $\geq 1-\delta$, we have $\sq_{\phi}(SA)$ and $\sq_{\phi}((SA)^\dagger)$ for some $\phi$ satisfying $\phi \leq 2\varphi$.
  If $\varphi = 1$, then for all $r$, we have $\sq(SA)$ and $\sq((SA)^\dagger)$.

  The runtimes for $\sq_\phi(SA)$ are $\qcb(SA) = \qcb(A)$, $\scb_\phi(SA) = \scb_\varphi(A)$, $\pcb_\phi(SA) = \pcb_\varphi(A)$, and $\ncb_\phi(SA)=\bigO{1}$, after $\bigO{\ncb_\varphi(A)}$ pre-processing cost.
  The runtimes for $\sq_\phi((SA)^\dagger)$ are $\qcb((SA)^\dagger) = \qcb(A)$, $\scb_\phi((SA)^\dagger) = \scb_\varphi(A)+r\pcb_\varphi(A)$, $\pcb_\phi((SA)^\dagger) = r\pcb_\varphi(A)$, and $\ncb_\phi((SA)^\dagger)=\ncb_\varphi(A)$.
\end{restatable}
\begin{proof}
By \cref{lem:sa-frob-norm}, $\|SA\|_\fr^2 \geq \|A\|_\fr^2 / 2$ with probability $\geq 1-\delta$.
Suppose this bound holds.
To get $\sq_{\phi}(SA)$, we take $\widetilde{SA} = S\tilde{A}$, which bounds $SA$ by inspection.
Further, $\|S\tilde{A}\|_\fr^2 = \|\tilde{A}\|_\fr^2$ by \cref{lem:sa-frob-norm}, so $\phi = \|S\tilde{A}\|_\fr^2/\|SA\|_\fr^2 = \varphi\|A\|_\fr^2/\|SA\|_\fr^2 \leq 2\varphi$.
Analogously, $(S\tilde{A})^\dagger$ works as a bound for $\sq_\phi((SA)^\dagger)$.
We can query an entry of $SA$ by querying the corresponding entry of $A$, so all that suffices is to show that we have $\sq(S\tilde{A})$ and $\sq((S\tilde{A})^\dagger)$ from $\sq(\tilde{A})$.
(When $\varphi = 1$, we can ignore the above argument: the rest of the proof will show that we have $\sq(SA)$ and $\sq((SA)^\dagger)$ from $\sq(A)$.)

We have $\sq(S\tilde{A})$.
Because the rows of $S\tilde{A}$ are rescaled rows of $\tilde{A}$, we have $\sq$ access to them from $\sq$ access to $\tilde{A}$.
Because $\|S\tilde{A}\|_\fr^2 = \|\tilde{A}\|_\fr^2$ and $\|[S\tilde{A}](i,\cdot)\|^2 = \|\tilde{A}\|_\fr^2/r$, after precomputing $\|\tilde{A}\|_\fr^2$, we have $\sq$ access to the vector of row norms of $S\tilde{A}$ (pulling samples simply by pulling samples from the uniform distribution).

We have $\sq((S\tilde{A})^\dagger)$.
(This proof is similar to one from \cite{frieze2004FastMonteCarloLowRankApx}.)
Since the rows of $(S\tilde{A})^\dagger$ are length $r$, we can respond to $\sq$ queries to them by reading all entries of the row and performing some linear-time computation.
$\|(S\tilde{A})^\dagger\|_\fr^2 = \|\tilde{A}\|_\fr^2$, so we can respond to a norm query by querying the norm of $\tilde{A}$.
Finally, we can sample according to the row norms of $(S\tilde{A})^\dagger$ by first querying an index $i \in [r]$ uniformly at random, then outputting the index $j \in [n]$ sampled from $[S\tilde{A}](i,\cdot)$ (which we can sample from because it is a row of $\tilde{A}$).
The distribution of the samples output by this procedure is correct: the probability of outputting $j$ is
\[
  \frac{1}{r}\sum_{i=1}^r \frac{|[S\tilde{A}](i,j)|^2}{\|[S\tilde{A}](i,\cdot)\|^2}
  = \sum_{i=1}^r \frac{|[S\tilde{A}](i,j)|^2}{\|S\tilde{A}\|_\fr^2}
  = \frac{\|[S\tilde{A}](\cdot,j)\|^2}{\|S\tilde{A}\|_\fr^2}. \qedhere
\]
\end{proof}

\subsection{Approximation results}\label{subsec:tool_box}

Here, we present approximation results on sketched matrices that we will use heavily throughout our results.
We begin with a fundamental observation: given sampling and query access to a matrix $A$, we can approximate the matrix product $A^\dagger B$ by a sum of rank-one outer products.
We formalize this with two variance bounds, which we can use together with Chebyshev's inequality.

\begin{restatable}[Asymmetric matrix multiplication to Frobenius norm error, {\cite[Lemma~4]{DKM06}}]{lemma}{mmhp} \label{lem:mmhp}
  Consider $X \in \bbc^{m\times n}, Y \in \bbc^{m\times p}$, and take $S \in \bbr^{r\times m}$ to be sampled according to $p \in \bbr^m$ a $\phi$-oversampled importance sampling distribution from $X$ or $Y$.
  Then,
  \[
  \E[\|X^\dagger S^\dagger SY - X^\dagger Y\|_\fr^2] \leq \frac{\phi}{r} \|X\|_\fr^2\|Y\|_\fr^2
  \quad \text{and} \quad
  \E\Big[\sum_{i=1}^r \|[SX](i,\cdot)\|^2\|[SY](i,\cdot)\|^2\Big] \leq \frac{\phi}{r}\|X\|_\fr^2\|Y\|_\fr^2.
  \]
\end{restatable}
\begin{proof}
  To show the first equation, we use that $\E[\|X^\dagger X^\dagger SY - X^\dagger Y\|_\fr^2]$ is a sum of variances, one for each entry $(i,j)$, since $\E[X^\dagger S^\dagger SY - XY]$ is zero in every entry.
  Furthermore, for every entry $(i,j)$, the matrix expression is the sum of $r$ independent, mean-zero terms, one for each row of $S$:
  \begin{equation*}
    [X^\dagger S^\dagger SY - XY](i,j)
    = \sum_{s=1}^r \Big([SX](s,i)^\dagger [SY](s, j) - \frac1r [X^\dagger Y](i,j)\Big).
  \end{equation*}
  So, we can use standard properties of variances\footnote{See the proof of \cref{lem:mm} in \cref{apx:proofs} for this kind of computation done with more detail.} to conclude that
  \begin{multline*}
  \E[\|X^\dagger S^\dagger SY - X^\dagger Y\|_\fr^2]
  = r\cdot \E[\|[SX](1,\cdot)^\dagger [SY](1,\cdot) - {\textstyle\frac1r} X^\dagger Y\|_\fr^2]
   \leq r\cdot \E[\|[SX](1,\cdot)^\dagger [SY](1,\cdot)\|_\fr^2] \\
  = r \sum_{i=1}^m p(i)\frac{\|X(i,\cdot)^\dagger Y(i,\cdot)\|_\fr^2}{r^2p(i)^2}
  = \frac1r \sum_{i=1}^m \frac{\|X(i,\cdot)\|^2\|Y(i,\cdot)\|^2}{p(i)}
  \leq \frac\phi r \|X\|_\fr^2\|Y\|_\fr^2.
  \end{multline*}
  The second other inequality follows by the same computation:
  \begin{equation*}
  \E\Big[\sum_{i=1}^r \|[SX](i,\cdot)\|^2\|[SY](i,\cdot)\|^2\Big]
  = r\cdot \E[\|[SX](1,\cdot)\|^2\|[SY](1,\cdot)\|^2]
  \leq \frac{\phi}{s}\|X\|_\fr^2\|Y\|_\fr^2. \qedhere
  \end{equation*}
\end{proof}

The above result shows that, given $\sq(X)$, $X^\dagger Y$ can be approximated by a sketch with constant failure probability.
If we have $\sq(X)$ \emph{and} $\sq(Y)$, we can make the failure probability exponential small.
To show this tighter error bound, we use an argument of Drineas, Kannan, and Mahoney for approximating matrix multiplication.
We state their result in a slightly stronger form, which is actually proved in their paper.
For completeness, a proof of this statement is in the appendix.

\begin{restatable}[Matrix multiplication by subsampling  {\cite[Theorem~1]{DKM06}}]{lemma}{lemmm} \label{lem:mm}
Suppose we are given $X \in \bbc^{n\times m}, Y \in \bbc^{n\times p}, r \in \N$ and a distribution $p \in \bbr^n$ satisfying the oversampling condition that, for some $\phi \geq 1$,
\begin{align*}
  p(k) \geq \frac{\|X(k,\cdot)\|\|Y(k,\cdot)\|}{\phi\sum_\ell \|X(\ell,\cdot)\|\|Y(\ell,\cdot)\|}.
\end{align*}
Let $S \in \bbr^{r \times n}$ be sampled according to $p$.
Then $X^\dagger S^\dagger SY$ is an unbiased estimator for $X^\dagger Y$ and
\begin{align*}
  \Pr\Big[\|X^\dagger S^\dagger SY - X^\dagger Y\|_\fr < \sqrt{\frac{8\phi^2\ln(2/\delta)}{r}}\underset{\leq\|X\|_\fr\|Y\|_\fr}{\underbrace{\sum_\ell \|X(\ell,\cdot)\|\|Y(\ell,\cdot)\|}}\Big] > 1-\delta.
\end{align*}
\end{restatable}

\noindent From a simple application of \cref{lem:mm}, we get a key lemma used frequently in \cref{sec:applications}.

\begin{restatable}[Approximating matrix multiplication to Frobenius norm error; corollary of {\cite[Theorem~1]{DKM06}}]{lemma}{apprmms}
    \label{prop:appr-mms}
    Consider $X \in \bbc^{m\times n}, Y \in \bbc^{m\times p}$, and take $S \in \bbr^{r\times m}$ to be sampled according to $q \coloneqq \frac{q_1+q_2}{2}$, where $q_1, q_2 \in \bbr^m$ are $\phi_1,\phi_2$-oversampled importance sampling distributions from $x, y$, the vector of row norms for $X$, $Y$, respectively.
    Then $S$ is a $2\phi_1,2\phi_2$-oversampled importance sampling sketch of $X, Y$, respectively.
    Further,
    \[
       \Pr\Big[\|X^\dagger S^\dagger SY - X^\dagger Y\|_\fr < \sqrt{\frac{8\phi_1\phi_2\log 2/\delta}{r}}\|X\|_\fr\|Y\|_\fr\Big] > 1-\delta.
    \]
\end{restatable}
\begin{proof}
First, notice that $2q(i) \geq q_1(i)$ and $2q(i) \geq q_2(i)$, so $q$ oversamples the importance sampling distributions for $X$ and $Y$ with constants $2\phi_1$ and $2\phi_2$, respectively.
We get the bound by using \cref{lem:mm}; $q$ satisfies the oversampling condition with $\phi = \frac{\sqrt{\phi_1\phi_2}\|X\|_\fr\|Y\|_\fr}{\sum_\ell \|X(\ell,\cdot)\|\|Y(\ell,\cdot)\|}$, using the inequality of arithmetic and geometric means:
\begin{align*}
  \frac{1}{q(i)}\frac{\|X(i,\cdot)\|\|Y(i,\cdot)\|}{\sum_\ell \|X(\ell,\cdot)\|\|Y(\ell,\cdot)\|}
  &= \frac{2}{q_1(i) + q_2(i)} \frac{\|X(i,\cdot)\|\|Y(i,\cdot)\|}{\sum_\ell \|X(\ell,\cdot)\|\|Y(\ell,\cdot)\|} \\
  &\leq \frac{1}{\sqrt{q_1(i)q_2(i)}} \frac{\|X(i,\cdot)\|\|Y(i,\cdot)\|}{\sum_\ell \|X(\ell,\cdot)\|\|Y(\ell,\cdot)\|} \\
  &\leq \frac{\sqrt{\phi_1\phi_2}\|X\|_\fr\|Y\|_\fr}{\|X(i,\cdot)\|\|Y(i,\cdot)\|} \frac{\|X(i,\cdot)\|\|Y(i,\cdot)\|}{\sum_\ell \|X(\ell,\cdot)\|\|Y(\ell,\cdot)\|} \\
  &= \frac{\sqrt{\phi_1\phi_2}\|X\|_\fr\|Y\|_\fr}{\sum_\ell \|X(\ell,\cdot)\|\|Y(\ell,\cdot)\|}. \qedhere
\end{align*}
\end{proof}

\begin{remark} \label{rmk:appr-mms-sq}
    \cref{prop:appr-mms} implies that, given $\sq_{\phi_1}(X)$ and $\sq_{\phi_2}(Y)$, we can get $\sq_\phi(M)$ for $M$ a sufficiently good approximation to $X^\dagger Y$, with $\phi \leq \phi_1\phi_2\frac{\|X\|_\fr^2\|Y\|_\fr^2}{\|M\|_\fr^2}$.
    This is an approximate closure property for oversampling and query access under matrix products.

    Given the above types of accesses, we can compute the sketch $S$ necessary for \cref{prop:appr-mms} by taking $p = \mathcal{D}_{\tilde{x}}$ and $q = \mathcal{D}_{\tilde{y}}$), thereby finding a desired $M \coloneqq X^\dagger S^\dagger SY$.
    We can compute entries of $M$ with only $r$ queries each to $X$ and $Y$, so all we need is to get $\sq(\tilde{M})$ for $\tilde{M}$ the appropriate bound.
    We choose $|\tilde{M}(i,j)|^2 \coloneqq r\sum_{\ell=1}^r |[S\tilde{X}](\ell,i)^\dagger [S\tilde{Y}](\ell,j)|^2$; showing that we have $\sq(M)$ follows from the proofs of \cref{lem:outersampling,lem:weighted-oversampling}, since $M$ is simply a linear combination of outer products of rows of $\tilde{X}$ with rows of $\tilde{Y}$.
    Finally, this bound has the appropriate norm.
    Notating the rows sampled by the sketch as $s_1,\ldots,s_r$, we have
    \begin{multline*}
        \|\tilde{M}\|_\fr^2 = r\sum_{\ell=1}^r \|[S\tilde{X}](\ell,\cdot)\|^2\|[S\tilde{Y}](\ell,\cdot)\|^2
        = r\sum_{\ell=1}^r \frac{\|\tilde{X}(s_\ell,\cdot)\|^2\|\tilde{Y}(s_\ell,\cdot)\|^2}{r^2(\frac{\|\tilde{X}(s_\ell,\cdot)\|^2}{2\|\tilde{X}\|_\fr^2} + \frac{\|\tilde{Y}(s_\ell,\cdot)\|^2}{2\|\tilde{Y}\|_\fr^2})^2}\\
        \leq \sum_{\ell=1}^r \frac{\|\tilde{X}(s_\ell,\cdot)\|^2\|\tilde{Y}(s_\ell,\cdot)\|^2}{r(\frac{\|\tilde{X}(s_\ell,\cdot)\|\|\tilde{Y}(s_\ell,\cdot)\|}{\|\tilde{X}\|_\fr\|\tilde{Y}\|_\fr})^2}
        = \|\tilde{X}\|_\fr^2\|\tilde{Y}\|_\fr^2
        = \phi_1\phi_2\|X\|_\fr^2\|Y\|_\fr^2.
    \end{multline*}
\end{remark}
If $X = Y$, we can get an improved spectral norm bound: instead of depending on $\|X\|_\fr^2$, error depends on $\|X\|\|X\|_\fr$.

\begin{restatable}[Approximating matrix multiplication to spectral norm error {\cite[Theorem~3.1]{rudelson2007sampling}}]{lemma}{apprmmspectral} \label{lem:AMP-spectral}
    Suppose we are given $A \in \bbr^{m\times n}, \eps > 0, \delta \in [0,1]$, and $S \in \bbr^{r\times n}$ a $\phi$-oversampled importance sampling sketch of $A$.
    Then
    \begin{align*}
    \Pr \Big[ \| A^\dagger S^\dagger S A - A^\dagger A \| \lesssim \sqrt{\frac{\phi^2\log r\log 1/\delta}{r}}\|A\|\|A\|_\fr \Big]  > 1-\delta.
    \end{align*}
\end{restatable}

The above results can be used to approximate singular values, simply by directly translating the bounds on matrix product error to bounds on singular value error.

\begin{restatable}[Approximating singular values]{lemma}{apprsvs}
    \label{prop:appr-svs}
    Given $\sq_\phi(A) \in \bbc^{m\times n}$ and $\eps \in (0,1]$, we can form importance sampling sketches $S \in \bbr^{r\times m}$ and $T^\dagger \in \bbr^{c\times n}$ in $\bigO{(r+c)\sqcb_\phi(A)}$ time satisfying the following property.
    Take $r, c \geq s$ for some sufficiently large $s = \bOt{\frac{\phi^2}{\eps^2}\log\frac{1}{\delta}}$.
    Then, if $\sigma_i$ and $\hat{\sigma}_i$ are the singular values of $A$ and $SAT$, respectively (where $\hat{\sigma}_i = 0$ for $i > \min(r,c)$), we have with probability $\geq 1-\delta$ that
    \[
    \sqrt{\sum_{i=1}^{\min(m,n)} (\hat{\sigma}_i^2 - \sigma_i^2)^2} \leq \eps\|A\|_\fr^2.
    \]
    If we additionally assume that $\eps \lesssim \|A\|/\|A\|_\fr$, we can conclude $\abs{\sigma_i^2 - \hat{\sigma}_i^2} \leq \eps\|A\|\|A\|_\fr$ for all $i$.
\end{restatable}
This result follows from results bounding the error between singular values by errors of matrix products.
For notation, let $\sigma_i(M)$ be the $i$th largest singular value of $M$.
We will use the following inequalities relating norm error of matrices to error in their singular values:
\begin{lemma}[Hoffman-Wielandt inequality {\cite[Lemma~2.7]{kannan2017RandAlgNumLinAlg}}] \label{lem:hwineq}
For symmetric $X, Y \in \bbr^{n\times n}$,
$$
  \sum \abs{\sigma_i(X) - \sigma_i(Y)}^2 \leq \|X - Y\|_\fr^2.
$$
\end{lemma}

\begin{lemma}[Weyl's inequality {\cite[Corollary III.2.2]{bhatia1997MatrixAnalysis}}] \label{lem:weylineq}
For $A, B \in \bbc^{m\times n}$, $\abs{\sigma_k(A) - \sigma_k(B)} \leq \|A - B\|$.
When $A, B$ are Hermitian, the same bound holds for their eigenvalues.\footnote{\cite[Corollary III.2.2]{bhatia1997MatrixAnalysis} actually proves the Hermitian version. The result about singular values is an easy consequence, see for example the blog of Terence Tao \cite[Exercise 22(iv)]{tao2010notesOnHermitianEigenvalues}.}
\end{lemma}

\begin{proof}[Proof of \cref{prop:appr-svs}]
We use known theorems, plugging in the values of $r$ and $c$.
Using \cref{prop:appr-mms} for the sketch $S$, we know that
\begin{equation*}
  \Pr\Big[\|A^\dagger S^\dagger S A - A^\dagger A\|_\fr \leq \frac\eps2\|A\|_\fr^2\Big] \geq 1-\delta;
\end{equation*}
by \cref{lem:sq-sketching}, $T^\dagger$ is an $\leq 2\phi$-oversampled importance sampling sketch of $(SA)^\dagger$, so by \cref{prop:appr-mms} for $T^\dagger$,
\begin{equation*}
  \Pr\Big[\|S A T T^\dagger A^\dagger S^\dagger - S A A^\dagger S^\dagger\|_\fr \leq \frac{\eps}{4}\|SA\|_\fr^2\Big] \geq 1-\delta,
\end{equation*}
and from \cref{lem:sa-frob-norm},
\begin{equation*}
  \Pr\Big[\|S A\|_\fr^2 \leq 2 \|A\|_\fr^2\Big] \geq 1-\delta.
\end{equation*}
By rescaling $\delta$ and union bounding, we can have all events happen with probability $\geq 1-\delta$.
Then, from triangle inequality followed by \cref{lem:hwineq},
\begin{align*}
  \sqrt{\sum \abs{\sigma_i(SAT)^2 - \sigma_i(A)^2}^2} &\leq
  \sqrt{\sum \abs{\sigma_i(SAT)^2 - \sigma_i(SA)^2}^2} +
  \sqrt{\sum \abs{\sigma_i(SA)^2 - \sigma_i(A)^2}^2} \\
  &\leq \|(SAT)(SAT)^\dagger - (SA)(SA)^\dagger\|_\fr + \|(SA)^\dagger(SA) - A^\dagger A\|_\fr \\
  &\leq \eps\|A\|_\fr^2.
\end{align*}
The analogous result holds for spectral norm via \cref{lem:AMP-spectral} and \cref{lem:weylineq}; the only additional complication is that we need to assert that $\|SA\| \lesssim \|A\|$.
We use the following argument, using the upper bound on $\eps$:
\begin{equation*}
  \|SA\|^2 = \|A^\dagger S^\dagger SA\| \leq \|A^\dagger S^\dagger SA - A^\dagger A\| + \|A^\dagger A\| \leq \|A\|^2 + \eps\|A\|\|A\|_\fr \lesssim \|A\|^2. \qedhere
\end{equation*}
\end{proof}

Finally, if we wish to approximate a vector inner product $u^\dagger v$, a special case of matrix product, we can do so with only sampling and query access to one of the vectors while still getting $\log\frac{1}{\delta}$ dependence on failure probability.
The proof of this statement is in \cref{apx:proofs}.

\begin{restatable}[Inner product estimation, {\cite[Proposition~4.2]{tang2018QuantumInspiredRecommSys}}]{lemma}{ipest}
    \label{lemma:inner-prod}
    Given $\sq_\phi(u), \q(v) \in \bbc^n$, we can output an estimate $c \in \bbc$ such that $\abs{c - \langle u, v\rangle} \leq \eps$ with probability $\geq 1-\delta$ in time $\bigO{\phi\|u\|^2\|v\|^2\frac{1}{\eps^2}\log\frac{1}{\delta} (\sqcb_\phi(u)+\qcb(v))}$.
\end{restatable}

\begin{remark} \label{claim:tr_prod} \label{lemma:xAy}
    \cref{lemma:inner-prod} also applies to higher-order tensor inner products:
    \begin{enumerate}[label=(\alph*)]
        \item (Trace inner products, {\cite[Lemma~11]{gilyen2018QInsLowRankHHL}}) Given $\sq_\phi(A) \in \mathbb{C}^{n\times n}$ and $\q(B) \in \mathbb{C}^{n\times n}$, we can estimate $\Tr[AB^\dagger]$ to additive error $\eps$ with probability at least $1-\delta$ by using
        \begin{align*}
        \bigO{\phi\frac{\|A\|_\fr^2\|B\|_\fr^2}{\eps^2} \big(\sqcb_\phi(A)+\qcb(B)\big)\log\frac{1}{\delta}}
        \end{align*}
        time.
        To do this, note that $\sq_\phi(A)$ and $\q(B)$ imply $\sq_\phi(\vvec(A))$ and $\q(\vvec(B))$.
        $\Tr[AB] = \langle \vvec(B), \vvec(A)\rangle$, so we can just apply \cref{lemma:inner-prod} to conclude.

        \item (Expectation values) Given $\sq_\phi(A) \in \mathbb{C}^{n\times n}$ and $\q(x), \q(y) \in \mathbb{C}^{n}$, we can estimate $x^\dagger Ay$ to additive error $\eps$ with probability at least $1-\delta$ in
        \begin{align*}
        \bigO{\phi\frac{\|A\|_\fr^2\|x\|^2\|y\|^2}{\eps^2} \big(\sqcb_\phi(A)+\qcb(x)+\qcb(y)\big)\log\frac{1}{\delta}}
        \end{align*}
        time.
        To do this, observe that $x^\dagger Ay = \Tr(x^\dagger Ay) = \Tr(Ayx^\dagger)$ and that $\q(yx^\dagger)$ can be simulated with $\q(x), \q(y)$.
        So, we just apply the trace inner product procedure.
    \end{enumerate}
\end{remark}

Finally, we observe a simple technique to convert importance sampling sketches into approximate isometries, by inserting the appropriate pseudoinverse.
This will be used in some of the more involved applications.

\begin{restatable}{lemma}{apporth} \label{lem:app-orth-expression}
  Given $A \in \bbc^{m\times n}$, $S \in \bbc^{r\times m}$ sampled from a $\phi$-oversampled importance sampling distribution of $A$, and $T^\dagger \in \bbc^{n\times c}$ sampled from an $\leq \phi$-oversampled importance sampling distribution of $(SA)^\dagger$, let $R \coloneqq SA$ and $C \coloneqq SAT$.
  Let $\sigma_k$ be the $k$th singular value of $A$.
  If, for $\alpha \in (0,1]$, $r = \tilde{\Omega}(\frac{\phi^2\|A\|^2\|A\|_\fr^2}{\sigma_k^4}\log\frac{1}{\delta})$ and $c = \tilde{\Omega}(\frac{\phi^2\|A\|^2\|A\|_\fr^2}{\sigma_k^4\alpha^2}\log\frac{1}{\delta})$, then with probability $\geq 1-\delta$, $((C_k)^+ R)^\dagger$ is an $\alpha$-approximate projective isometry onto the image of $(C_k)^+$.
  Further, $(DV^\dagger R)^\dagger$ is an $\alpha$-approximate isometry, where $C_k^+ = UDV^\dagger$ is a singular value decomposition truncated so that $D \in \bbr^{k' \times k'}$ is full rank (so $k' \leq \min(k, \rank(A))$).
\end{restatable}
\begin{proof}
The following occurs with probability $\geq 1-\delta$.
By \cref{prop:appr-svs}, $\|C_k^+\| \lesssim \frac{1}{\sigma_k^2}$.
By \cref{lem:AMP-spectral}, $\|R^\dagger R - A^\dagger A\| \lesssim \|A\|^2$, which implies that $\|R\| \lesssim \|A\|$, and by \cref{lem:sa-frob-norm}, $\|R\|_\fr \lesssim \|A\|_\fr$.
Further, $\|RR^\dagger - CC^\dagger\| \leq \alpha\sigma_k^2\frac{\|R\|\|R\|_\fr}{\|A\|\|A\|_\fr} \lesssim \alpha\sigma_k^2$.
Finally, $C_k^+ C = C_k^+ C_k$ is an orthogonal projector.
So, with probability $\geq 1-\delta$,
\begin{align*}
  \|(C_k^+ R)(C_k^+ R)^\dagger - (C_k^+ C)(C_k^+ C)^\dagger \|
  &= \| C_k^+(RR^\dagger - CC^\dagger) (C_k^+)^\dagger\|
  \leq \|C_k^+\|^2\|RR^\dagger - CC^\dagger\|
  = O(\alpha).
\end{align*}
We get the computation for the $\alpha$-approximate isometry by restricting attention to the span of $U$:
\begin{align*}
  \|(DV^\dagger R)(DV^\dagger R)^\dagger - I \|
  &= \| DV^\dagger(RR^\dagger - CC^\dagger)VD^\dagger\|
  \leq \|UDV^\dagger\|^2\|RR^\dagger - CC^\dagger\|
  = O(\alpha). \qedhere
\end{align*}
\end{proof}

One can also observe that, for a sufficiently good sketch $C$, $R \approx C_k(C_k)^+ R$ in spectral norm, giving a generic way to approximate a sketch $R$ by a product of a small matrix with an approximate projective isometry.
We do not need it in our proofs, so this computation is not included. 


\section{Singular value transformation}\label{subsec:main-svt}

Our main result is that, given $\sq_\phi(A)$ and a smooth function $f$, we can approximate $f(A^\dagger A)$ by a decomposition $R^\dagger U R + f(0)I$.
This primitive is based on the even singular value transformation used by Gilyén, Su, Low, and Wiebe~\cite{gilyen2018QSingValTransf}.

\begin{restatable}[Even singular value transformation]{theorem}{evensvt} \label{thm:evenSing}
Let $A \in \mathbb{C}^{m\times n}$ and $f\colon\mathbb{R}^+\to \mathbb{C}$ be such that $f(x)$ and $\bar{f}(x)\coloneqq (f(x)-f(0))/x$ are $L$-Lipschitz and $\bar{L}$-Lipschitz, respectively, on $\cup_{i=1}^{\min(m,n)} [\sigma_i^2 - d, \sigma_i^2 + d]$ for some $d > 0$.
Take parameters $\eps$ and $\delta$ such that $0 < \eps \lesssim \min(L\|A\|_*^2,\bar{L}\|A\|_*^2\|A\|^2)$ and $\delta\in(0,1]$.
Choose a norm $* \in \{\fr,\op\}$.

Suppose we have $\sq_\phi(A)$.
Consider the importance sampling sketch $S \in \bbr^{r\times m}$ corresponding to $\sq_\phi(A)$ and the importance sampling sketch $T^\dagger \in \bbr^{c\times n}$ corresponding to $\sq_{\leq 2\phi}((SA)^\dagger)$ (which we have by \cref{lem:sq-sketching}).
Then, for $R \coloneqq SA$ and $C \coloneqq SAT$, we can achieve the bound
\begin{equation}\label{eq:evenSVTbound}
    \Pr\Big[\|R^\dagger \bar{f}(CC^\dagger) R+f(0)I - f(A^\dagger A)\|_* > \eps \Big] < \delta,
\end{equation}
if $r,c > \|A\|^2\|A\|_\fr^2\phi^2\frac{1}{d^2}\log\frac{1}{\delta}$ (or, equivalently, $d > \bar{\eps} \coloneqq \|A\|_*\|A\|_\fr(\frac{\phi^2\log(1/\delta)}{\min(r,c)})^{\!1/2}\!$) and
\begin{equation}\label{eq:evenSVTrc}
    r = \tilde{\Omega}\Big(\phi^2 L^2\|A\|_*^2\|A\|_\fr^2\frac{1}{\eps^2}\log\frac1\delta\Big) \qquad
    c = \tilde{\Omega}\Big(\phi^2 \bar{L}^2\|A\|^4\|A\|_*^2\|A\|_\fr^2\frac{1}{\eps^2}\log\frac{1}{\delta}\Big).
\end{equation}
\end{restatable}

First, we make some technical remarks.
The assumption that $\eps \lesssim L\|A\|_*^2$ is for non-degeneracy: if $\eps \geq L\|A\|^2$, then the naive approximation $f(0)I$ of $f(A^\dagger A)$ would suffice, since $\|f(0)I - f(A^\dagger A)\| \leq L\|A\|^2 \leq \eps$ as desired.\footnote{The choice $f(0)I$ assumes that $f$ is Lipschitz on $\{0,\|A\|^2\}$. More generally, we can choose $f(x)I$ for any $x\in\cup_{i=1}^{\min(m,n)} [\sigma_i^2 - d, \sigma_i^2 + d]$ in order to get a sufficiently good naive approximation.}
The parameter $d$ (or, rather, the parameter $\bar{\eps}$) specifies the domain where $f(x)$ and $\bar{f}(x)$ should be smooth: the condition in the theorem is that they should be Lipschitz on the spectrum of $A^\dagger A$, with $\bar{\eps}$ room for approximation.
This will not come into play often, though, since we can often design our singular value transforms such that we can take $d = \infty$.
For example, if our desired transform $f$ becomes non-smooth outside the relevant interval $[0,\|A\|^2]$, we can apply \cref{thm:evenSing} with $d = \infty$ and the function $g$ such that $g(x) = g(\|A\|^2)$ for $x \geq \|A\|^2$ and $g(x) = f(x)$ otherwise.
Then $g(A^\dagger A) = f(A^\dagger A)$ and $g$ is smooth everywhere, so we do not need to worry about the $d$ parameter.
Finally, we note that no additional log terms are necessary (i.e., $\tilde{\Omega}$ becomes $\Omega$) when the Frobenius norm is used.

By our discussion in \cref{subsec:sketch}, finding the sketches $S$ and $T$ for \cref{thm:evenSing} takes time $\bigO{(r+c)\scb_\phi(A) + rc\qcb_\phi(A) + \ncb_\phi(A)}$, querying for all of the entries of $C$ takes additional time $\bigO{rc\qcb(A)}$, and computing $\bar{f}(CC^\dagger)$ takes additional time $\bigO{\min(r^2c,rc^2)}$ (if done naively).
For our applications, this final matrix function computation will dominate the runtime, and the rest of the cost we will treat as $\bigO{rc\sqcb_\phi(A)}$.

For some intuition on error bounds and time complexity, we consider how the parameters in our main theorem behave in a restricted setting: suppose we have $\sq(A)$ with minimum singular value $\sigma$ and such that $\|A\|_\fr / \sigma$ is dimension-independent.\!\footnote{By a dimension-independent or dimensionless quantity, we mean a quantity that is both independent of the size of the input matrix and is scale-invariant, i.e., does not change under scaling $A\leftarrow \alpha A$.}
This condition simultaneously bounds the rank and condition number of $A$.
Further suppose\footnote{This criterion is fairly reasonable. For example, the polynomials used in QSVT satisfy it.} that $f$ is $L$-Lipschitz on the interval $[0,\|A\|^2]$ and satisfies
\[
    L\|A\|^2 < \Gamma D \text{ where } D \coloneqq \max_{x \in [0,\|A\|^2]} f(x) - \min_{y \in [0,\|A\|^2]} f(y),
\]
for some dimension-independent $\Gamma$.
$\Gamma$ must be at least one, so we can think about such an $f$ as being at most $\Gamma$ times ``steeper'' compared to the least possible ``steepness''.
Under these assumptions, we can get a decomposition satisfying
\begin{align*}
  \|R^\dagger \bar{f}(CC^\dagger)R + f(0)I - f(A^\dagger A)\| > \eps D
\end{align*}
with probability $\geq 1-\delta$ by taking
\begin{align*}
  r = \tilde{\Theta}\Big(\Gamma^2\frac{\|A\|_\fr^2}{\|A\|^2}\frac{1}{\eps^2}\log\frac1\delta\Big) \text{ and }
  c = \tilde{\Theta}\Big(\Gamma^2\frac{\|A\|^2\|A\|_\fr^2}{\sigma^4}\frac{1}{\eps^2}\log\frac{1}{\delta}\Big).
\end{align*}
The time to compute the decomposition is
\begin{align*}
    \bOt{\frac{\|A\|_\fr^6}{\|A\|^2\sigma^4}\frac{\Gamma^6}{\eps^6}\log^3\frac{1}{\delta}}.
\end{align*}
These quantities are all dimensionless.
Dependence on $\sigma$ arises because we bound $\bar{L} \leq L/\sigma^2$: our algorithm's dependence on $\bar{L}$ implicitly enforces a low-rank constraint in this case.
All of our analyses give qualitatively similar results to this, albeit in more general settings allowing approximately low-rank input.

To perform error analyses, we will need bounds on the norms of the matrices in our decomposition.
The following lemma gives the bounds we need for \cref{sec:applications}.

\begin{restatable}[Norm bounds for even singular value transformation]{lemma}{evensvtbounds} \label{lem:evenSingBounds}
Suppose the assumptions from \cref{thm:evenSing} hold.
Then with probability at least $1-\delta$, the event in \cref{eq:evenSVTbound} occurs (that is, $R^\dagger \bar{f}(CC^\dagger) R \approx f(A^\dagger A) - f(0)I$) and moreover, \emph{the following bounds also hold}:
\begin{gather}
    \label{eq:evenSVTbdR}
    \|R\| = \bigO{\|A\|} \quad \text{and} \quad \|R\|_\fr = \bigO{\|A\|_\fr},\\
    \label{eq:evenSVTbdC}
    \|\bar{f}(CC^\dagger)\| \leq \max \Big\{ \abs{\bar{f}(x)} \,\Big\vert\, x \in \bigcup_{i=1}^{\min(r,c)} [\sigma_i^2 - \bar{\eps}, \sigma_i^2 + \bar{\eps}] \Big\}, \\
    \label{eq:evenSVTbdRC}
    \text{when $* = \op$, } \Big\|R^\dagger \sqrt{\bar{f}(CC^\dagger)}\Big\| \leq \sqrt{\|f(A^\dagger A) - f(0)I\| + \eps}.
\end{gather}
\end{restatable}

\cref{eq:evenSVTbdRC} is typically a better bound than combining \cref{eq:evenSVTbdR,eq:evenSVTbdC}.
For intuition, notice this is true if $\eps, \bar{\eps} = 0$: the left-hand and right-hand sides of the following inequality are the two ways to bound $\|R^\dagger \sqrt{\bar{f}(CC^\dagger)}\|^2$, up to constant factors ($\sigma$ below runs over the singular values of $A$):
\[
    \|f(A^\dagger A) - f(0)I\| \leq \max_{\sigma}|f(\sigma^2) - f(0)| \leq \max_{\sigma} \sigma^2 \max_{\sigma} \frac{|f(\sigma^2) - f(0)|}{\sigma^2}
    = \|A\|^2\max_{\sigma} |\bar{f}(\sigma^2)|.
\]
The rest of this section will be devoted to proving \cref{thm:evenSing} and \cref{lem:evenSingBounds}.
A mathematical tool we will need is a matrix version of the defining inequality of $L$-Lipschitz functions, $\abs{f(x) - f(y)} \leq L\abs{x - y}$ when $f$ is $L$-Lipschitz.
The Frobenius norm version of this bound (\cref{lem:lipschitz-frob}) follows by computing matrix derivatives; the spectral norm version (\cref{lem:lipschitz-oper}) has a far less obvious proof.

\begin{lemma}[{\cite[Corollary 2.3]{Gil10}}] \label{lem:lipschitz-frob}
  Let $A$ and $B$ be Hermitian matrices and let $f\colon \bbr \rightarrow \bbc$ be $L$-Lipschitz continuous on the eigenvalues of $A$ and $B$.
  Then $\|f^\evt(A) - f^\evt(B)\|_\fr \leq L\|A - B\|_\fr$.
\end{lemma}

\begin{lemma}[{\cite[Theorem 11.2]{aleksandrov2011EstimatesOfOperatorModuli}}] \label{lem:lipschitz-oper}
  \label{lemma:lipschitz}
  Let $A$ and $B$ be Hermitian matrices and let $f\colon \bbr \rightarrow \bbc$ be $L$-Lipschitz continuous on the eigenvalues of $A$ and $B$.
  Then
  \begin{align*}
    \norm{f^\evt(A)-f^\evt(B)} \lesssim L\|A-B\|\log\min(\rank A, \rank B). 
  \end{align*}
\end{lemma}

\begin{proof}[Proof of \cref{thm:evenSing} and \cref{lem:evenSingBounds}]
Since $g(A^\dagger A) = f(A^\dagger A) + g(0)I$ for $f(x)\coloneqq g(x)-g(0)$, we can assume without loss of generality that $f(0) = 0$.
As a reminder, in the statement of \cref{thm:evenSing} we take
\begin{equation*}
    r = \tilde{\Omega}\Big(\phi^2 L^2\|A\|_*^2\|A\|_\fr^2\frac{1}{\eps^2}\log\frac1\delta\Big) \qquad
    c = \tilde{\Omega}\Big(\phi^2 \bar{L}^2\|A\|^4\|A\|_*^2\|A\|_\fr^2\frac{1}{\eps^2}\log\frac{1}{\delta}\Big).
\end{equation*}
These values are chosen such that the following holds with probability $\geq 1-\delta$ simultaneously.
\begin{enumerate}
  \item The $i$th singular value of $CC^\dagger$ does not differ from the $i$th singular value of $A^\dagger A$ by more than $\bar{\eps}$.
  This follows from \cref{prop:appr-svs} with error parameter $\eps\|A\|_\fr^{-1}\|A\|_*^{-1}\max(L^{-1},\bar{L}^{-1}\|A\|^{-2})$.
  This immediately implies \cref{eq:evenSVTbdC}.

  \item $\|R\|^2 = \bigO{\|A\|^2}$.
  This is the spectral norm bound in \cref{eq:evenSVTbdR} (the Frobenius norm bound follows from \cref{lem:sa-frob-norm}).
  We use \cref{lem:AMP-spectral}:
  \begin{align*}
    \|R\|^2 &\leq \|A\|^2 + \|R^\dagger R - A^\dagger A\| \leq \|A\|^2 + \frac{\eps\|A\|^2}{L\|A\|_*^2} = O(\|A\|^2).
  \end{align*}
  \item $\|f(R^\dagger R) - f(A^\dagger A)\|_* = \bigO{\eps}$.
  We need the polylog factors in our number of samples to deal with the $\log r$ that arises from \cref{lemma:lipschitz} in the spectral norm case.
  \begin{align*}
    &\|f(R^\dagger R) - f(A^\dagger A)\|_* \\
    &\lesssim L\|R^\dagger R - A^\dagger A\|_*\log \rank(R^\dagger R) \tag{\cref{lemma:lipschitz} or \cref{lem:lipschitz-frob}} \\
    &\lesssim L\sqrt{\frac{\phi^2\log r\log(1/\delta)}{r}}\|A\|_*\|A\|_\fr\log r \tag{\cref{lem:AMP-spectral} or \cref{prop:appr-mms}} \\
    &\lesssim \eps \tag{plugging in value for $r$}.
  \end{align*}
  \item $\|\bar{f}(CC^\dagger) - \bar{f}(RR^\dagger)\|_* = \bigO{\eps/\|A\|^2}$.
  This follows similarly to the above point.
\end{enumerate}
When all of the above bounds hold, we can conclude:
\begin{align*}
  &\|R^\dagger \bar{f}(CC^\dagger) R - f(A^\dagger A)\|_* \\
  &\leq \|R^\dagger \bar{f}(RR^\dagger) R - f(A^\dagger A)\|_* + \|R^\dagger (\bar{f}(RR^\dagger) - \bar{f}(CC^\dagger))R\|_* \\
  &= \|f(R^\dagger R) - f(A^\dagger A)\|_* + \|R^\dagger (\bar{f}(RR^\dagger) - \bar{f}(CC^\dagger))R\|_* \tag{Definition of $\bar{f}$} \\
  &\leq \|f(R^\dagger R) - f(A^\dagger A)\|_* + \|R\|^2\|\bar{f}(RR^\dagger) - \bar{f}(CC^\dagger)\|_* \\
  &\lesssim \eps + \|R\|^2\eps/\|A\|^2 \\
  &\lesssim \eps.
\end{align*}
This gives \cref{eq:evenSVTbound} after rescaling $\eps$ by an appropriate constant factor.
When $* = \op$, we also have \cref{eq:evenSVTbdRC}, since
\begin{equation*}
  \Big\|R^\dagger \sqrt{\bar{f}(CC^\dagger)}\Big\|
  = \sqrt{\|R^\dagger \bar{f}(CC^\dagger) R\|}
  \leq \sqrt{\|f(A^\dagger A) - f(0)I\| + \eps}. \qedhere
\end{equation*}
\end{proof}

We remark here that the $\log$ term in \cref{lem:lipschitz-oper} unfortunately cannot be removed (because some Lipschitz functions are not operator Lipschitz).
However, several bounds hold under various mild assumptions, and for particular functions, the $\log$ term can be improved to $\log \log$ or completely removed.
For example, the QSVT literature~\cite{gilyen2018QSingValTransf} cites the following result:
\begin{lemma}[{\cite[Corollary 7.4]{aleksandrov2009OperatorHolderZygmundFunctions}}] \label{lem:lipschitz-oper-moduli}
  Let $A$ and $B$ be Hermitian matrices such that $a I \preceq A, B \preceq b I$, and let $f\colon \bbr \rightarrow \bbc$ be $L$-Lipschitz continuous on the interval $[a,b]$.
  Then
  \begin{align*}
  \norm{f^\evt(A)-f^\evt(B)} \lesssim L\|A-B\|\log\left(e\frac{b-a}{\nrm{A-B}}\right). 
  \end{align*}
\end{lemma}
Though we will not use it, we can extend these results on eigenvalue transformation of Hermitian matrices to singular value transformation of general matrices via the reduction from {\cite[Corollary 21]{gilyen2018QSingValTransf}}.
For example, \cref{lem:lipschitz-oper} implies the following:

\begin{lemma}
  \label{lemma:lipschitz-svf}
  Let $A, B \in \bbc^{m\times n}$ be matrices and let $f\colon [0,\infty) \rightarrow \bbc$ be $L$-Lipschitz continuous on the singular values of $A$ and $B$ such that $f(0) = 0$.
  Then
  \begin{align*}
    \|f^\svt(A)-f^\svt(B)\| \lesssim L\|A - B\|\log \min(\rank A, \rank B).
  \end{align*}
\end{lemma}

In \cref{sec:cur}, we prove results on generic singular value transformation and eigenvalue transformation by bootstrapping \cref{thm:evenSing}.
Since these are slower, though, we will use primarily the even singular value transformation results that we just proved to recover ``dequantized QML'' results.
This will be the focus of next section. 


\section{Applying the framework to dequantizing QML algorithms}\label{sec:applications}

Now, with our framework, we can recover previous dequantization results: recommendation systems (\cref{sec:recommendation-systems}), supervised clustering (\cref{sec:supervised-clustering}), principal component analysis (\cref{sec:PCA}), low-rank matrix inversion (\cref{sec:matrix-inversion}), support-vector machines (\cref{sec:SVM}), and low-rank semidefinite programs (\cref{sec:SDP}).
We also propose new quantum-inspired algorithm for other applications, including QSVT (\cref{subsec:dequant}), Hamiltonian simulation (\cref{sec:Hamiltonian-simulation}), and discriminant analysis (\cref{sec:discriminant-analysis}).
We give applications in roughly chronological order; this also happens to be a rough difficulty curve, with applications that follow more easily from our main results being first.

Everywhere it occurs, $K \coloneqq \|A\|_\fr^2/\sigma^2$, where $A$ is the input matrix.
$\kappa \coloneqq \|A\|_2^2/\sigma^2$.
For simplicity, we will often describe our runtimes as if we know spectral norms of input matrices (so, for example, we know $\kappa$).
If we do not know the spectral norm, we can run \cref{prop:appr-svs} repeatedly with multiplicatively decreasing $\eps$ until we find a constant factor upper bound on the spectral norm, which suffices for our purposes.
Alternatively, we can bound the spectral norm by the Frobenius norm, which we know from sampling and query access to input.

\subsection{Dequantizing QSVT}\label{subsec:dequant}
We begin by dequantizing the quantum singular value transformation described by Gily\'{e}n, Su, Low, and Wiebe~\cite{gilyen2018QSingValTransf} for close-to-low-rank matrices.

\begin{definition} \label{def:qsvt}
For a matrix $A \in \bbc^{m\times n}$ and $p(x) \in \bbc[x]$ degree-$d$ polynomial of parity-$d$ (i.e., even if $d$ is even and odd if $d$ is odd), we define the notation $p^\qsvt(A)$ in the following way:
\begin{enumerate}
    \item If $p$ is \emph{even}, meaning that we can express $p(x) = q(x^2)$ for some polynomial $q(x)$, then $$p^\qsvt(A) \coloneqq q(A^\dagger A)= p(\sqrt{A^\dagger A}).$$
    \item If $p$ is \emph{odd}, meaning that we can express $p(x) = x\cdot q(x^2)$ for some polynomial $q(x)$, then $$p^\qsvt(A) \coloneqq A\cdot q(A^\dagger A).$$
\end{enumerate}
\end{definition}

\begin{theorem} \label{thm:deq-qsvt}
Suppose we are given a matrix $A \in \bbc^{m\times n}$ satisfying $\|A\|_\fr = 1$ via the oracles for $\sq(A)$ and $\sq(A^\dagger)$ with $\sqcb(A), \sqcb(A^\dagger) = \bigO{\log(mn)}$, a vector $\sq(b) \in \bbc^n$ with $\|b\| = 1$ and $\sqcb(b) = \bigO{\log n}$, and a degree-$d$ polynomial $p(x)$ of parity-$d$ such that $\abs{p(x)} \leq 1$ for all $x \in [-1,1]$.

Then with probability $\geq 1-\delta$, for $\eps$ a sufficiently small constant, we can get $\sq_\phi(v) \in \bbc^n$ such that $\|v - p^\qsvt(A)b\| \leq \eps\|p^\qsvt(A)b\|$ in $\poly\Big(d,\frac{1}{\|p^\qsvt(A)b\|},\frac{1}{\eps},\frac{1}{\delta}, \log mn\Big)$ time.

Specifically, for $p$ even, the runtime is
\begin{align*}
    \bOt{\frac{d^{16}\|A\|^{10}}{(\eps\|p^\qsvt(A)b\|)^6}\log^3\frac{1}{\delta} + \frac{d^{12}\|A\|^8 + d^6\|A\|^2}{(\eps\|p^\qsvt(A)b\|)^4}\log^2\frac{1}{\delta}\log(mn)}
\end{align*}
with
\begin{align*}
    \sqrun(v) &= \bOt{\frac{d^{12}\|A\|^4}{\eps^4\|p^\qsvt(A)b\|^6}\log(mn)\log^3\frac{1}{\delta}},
\end{align*}
and for $p$ odd, the runtime is
\begin{align*}
    \bOt{\frac{d^{22}\|A\|^{16}}{(\eps\|p^\qsvt(A)b\|)^{6}} + \frac{d^{16}\|A\|^{12} + d^{10}\|A\|^4\delta^{-1}}{(\eps\|p^\qsvt(A)b\|)^4}\log(mn)}
\end{align*}
with
\begin{align*}
    \sqrun(v) &= \bOt{\frac{d^8}{\eps^2\delta\|p^\qsvt(A)b\|^4}\log(mn)}.
\end{align*}
\end{theorem}

From this result it follows that QSVT, as described in {\cite[Theorem~17]{gilyen2018QSingValTransf}}, has no exponential speedup when the block-encoding of $A$ comes from a quantum-accessible ``QRAM'' data structure as in {\cite[Lemma~50]{gilyen2018QSingValTransf}}.
In the setting of QSVT, given $A$ and $b$ in QRAM, one can prepare $|b\rangle$ and construct a block-encoding for $A/\|A\|_\fr = A$ in $\polylog(mn)$ time.
Then one can apply (quantum) SVT by a degree-$d$ polynomial on $A$ and apply the resulting map to $|b\rangle$ with $d\cdot \polylog(mn)$ gates and finally project down to get the state $|p^\qsvt(A)b\rangle$ with probability $\geq 1-\delta$ after $\Theta\big(\frac{1}{\|p^\qsvt(A)b\|}\log\frac{1}{\delta}\big)$ iterations of the circuit.
So, getting a sample from $|p^\qsvt\!(A)b\rangle$ takes $\Theta\big(d\frac{1}{\|p^\qsvt(A)b\|}\polylog(mn/\delta)\big)$ time.
This circuit gives an exact outcome, possibly with some $\log(1/\eps)$ factors representing the discretization error in truncating real numbers to finite precision (which we ignore, since we do not account for them in our classical algorithm runtimes).

Analogously, by \cref{rmk:when-sq-access}, having $A$ and $b$ in (Q)RAM implies having $\sq(A)$ and $\sq(b)$ with $\sqcb(A) = \bigO{\log mn}$ and $\sqcb(b) = \bigO{\log n}$.
Since QSVT also needs to assume $\max_{x \in [-1,1]}\abs{p(x)} \leq 1$, the classical procedure matches the assumptions for QSVT.
Our algorithm runs only polynomially slower than the quantum algorithm, since the quantum runtime clearly depends on $d$, $\frac{1}{\|p^\qsvt(A)b\|}$, and $\log(mn)$.
We are exponentially slower in $\eps$ and $\delta$ (these errors are conflated for the quantum algorithm).
However, this exponential advantage vanishes if the desired output is not a quantum state but some fixed value (or an estimate of one).
In that case, the quantum algorithm must also pay $\frac{1}{\eps}$ during the sampling or tomography procedures and the classical algorithm can boost a constant success probability to $\geq 1-\delta$, only paying a $\log\frac{1}{\delta}$ factor.
Note that, unlike in the quantum output, we can query entries of the output, which a quantum algorithm cannot do without paying at least a $\frac{1}{\eps}$ factor.

\cref{thm:deq-qsvt} also dequantizes QSVT for block-encodings of density operators when the density operator comes from some well-structured classical data.
Indeed, \cite[Lemma~45]{gilyen2018QSingValTransf} assumes we can efficiently prepare a purification of the density operator $\rho$.
The rough classical analogue is the assumption that we have sampling and query access to some $A \in \bbc^{m\times n}$ with $\rho = A^\dagger A$.
Since $\Tr(\rho) = 1$, we have $\|A\|_\fr = 1$.
Then, $p^\qsvt(\rho) = r^\qsvt(A)$ for $r(x) = p(x^2)$ and $\|\rho\| = \|A\|^2$, so we can repeat the above argument to show the lack of exponential speedup for this input model too.

We can mimic the quantum algorithm with our techniques because low-degree polynomials are smooth, in the sense that we formalize with the following lemma (proven in \cref{apx:proofs}).

\begin{restatable}{lemma}{lowdeg} \label{lem:low-deg-lipschitz}
Consider $p(x)$ a degree-$d$ polynomial of parity-$d$ such that $\abs{p(x)} \leq 1$ for $x \in [-1,1]$.
Recall that, for a function $f: \bbc \to \bbc$, we define $\bar{f}(x) \coloneqq (f(x)-f(0))/x$ (and $\bar{f}(0) = f'(0)$ when $f$ is differentiable at zero).
\begin{itemize}
\item If $p$ is even, then $\displaystyle\max_{x \in [0,1]} \abs{q(x)} \leq 1, \max_{x \in [-1,1]} \abs{q'(x)} \lesssim d^2\!, \max_{x \in [-1,1]} \abs{\bar{q}(x)} \lesssim d^2\!$, and $\displaystyle\max_{x \in [-1,1]} \abs{\bar{q}'(x)} \lesssim d^4\!\!$.

\item If $p$ is odd, then $\displaystyle \!\max_{x \in [-1,1]} \abs{q(x)} \lesssim d,\, \!\max_{x \in [-1,1]} \abs{q'(x)} \lesssim d^3\!,\max_{x \in [-1,1]} \abs{\bar{q}(x)} \lesssim d^3\!$, and $\displaystyle\!\max_{x \in [-1,1]} \abs{\bar{q}'(x)} \lesssim d^5\!\!$.
\end{itemize}
\end{restatable}

These bounds are tight for Chebyshev polynomials.
In general, these bounds can be loose, so for any particular QML application we recommend using our main results for faster algorithms.

\begin{proof}[Proof of \cref{thm:deq-qsvt}]
Consider the even case: take $p(x) = q(x^2)$ for $q$ a degree-$d/2$ polynomial, so $p^\qsvt(A) = q(A^\dagger A)$, and we have the correct form to apply \cref{thm:evenSing}.
$q$ is uncontrolled outside of $[-1,1]$, so we instead apply the singular value transformation which is constant outside of $[-1,1]$:
\begin{align*}
    f(x) \coloneqq \begin{cases}
        q(-1) & x \leq -1 \\
        q(x) & -1 \leq x \leq 1 \\
        q(1) & 1 \leq x
    \end{cases}.
\end{align*}
We can do this because the singular values of $A$ lie in $[0,1]$, so $q(A^\dagger A) = f(A^\dagger A)$.
Then, by \cref{lem:low-deg-lipschitz}, $f$ and $\bar{f}$ are Lipschitz with $L = \bigO{d^2}, \bar{L} = \bigO{d^4}$.
So, by \cref{thm:evenSing}, we can get $R \in \bbc^{r\times n}$ and $C \in \bbc^{r\times c}$ such that $\|R^\dagger \bar{f}(CC^\dagger) R + f(0)I - f(A^\dagger A)\| \leq \eps$, where
\begin{align*}
    r = \bOt{d^4\|A\|^2\|A\|_\fr^2\frac{1}{\eps^{2}}\log\frac{1}{\delta}} \qquad \text{and} \qquad c = \bOt{d^8\|A\|^6\|A\|_\fr^2\frac{1}{\eps^{2}}\log\frac{1}{\delta}}.
\end{align*}
(We will later rescale $\eps$; note that $\eps\|p^\qsvt(A)b\| \lesssim L\|A\|^2\|b\|$, so $\eps$ is small enough for the theorem assumption.)
This reduces the problem to approximating $R^\dagger \bar{f}(CC^\dagger) Rb + f(0)b$.
We further approximate $Rb \approx u \in \bbc^r$ such that $\|Rb - u\| \leq \eps/d$.
Using \cref{prop:appr-mms}, this needs $\bigO{\|A\|_\fr^2\|b\|^2\frac{d^2}{\eps^2}\log\frac{1}{\delta}}$ samples, which can be done in time $\bigO{\|A\|_\fr^2\|b\|^2\frac{d^2}{\eps^2}r\log(mn)\log\frac{1}{\delta}}$, using that $\|R\|_\fr \lesssim \|A\|_\fr$ (\cref{eq:evenSVTbdR}) and $\sqcb(R^\dagger) = \bigO{r\sqcb(A)}$ (\cref{lem:sq-sketching}).
This suffices to maintain the error bound because (using \cref{eq:evenSVTbdRC,eq:evenSVTbdC,lem:low-deg-lipschitz}),
\begin{align*}
    \|R^\dagger \bar{f}(CC^\dagger)(Rb - u)\| &\leq \Big\|R^\dagger\sqrt{\bar{f}(CC^\dagger)}\Big\|\Big\|\sqrt{\bar{f}(CC^\dagger)}\Big\|\|Rb - u\| \\
    &\leq \sqrt{\|f(A^\dagger A) - f(0)I\| + \eps}\sqrt{\max_x \bar{f}(x)}\frac{\eps}{d}
    \leq \sqrt{2 + \eps}d\frac{\eps}{d}
    \lesssim \eps.
\end{align*}
As a consequence, $v \coloneqq R^\dagger \bar{f}(CC^\dagger) u+ f(0)b$ satisfies $\|v - p^\qsvt(A)b\| \leq \eps$.
Via \cref{lemma:sample-Mv}, we can get $\sq_\phi(v)$ with
\begin{align*}
    \sqrun(v) &= \phi\sq_\phi(v)\log\frac{1}{\delta} \\
    &= \Big((r+1)\frac{\|R\|_\fr^2\|\bar{f}(CC^\dagger) u\|^2 + p(0)^2\|b\|^2}{\|v\|^2}\Big)\Big((r+1)\log(mn)\Big)\log\frac{1}{\delta} \\
    &\lesssim r^2\frac{\|\bar{f}(CC^\dagger)\|^2(\|Rb\|+\|Rb-u\|)^2 + p(0)^2}{(\|p^\qsvt(A)b\|-\|v-p^\qsvt(A)\|)^2}\log (mn)\log\frac{1}{\delta} \\
    &\lesssim r^2\frac{d^4(1 + \eps/d)^2 + 1}{(\|p^\qsvt(A)b\|-\eps)^2}\log\frac{1}{\delta} \\
    &\lesssim \frac{r^2d^4}{\|p^\qsvt(A)b\|^2}\log(mn)\log\frac{1}{\delta}.
\end{align*}
In the last step, we use that $\eps \lesssim \|p^\qsvt(A)b\|$; if we do not have that assumption, $\sqrun(v) \lesssim \frac{r^2d^4}{\|v\|^2}\log(mn)\log\frac{1}{\delta}$.
We rescale $\eps \leftarrow \eps\|p^\qsvt(A)b\|$ to get the desired bound.
The runtime is dominated by finding $C$ in $\bigO{rc\log(mn)}$ time, computing $\bar{f}(CC^\dagger)$ in $\bigO{r^2c}$ time, and estimating $R^\dagger b$ in $\bigO{r\frac{d^2}{\eps^2}\log(mn)\log\frac{1}{\delta}}$ time.
We also need to compute the matrix-vector product $\bar{f}(CC^\dagger) u$, but this can be done in $\bigO{rc}$ time by instead multiplying through with the expression $U\bar{f}(D^2)U^\dagger = \bar{f}(CC^\dagger)$, where $U \in \bbc^{r\times c}$ comes from the SVD of $C$.

Now for the odd case: we could use \cref{thm:gen-svt} here, but we will continue to use \cref{thm:evenSing} here.
Similarly to the even case, we take $g(x)$ to be $q(x)$ in $[-1,1]$ and held constant ouside it, so $p^\qsvt(A) = A \cdot g(A^\dagger A)$.
Then, by plugging in the smoothness parameters from \cref{lem:low-deg-lipschitz}, we get $R, C$ such that $\|R^\dagger \bar{g}(CC^\dagger) R + g(0)I - g(A^\dagger A)\| < \frac{\eps}{\|A\|}$ with probability $\geq 1-\delta$ where
\begin{align*}
    r = \bOt{d^6\|A\|^4\|A\|_\fr^2\frac{1}{\eps^2}\log\frac{1}{\delta}} \qquad
    c = \bOt{d^{10}\|A\|^8\|A\|_\fr^2\frac{1}{\eps^2}\log\frac{1}{\delta}}.
\end{align*}
We now use the approximating matrix product lemmas \cref{lem:mmhp,prop:appr-mms} three times.
\begin{enumerate}
    \item We use \cref{lem:mmhp} to approximate $AR^\dagger \approx A'R'^\dagger$ such that $\|AR^\dagger - A'R'^\dagger\| \leq \eps d^{-2}$.
    We can do this since we have $\sq(A^\dagger)$ (by assumption) and $\sq(R^\dagger)$, in $\bigO{\|A\|_\fr^2\|R\|_\fr^2d^4\eps^{-2}\delta^{-1}} = \bigO{d^4\eps^{-2}\delta^{-1}}$ samples, with each sample costing $\bigO{r\log(mn)}$ time (by \cref{lem:sq-sketching}).
    Then using \cref{lem:evenSingBounds} and the bounds on $\abs{q(x)}$ and $\abs{\bar{q}(x)}$ in \cref{lem:low-deg-lipschitz},
    \begin{align*}
        \|(AR^\dagger - A'R'^\dagger)\bar{g}(CC^\dagger)R\|
        &\leq \|AR^\dagger - A'R'^\dagger\|\|\sqrt{\bar{g}(CC^\dagger)}\|\|\sqrt{\bar{g}(CC^\dagger)}R\| \\
        &\leq (\eps d^{-2})\sqrt{d^3}\sqrt{\|g(A^\dagger A) - g(0)I\| + \frac{\eps}{\|A\|}}
        \lesssim \eps.
    \end{align*}
    \item We use \cref{prop:appr-mms} to approximate $Rb \approx u$ such that $\|Rb - u\| \leq \frac{\eps}{d^2\|A\|}$, where we use $\bigO{\|R\|_\fr^2\|b\|^2\frac{d^4\|A\|^2}{\eps^2}\log\frac{1}{\delta}} = \bigO{d^4\|A\|^2\eps^{-2}\log\frac{1}{\delta}}$ samples.
    \begin{align*}
        \|A'R'^\dagger \bar{g}(CC^\dagger) (Rb - u)\|
        &\leq \|AR^\dagger\bar{g}(CC^\dagger)(Rb - u)\| + \|(A'R'^\dagger - AR^\dagger)\bar{g}(CC^\dagger)Rb\| \\
        &\lesssim \|A\|\|R^\dagger\bar{g}(CC^\dagger)\|\|Rb - u\| + \eps
        \lesssim \|A\|d^2(\eps d^{-2}\|A\|^{-1}) + \eps
        \lesssim \eps.
    \end{align*}
    \item Using $\sq(b)$ and \cref{lem:mmhp}, we approximate $Ab \approx A''b''$ such that $\|Ab - A''b''\| \leq \eps/d$ (and consequently, $q(0)\|Ab - A''b''\| \leq \eps$) with $\bigO{\|A\|_\fr^2\|b\|^2d^2\eps^{-2}\delta^{-1}} = \bigO{d^2\eps^{-2}\delta^{-1}}$ samples.
\end{enumerate}

So, we have shown that $v \coloneqq A'R'^\dagger \bar{g}(CC^\dagger)u + q(0)A''b''$ satisfies $\|v - p^\qsvt(A)\| \lesssim \eps$.
$v$ is a linear combination of columns of $A$; via \cref{lemma:sample-Mv}, we can get $\sq_\phi(v)$ with
\begin{align*}
\sqrun(v)&= \phi\sq_\phi(v)\log\frac{1}{\delta} \\
    &= \bOt{\frac{\sum_i \|A'(\cdot, i)\|^2\|R'(\cdot,i)^\dagger \bar{g}(CC^\dagger)u\|^2 + \sum_j q(0)^2\|A''(\cdot,j)\|^2\|b''(j)\|^2}{\|v\|^2}\Big(\frac{d^4+d^2}{\eps^{2}\delta}\Big)^2\log(mn)} \\
    &= \bOt{\frac{\frac{\|A\|_\fr^2\|R^\dagger\|_\fr^2}{d^4\eps^{-2}\delta^{-1}}(\|\bar{g}(CC^\dagger)R\|\|b\| + \|\bar{g}(CC^\dagger)\|\|Rb-u\|)^2 + q(0)^2\frac{\|A\|_\fr^2\|b\|^2}{d^2\eps^{-2}\delta^{-1}}}{(\|p^\qsvt(A)b\| - \|p^\qsvt(A)b - v\|)^2}\frac{d^8}{\eps^4\delta^2}\log(mn)} \\
    &= \bOt{\frac{d^{-4}(d^2 + d^3\eps d^{-2}\|A\|^{-1})^2 + 1}{(\|p^\qsvt(A)b\| - \eps)^2}\frac{d^8}{\eps^2\delta}\log(mn)} \\
    &= \bOt{d^8\|p^\qsvt(A)b\|^{-2}\eps^{-2}\delta^{-1}\log(mn)}.
\end{align*}
Above, we used that $\eps \lesssim \|p^\qsvt(A)b\| \leq \|A\|\|q(A^\dagger A)\|\|b\| \leq d\|A\|$.
Now, we rescale $\eps \leftarrow \eps\|p^\qsvt(A)b\|$ to get the desired statement.
The runtime is dominated by the sampling for $C$ in $\bigO{rc\log(mn)}$ time, the computation of $\bar{g}(CC^\dagger)$ in $\bigO{r^2c}$ time, and the approximation of $AR^\dagger \approx A'R'^\dagger$ in $\bigO{rd^4\eps^{-2}\delta^{-1}\log(mn)}$ time.
\end{proof}

\begin{remark}\label{rmk:a-to-a-dagger}
Here, we make a brief remark about a technical detail we previously elided.
Technically, QSVT can use $A^\dagger$ in QRAM instead of $A$ (cf.~{\cite[Lemma~50]{gilyen2018QSingValTransf}}), leaving open the possibility that there is a quantum algorithm that does not give an exponential speedup when $A$ is in QRAM, but does when $A^\dagger$ is in QRAM.
We sketch an argument why this is impossible by showing that, given $\sq(A)$, we can simulate $\sq_\phi(B)$ (and $\sq_\phi(B^\dagger)$) for $B$ such that $\|B - A^\dagger\| \leq \eps\|A\|$ with probability $\geq 1-\delta$.
Unfortunately, this argument is fairly involved, so we defer it to \cref{apx:sketch}.
\end{remark}

\subsection{Recommendation systems}\label{sec:recommendation-systems}

Our framework gives a simpler and faster variant of Tang's dequantization \cite{tang2018QuantumInspiredRecommSys} of Kerenidis and Prakash's quantum recommendation systems \cite{kerenidis2016QRecSys}.
Tang's result is notable for being the first result in this line of work and for dequantizing what was previously believed to be the strongest candidate for practical exponential quantum speedups for a machine learning problem \cite{preskill2018QuantCompNISQEra}.

We want to find a product $j \in [n]$ that is a good recommendation for a particular user $i \in [m]$, given incomplete data on user-product preferences.
If we store this data in a matrix $A \in \bbr^{m\times n}$ with sampling and query access, in the strong model described by Kerenidis and Prakash~\cite{kerenidis2016QRecSys}, finding good recommendations reduces to the following:

\begin{prob} \label{prob:rec-systems}
  For a matrix $A \in \bbr^{m\times n}$, given $\sq(A)$ and a row index $i \in [m]$, sample from $\hat{A}(i, \cdot)$ up to $\delta$ error in total variation distance, where $\|\hat{A} - A_{\sigma,\eta}\|_\fr \leq \eps\|A\|_\fr$.
\end{prob}

Here, $A_{\sigma,\eta}$ is a certain type of low-rank approximation to $A$.
The standard notion of low-rank approximation is that of $A_r \coloneqq \sum_{i=1}^r \sigma_iU(\cdot, i)V(\cdot, i)^\dag$, which is the rank-$r$ matrix closest to $A$ in spectral and Frobenius norms.
Using singular value transformation, we define an analogous notion thresholding singular values instead of rank.

\begin{definition}[$A_{\sigma,\eta}$]
We define $A_{\sigma,\eta}$ as a singular value transform of $A$ satisfying:
\begin{align*}
  A_{\sigma,\eta} &\coloneqq P_{\sigma,\eta}^{(\mathrm{SV})}(A) \qquad P_{\sigma, \eta}(\lambda)  \begin{cases} = \lambda & \lambda \geq \sigma(1+\eta) \\ = 0 & \lambda < \sigma(1-\eta) \\ \in [0,\lambda] & \text{otherwise} \end{cases}.
\end{align*}
Note that $P_{\sigma,\eta}$ is not fully specified in the range $[\sigma(1-\eta),\sigma(1+\eta))$, so $A_{\sigma,\eta}$ is any of a family of matrices with error $\eta$.
\end{definition}

For intuition, $P_{\sigma,\eta}^{(\mathrm{SV})}(A)$ is $A$ for (right) singular vectors with value $\geq \sigma(1+\eta)$, zero for those with value $< \sigma(1-\eta)$, and something in between for the rest.
Our analysis simplifies the original algorithm, which passes through the low-rank approximation guarantee of Frieze, Kannan, and Vempala~\cite{frieze2004FastMonteCarloLowRankApx}.

Our algorithm uses that we can rewrite our target low-rank approximation as $A\cdot t(A^\dagger A)$, where $t$ is a smoothened projector.
So, we can use our main theorem, \cref{thm:evenSing}, to approximate $t(A^\dagger A)$ by some $R^\dagger U R$.
Then, the $i^{\text{th}}$ row of our low-rank approximation is $A(i,\cdot) R^\dagger U R$, which is a product of a vector with an RUR decomposition.
Thus, using the sampling techniques described in \cref{subsec:tool_box}, we have $\sq_\phi(A(i,\cdot) R^\dagger U R)$, so we can get the sample from this row as desired.

\begin{corollary}\label{corollary:rec-systems}
Suppose $0 < \eps \lesssim \|A\|/\|A\|_\fr$ and $\eta \leq 0.99$.
A classical algorithm can solve \cref{prob:rec-systems} in time
\[
  \bOt{\frac{K^3\kappa^5}{\eta^6\eps^6}\log^3\frac{1}{\delta} + \frac{K^2\kappa\|A(i,\cdot)\|^2}{\eta^2\eps^2\|\hat{A}(i,\cdot)\|^2}\log^2\frac{1}{\delta}}.
\]
\end{corollary}
The assumption on $\eps$ is a weak non-degeneracy condition in the low-rank regime.
For reference, $\eta = 1/6$ in the application of this algorithm to recommendation systems.
So, supposing the first term of the runtime dominates, the runtime is $\bOt{\frac{\|A\|_\fr^6\|A\|^{10}}{\sigma^{16}\eps^6}\log^3\frac{1}{\delta}}$, which improves on the previous runtime $\bOt{\frac{\|A\|_\fr^{24}}{\sigma^{24}\eps^{12}}\log^3\frac{1}{\delta}}$ of~\cite{tang2018QuantumInspiredRecommSys}.
The quantum runtime for this problem is $\bOt{\frac{\|A\|_\fr}{\sigma}}$, up to $\polylog(m,n)$ terms~\cite[Theorem~27]{chakraborty2018BlockMatrixPowers}.
\begin{proof}
Note that $A_{\sigma,\eta} = A\cdot t(A^\dagger A)$, where $t$ is the thresholding function shown below.
\begin{align*}
t(x) = \begin{cases}
  0 & x < (1-\eta)^2\sigma^2 \\
  \frac{1}{4\eta\sigma^2}(x - (1-\eta)^2\sigma^2) & (1-\eta)^2\sigma^2 \leq x < (1+\eta)^2\sigma^2 \\
  1 & x\geq (1+\eta)^2\sigma^2 \end{cases}.
\end{align*}
We will apply \cref{thm:evenSing} with error parameter $\eps$ to get matrices $R, C$ such that $AR^\dagger \bar{t}(CC^\dagger) R$ satisfies
\begin{align}
  \|A_{\sigma,1/6} - AR^\dagger \bar{t}(CC^\dagger) R\|_\fr &\leq \|A\|_\fr\|t(A^\dagger A) - R^\dagger \bar{t}(CC^\dagger)R\| \leq \eps\|A\|_\fr \label{eqn:recs-svt}.
\end{align}
Since $t(x)$ is $(4\eta\sigma^2)^{-1}$-Lipschitz and $t(x)/x$ is $(4\eta(1-\eta)^2\sigma^4)^{-1}$-Lipschitz, the sizes of $r$ and $c$ are
\begin{align*}
  r &= \bOt{L^2\|A\|^2\|A\|_\fr^2\frac{1}{\eps^2}\log\frac{1}{\delta}}
  = \bOt{\frac{\|A\|^2\|A\|_\fr^2}{\sigma^4\eta^2\eps^2}\log\frac{1}{\delta}}
  = \bOt{\frac{K\kappa}{\eta^2\eps^2}\log\frac1\delta}; \\
  c &= \bOt{\bar{L}^2\|A\|^6\|A\|_\fr^2\frac{1}{\eps^2}\log\frac{1}{\delta}}
  = \bOt{\frac{\|A\|^6\|A\|_\fr^2}{\sigma^8\eta^2\eps^2}\log\frac{1}{\delta}}
  = \bOt{\frac{K\kappa^3}{\eta^2\eps^2}\log\frac{1}{\delta}}.
\end{align*}
So, it suffices to compute the SVD of an $r\times c$ matrix, which has a runtime of
\[
  \bOt{\frac{K^3\kappa^5}{\eta^6\eps^6}\log^3\frac{1}{\delta}}.
\]
Next, we want to approximate $AR^\dagger \approx A'R'^\dagger$.
If we had $\sq(A^\dagger)$ (in particular, if we could compute column norms $\|A(\cdot,j)\|$), we could do this via \cref{prop:appr-mms}, and if we were okay with paying factors of $\frac{1}{\delta}$, we could do this via \cref{lem:mmhp}.
Here, we will instead implicitly define an approximation by approximating each row $[AR^\dagger](i,\cdot) = A(i,\cdot)R^\dagger$ via \cref{prop:appr-mms}, since we then have $\sq(A(i,\cdot)^\dagger)$ and $\sq(R^\dagger)$.
With this proposition, we can estimate $[AR^\dagger](i,\cdot) \approx (A(i,\cdot)S^\dagger S)R^\dagger$ to $\frac{\eps}{\sqrt{K}}\|A(i,\cdot)\|\|R^\dagger\|_\fr = \eps\|A(i,\cdot)\|\sigma$ error using $r' \coloneqq O(\eps^{-2}K\log\frac{1}{\delta})$ samples\footnote{Formally, to get a true approximation $AR \approx A'R$, we need to union bound the failure probability for each row, paying a $\log m$ factor in runtime. However, we will ignore this consideration: our goal is to sample from one row, so we only need to succeed in our particular row.}.
Here, $A'(i,\cdot) \coloneqq A(i,\cdot)S^\dagger S$ is our $r'$-sparse approximation, giving that
\begin{align}
  \|AR^\dagger - A'R^\dagger\|_\fr
  = \sqrt{\textstyle\sum_{i=1}^m \|[AR^\dagger](i,\cdot) - [A'R^\dagger](i,\cdot)\|^2}
  \leq \sqrt{\textstyle\sum_{i=1}^m \eps^2\|A(i,\cdot)\|^2\sigma^2} = \eps\sigma\|A\|_\fr. \label{eqn:recs-mm}
\end{align}
Using this and the observation that $\max_{x} \bar{t}(x) = (1+\eta)^{-2}\sigma^{-2} \leq \sigma^{-2}$, we can bound the quality of our final approximation as
\begin{align*}
  \|\hat{A} - A_{\sigma, \eta}\|_\fr
  &\leq \|(A'R^\dagger - AR^\dagger)\bar{t}(CC^\dagger)R\|_\fr + \|AR^\dagger\bar{t}(CC^\dagger) R - A_{\sigma,\eta}\|_\fr \tag*{by triangle inequality} \\
  &\leq \|A'R^\dagger - AR^\dagger\|_\fr\Big\|\sqrt{\bar{t}(CC^\dagger)}\Big\|\Big\|\sqrt{\bar{t}(CC^\dagger)}R\Big\| + \eps\|A\|_\fr \tag*{by \cref{eqn:recs-svt}} \\
  &\leq \eps\sigma\|A\|_\fr\sigma^{-1}\sqrt{1+\eps} + \eps\|A\|_\fr \lesssim \eps\|A\|_\fr \tag*{by \cref{lem:evenSingBounds,eqn:recs-mm}}.
\end{align*}
We can sample from $\hat{A}(i,\cdot) = A'(i,\cdot)R^\dagger \bar{f}(CC^\dagger) R$ by naively computing $x \coloneqq A'(i,\cdot)R^\dagger \bar{f}(CC^\dagger)$, taking $O(r'r + rc)$ time.
Then, we use \cref{lemma:sample-Mv,lem:b-sq-approx} to get a sample from $xR$ with probability $\geq 1-\delta$ in $\bigO{\sqrun_\phi(xR)}$ time, which is $\bigO{\phi \sqcb_\phi(xR)\log\frac{1}{\delta}}$, where $\sqcb_\phi(xR) = \bigO{r}$ and
\begin{equation*}
  \phi
  = r\frac{\sum_{j=1}^r\abs{x(j)}^2\|R(j,\cdot)\|^2}{\|xR\|^2}
  \lesssim r\frac{\sum_{j=1}^r\abs{x(j)}^2\|A\|_\fr^2}{\|\hat{A}(i,\cdot)\|^2r}
  = \frac{\|x\|^2\|A\|_\fr^2}{\|\hat{A}(i,\cdot)\|^2}.
\end{equation*}
Then, using previously established bounds and bounds from \cref{lem:evenSingBounds}, we have
\begin{align*}
  \frac{\|x\|^2\|A\|_\fr^2}{\|\hat{A}(i,\cdot)\|^2}
  &= \frac{\|A'(i,\cdot)R^\dagger\bar{t}(CC^\dagger)\|^2\|A\|_\fr^2}{\|\hat{A}(i,\cdot)\|^2} \\
  &\leq \Big(\|A(i,\cdot)\|\Big\|R^\dagger\sqrt{\bar{t}(CC^\dagger)}\Big\|\Big\|\sqrt{\bar{t}(CC^\dagger)}\Big\|+\|A(i,\cdot)'R^\dagger - A(i,\cdot)R^\dagger\|\|\bar{t}(CC^\dagger)\|\Big)^2\frac{\|A\|_\fr^2}{\|\hat{A}(i,\cdot)\|^2}  \\
  &\lesssim (\|A(i,\cdot)\|\sigma^{-1}+\eps\sigma\|A(i,\cdot)\|\sigma^{-2})^2\frac{\|A\|_\fr^2}{\|\hat{A}(i,\cdot)\|^2} \\
  &\lesssim \frac{\|A(i,\cdot)\|^2\|A\|_\fr^2}{\|\hat{A}(i,\cdot)\|^2\sigma^2}.
\end{align*}
This sampling procedure and the SVD dominate the runtime.
Since the sampling is exact, the only error in total variation distance is the probability of failure.
\end{proof}

\begin{remark} \label{rmk:spectral-low-rank-approximation}
%
This algorithm implicitly assumes that the important singular values are $\geq \sigma$.
Without such an assumption, we can take $\sigma = \eps\|A\|_\fr$ and $\eta = 1/2$, and have meaningful bounds on the output matrix $\hat{A}$.
Observe that, for $p(x) = x(t(\sqrt{x}) - 1)$,
\begin{align*}
  \|A \cdot t(A^\dagger A) - A\| = \|p^\svt(A)\| \leq \frac{3}{2}\eps\|A\|_\fr.
\end{align*}
So, our low-rank approximation output $\hat{A}$ satisfies $\|\hat{A} - A\| \lesssim \eps\|A\|_\fr$, with no assumptions on $A$, in $\bOt{\frac{\|A\|_\fr^6}{\|A\|^6\eps^{22}}\log^3\frac{1}{\delta}}$ time.
This can be subsequently used to get $\sq_\phi(\hat{A}(i,\cdot)) = \sq_\phi(e_i\hat{A})$ where $\|\hat{A}(i,\cdot) - A(i,\cdot)\| \lesssim \eps\|A\|_\fr$ (in a myopic sense, solving the same problem as \cref{prob:rec-systems}), or more generally, any product of $\hat{A}$ with a vector, in time independent of dimension.
\end{remark}

\subsection{Supervised clustering}\label{sec:supervised-clustering}

The 2013 paper of Lloyd, Mohseni, and Rebentrost~\cite{lloyd2013Clustering} gives two algorithms for the machine learning problem of clustering.
The first algorithm is a simple swap test procedure that was dequantized by Tang \cite{tang2018QInspiredClassAlgPCA} (the second is an application of the quantum adiabatic algorithm with no proven runtime guarantees).
We will reproduce the algorithm from \cite{tang2018QInspiredClassAlgPCA} here: since the dequantization just uses the inner product protocol, so it rather trivially fits into our framework.

We have a dataset of points in $\bbr^d$ grouped into clusters, and we wish to classify a new data point by assigning it to the cluster with the nearest average, aka \emph{centroid}.
We do this by estimating the distance between the new point $p \in \bbr^d$ to the centroid of a cluster of points $q_1,\ldots,q_{n-1} \in \bbr^d$, namely, $\|p - \frac{1}{n-1}(q_1 + \cdots + q_{n-1})\|^2$.
This is equal to $\|wM\|^2$, where
\begin{align*}
M\coloneqq \left[\begin{smallmatrix} p/\|p\| \\ -q_1/(\|q_1\|\sqrt{n-1}) \\ \vdots \\ -q_{n-1}/(\|q_{n-1}\|\sqrt{n-1})\end{smallmatrix}\right]\in \mathbb{R}^{n\times d},\qquad w \coloneqq \left[\|p\|,\frac{\|q_1\|}{\sqrt{n-1}},\dots,\frac{\|q_{n-1}\|}{\sqrt{n-1}}\right]\in \mathbb{R}^{n}.
\end{align*}
Because the quantum algorithm assumes input in quantum states, we can assume sampling and query access to the data points, giving the problem
\begin{prob} \label{prob:cluster}
Given $\sq(M) \in \bbr^{n\times d}, \q(w) \in \bbr^n$, approximate $(wM)(wM)^T$ to additive $\eps$ error with probability at least $1-\delta$. 
\end{prob}

\begin{corollary}[{\cite[Theorem 4]{tang2018QInspiredClassAlgPCA}}]\label{cor:clustering}
There is a classical algorithm to solve \cref{prob:cluster} in $\bigO{\|M\|_\fr^4\|w\|^4\frac{1}{\eps^2}\log\frac{1}{\delta}}$ time.
\end{corollary}

Note that $\|M\|_\fr^2 = 2$ and $\|w\|^2 = \|p\|^2 + \frac{1}{n-1}\sum_{i=1}^{n-1} \|q_i\|^2$.
The quantum algorithm has a quadratically faster runtime of $\bigO{\|M\|_\fr^2\|w\|^2\frac{1}{\eps}}$, ignoring $\polylog(n,d)$ factors~\cite{lloyd2013Clustering,tang2018QInspiredClassAlgPCA}.

\begin{proof}
Recall our notation for the vector of row norms $m\coloneqq \left[\|M(1,\cdot)\|,\dots,\|M(n,\cdot)\|\right]$ coming from \cref{defn:sampling-A}. We can rewrite $(wM)(wM)^T$ as an inner product $\langle u, v\rangle$ where
\begin{align*}
    u &\coloneqq \sum_{i=1}^n\sum_{j=1}^d\sum_{k=1}^n M(i,j)\|M(k,\cdot)\| e_i\otimes e_j\otimes e_k = M \otimes m \\
    v &\coloneqq \sum_{i=1}^n\sum_{j=1}^d\sum_{k=1}^n \frac{w_iw_kM(j,k)}{\|M(k,\cdot)\|} e_i \otimes e_j \otimes e_k,
\end{align*}
where $u$ and $v$ are three-dimension tensors.
By flattening $u$ and $v$, we can represent them as two vectors in $\mathbb{R}^{(n\cdot d\cdot n)\times 1}$.
We clearly have $\q(v)$ from queries to $M$ and $w$.
As for getting $\sq(u)$ from $\sq(M)$: to sample, we first sample $i$ according to $m$, sample $j$ according to $M(i,\cdot)$, and sample $k$ according to $m$; to query, compute $u_{i,j,k} = M(i,j)m(k)$.
Finally, we can apply \cref{lemma:inner-prod} to estimate $\langle u,v\rangle$. $\|u\|=\|M\|^2_\fr$ and $\|v\| = \|w\|^2$, so estimating $\langle u,v\rangle$ to $\eps$ additive error with probability at least $1-\delta$ requires $O(\|M\|^4_\fr\|w\|^4\eps^{-2}\log \frac{1}{\delta})$ samples.
\end{proof}

\subsection{Principal component analysis}\label{sec:PCA}

Principal component analysis (PCA) is an important data analysis tool, first proposed to be feasible via quantum computation by Lloyd, Mohseni, and Rebentrost~\cite{lloyd2013QPrincipalCompAnal}.
Given copies of states with density matrix $\rho = X^\dagger X$, the quantum PCA algorithm can prepare the state $\sum \lambda_i|v_i\rangle \langle v_i| \otimes |\hat{\lambda}_i\rangle \langle \hat{\lambda}_i|$, where $\lambda_i$ and $v_i$ are the eigenvalues and eigenvectors of $X^\dagger X$, and $\hat{\lambda}_i$ are eigenvalue estimates (up to additive error).
See Prakash's PhD thesis~\cite[Section 3.2]{prakash2014QLinAlgAndMLThesis} for a full analysis and Chakraborty, Gily\'{e}n, and Jeffery\ for a faster version of this algorithm in the block-encoding model \cite{chakraborty2018BlockMatrixPowers}.
Directly measuring the eigenvalue register is called \emph{spectral sampling}, but such sampling is not directly useful for machine learning applications.

Though we do not know how to dequantize this protocol exactly, we can dequantize it in the low-rank setting, which is the only useful poly-logarithmic time application that Lloyd, Mohseni, and Rebentrost~\cite{lloyd2013QPrincipalCompAnal} suggests for quantum PCA.

\begin{prob}[PCA for low-rank matrices] \label{prob:pca}
Given a matrix $\sq(X)\in\bbc^{m\times n}$ such that $X^\dagger X$ has top $k$ eigenvalues $\{\lambda_i\}_{i=1}^{k}$ and eigenvectors $\{v_i\}_{i=1}^{k}$, with probability $\geq 1-\delta$, compute eigenvalue estimates $\{\hat{\lambda}_i\}_{i=1}^k$ such that $\sum_{i=1}^k \abs{\hat{\lambda}_i - \lambda_i} \leq \eps\Tr(X^\dagger X)$ and eigenvectors $\{\sq_\phi(\hat{v}_i)\}_{i=1}^k$ such that $\|\hat{v}_i - v_i\| \leq \eps$ for all $i$.
\todo{change this problem to use $\eps\|X\|^2$.}
\end{prob}

Note that we should think of $\lambda_i$ as $\sigma_i^2$, where $\sigma_i$ is the $i$th largest singular value of $X$.
To robustly avoid degeneracy conditions, our runtime must depend on parameters for condition number and spectral gap:
\begin{align}\label{eqn:K-eta-defn}
K \coloneqq \Tr(X^\dagger X)/\lambda_k \geq k\quad\text{and}\quad\eta \coloneqq \min_{i \in [k]} |\lambda_i - \lambda_{i+1}|/\|X\|^2.
\end{align}
We also denote $\kappa \coloneqq \|X\|^2/\lambda_k$.
Dependence on $K$ and $\eta$ are necessary to reduce \cref{prob:pca} to spectral sampling.
If $K = \poly(n)$, then $\lambda_k = \Tr(X^\dagger X)/\poly(n)$, so distinguishing $\lambda_k$ from $\lambda_{k+1}$ necessarily takes $\poly(n)$ samples, and even sampling $\lambda_k$ once takes $\poly(n)$ samples. As a result, learning $v_k$ is also impossible.
A straightforward coupon collector argument (given e.g.\ by Tang \cite{tang2018QInspiredClassAlgPCA}) shows that \cref{prob:pca} can be solved by a quantum algorithm performing spectral sampling\footnote{The quantum analogue to $\sq(X)$ is efficient state preparation of $X$, a purification of $\rho$.}, with runtime depending polynomially on $K$ and $\frac{1}{\eta}$.
We omit this argument for brevity.
Classically, we can solve this PCA problem with quantum-inspired techniques, as first noted in \cite{tang2018QInspiredClassAlgPCA}.

\begin{corollary}\label{cor:PCA}
For $0 < \eps \lesssim \eta\|X\|^2/\|X\|_\fr^2$, we can solve \cref{prob:pca} in $\bOt{\frac{\|X\|_\fr^6}{\lambda_k^2\|X\|^2}\eta^{-6}\eps^{-6}\log^3\frac{k}{\delta}}$ time to get $\sq_\phi(\hat{v}_i)$ where $\sqrun(\hat{v}_i) = \bOt{\frac{\|X\|_\fr^4}{\lambda_i\|X\|^2}\eta^{-2}\eps^{-2}\log^2\frac1\delta}$.
\end{corollary}
This improves significantly over prior work~\cite[Theorem~8]{tang2018QInspiredClassAlgPCA}, which achieves the runtime of $\bOt{\frac{\|X\|_\fr^{36}}{\|X\|^{12}\lambda_k^{12}}\eta^{-6}\eps^{-12}\log^3\frac{k}{\delta}}$.\footnote{This runtime comes from taking $\eps_\sigma = \eps_v = \eps$ and changing the normalization of the gap parameter $\eta = \eta\|X\|^2/\|X\|_\fr^2$ to correspond to the problem as formulated here.}
The best quantum algorithm for this problem runs in $\bOt{\frac{\|X\|_\fr\|X\|}{\lambda_k\eps}}$ time, up to factors of $\polylog(m, n)$~\cite[Theorem~27]{chakraborty2018BlockMatrixPowers}.\!\footnote{Given $X$ in QRAM, this follows from applying Theorem~27 to a quantum state with density matrix of $X^\dagger X$ with $\alpha = \|X\|_\fr$ and $\Delta = \frac{\eps\|X\|_\fr^2}{\|X\|} \lesssim \frac{\eta\|X\|}{\|X\|_\fr}$. The output is some estimate of $\sqrt{\lambda_i}$ to $\Delta$ error, which when squared is an estimate of $\lambda_i$ to $\Delta \|X\| = \eps\|X\|_\fr^2$ error as desired.
Then, the density matrix is a probability distribution over eigenvectors with their corresponding eigenvalue estimate (which is enough to identify the eigenvector).
The coupon collector argument mentioned above gives us access to all the top $k$ eigenvalues and eigenvectors by running this algorithm $\|X\|_\fr^2/\lambda_k$ times~\cite{tang2018QInspiredClassAlgPCA}.}

We approach the problem as follows.
First, we use that an importance-sampled submatrix of $X$ has approximately the same singular values as $X$ itself (\cref{prop:appr-svs}) to get our estimates $\{\hat{\lambda}_i\}_{i=1}^k$.
With these estimates, we can define smoothened step functions $f_i$ for $i \in [k]$ such that $f_i(X^\dagger X) = v_i^\dagger v_i$.
We can then use our main theorem to find an RUR decomposition for $f_i(X^\dagger X)$.
We use additional properties of the RUR description to argue that it is indeed a rank-1 outer product $\hat{v}_i^\dagger \hat{v}_i$, which is our desired approximation for the eigenvector.
We have sampling and query access to $\hat{v}_i$ because it is $R^\dagger x$ for some vector $x$.
Our runtime is quite good because these piecewise linear step functions have relatively tame derivatives, as opposed to the thresholded inverse function, whose Lipschitz constants must incur quadratic and cubic overheads in terms of condition number.

\begin{proof}
We will assume that we know $\lambda_k$ and $\eta$.
If both are unknown, then we can estimate them with the singular value estimation procedure described below (\cref{prop:appr-svs}).

Notice that $\eta\|X\|^2 \leq \lambda_k$ follows from our definition of $\eta$.
The algorithm will proceed as follows: first, consider $C \coloneqq SXT \in \bbc^{c\times r}$ as described in \cref{thm:evenSing}, with parameters
\begin{align*}
  r &\coloneqq \bOt{\frac{\|X\|_\fr^2}{\eta^2\|X\|^2\eps^2}\log\frac{k}{\delta}}
  \qquad c\coloneqq \bOt{\frac{\|X\|_\fr^2\|X\|^2}{\eta^2\lambda_k^2\eps^2}\log\frac{k}{\delta}}.
\end{align*}
Consider computing the eigenvalues of $CC^\dagger$; denote the $i^{\text{th}}$ eigenvalue $\hat{\lambda}_i$.
Since $r,c = \Omega(\frac{\|X\|_\fr^2}{\lambda_k\eps^2}\log\frac{1}{\delta})$, by \cref{prop:appr-svs} with error parameter $\frac{\eps\sqrt{\lambda_k}}{8\|X\|_\fr}$, with probability $\geq 1-\delta$,
\begin{align*}
  \sqrt{\sum\nolimits_{i=1}^{\min(m,n)} (\hat{\lambda}_i - \lambda_i)^2} &\leq \frac{\eps\sqrt{\lambda_k}}{8\|X\|_\fr}\|X\|_\fr^2.
\end{align*}
These $\hat{\lambda}_i$'s for $i \in [k]$ have the desired property for eigenvalue estimates:
\begin{align*}
  \sum_{i=1}^k |\hat{\lambda}_i - \lambda_i|
  \leq \sqrt{k}\sqrt{\sum\nolimits_{i=1}^{k} (\hat{\lambda}_i - \lambda_i)^2}
  \leq \eps\sqrt{k\lambda_k}\|X\|_\fr \leq \eps\|X\|_\fr^2.
\end{align*}
This bound also implies that, for all $i$, $|\hat{\lambda}_i - \lambda_i| \leq \frac{\eps}{8}\|X\|_\fr^2$.
Next, consider the eigenvalue transformations $f_i$ for $i \in [k]$, defined
\[
  f_i(x) \coloneqq \begin{cases}
    0 & x - \hat{\lambda}_i < - \frac14\eta\|X\|^2 \\
    2 + \frac{8}{\eta\|X\|^2}(x - \hat{\lambda}_i) & - \frac14\eta\|X\|^2 \leq x - \hat{\lambda}_i < - \frac18\eta\|X\|^2 \\
    1 & -\frac18\eta\|X\|^2 \leq x - \hat{\lambda}_i < \frac18\eta\|X\|^2 \\
    2 - \frac{8}{\eta\|X\|^2}(x - \hat{\lambda}_i) & \frac18\eta\|X\|^2 \leq x - \hat{\lambda}_i < \frac14\eta\|X\|^2 \\
    0 & \frac14\eta\|X\|^2 \leq x - \hat{\lambda}_i
  \end{cases}.
\]
This is a function that is one when $|x - \hat{\lambda}_i| \leq \frac18\eta\|X\|^2$, zero when $|x - \hat{\lambda}_i| \geq \frac14\eta\|X\|^2$, and interpolates between them otherwise.
From the eigenvalue gap and the aforementioned bound $|\hat{\lambda}_i - \lambda_i| \leq \frac18\eta\|X\|^2$, we can conclude that $f_i(X^\dagger X) = v_iv_i^\dagger$ exactly.
Further, by \cref{thm:evenSing}, we can conclude that $R^\dagger \bar{f}_i(CC^\dagger) R$ approximates $v_iv_i^\dagger$, with $C, R$ the exact approximations used to estimate singular values.
The conditions of \cref{thm:evenSing} are satisfied because $\eps \lesssim 8 \leq \frac{8}{\eta} = L\|X\|^2$ for $L$ the Lipschitz constant of $f_i$.
The values of $r,c$ are chosen so that $\|R^\dagger \bar{f}_i(CC^\dagger) R - f_i(X^\dagger X)\| \leq \eps/2$ (note $f_i(0) = 0$):
\begin{align*}
  r &= \bOt{L^2\|X\|^2\|X\|_\fr^2\frac{1}{\eps^2}\log\frac1\delta}
  = \bOt{\frac{\|X\|_\fr^2}{\|X\|^2\eta^2\eps^2}\log\frac{1}{\delta}} \\
  c &= \bOt{\bar{L}^2\|X\|^6\|X\|_\fr^2\frac{1}{\eps^2}\log\frac{1}{\delta}}
  = \bOt{\frac{\|X\|^6\|X\|_\fr^2}{\eta^2\|X\|^4(\hat{\lambda}_i - \frac14\eta\|X\|^2)^2\eps^2}\log\frac{1}{\delta}}
  = \bOt{\frac{\|X\|^2\|X\|_\fr^2}{\eta^2\lambda_k^2\eps^2}\log\frac{1}{\delta}}.
\end{align*}
Further, $f_i$ is chosen with respect to $\hat{\lambda}_i$ such that $R^\dagger \bar{f}_i(CC^\dagger)R$ is rank one, since $CC^\dagger$ has one eigenvalue between $\hat{\lambda}_i - \frac{1}{4}\eta\|X\|^2$ and $\hat{\lambda}_i + \frac{1}{4}\eta\|X\|^2$.
Thus, this approximation is an outer product, $R^\dagger \bar{f}_i(CC^\dagger)R = \hat{v}_i\hat{v}_i^\dagger$, and we take the corresponding vector to be our eigenvector estimate: $\|\hat{v}_i\| \leq \sqrt{1+\eps/2} \leq 1+\eps/4$, so
\begin{align*}
  \eps/2 &\geq \|(\hat{v}_i\hat{v}_i^\dagger - v_iv_i^\dagger)v_i\| \tag*{by definition} \\
  &= \|\langle \hat{v}_i, v_i\rangle \hat{v}_i - v_i\| \tag*{by $\|v_i\|^2 = 1$} \\
  &\geq \|\hat{v}_i - v_i\| - (\langle \hat{v}_i, v_i\rangle - 1)\|\hat{v}_i\| \tag*{by triangle inequality} \\
  &\geq \|\hat{v}_i - v_i\| - (\|\hat{v}_i\|\|v_i\| - 1)\|u\| \tag*{by Cauchy--Schwarz} \\
  &\geq \|\hat{v}_i - v_i\| - (1+\eps/4 - 1)(1+\eps/4) \tag*{by $\|\hat{v}_i\| \leq 1+\eps/4$} \\
  &\geq \|\hat{v}_i - v_i\| - \eps/2,
\end{align*}
which is the desired bound.
By choosing failure probability $\delta/k$, the bound can hold true for all $k$ with probability $\geq 1-\delta$.

Finally, we can get access to $\hat{v}_i = R^\dagger \bar{v}_i$, where $\bar{v}_i \in \bbc^r$ satisfies $\bar{v}_i^\dagger\bar{v}_i = \bar{f}_i(CC^\dagger)$.
Since $\|\bar{v}_i^\dagger\| \leq \sqrt{\max_x \bar{f}_i(x)} \lesssim \lambda_i^{-\frac12}$, using \cref{lem:b-sq-approx,lemma:sample-Mv}, we have $\sq_\phi(\hat{v}_i)$ with
\begin{align*}
  \phi &= r\frac{\sum_{s=1}^r\abs{\hat{v}_i(s)}^2\|R(s,\cdot)\|^2}{\|R^\dagger \bar{v}_i\|^2}
  = r\frac{\sum_{s=1}^r\abs{\hat{v}_i(s)}^2\|X\|_\fr^2}{\|R^\dagger \bar{v}_i\|^2r}
  = \frac{\|\hat{v}_i\|^2\|X\|_\fr^2}{\|R^\dagger \bar{v}_i\|^2}
  \lesssim \frac{\|X\|_\fr^2}{\lambda_i(1-\eps)^2}
  \lesssim \frac{\|X\|_\fr^2}{\lambda_i},
\end{align*}
so $\sqrun_\phi(\hat{v}_i) = \phi\sqcb_\phi(v)\log\frac{1}{\delta} \lesssim \frac{\|X\|_\fr^2}{\lambda_i}r\log\frac{1}{\delta}$.
\end{proof}

\subsection{Matrix inversion and principal component regression}\label{sec:matrix-inversion}

The low-rank matrix inversion algorithms given by Gily\'{e}n, Lloyd, and Tang~\cite{gilyen2018QInsLowRankHHL} and Chia, Lin, and Wang~\cite{chia2018QInspiredSubLinLowRankLinEqSolver} dequantize Harrow, Hassidim, and Lloyd's quantum matrix inversion algorithm (HHL) \cite{harrow2009QLinSysSolver} in the regime where the input matrix is \emph{low-rank} instead of sparse.
The corresponding quantum algorithm in this regime is given by Chakraborty, Gily\'{e}n, and Jeffery~\cite{chakraborty2018BlockMatrixPowers}, among others.
Since sparse matrix inversion is BQP-complete, it is unlikely that one can efficiently dequantize it.
However, the variant of low-rank (non-sparse) matrix inversion appears often in quantum machine learning \cite{prakash2014QLinAlgAndMLThesis,wossnig2018QLinSysAlgForDensMat,rebentrost2014QSVM,cong2016quantum,rebentrost2018QuantumFinance}, making it an influential primitive in its own right.

Using our framework, we can elegantly derive the low-rank matrix inversion algorithm in a manner similar to prior quantum-inspired work \cite{gilyen2018QInsLowRankHHL,chia2018QInspiredSubLinLowRankLinEqSolver}.
Moreover, we can also handle the approximately low-rank regime and only invert the matrix on a well-conditioned subspace, solving principal component regression---for more discussion see~\cite{gilyen2018QSingValTransf}.
Namely, we can find a thresholded pseudoinverse of an input matrix:
\begin{definition}[$A_{\sigma,\eta}^+$] \label{def:psinv}
We define $A_{\sigma,\eta}^+$ to be any singular value transform of $A$ satisfying:
\begin{align}
  A_{\sigma,\eta}^+ &\coloneqq \operatorname{tinv}_{\sigma,\eta}^{(\mathrm{SV})}(A) \qquad \operatorname{tinv}_{\sigma, \eta}(\lambda)  \begin{cases} = 1/\lambda & \lambda \geq \sigma \\ = 0 & \lambda < \sigma(1-\eta) \\ \in [0,\sigma^{-1}] & \text{otherwise} \end{cases}.
\end{align}
\end{definition}
This definition is analogous to $A_{\sigma,\eta}$ in \cref{sec:recommendation-systems}: it is $A^+$ for singular vectors with value $\geq \sigma$, zero for singular vectors with value $\leq \sigma(1-\eta)$, and a linear interpolation between the two in between.

\begin{prob} \label{prob:inv}
Given $\sq_\varphi(A) \in \mathbb{C}^{m\times n}, \q(b) \in \mathbb{C}^m$, with probability $\geq 1-\delta$, get $\sq_\phi(\hat{x})$ such that $\|\hat{x} - x^*\| \leq \eps\|A\|^{-1}\|b\|$, where $x^* \coloneqq A_{\sigma,\eta}^+b$.
\end{prob}

\begin{corollary} \label{cor:inv}
For $0 < \eps \lesssim \frac{\|A\|^2}{\sigma^2}$ and $\eta \leq 0.99$, we can solve \cref{prob:inv} in $\bOt{\frac{\varphi^6 K^3\kappa^{11}}{\eta^6\eps^6}\log^3\frac1\delta}$ time to give $\sq_{\phi}(\hat{x})$ for $\sqrun_\phi(\hat{x}) = \bOt{\frac{\varphi^4 K^2\kappa^5}{\eta^2\eps^2}\frac{\|x^*\|^2}{\|\hat{x}\|^2}\log^2\frac{1}{\delta}}$.
\end{corollary}

This should be compared to \cite{gilyen2018QInsLowRankHHL}, which applies only to strictly rank-$k$ $A$ with $\varphi = 1$ and gets the incomparable runtime of $\bOt{\frac{K^3\kappa^8k^6}{\eta^6\eps^6}\log^3\frac{1}{\delta}}$.
The corresponding quantum algorithm using block-encodings takes $\bigO{\|A\|_\fr/\sigma}$ time, up to $\polylog(m,n)$ factors, to get this result for constant $\eta$~\cite[Theorem~41]{gilyen2018QSingValTransf}.

If we further assume that $\eps < 0.99$ and $b$ is in the image of $A$, then $\sqrun_\phi(\hat{x})$ can be simplified, since $\|\hat{x}\| \geq \|x^*\| - \eps\|A\|^{-1}\|b\| \geq (1-\eps)\|x^*\|$, so $\frac{\|x^*\|}{\|\hat{x}\|} \leq 100$.
However, this algorithm also works for larger $\eps$; namely, if we only require that $\|\hat{x} - x^*\| \leq \eps\sigma^{-1}\|b\|$ (a ``worst-case'' error bound), then this algorithm works with runtime smaller by a factor of $\kappa^3$ (and $\sqrun_\phi(\hat{x})$ smaller by a factor of $\kappa$).

The algorithm comes from rewriting $A_{\sigma,\eta}^+b = \iota(A^\dagger A)A^\dagger b$ for $\iota$ a function encoding a thresholded inverse.
Namely, $\iota(x) = 1/x$ for $x \geq \sigma^2$, $\iota(x) = 0$ for $x \leq (1-\eta)^2\sigma^2$, and is a linear interpolation between the endpoints for $x \in [(1-\eta)^2\sigma^2, \sigma^2]$.
By our main theorem, we can find an RUR decomposition for $\iota(A^\dagger A)$, from which we can then get $\sq(R^\dagger URA^\dagger b)$ via sampling techniques.

\begin{proof}
We will solve our problem for $x^* = A_{\sigma,\eta}^+b = \iota(A^\dagger A)A^\dagger b$ where
\[
  \iota(x) \coloneqq \begin{cases}
    0 & x < \sigma^2(1-\eta)^2 \\
    \frac{1}{(2\eta-\eta^2)\sigma^4}(x - \sigma^2(1-\eta)^2) & \sigma^2(1-\eta)^2 \leq x < \sigma^2 \\
    \frac{1}{x} & \sigma^2 \leq x
  \end{cases}.
\]
So, if we can estimate $\iota(A^\dagger A)$ such that $\|\iota(A^\dagger A) - R^\dagger \bar{\iota}(CC^\dagger) R\| \leq \frac{\eps}{\|A\|^2}$, then as desired,
\[
  \|A_{\sigma,\eta}^+b - R^\dagger \bar{\iota}(CC^\dagger) R A^\dagger b\| \leq \frac{\eps}{\|A\|}\|b\| \leq \eps\|A_{\sigma,\eta}^+b\|.
\]
By \cref{thm:evenSing} with $L = \frac{1}{(2\eta - \eta^2)\sigma^4}$ and $\bar{L} = \frac{1}{(1-\eta)^2(2\eta-\eta^2)\sigma^6}$, we can find such $R$ and $C$ with
\begin{align*}
  r &
  = \bOt{\varphi^2\frac{\|A\|^2\|A\|_\fr^2}{(2\eta-\eta^2)^2\sigma^8\frac{\eps^2}{\|A\|^4}}\log\frac1\delta}
  = \bOt{\frac{\varphi^2K\kappa^3}{\eta^2\eps^2}\log\frac1\delta} \\
  c &= \bOt{\varphi^2\frac{\|A\|^6\|A\|_\fr^2}{(1-\eta)^4(2\eta-\eta^2)^2\sigma^{12}\frac{\eps^2}{\|A\|^4}}\log\frac1\delta}
  = \bOt{\frac{\varphi^2K\kappa^5}{\eta^2\eps^2}\log\frac1\delta}.
\end{align*}
Computing the SVD of a matrix of this size dominates the runtime, giving the complexity in the theorem statement.
Next, we would like to further approximate $R^\dagger \bar{\iota}(CC^\dagger) RA^\dagger b$.
We will do this by estimating $RA^\dagger b$ by some vector $u$ to $\eps\sigma^3\|A\|^{-1}\|b\| = \eps\|A\|_\fr^2\|b\|K^{-1}\kappa^{-\frac12}$ error, since then, using the bounds from \cref{lem:evenSingBounds},
\begin{align*}
  \|R^\dagger \bar{\iota}(CC^\dagger) RA^\dagger b - R^\dagger \bar{\iota}(CC^\dagger) u\| &\leq \Big\|R^\dagger\sqrt{\bar{\iota}(CC^\dagger)}\Big\|\Big\|\sqrt{\bar{\iota}(CC^\dagger)}\Big\|\|RA^\dagger b - u\| \\
  &\lesssim \sqrt{\textstyle\sigma^{-2} + \frac{\eps}{\|A\|^2}}\sigma^{-2}(\eps\sigma^3\|A\|^{-1}\|b\|)
  \lesssim \eps\|A\|^{-1}\|b\|.
\end{align*}
We use \cref{lemma:xAy} to estimate $u(i) = R(i,\cdot)A^\dagger b$, for all $i \in [r]$, to $\eps\|R(i,\cdot)\|\|A\|_\fr\|b\|K^{-1}\kappa^{-\frac12}$ error, with probability $\geq 1-\delta/r$.
This takes $\bigO{\varphi\frac{K^2\kappa}{\eps^2}\log\frac{r}{\delta}}$ samples for each of the $r$ entries.
This implies that $\hat{x} \coloneqq R^\dagger \bar{\iota}(CC^\dagger)u$ has the desired error and failure probability.
Finally, we can use \cref{lem:b-sq-approx,lemma:sample-Mv} with matrix $R^\dagger$ and vector $\bar{\iota}(CC^\dagger)u$ to get $\sq_\phi(\hat{x})$ for
\begin{align*}
  \phi &= \varphi r\frac{\sum_{s=1}^r\abs{[\bar{\iota}(CC^\dagger)u](s)}^2\|R(s,\cdot)\|^2}{\|\hat{x}\|^2} \\
  &= \varphi^2\frac{\|\bar{\iota}(CC^\dagger)u\|^2\|A\|_\fr^2}{\|\hat{x}\|^2} \tag*{by $\|R(s,\cdot)\| \leq \|A\|_\fr\sqrt{\varphi/r}$} \\
  &\leq \varphi^2\frac{(\|\bar{\iota}(CC^\dagger)R\|\|A^\dagger\|\|b\| + \|\bar{\iota}(CC^\dagger)\|\|RA^\dagger b - u\|)^2\|A\|_\fr^2}{\|\hat{x}\|^2} \tag*{by linear algebra} \\
  &\lesssim \varphi^2\frac{(\sigma^{-3}\|A\|\|b\| + \sigma^{-4}\eps\sigma^3\|b\|/\|A\|)^2\|A\|_\fr^2}{\|\hat{x}\|^2} \tag*{by prior bounds} \\
  &\lesssim \varphi^2\frac{\sigma^{-6}\|A\|^2\|b\|^2\|A\|_\fr^2}{\|\hat{x}\|^2} \tag*{by $\eps \lesssim \|A\|^2/\sigma^2$} \\
  &\leq \varphi^2K\kappa^2\frac{\|x^*\|^2}{\|\hat{x}\|^2}, \tag*{by $\|A\|^{-1}\|b\| \leq \|x^*\|$}
\end{align*}
so $\sqrun_\phi(\hat{x}) = \phi\sqcb_\phi(\hat{x})\log\frac{1}{\delta} = \bigO{r\varphi^2 K\kappa^2\frac{\|x^*\|^2}{\|\hat{x}\|^2}\log\frac{1}{\delta}}$.
\end{proof}

\subsection{Support vector machines}\label{sec:SVM}
In this section, we use our framework to dequantize Rebentrost, Mohseni, and Lloyd's quantum support vector machine \cite{rebentrost2014QSVM}, which was previously noted to be possible by Ding, Bao, and Huang~\cite{ding2019SVM}.
Mathematically, the support vector machine is a simple machine learning model attempting to label points in $\bbr^m$ as $+1$ or $-1$. Given input data points $x_1, \ldots, x_m \in \bbr^n$ and their corresponding labels $y \in \{\pm 1\}^m$. Let $w \in \bbr^n$ and $b \in \bbr$ be the specification of hyperplanes separating these points. It is possible that no such hyperplane satisfies all the constraints. To resolve this, we add a slack vector $e \in \bbr^m$ such that $e(j) \geq 0$ for $j \in [m]$. We want to minimize the squared norm of the residuals:
\begin{align*}
  \min_{w, b} &\quad \frac{1}{2}\frac{\|w\|^2}{2} + \frac{\gamma}{2}\|e\|^2 \\
  \text{s.t.} &\quad y(i)(w^Tx_i+ b) = 1 - e(i), \quad \forall i \in [m].
\end{align*}
The dual of this problem is to maximize over the Karush-Kuhn-Tucker multipliers of a Lagrange function, taking partial derivatives of which yields the linear system
\begin{align}
\normalsize
    \label{eq:svm}
    \begin{bsmallmatrix}
      0 & \vec{1}^T \\
      \vec{1} & XX^T + \gamma^{-1}I
    \end{bsmallmatrix}
    \begin{bsmallmatrix} b \\ \alpha \end{bsmallmatrix} = \begin{bsmallmatrix} 0 \\ y \end{bsmallmatrix},
\end{align}
where $\vec{1}$ is the all-ones vector and $X = \{x_1, \ldots, x_m\} \in \bbc^{m \times n}$.
Call the above $m+1\times m+1$ matrix $F$, and $\hat{F} \coloneqq F/\Tr(F)$.

The quantum algorithm, given $X$ and $y$ in QRAM, outputs a quantum state $|\hat{F}_{\lambda,0.01}^+[\begin{smallmatrix} 0 \\ y \end{smallmatrix}]\rangle$ (\cref{def:psinv}) in $\bOt{\frac{1}{\lambda^3\eps^3}\polylog(mn)}$ time.
The quantum-inspired analogue is as follows.
\begin{prob}
  \label{prob:svm}
  Given $\sq(X) \in \bbr^{m\times n}$ and $\sq(y) \in \bbr^{m}$, for $\|\hat{F}\| \leq 1$, output $\sq_\phi(v) \in \bbr^{m+1}$ such that $\|\hat{x} - \hat{F}_{\lambda,\eta}^+[\begin{smallmatrix} 0 \\ y \end{smallmatrix}]\| \leq \eps\|\hat{F}_{\lambda,\eta}^+[\begin{smallmatrix} 0 \\ y \end{smallmatrix}]\|$ with probability $\geq 1-\delta$.
\end{prob}
Note that we must assume $\|\hat{F}\| \leq 1$; the quantum algorithm makes the same assumption\footnote{The algorithm as written in \cite{rebentrost2014QSVM} assumes that $\|F\| \leq 1$; we confirmed with an author that this is a typo.}.
Another dequantization was reported in~\cite{ding2019SVM}, which, assuming $X$ is strictly low-rank (with minimum singular value $\sigma$), outputs a description of $(XX^T)^+y$ that can be used to classify points.
This can be done neatly in our framework: express $(XX^T)^+$ (or, more generally, $(XX^T)_{\sigma,\eta}^+$) as $X f(X^TX) X^T$ for the appropriate choice of $f$.
Then, use \cref{thm:evenSing} to approximate $f(X^TX) \approx R^T ZR$ and use \cref{lem:mmhp} to approximate $XR^T \approx CW^T$.
This gives an approximate ``CUC'' decomposition of the desired matrix, since $Xf(X^TX)X^T \approx XR^TZRX^T \approx CW^TZWC^T$, which we can use for whatever purpose we like.

For our solution to \cref{prob:svm}, though, we simply reduce to matrix inversion as described in \cref{sec:matrix-inversion}: we first get $\sq_\phi(\hat{F})$, and then we apply \cref{cor:inv} to complete.
Section VI.C of~\cite{ding2019SVM} claims to dequantize this version, but gives no correctness bounds\footnote{The correctness of this dequantization is unclear, since the approximations performed in this section incur significant errors.} or runtime bounds (beyond arguing it is polynomial in the desired parameters).

\begin{corollary}\label{cor:SVM}
  For $0 < \eps \lesssim 1$ and $\eta \leq 0.99$, we can solve \cref{prob:svm} in $\bOt{\lambda^{-28}\eta^{-6}\eps^{-6}\log^3\frac{1}{\delta}}$ time, where we get $\sq_\phi(v)$ for $\sqrun_\phi(v) = \bOt{\lambda^{-14}\eta^{-2}\eps^{-4}\log^2(\frac{1}{\delta})\log(\frac{m}{\delta})}$.
\end{corollary}
The runtimes in the statement are not particularly tight, but we chose the form to mirror the runtime of the QSVM algorithm, which similarly depends polynomially on $\frac{1}{\lambda}$ and $\frac{1}{\eta}$.
\begin{proof}
  Consider constructing $\sq_\varphi(K) \in \bbc^{m\times m}$ as follows.
  To query an entry $K(i,j)$, we estimate $X(i,\cdot)X(j,\cdot)^T$ to $\eps\|X(i,\cdot)\|\|X(j,\cdot)\|$ error.
  We define $K(i,j)$ to be this estimate.
  Using \cref{lemma:inner-prod}, we can do this in $\bigO{\frac{1}{\eps^2}\log\frac{q}{\delta}}$ time.
  $q$ here refers to the number of times the query oracle is used, so in total the subsequent algorithm will only have an errant query with probability $\geq 1-\delta$. ($q$ will not appear in the runtime because it's folded into a polylog term.)
  Then, we can take $\tilde{K} \coloneqq xx^T$, where $x \in \bbr^m$ is the vector of row norms of $X$, since by Cauchy--Schwarz,
  \begin{align*}
    K(i,j) \leq X(i,\cdot)X(j,\cdot)^T + \eps\|X(i,\cdot)\|\|X(j,\cdot)\| \leq (1+\eps)\|X(i,\cdot)\|\|X(j,\cdot)\| = \tilde{K}(i,j).
  \end{align*}
  Since we have $\sq(x)$ from $\sq(X)$, we have $\sq(\tilde{K})$ with $\sqcb(\tilde{K}) = \bigO{1}$ by \cref{lem:outersampling}.
  $\|\tilde{K}\|_\fr^2 = (1+\eps)^2\|X\|_\fr^4$, so we have $\sq_\varphi(K)$ for $\varphi = (1+\eps)^2\frac{\|X\|_\fr^4}{\|K\|_\fr^2}$.
  We can trivially get $\sq(L)$ for $L \coloneqq \big[\begin{smallmatrix}0 & \vec{1}^T \\ \vec{1} & \gamma^{-1}I\end{smallmatrix}\big]$ with $\sqcb(L) = \bigO{1}$.
  Our approximation to $\hat{F}$ is
  \begin{align*}
    M \coloneqq \frac{1}{\Tr(F)}\Big(L + \Big[\begin{smallmatrix}0 & \vec{0}^T \\ \vec{0} & K \end{smallmatrix}\Big]\Big); \qquad
    \|M - \hat{F}\| \leq \frac{1}{\Tr(F)}\|K - XX^T\|_\fr \leq \frac{1}{\Tr(F)}\eps\|X\|_\fr^2 \leq \eps.
  \end{align*}
  Using \cref{lem:weighted-oversampling}, we have $\sq_{\varphi'}(M)$ with
  \begin{align*}
    \varphi' = \frac{2((1+\eps)^2\frac{\|X\|_\fr^4}{\|K\|_\fr^2}\|K\|_\fr^2+\|L\|_\fr^2)}{\Tr(F)^2\|M\|_\fr^2}
    \lesssim \frac{\|X\|_\fr^4+\gamma^{-2}m+2m}{(\|X\|_\fr^2 + m\gamma^{-1})^2\|M\|_\fr^2}
    \lesssim \frac{1}{\|M\|_\fr^2},
  \end{align*}
  where the last inequality uses that $\Tr(F) \geq \sqrt{m}$, which follows from $\|\hat{F}\| \leq 1$:
  \begin{align*}
    1 = \|\hat{F}\|\big\|\big[\begin{smallmatrix}0\\ \vec{1}/\sqrt{m} \end{smallmatrix}\big]\big\|
    \geq \big\|\hat{F}\big[\begin{smallmatrix}0\\ \vec{1}/\sqrt{m} \end{smallmatrix}\big]\big\|
    \geq \frac{\sqrt{m}}{\Tr(F)}.
  \end{align*}
  Note that we can compute $\Tr(F)$ given $\sq(X)$.
  So, applying \cref{cor:inv}, we can get the desired $\sq_\phi(v)$ in runtime
  \begin{align*}
  \bOt{\frac{\varphi^6\|M\|_\fr^6\|M\|^{22}}{\lambda^{28}\eta^6\eps^6}\log^3\frac{1}{\delta}}
  \lesssim \bOt{\frac{\|M\|^{22}}{\|M\|_\fr^{6}\lambda^{28}\eta^6\eps^6}\log^3\frac{1}{\delta}}
  \lesssim \bOt{\frac{1}{\lambda^{28}\eta^6\eps^6}\log^3\frac{1}{\delta}}.
  \end{align*}
  Here, we used that $\|M\| \leq \|M\|_\fr \lesssim 1$, which we know since $\varphi' \geq 1$ (by our definition of oversampling and query access).
  That $\q(M) = \bigO{\frac{1}{\eps^2}\log\frac{q}{\delta}}$ does not affect the runtime, since the dominating cost is still the SVD.
  On the other hand, this does come into play for the runtime for sampling:
  \begin{align*}
    \sqrun_\phi(v) = \bOt{\frac{\varphi^4\|M\|_\fr^4\|M\|^{10}}{\eta^2\eps^2}\log^2\big(\frac{1}{\delta}\big)\frac{1}{\eps^2}\log\big(\frac{m}{\delta}\big)}.
  \end{align*}
  We take $q = m$ to guarantee that all future queries will be correct with probability $\geq 1-\delta$.
\end{proof}

The normalization used by the quantum and quantum-inspired SVM algorithms means that these algorithms fail when $X$ has too small Frobenius norm, since then the singular values from $XX^T$ are all filtered out.
In \cref{apx:proofs}, we describe an alternative method that relies less on normalization assumptions, instead simply computing $F^+$.
This is possible if we depend on $\|X\|_\fr^2\gamma$ in the runtime.
Recall from \cref{eq:svm} that we regularize by adding $\gamma^{-1}I$, so $\gamma^{-1}$ acts as a singular value lower bound and $\|X\|_\fr^2\gamma$ implicitly constrains.

\begin{restatable}{corollary}{svmrev}\label{cor:SVM2}
  Given $\sq(X^T)$ and $\sq(y)$, with probability $\geq 1-\delta$, we can output a real number $\hat{b}$ such that $\abs{b - \hat{b}} \leq \eps(1+b)$ and $\sq_\phi(\hat{\alpha})$ such that $\|\hat{\alpha} - \alpha\| \leq \eps\gamma\|y\|$, where $\alpha$ and $b$ come from \cref{eq:svm}.
  Our algorithm runs in $\bOt{\|X\|_\fr^6\|X\|^{16}\gamma^{11}\eps^{-6}\log^3\frac1\delta}$ time, with $\sqrun_\phi(\hat{\alpha}) = \bOt{\|X\|_\fr^4\|X\|^6\gamma^5\frac{\gamma^2m}{\|\hat{\alpha}\|^2}\eps^{-4}\log^2\frac{1}{\delta}}$.
  Note that when $\gamma^{-1/2}$ is chosen to be sufficiently large (e.g.\ $O(\|X\|_\fr)$) and $\|\alpha\| = \Omega(\gamma\|y\|)$, this runtime is dimension-independent.
\end{restatable}

Notice that $\eps\gamma\|y\|$ is the right notion, since $\gamma$ is an upper bound on the spectral norm of the inverse of the matrix in \cref{eq:svm}.
We assume $\sq(X^T)$ instead of $\sq(X)$ for convenience, though both are possible via the observation that $f(XX^T) = X\bar{f}(X^TX)X^T$.

\subsection{Hamiltonian simulation}\label{sec:Hamiltonian-simulation}

The problem of simulating the dynamics of quantum systems was the original motivation for quantum computers proposed by Feynman~\cite{feynman1982SimQPhysWithComputers}.
Specifically, given a Hamiltonian $H$, a quantum state $|\psi\rangle$, a time $t>0$, and a desired error $\eps > 0$, we ask to prepare a quantum state $|\psi_{t}\rangle$ such that
\begin{align*}
\||\psi_{t}\rangle-e^{iHt}|\psi\rangle\|\leq\eps.
\end{align*}
This problem, known as Hamiltonian simulation, sees wide application, including in quantum physics and quantum chemistry.
A rich literature has developed on quantum algorithms for Hamiltonian simulation~\cite{lloyd1996UnivQSim,aharonov2003adiabatic,berry2015HamSimNearlyOpt}, with an optimal quantum algorithm for simulating sparse Hamiltonians given in~\cite{low2016HamSimQSignProc}.
In this subsection, we apply our framework to develop classical algorithms for Hamiltonian simulation.
Specifically, we ask:
\begin{prob} \label{prob:Hamiltonian-simulation}
Consider a Hermitian matrix $H \in \bbc^{n\times n}$, a unit vector $b\in\bbc^{n}$, and error parameters $\eps, \delta > 0$.
Given $\sq(H)$ and $\sq(b)$, output $\sq_\phi(\hat{b})$ with probability $\geq 1-\delta$ for some $\hat{b} \in \bbc^n$ satisfying $\|\hat{b} - e^{iH}b\| \leq \eps$.
\end{prob}

We give two algorithms that are fundamentally the same, but operate in different regimes: the first works for low-rank $H$, and the second for arbitrary $H$.

\begin{corollary}\label{corollary:hamiltonian-low-rank}
Suppose $H$ has minimum singular value $\sigma$ and $\eps < \min(0.5, \sigma)$.
We can solve \cref{prob:Hamiltonian-simulation} in $\bOt{\frac{\|H\|_\fr^6\|H\|^{16}}{\sigma^{16}\eps^6}\log^3\frac{1}{\delta}}$ time, giving $\sq_\phi(\hat{b})$ with $\sqrun_\phi(\hat{b}) = \bOt{\frac{\|H\|_\fr^4\|H\|^8}{\sigma^8\eps^4}\log^3\frac{1}{\delta}}$.
\end{corollary}

This runtime is dimensionless in a certain sense.
The natural error bound to require is that $\|\hat{b} - e^{iH}b\| \leq \eps\|H\|$, since $\abs{\frac{d}{dx}(e^{-i\|H\|x})} = \|H\|$.
So, if we rescale $\eps$ to $\eps\|H\|$, the runtime is $\bOt{\frac{\|H\|_\fr^6\|H\|^{10}}{\sigma^{16}\eps^6}\log^3\frac{1}{\delta}}$, which is dimensionless.
The runtime of the algorithm in the following corollary does not have this property, so its scaling with $\|H\|$ is worse, despite being faster for, say, $\|H\| = 1$.

\begin{corollary}\label{corollary:Hamiltonian-simulation}
For $\eps < \min(0.5, \|H\|^3)$, we can solve \cref{prob:Hamiltonian-simulation} in $\bOt{\|H\|^{16}\|H\|_\fr^6\eps^{-6}\log^3\frac{1}{\delta}}$ time, giving $\sq_\phi(\hat{b})$ with $\sqrun_\phi(\hat{b}) = \bOt{\|H\|^8\|H\|_\fr^4\eps^{-4}\log^3\frac{1}{\delta}}$.
\end{corollary}

Our strategy proceeds as follows: consider a generic function $f(x)$ and Hermitian $H$.
We can write $f(x)$ as a sum of an even function $a(x) \coloneqq \frac{1}{2}(f(x) + f(-x))$ and an odd function $b(x) \coloneqq \frac{1}{2}(f(x) - f(-x))$.
For the even function, we can use \cref{thm:evenSing} to approximate it via the function $f_a(x) \coloneqq a(\sqrt{x})$; the odd function can be written as $H$ times an even function, which we approximate using \cref{thm:evenSing} for $f_b(x) \coloneqq b(\sqrt{x})/\sqrt{x}$.
In other words, $f(H) = f_a(H^\dagger H) + f_b(H^\dagger H)H$.
Since $\abs{a'(x)},\,\abs{b'(x)} \leq \abs{f'(x)}$, the Lipschitz constants don't blow up by splitting $f$ into even and odd parts.

Now, we specialize to Hamiltonian simulation.
We first rewrite the problem, using the function $\sinc(x) \coloneqq \sin(x)/x$.
\[
  e^{iH}b = \cos(H)b + i\cdot\sinc(H)Hb = f_{\cos}(H^\dagger H) b + f_{\sinc}(H^\dagger H) Hb,
\]
where $f_{\cos}(\lambda) \coloneqq \cos(\sqrt{\lambda})$ and $f_{\sinc}(\lambda) \coloneqq i\cdot \sinc(\sqrt{\lambda})$.
When applying \cref{thm:evenSing} on $f_{\cos}$ and $f_{\sinc}$, we will use the following bounds on the smoothness of $f_{\cos}$ and $f_{\sinc}$.
\begin{align*}
  |f_{\cos}'(x)| &= \Big|\frac{\sin(\sqrt{x})}{2\sqrt{x}}\Big| \leq \min\Big(\frac12,\frac{1}{2\sqrt{x}}\Big) \\
  |\bar{f}_{\cos}'(x)| &= \Big|\frac{2-2\cos(\sqrt{x}) - \sqrt{x}\sin(\sqrt{x})}{2x^2}\Big| \leq \min\Big(\frac{1}{24},\frac{5}{2x^{3/2}}\Big) \\
  |f_{\sinc}'(x)| &= \Big|\frac{\sqrt{x}\cos(\sqrt{x}) - \sin(\sqrt{x})}{2x^{3/2}}\Big| \leq \min\Big(\frac{1}{4},\frac{1}{x}\Big) \\
  |\bar{f}_{\sinc}'(x)| &= \Big|\frac{2\sqrt{x} + \sqrt{x}\cos(\sqrt{x}) - 3\sin(\sqrt{x})}{2x^{5/2}}\Big| \leq \min\Big(\frac{1}{60},\frac{3}{x^2}\Big)
\end{align*}
We separate these bounds into the case where $x \geq 1$, which we use when we assume $H$ has a minimum singular value, and the case where $x < 1$, which we use for arbitrary $H$.

\begin{proof}[Proof of \cref{corollary:Hamiltonian-simulation}]
Using the Lipschitz bounds above with \cref{thm:evenSing}, we can find $R_{\cos} \in \bbc^{r_{\cos} \times n}, C_{\cos} \in \bbc^{r_{\cos} \times c_{\cos}}, R_{\sinc} \in \bbc^{r_{\sinc} \times n}, C_{\sinc} \in \bbc^{r_{\sinc} \times c_{\sinc}}$ such that
\begin{gather}
  \|R_{\cos}^\dagger \bar{f}_{\cos}(C_{\cos}C_{\cos}^\dagger) R_{\cos} + I - f_{\cos}(H^\dagger H)\| \leq \eps \label{eqn:cosapprox} \\
  \|R_{\sinc}^\dagger \bar{f}_{\sinc}(C_{\sinc}C_{\sinc}^\dagger) R_{\sinc} + i\cdot I - f_{\sinc}(H^\dagger H)\| \leq \frac{\eps}{\|H\|} \label{eqn:sincapprox}
\end{gather}
where, using that our Lipschitz constants are all bounded by constants,
\begin{align*}
  r_{\cos} &= \bOt{\|H\|_\fr^2\|H\|^2\eps^{-2}\log\frac{1}{\delta}}
  &c_{\cos} &= \bOt{\|H\|_\fr^2\|H\|^6\eps^{-2}\log\frac{1}{\delta}} \\
  r_{\sinc} &= \bOt{\|H\|_\fr^2\|H\|^4\eps^{-2}\log\frac{1}{\delta}}
  &c_{\sinc} &= \bOt{\|H\|_\fr^2\|H\|^8\eps^{-2}\log\frac{1}{\delta}}.
\end{align*}
As a consequence,
\[
  \Big\|e^{iH}b - \Big(R_{\cos}^\dagger \bar{f}_{\cos}(C_{\cos}C_{\cos}^\dagger) R_{\cos}b + b + R_{\sinc}^\dagger \bar{f}_{\sinc}(C_{\sinc}C_{\sinc}^\dagger) R_{\sinc} Hb + iHb\Big)\Big\| \lesssim \eps.
\]
Note that, by \cref{lem:evenSingBounds}, $\|R_{\cos}\| \lesssim \|H\|$, $\|\bar{f}_{\cos}(C_{\cos}C_{\cos}^\dagger)\| \lesssim 1$, and $\|R_{\cos}^\dagger \sqrt{\bar{f}_{\cos}(C_{\cos}C_{\cos}^\dagger)}\| \lesssim 1$; the same bounds hold for the $\sinc$ analogues.
We now approximate using \cref{prop:appr-mms} four times.
\begin{enumerate}
  \item We approximate $R_{\cos} b \approx u$ to $\eps \|b\|$ error, requiring $\bigO{\|H\|_\fr^2\eps^{-2}\log\frac{1}{\delta}}$ samples.
  \item We approximate $R_{\sinc}H \approx WC$ to $\eps$ error, requiring $\bigO{\|H\|_\fr^4\eps^{-2}\log\frac{1}{\delta}}$ samples.
  \item We approximate $Cb \approx v$ to $\eps\|H\|_\fr^{-1}\|b\|$ error, requiring $\bigO{\|H\|_\fr^4\eps^{-2}\log\frac{1}{\delta}}$ samples.
  \item We approximate $Hb \approx R^\dagger w$ to $\eps\|b\|$ accuracy, requiring $r \coloneqq \bigO{\|H\|_\fr^2\eps^{-2}\log\frac{1}{\delta}}$ samples.
\end{enumerate}
Our output will be
\[
  \hat{b} \coloneqq R_{\cos}^\dagger \bar{f}_{\cos}(C_{\cos}C_{\cos}^\dagger)  u + b + R_{\sinc}^\dagger \bar{f}_{\sinc}(C_{\sinc}C_{\sinc}^\dagger)  Wv + i R^\dagger w,
\]
which is close to $e^{iH}b$ because
\begin{align*}
  &\Big\|\hat{b} - \Big(R_{\cos}^\dagger \bar{f}_{\cos}(C_{\cos}C_{\cos}^\dagger) R_{\cos}b + b + R_{\sinc}^\dagger \bar{f}_{\sinc}(C_{\sinc}C_{\sinc}^\dagger) R_{\sinc} Hb + iHb\Big)\Big\| \\
  &\leq \|R_{\cos}^\dagger \bar{f}_{\cos}(C_{\cos}C_{\cos}^\dagger)(u - R_{\cos}b)\| + \|R_{\sinc}^\dagger \bar{f}_{\sinc}(C_{\sinc}C_{\sinc}^\dagger)(R_{\sinc}H - WC)b\|\\
  &\qquad+ \|R_{\sinc}^\dagger \bar{f}_{\sinc}(C_{\sinc}C_{\sinc}^\dagger) W (Cb - v)\| + \|iR^\dagger w - iHb\| \\
  & \lesssim \|u - R_{\cos}b\| + \|R_{\sinc}H - WC\|\|b\| + \|H\|_\fr\|Cb - v\| + \|R^\dagger w - Hb\|
  \leq 4\eps\|b\|.
\end{align*}

Now, we have expressed $\hat{b}$ as a linear combination of a small number of vectors, all of which we have sampling and query access to.
We can complete using \cref{lem:b-sq-approx,lemma:sample-Mv}, where the matrix is the concatenation $(R_{\cos}^\dagger \mid b \mid R_{\sinc}^\dagger \mid i\cdot R^\dagger)$, and the vector is the concatenation $(\bar{f}_{\cos}(C_{\cos}C_{\cos}^\dagger)u \mid 1 \mid \bar{f}_{\sinc}(C_{\sinc}C_{\sinc}^\dagger)Wv \mid w)$.
The length of this vector is $r_{\cos} + 1 + r_{\sinc} + r \lesssim r_{\sinc}$.
We get $\sq_\phi(\hat{b})$ where
\begin{align*}
  \phi &\lesssim r_{\sinc} \Big(\frac{\|H\|_\fr^2}{r_{\cos}}\|\bar{f}_{\cos}(C_{\cos}C_{\cos}^\dagger)u\|^2 + \|b\|^2 + \frac{\|H\|_\fr^2}{r_{\sinc}}\|\bar{f}_{\sinc}(C_{\sinc}C_{\sinc}^\dagger)Wv\|^2 + \frac{\|H\|_\fr^2}{r}\|w\|^2\Big)\|\hat{b}\|^{-2} \\
  &\lesssim \Big(\frac{r_{\sinc}}{r_{\cos}}\|H\|_\fr^2(1+\eps)^2\|b\|^2 + r_{\sinc}\|b\|^2 + \|H\|_\fr^2(1+\eps)^2\|b\|^2 + \frac{r_{\sinc}}{r}\|H\|_\fr^2\|b\|^2\Big)\|b\|^{-2} \\
  &= \bOt{\|H\|_\fr^2\|H\|^2 + r_{\sinc} + \|H\|_\fr^2 + \|H\|_\fr^2\|H\|^4} = \bOt{r_{\sinc}}.
\end{align*}
In the second inequality, we use the same bounds for proving $\|\hat{b} - e^{iH}b\| \leq \eps$, repurposed to argue that all approximations are sufficiently close to the values they are estimating, up to relative error.
So, $\sqrun_\phi(\hat{b}) = \bOt{r_{\sinc}^2\log\frac{1}{\delta}}$.
\end{proof}

\begin{proof}[Proof of \cref{corollary:hamiltonian-low-rank}]
Our approach is the same, though with different parameters.
For \cref{thm:evenSing}, we use that in the interval $[\sigma^2/2,\infty)$, $f_{\cos}$ has Lipschitz constants of $L = O(1/\sigma)$ and $\bar{L} = O(1/\sigma^3)$ and $f_{\sinc}$ has $L = O(1/\sigma^2)$ and $\bar{L} = O(1/\sigma^4)$.
So, if we take
\begin{align*}
  r_{\cos} &= \bOt{\|H\|^2\frac{\|H\|_\fr^2}{\sigma^2}\eps^{-2}\log\frac{1}{\delta}}
  &c_{\cos} &= \bOt{\|H\|^2\frac{\|H\|_\fr^2\|H\|^4}{\sigma^6}\eps^{-2}\log\frac{1}{\delta}} \\
  r_{\sinc} &= \bOt{\|H\|^2\frac{\|H\|_\fr^2\|H\|^2}{\sigma^4}\eps^{-2}\log\frac{1}{\delta}}
  &c_{\sinc} &= \bOt{\|H\|^2\frac{\|H\|_\fr^2\|H\|^6}{\sigma^8}\eps^{-2}\log\frac{1}{\delta}},
\end{align*}
all the conditions of \cref{thm:evenSing} are satisfied: in particular, $\sigma^2/2 > \bar{\eps}$ in both cases, up to rescaling $\eps$ by a constant factor:
\begin{align*}
  \bar{\eps}_{\cos} &\lesssim \|H\|\|H\|_\fr\frac{\eps\sigma}{\|H\|\|H\|_\fr} = \eps\sigma \leq \sigma^2 \\
  \bar{\eps}_{\sinc} &\lesssim \|H\|\|H\|_\fr\frac{\eps\sigma^2}{\|H\|^2\|H\|_\fr} = \eps\sigma^2\|H\|^{-1} \leq \sigma^2
\end{align*}
Here, we used our initial assumption that $\eps \leq \sigma$.
So, the bounds \cref{eqn:cosapprox,eqn:sincapprox} hold.
Note that, by \cref{lem:evenSingBounds}, $\|R_{\cos}\| \lesssim \|H\|$, $\|\bar{f}_{\cos}(C_{\cos}C_{\cos}^\dagger)\| \lesssim \sigma^{-2}$, and $\|R_{\cos}^\dagger \sqrt{\bar{f}_{\cos}(C_{\cos}C_{\cos}^\dagger)}\| \leq 1$; the same bounds hold for the $\sinc$ analogues.
We now approximate using \cref{prop:appr-mms} four times.
\begin{enumerate}
  \item We approximate $R_{\cos} b \approx u$ to $\eps\sigma\|b\|$ error, requiring $\bigO{\|H\|_\fr^2\sigma^{-2}\eps^{-2}\log\frac{1}{\delta}}$ samples.
  \item We approximate $R_{\sinc}H \approx WC$ to $\eps\sigma$ error, requiring $\bigO{\|H\|_\fr^4\sigma^{-2}\eps^{-2}\log\frac{1}{\delta}}$ samples.
  \item We approximate $Cb \approx v$ to $\eps\sigma\|H\|_\fr^{-1}\|b\|$ error, requiring $\bigO{\|H\|_\fr^4\sigma^{-2}\eps^{-2}\log\frac{1}{\delta}}$ samples.
  \item We approximate $Hb \approx R^\dagger w$ to $\eps\|b\|$ accuracy, requiring $r \coloneqq \bigO{\|H\|_\fr^2\eps^{-2}\log\frac{1}{\delta}}$ samples.
\end{enumerate}
Our output will be
\[
  \hat{b} \coloneqq R_{\cos}^\dagger \bar{f}_{\cos}(C_{\cos}C_{\cos}^\dagger)  u + b + R_{\sinc}^\dagger \bar{f}_{\sinc}(C_{\sinc}C_{\sinc}^\dagger)  Wv + i R^\dagger w,
\]
which is close to $e^{iH}b$ by the argument
\begin{align*}
  &\Big\|\hat{b} - \Big(R_{\cos}^\dagger \bar{f}_{\cos}(C_{\cos}C_{\cos}^\dagger) R_{\cos}b + b + R_{\sinc}^\dagger \bar{f}_{\sinc}(C_{\sinc}C_{\sinc}^\dagger) R_{\sinc} Hb + iHb\Big)\Big\| \\
  &\leq \|R_{\cos}^\dagger \bar{f}_{\cos}(C_{\cos}C_{\cos}^\dagger)(u - R_{\cos}b)\| + \|R_{\sinc}^\dagger \bar{f}_{\sinc}(C_{\sinc}C_{\sinc}^\dagger)(R_{\sinc}H - WC)b\| \\
  &\qquad+ \|R_{\sinc}^\dagger \bar{f}_{\sinc}(C_{\sinc}C_{\sinc}^\dagger) W (Cb - v)\| + \|iR^\dagger w - iHb\| \\
  &\lesssim \sigma^{-1}\|u - R_{\cos}b\| + \sigma^{-1}\|R_{\sinc}H - WC\|\|b\| + \sigma^{-1}\|H\|_\fr\|Cb - v\| + \|R^\dagger w - Hb\|
  \leq 4\eps\|b\|
\end{align*}

Now, we have expressed $\hat{b}$ as a linear combination of a small number of vectors, all of which we have sampling and query access to.
We can complete using \cref{lem:b-sq-approx,lemma:sample-Mv}, where the matrix is the concatenation $(R_{\cos}^\dagger \mid b \mid R_{\sinc}^\dagger \mid i\cdot R^\dagger)$, and the vector is the concatenation $(\bar{f}_{\cos}(C_{\cos}C_{\cos}^\dagger)u \mid 1 \mid \bar{f}_{\sinc}(C_{\sinc}C_{\sinc}^\dagger)Wv \mid w)$.
The length of this vector is $r_{\cos} + 1 + r_{\sinc} + r \lesssim r_{\sinc}$.
We get $\sq_\phi(\hat{b})$ where
\begin{align*}
  \phi &\lesssim r_{\sinc} \Big(\frac{\|H\|_\fr^2}{r_{\cos}}\|\bar{f}_{\cos}(C_{\cos}C_{\cos}^\dagger)u\|^2 + \|b\|^2 + \frac{\|H\|_\fr^2}{r_{\sinc}}\|\bar{f}_{\sinc}(C_{\sinc}C_{\sinc}^\dagger)Wv\|^2 + \frac{\|H\|_\fr^2}{r}\|w\|^2\Big)\|\hat{b}\|^{-2} \\
  &\lesssim \Big(\frac{r_{\sinc}}{r_{\cos}}\|H\|_\fr^2\sigma^{-2}\|b\|^2 + r_{\sinc}\|b\|^2 + \|H\|_\fr^2\sigma^{-2}\|b\|^2 + \frac{r_{\sinc}}{r}\|H\|_\fr^2\|b\|^2\Big)\|b\|^{-2} \\
  &=\bOt{\|H\|_\fr^2\|H\|^2\sigma^{-4} + r_{\sinc} + \|H\|_\fr^2\sigma^{-2} + \|H\|^4\sigma^{-4}\|H\|_\fr^2}
  =\bOt{r_{\sinc} + \frac{t^2\|H\|_\fr^2}{\sigma^{4}}}.
\end{align*}
So, $\sqrun_\phi(\hat{b}) = \bOt{r_{\sinc}(r_{\sinc} + \|H\|_\fr^2\|H\|^2\sigma^{-4})\log\frac{1}{\delta}}$.
Since $\eps < \sigma$, the $r_{\sinc}^2$ term dominates.
\end{proof}

\begin{remark}
In the case where $H$ is not low-rank, we could still run a modified version of \cref{corollary:hamiltonian-low-rank} to compute a modified ``$\exp_{\sigma,\eta}(iH)$'' where singular values below $\sigma$ are smoothly thresholded away.
Following the same logic as \cref{def:psinv}, we could redefine $f_{\cos}$ such that $f_{\cos}(x) = 1$ for $x < \sigma^2(1-\eta)$, $f_{\cos}(x) = \cos(\sqrt{\lambda})$ for $x \geq \sigma^2$, and is a linear interpolation between the endpoints for the $x$ in between (and $f_{\sinc}$ similarly).
These functions have the same Lipschitz constants as their originals, up to factors of $\frac1\eta$, and give the desired behavior of ``smoothing away'' small singular values (though we do keep the 0th and 1st order terms of the exponential).
\end{remark}

\begin{remark}
Our result generalizes those of Ref.~\cite{rudi2018nystrom}, which achieves essentially the same result only in the much easier regime where $H$ and $b$ are sparse.
They achieve a significant speedup due to these assumptions: note that when $H$ is sparse, and a subsample of rows $R$ is taken, $RR^\dagger$ can be computed in time independent of dimension; so, we only need to take a subsample of rows, and not of columns.
More corners can be cut from our algorithm in this fashion.
In summary, though our algorithm is significantly slower, their sparsity assumptions are essential for their fast runtime, and our framework can identify where these tradeoffs occur.
\end{remark}

\subsection{Semidefinite program solving}\label{sec:SDP}
A recent line of inquiry in quantum computing \cite{brandao2016QSDPSpeedup,apeldoorn2017QSDPSolvers,brandao2017QSDPSpeedupsLearning,apeldoorn2018ImprovedQSDPSolving} focuses on finding quantum speedups for \emph{semidefinite programs} (SDPs), a central topic in the theory of convex optimization with applications in algorithms design, operations research, and approximation algorithms.
Chia, Li, Lin, and Wang~\cite{chia2019QInspiredSubLinLowRankSDPSolver} first noticed that quantum-inspired algorithms could dequantize these quantum algorithms in certain regimes.
We improve on their result, giving an algorithm which is as general as the quantum algorithms, if the input is given classically (e.g., in a data-structure in RAM).
Our goal is to solve the $\eps$-feasibility problem; solving an SDP reduces by binary search to solving $\log(1/\eps)$ instances of this feasibility problem.

\begin{prob}[SDP $\eps$-feasibility]\label{prob:sdp}
Given an $\eps>0$, $m$ real numbers $b_{1},\ldots,b_{m}\in\bbr$, and Hermitian $n\times n$ matrices $\sq(A^{(1)}),\ldots,\sq(A^{(m)})$ such that $-I\preceq A^{(i)}\preceq I$ for all $i\in\range{m}$, we define $\mathcal{S}_{\eps}$ as the set of all $X$ satisfying\footnote{For simplicity, we assume here that $X$ is normalized to have trace $1$. This can be relaxed; for an example, see~\cite{apeldoorn2017QSDPSolvers}.}
\begin{align*}
  \Tr[A^{(i)} X] &\leq b_{i}+\eps\quad\forall\,i\in\range{m}; \\ 
  X&\succeq 0; \\ 
  \Tr[X]&=1. 
\end{align*}
If $\mathcal{S}_{\eps}=\emptyset$, output ``infeasible".
If $\mathcal{S}_{0}\neq\emptyset$, output an $X\in\mathcal{S}_{\eps}$.
(If neither condition holds, either output is acceptable.)
\end{prob}

\begin{corollary}\label{cor:SDP}
Let $F \geq \max_{j \in [m]}(\|A^{(j)}\|_\fr)$, and suppose\footnote{Because of the normalization assumption that $\|A^{(\cdot)}\| \leq 1$, $F$ is effectively a dimensionless ``stable rank''-type constant, normalized by $\max_i \|A^{(i)}\|$.} $F = \Omega(1)$.
Then we can solve \cref{prob:sdp} with success probability $\geq 1-\delta$ in cost
  \begin{equation*}
  \bOt{\Big(\frac{F^{18}}{\eps^{40}}\log^{20}(n)\sqcb(A)
    +\frac{F^{22}}{\eps^{46}}\log^{23}(n)
    +m\frac{F^8}{\eps^{18}}\log^{8}(n)\qcb(A)
    +m\frac{F^{14}}{\eps^{28}}\log^{13}(n)\Big)\log^3\frac{1}{\delta}},
  \end{equation*}
providing sampling and query access to a solution.
\end{corollary}

Assuming $\sqcb(A) = \bOt{1}$, this runtime is
$
\widetilde{\mathcal{O}}\Big(\frac{F^{22}}{\eps^{46}}\log^{23}(n) + m\frac{F^{14}}{\eps^{28}}\log^{13}(n)\Big).
$
For the same feasibility problem, the previous quantum-inspired SDP solver~\cite{chia2019QInspiredSubLinLowRankSDPSolver} proved a complexity bound $\bOt{mr^{57}\eps^{-92}\log^{37}(n)}$, assuming that the constraint matrices have rank at most~$r$. Since the rank constraint implies that $\nrm{A^{(\cdot)}}_\fr\leq \sqrt{r}$, under this assumption our algorithm has complexity $\bOt{r^{11}\eps^{-46}\log^{23}(n)+m r^{7}\eps^{-28}\log^{13}(n)}$.
So, our new algorithm both solves a more general problem and also greatly improves the runtime.
The paper with the current best runtime for SDP solving does not discuss this precise model, but if we use the runtime they achieve in quantum state input model, making reasonable substitutions of $\gamma \to \frac{1}{\eps}$ and $B \to F^2$, the corresponding quantum runtime is $\bOt{\frac{F^7}{\eps^{7.5}} + \frac{\sqrt{m}F^2}{\eps^4}}$, up to $\polylog(n)$ factors.

Like prior work on quantum algorithms for SDP-solving, we use the matrix multiplicative weights (MMW) framework \cite{arora2016CombPrimDualSDP,kale2007efficient} to solve \cref{prob:sdp}.
\cref{cor:SDP} immediately follows from running the algorithm this framework admits (\cref{alg:MMW}), where we solve an instance of the problem described in \cref{cor:traceEst} with precision $\theta=\eps/4$ in each of the $\bigO{\log(n)/\eps^2}$ iterations.

\begin{algorithm}
Set $X_{1}\coloneqq \frac{I_{n}}{n}$, and the number of iterations $T\coloneqq \frac{16\log n}{\eps^{2}}$\;
\For{$t=1,\ldots,T$}{
  \textbf{find} a $j_{t}\in\range{m}$ such that $\Tr[A^{(j_{t})}X_{t}]>b_{j_{t}}+\frac{\eps}{2}$ \label{step:either}\\
  \quad\textbf{or} conclude correctly that $\Tr[A^{(j_{t})}X_{t}]\leq b_{j_{t}}+\eps$ for all $j\in [m]$\label{step:or}\\
  \textbf{if} a $j_{t}\in\range{m}$ is found \textbf{then}\\
  $\quad X_{t+1} \coloneqq \exp[-\frac{\eps}{4}\sum_{i=1}^{t}A^{(j_{i})}] / \Tr[\exp[-\frac{\eps}{4}\sum_{i=1}^{t}A^{(j_{i})}]]$ \label{step:Gibbs}\\
  \textbf{else} conclude that $X_{t}\in\mathcal{S}_{\eps}$\\
  $\quad$ \textbf{return} $X_t$
}
If no solution found, conclude that the SDP is infeasible and terminate the algorithm
\caption{MMW based feasibility testing algorithm for SDPs}
\label{alg:MMW}
\end{algorithm}

MMW works as a zero-sum game with two players, where the first player wants to provide an $X \in \mathcal{S}_{\eps}$, and the second player wants to find a violation for any proposed $X$, i.e., a $j\in\range{m}$ such that $\Tr[A^{(j)}X]>b_{j}+\eps$. At the $t^{\text{th}}$ round of the game, if the second player points out a violation $j_{t}$ for the current solution $X_{t}$, the first player proposes a new solution
\begin{equation*}
  X_{t+1}\propto\exp[-\eps(A^{(j_{1})}+\cdots+A^{(j_{t})})].
\end{equation*}
Solutions of this form are also known as \emph{Gibbs states}. It is known that MMW solves the SDP $\eps$-feasibility problem in $\bigO{\frac{\log n}{\eps^{2}}}$ iterations;
a proof can be found, e.g., in the work of Brand\~ao, Kalev, Li, Lin, Svore, and Wu~\cite[Theorem 3]{brandao2017QSDPSpeedupsLearning} or in Lee, Raghavendra and Steurer~\cite[Lemma 4.6]{lee2015LowerBoundSDPRelax}.

Our task is to execute \cref{step:either,step:or} of \cref{alg:MMW}, for an implicitly defined matrix with the form given in \cref{step:Gibbs}.

\begin{lemma}[``Efficient'' trace estimation]\label{cor:traceEst}
  Consider the setting described in \cref{cor:SDP}.
  Given $\theta\in(0,1]$, $t\leq\frac{\log(n)}{\theta^2}$ and $j_{i}\in\range{m}$ for $i\in[t]$, defining $H\coloneqq \exp[-\theta\sum_{i=1}^{t}A^{(j_{i})}]$, we can estimate $\Tr(A^{(i)}H)/\Tr(H)$ with success probability $\geq 1-\delta$ for all $i\in [m]$ to precision $\theta$ in cost
  \[
  \bOt{\left[\frac{F^{18}}{\theta^{38}}\log^{19}(n)\sqcb(A)
    +\frac{F^{22}}{\theta^{44}}\log^{22}\kern-0.2mm(n)
    +m\frac{F^8}{\theta^{16}}\log^{7}\kern-0.4mm(n)\qcb(A)
    +m\frac{F^{14}}{\theta^{26}}\log^{12}(n)\right]\log^3\frac1\delta+ \frac{\log(n)}{\theta^2}\ncb(A)\!}\!,
  \]
  where $\sqcb(A)=\max_{j\in [m]}\sqcb(A^{(j)})$, and $\scb(A)$, $\qcb(A)$, $\ncb(A)$ are defined analogously.
\end{lemma}

To estimate $\Tr[A^{(i)}H]$, we first notice that we have $\sq_\phi(\theta\sum_{i=1}^t A^{(j_i)})$, since it is a linear combination of matrices that we have sampling and query access to (\cref{lem:weighted-oversampling}).
Then, we can find approximations of the Gibbs state by applying eigenvalue transformation (\cref{thm:eig-svt}) according to the exponential function to get $\exp[-\theta\sum_{i=1}^t A^{(j_i)}]$ as an RUR decomposition.
Then the estimation of $\Tr[A^{(i)}H]$ can be performed by usual techniques (namely, \cref{claim:tr_prod}).

In order to understand how precisely we need to approximate the matrix in \cref{step:Gibbs} we prove the following lemmas.
Our first lemma will show that, to estimate $\Tr(A^{(i)}H)/\Tr(H)$ to $\theta$ precision, it suffices to estimate both $\Tr(A^{(i)}H)$ and $\Tr(H)$ to $\frac13\theta\Tr(H)$ precision.

\begin{lemma}\label{lem:traceRelativeErrors}
  Suppose that $\theta\in[0,1]$ and $a,\tilde{a},Z,\tilde{Z}$ are such that $|a|\leq Z$, $|a-\tilde{a}|\leq \frac\theta3 Z$, and $|Z-\tilde{Z}|\leq \frac\theta3 Z$, then
  \begin{equation*}
    \left|\frac{\tilde{a}}{\tilde{Z}}-\frac{a}{Z}\right|\leq \theta.
  \end{equation*}
\end{lemma}
\begin{proof}
\begin{equation*}
\left|\frac{\tilde{a}}{\tilde{Z}}-\frac{a}{Z}\right|
=\left|\frac{\tilde{a}Z}{Z\tilde{Z}}-\frac{a\tilde{Z}}{Z\tilde{Z}}\right|
\leq \left|\frac{\tilde{a}Z-a Z}{Z\tilde{Z}}\right| + \left|\frac{aZ -a\tilde{Z}}{Z\tilde{Z}}\right|
\leq \frac{3}{2Z}\abs{\tilde{a} - a} + \frac{3a}{2Z^2}\abs{Z - \tilde{Z}}
\leq \frac{1}{2}\theta + \frac{1}{2}\theta
\leq\theta. \qedhere
\end{equation*}
\end{proof}

Next, we will prove that the approximations we will use to $\Tr(A^{(i)}H)$ and $\Tr(H)$ suffice.
We introduce some useful properties of matrix norms.
For a matrix $A \in \mathbb{C}^{m\times n}$ and $p\in[1,\infty]$, we denote by $\nrm{A}_p$ the \emph{Schatten $p$-norm}, which is the $\ell^p$-norm of the singular values $(\sum_i \sigma_i^p(A))^{1/p}$.
In particular, $\|A\|_\fr =\|A\|_2$ and $\|A\|_\op =\|A\|_\infty$.
We recall some useful inequalities~\cite[Section IV.2]{bhatia1997MatrixAnalysis}.
\emph{Hölder's inequality} states that for all $B\in\bbC^{n\times k}$ and $r,p,q\in(0,\infty]$ such that $\frac1p+\frac1q=\frac1r$, we have $\nrm{AB}_r\leq\nrm{A}_p\nrm{B}_q$.
The \emph{trace-norm inequality} states that if $n=m$, then $\abs{\Tr(A)}\leq\nrm{A}_1$.

\begin{restatable}[Perturbations of the partition function]{lemma}{partpert}\label{lem:partitionError}
  For all Hermitian matrices $H,\tilde{H}\in\bbC^{n\times n}$,
  \begin{equation*}
  \abs*{\Tr(e^{\tilde{H}}) - \Tr(e^H)}\leq\nrm{e^{\tilde{H}}- e^H}_1 \leq\left(e^{\|\tilde{H}-H\|}-1\right)\Tr(e^H).
  \end{equation*}
\end{restatable}

The bound in the above lemma is tight, as shown by the example $\tilde{H}\coloneqq H+\eps I$.
The proof is in the appendix.
Before proving the following lemma, we observe that for any Hermitian matrix $H\in\bbC^{n \times n}$ with $\nrm{H}_\fr^2\leq \frac{n}4$, we have by Hölder's inequality that
\begin{equation}\label{eq:traceLB}
\Tr(e^{H}) = n+ \Tr(e^{H}-I)=n+\sum_i (e^{\lambda_i}-1) \geq n+\sum_i \lambda_i = n + \Tr(H) \geq n - \sqrt{n}\norm{H}_\fr\geq n/2.
\end{equation}

\begin{lemma}\label{lem:UDUTraceError}
  Consider a Hermitian matrix $H\in\bbC^{n \times n}$ such that $\nrm{H}_\fr^2\leq \frac n4$.
  Let $H$ have an approximate eigendecomposition in the following sense: for $r \leq n$, suppose we have a diagonal matrix $D \in \bbr^{r\times r}$ and $\widetilde{U} \in \bbc^{r\times n}$ that satisfy $\|\widetilde{U}\widetilde{U}^\dagger - I\| \leq \delta$ and $\|H-\widetilde{U}^\dagger D \widetilde{U}\|\leq \eps$ for $\eps \leq \frac12$ and $\delta \leq \min(\frac{\eps}{4(\|H\|+\eps)}, \frac{\eps}{2})$.
  Then we have
  \begin{equation}\label{eq:traceZNew}
  \abs{(\Tr(e^{D}) + n-r) - \Tr(e^H)}\leq 2(e-1)\eps\Tr(e^H),
  \end{equation}  
  and, moreover, for all $A\in\bbC^{n \times n}$ we have
  \[
  \abs{\Tr(A \widetilde{U}^\dagger(e^{D}-I)\widetilde{U})+\Tr(A)-\Tr(A e^H)}\lesssim \eps\nrm{A}\Tr(e^H).
  \]
\end{lemma}
\begin{proof}
  First, recall that, by \cref{lem:apporth-facts}, there is unitary $U$ such that $\|\widetilde{U} - U\| \leq \delta$.
  Consequently, also using facts from \cref{lem:apporth-facts}, along with bounds on $\delta$,
  \begin{align}
  \|H - U^\dagger DU\|
  \leq \|H - \widetilde{U}^\dagger D \widetilde{U}\| + \delta\frac{2 - \delta}{(1-\delta)^2}\|\widetilde{U}^\dagger D \widetilde{U}\|
  \leq \eps + 4\delta(\|H\| + \eps) \leq 2\eps.
  \label{eq:IdealUCloseNew}
  \end{align}
  By \cref{lem:partitionError} we have
  \begin{equation*}
    \big\|e^{U^\dagger D U} - e^H\big\|_1 \leq (e^{2\eps} - 1)\Tr(e^H) \leq 2(e-1)\eps\Tr(e^H),
  \end{equation*}
  and since $e^{U^\dagger D U}=U^\dagger (e^{D}-I)U+I$, by the linearity of trace, the trace-norm inequality, and Hölder's inequality,
  \begin{multline}\label{eq:traceMiddleNew}
    \abs{\Tr(A U^\dagger (e^{D}-I)U) + \Tr(A) - \Tr(Ae^H)} \\
    = \abs{\Tr(A (e^{U^\dagger DU} - e^H))}
    \leq \|A\|\|e^{U^\dagger DU} - e^H\|_1
    \leq 2(e-1)\|A\|\eps\Tr(e^H).
  \end{multline}
  In particular, setting $A=I$, we get the first desired bound
  \begin{equation*}
    \abs{(\Tr(e^{D})+ n-r)-\Tr(e^H)}
    = \abs{\Tr(U^\dagger(e^{D}-I)U + I)-\Tr(e^H)}
    \leq 2(e-1)\eps\Tr(e^H).
  \end{equation*}  
  Note that the two identity matrices in the equation above refer to identities of two different sizes.
  Now, if we show that $\Tr(A U^\dagger (e^{D}-I)U) - \Tr(A \widetilde{U}^\dagger (e^{D}-I)\widetilde{U})$ is sufficiently small, then the second desired bound follows by \cref{eq:traceMiddleNew} and triangle inequality.
  \begin{align*}
  &|\Tr(A U^\dagger (e^{D}-I)U) - \Tr(A \widetilde{U}^\dagger (e^{D}-I)\widetilde{U})| \\
  &=|\Tr((U A U^\dagger-\widetilde{U} A \widetilde{U}^\dagger)(e^{D}-I))|\\  
  &\leq \big\|U A U^\dagger-\widetilde{U} A \widetilde{U}^\dagger\big\|\nrm{e^{D}-I}_1 \tag*{by trace-norm and Hölder's inequality}\\  
  &\leq (2\delta + \delta^2)\|A\|\nrm{e^{D}-I}_1 \tag*{by \cref{lem:apporth-facts}}\\
  &\leq 2\eps\|A\|\nrm{e^{D}-I}_1 \tag*{by assumption that $\delta \leq \eps/2$} \\
  &\leq 2\eps\nrm{A}\left(\Tr(e^{D})+r\right) \tag*{by triangle inequality}\\    
  &\lesssim \eps\nrm{A}\Tr(e^H). \tag*{by \cref{eq:traceZNew,eq:traceLB}}
  \end{align*}
\end{proof}

Now we are ready to devise our upper bound on the trace estimation subroutine.
\begin{proof}[Proof of \cref{cor:traceEst}]
  By \cref{lem:traceRelativeErrors}, it suffices to find estimates of $\Tr(e^H)$ and $\Tr(A e^H)$ for all $A = A^{(i)}$, to $\frac{\theta}{3}\Tr(e^H)$ additive precision.
  Recall from the statement that $H\coloneqq -\theta\sum_{i=1}^{t}A^{(j_{i})}$.
  By triangle inequality, $\nrm{H}_\fr\leq \frac{F}{\theta}\log(n)$.
  Because $H$ is a linear combination of matrices, by \cref{lem:weighted-oversampling}, after paying $\frac{\log(n)}{\theta^2}\ncb(A)$ cost, we can obtain $\sq_\phi(H)$ for $\phi\leq\frac{F^2\log^2(n)}{\theta^2\nrm{H}_\fr^2}$ with $\qcb(H)=\pcb_\phi(H)\leq\frac{\log(n)}{\theta^2}\qcb(A)$ and $\scb_\phi(H)=\scb(A)$.
  
  If $\frac{F}{\theta}\log(n)> \sqrt{n}/18$, then we simply compute the sum $H$ by querying all matrix elements of every $A^{(j_{i})}$ in the sum, costing $\bigO{t n^2 \qcb(A)}$. Then we compute $e^H$ 
  and its trace $\Tr(e^H)$ all in time $\bigO{n^3}$~\cite{pan1999ComplexityMatEigenProb}. Finally, we compute all the traces $\Tr(e^H A^{(m)})$ in time $\bigO{m n^2}$. The overall complexity is $\bigO{n^2 (t \qcb(A)+n+m)}=\bOt{\frac{F^6}{\theta^6}\qcb(A)\log^6(n)+m\frac{F^4}{\theta^4}\log^4(n)}$.
  
  If $\frac{F}{\theta}\log(n)\leq \sqrt{n}/18$ we do the following.
  Note that if $\nrm{H}\leq 1$, then $\Tr(e^H)\geq n/e$ and $\Tr(A^{(i)}e^H)\leq \nrm{A^{(i)}}_\fr\nrm{e^H}_\fr\leq F e\sqrt{n}$, so $\Tr(A^{(i)}e^H)/\Tr(e^H)\leq e^2 F / \sqrt{n}\leq\theta$, and outputting $0$ as estimates is acceptable. 
  We use \cref{thm:eig-svt} (with $f(x)=x$, so that $L=1$, and choosing $\eps\coloneqq \Theta(\theta)$) to find a diagonal matrix $D\!\in\!\bbR^{s\times s}$ with $s=\bOt{\!\phi^2\nrm{H}_\fr^6/\eps^6\log(1/\delta)\!}\!=\bOt{F^6\theta^{-6}\log^6(n)\eps^{-6}\log(1/\delta)}=\bOt{F^6\theta^{-12}\log^{6}(n)\log(1/\delta)}$ together with an approximate isometry $\widetilde{U}=N(SH)\in\bbC^{s \times n}$ such that $\nrm{H-\widetilde{U}^\dagger D \widetilde{U}}\leq \bigO{\eps}$.
  If every diagonal element is less than $3/4$, then we conclude that $\nrm{H}\leq 1$, and return $0$. Otherwise we have $\nrm{H}\geq 1/2$ and thus by \cref{thm:eig-svt} we have $\nrm{\widetilde{U}\widetilde{U}^\dagger - I} \lesssim\eps^3\nrm{H}^{-3} \lesssim\frac{\eps}{\nrm{H}+\eps} + \eps$ with probability at least $1-\frac\delta2$.
  As per \cref{thm:eig-svt}, the cost of this is $\log^3(1/\delta)$ times at most
  \begin{align*}
  \bOt{\frac{\|H\|_\fr^{18}}{\eps^{18}}\phi^7\sqcb_\phi(H) +\frac{\|H\|_\fr^{22}}{\eps^{22}}\phi^6}
  &=\bOt{\frac{\|H\|_\fr^{4}}{\eps^{18}}\frac{F^{14}}{\theta^{14}}\log^{14}(n)\sqcb_\phi(H) +\frac{\|H\|_\fr^{10}}{\eps^{22}}\frac{F^{12}}{\theta^{12}}\log^{12}(n)}\\
  &=\bOt{\frac{\|H\|_\fr^{4}}{\eps^{18}}\frac{F^{14}}{\theta^{16}}\log^{15}(n)\sqcb(A) +\frac{\|H\|_\fr^{10}}{\eps^{22}}\frac{F^{12}}{\theta^{12}}\log^{12}(n)}\\      
  &=\bOt{\frac{1}{\eps^{18}}\frac{F^{18}}{\theta^{20}}\log^{19}(n)\sqcb(A) +\frac{1}{\eps^{22}}\frac{F^{22}}{\theta^{22}}\log^{22}(n)}\\
  &=\bOt{\frac{F^{18}}{\theta^{38}}\log^{19}(n)\sqcb(A) +\frac{F^{22}}{\theta^{44}}\log^{22}(n)}.  
  \end{align*}
  
  By \cref{lem:UDUTraceError} we have\footnote{In case applying \cref{thm:eig-svt} would result in $s > n$, we instead directly diagonalize $H$ ensuring $s\leq n$.} that $\Tr(e^D)+(n-s)$ is a multiplicative $\frac\theta{3}$-approximation of $\Tr(e^H)$ as desired, and for all $A=A^{(i)}$, $\Tr((e^{D}-I)\widetilde{U} A \widetilde{U}^\dagger )+\Tr(A)$ is an additive $(\frac\theta{9}\Tr(e^H))$-approximation of $\Tr(A e^H)$.
  We can ignore the $\Tr(A)$ in our approximation: by \cref{eq:traceLB} we have
  \[
    \Tr(A) \leq \|A\|_\fr\|I\|_\fr \leq F\sqrt{n} \leq \theta n/18 \leq \theta\Tr(e^H)/9,
  \]
  so $|\Tr((e^{D}-I)\widetilde{U} A \widetilde{U}^\dagger) - \Tr(Ae^H)| \leq \frac{2\theta}{9}\Tr(e^H))$.
  So, it suffices to compute an additive $(\frac{\theta}{9}\Tr(e^H))$-approximation of $\Tr((e^{D}-I)\widetilde{U} A \widetilde{U}^\dagger)=\Tr( A \widetilde{U}^\dagger (e^{D}-I)\widetilde{U})$ to obtain the $(\frac{\theta}{3}\Tr(e^H))$-approximation of $\Tr(A e^H)$ we seek.
  
  We use \cref{claim:tr_prod} to estimate $\Tr( A \widetilde{U}^\dagger (e^{D}-I)\widetilde{U})$ to additive precision $(\frac{\theta}{9}\Tr(e^H))$. Note that by  \cref{lem:UDUTraceError} and \cref{eq:traceLB} we have
  \begin{equation*}
  \nrm{\widetilde{U}^\dagger (e^{D}-I)\widetilde{U}}_\fr
  \leq\|\tilde{U}\|^2\|e^{D}-I\|_\fr
  \leq 2\|e^{D}-I\|_\fr
  \leq 2\|e^{D}-I\|_1
  \lesssim \Tr(e^H),
  \end{equation*}  
  and since  $s=\bOt{F^6\theta^{-12}\log^{6}(n)\log(1/\delta)}$ and $\qcb(H)\leq\frac{\log(n)}{\theta^2}\qcb(A)$, we also have
  \begin{align*}
    \qcb(\widetilde{U}^\dagger (e^{D}-I)\widetilde{U})
    &=\qcb((SH)^\dagger N^\dagger (e^{D}-I)N(SH))\\
    &=\bigO{s \cdot \qcb(H)+s^2}\\
    &=\bOt{F^6\theta^{-14}\log^{7}(n)\log(1/\delta)\qcb(A)+F^{12}\theta^{-24}\log^{12}(n)\log^2(1/\delta)}.
  \end{align*}
  Therefore, \cref{claim:tr_prod} tells us that given $\sq(A)$, a  $(\frac{\theta}{9}\Tr(e^H))$-approximation of $\Tr( A \widetilde{U}^\dagger (e^{D}-I)\widetilde{U})$ can be computed with success probability at least $1-\frac\delta{2m}$ in time
  \begin{equation*}
  \bigO{\frac{\|A\|_\fr^2}{\theta^2} \big(\sqcb(A)+s \cdot \qcb(H)+s^2  \big)\log\frac{m}{\delta}}.
  \end{equation*}
  Since we do this for all $i\in[m]$, the overall complexity of obtaining the desired estimates $\Tr(A^{(i)} e^H)$ with success probability at least $1-\frac\delta2$ is $m$ times
  \begin{equation*}
  \bOt{\frac{F^8}{\theta^{16}}\log^{7}(n)\log(1/\delta)\log(m/\delta)\qcb(A)+\frac{F^{14}}{\theta^{26}}\log^{12}(n)\log^2(m/\delta)\log(m/\delta)}.\qedhere
  \end{equation*}
\end{proof}

\subsection{Discriminant analysis}\label{sec:discriminant-analysis}

Discriminant analysis is used for dimensionality reduction and classification over large data sets.
Cong and Duan introduced a quantum algorithm to perform both with Fisher's linear discriminant analysis~\cite{cong2016quantum}, a generalization of principal component analysis to data separated into classes.

The problem is as follows: given classified data, we wish to project our data onto a subspace that best explains between-class variance, while minimizing within-class variance.
Suppose there are $M$ input data points $\{x_i \in \bbr^N : 1 \leq i \leq M\}$ each belonging to one of $k$ classes. Let $\mu_c$ denote the centroid (mean) of class $c\in [k]$, and $\bar{x}$ denote the centroid of all data points. Following the notation of~\cite{cong2016quantum}, let
\begin{align*}
S_B = \sum_{c=1}^k (\mu_c-\bar{x})(\mu_c-\bar{x})^T \text{ and }
S_W = \sum_{c=1}^k \sum_{x\in c} (\mu_c-x)(\mu_c-x)^T.
\end{align*}
denote the between-class and within-class scatter matrices of the dataset respectively.
The original goal is to solve the generalized eigenvalue problem $S_B v_i = \lambda_i S_W v_i$ and output the top eigenvalues and eigenvectors; for dimensionality reduction using linear discriminant analysis, we would project onto these top eigenvectors. If $S_W$ would be full-rank, this problem would be equivalent to finding the eigenvalues of $S_W^{-1}S_B$. However, this does not happen in general, and therefore various relaxations are considered in the literature~\cite{Belhumeur1997EigenFisherFaces,Welling2009FisherLDA}. For example, Welling~\cite{Welling2009FisherLDA} considers the eigenvalue problem of
\begin{equation}\label{eq:welling}
S_B^{\frac12}S_W^{-1}S_B^{\frac12}.
\end{equation}
Cong and Duan further relax the problem, as they ignore small eigenvalues of $S_W$ and $S_B$, and only compute approximate eigenvalues of \cref{eq:welling} (after truncating eigenvalues), leading to inexact eigenvectors. We construct a classical analogue of their quantum algorithm.\footnote{Analyzing whether or not the particular relaxation used in this and other quantum machine learning papers provides a meaningful output is unfortunately beyond the scope of our paper.}
Cong and Duan also describe a quantum algorithm for discriminant analysis classification; this algorithm does a matrix inversion procedure very similar to those described in \cref{sec:matrix-inversion} and \cref{sec:SVM}, so for brevity we will skip dequantizing this algorithm.

To formally analyze this algorithm, we could, as in \cref{sec:PCA}, assume the existence of an eigenvalue gap, so the eigenvectors are well-conditioned.
However, let us instead use a different convention: if we can find diagonal $D$ and an approximate isometry $U$ such that $S_B^{\frac12}S_W^{-1}S_B^{\frac12}U \approx UD$, then we say we have found approximate eigenvalues and eigenvectors of $S_W^+S_B$.

\begin{prob}[Linear discriminant analysis] \label{prob:disc}
Consider the functions
\begin{align*}
  \fsqrt(x) = \begin{cases}
    0 & x < \sigma^2/2 \\
    2x/\sigma - \sigma & \sigma^2/2 \leq x < \sigma^2 \\
    \sqrt{x} & x \geq \sigma^2
  \end{cases} \qquad
  \finv(x) = \begin{cases}
    0 & x < \sigma^2/2 \\
    2x/\sigma^4 - 1/\sigma^2 & \sigma^2/2 \leq x < \sigma^2 \\
    1/x & x \geq \sigma^2
  \end{cases}.
\end{align*}
Given $\sq(B, W) \in \bbc^{m\times n}$, with $S_W \coloneqq W^\dagger W$ and $S_B \coloneqq B^\dagger B$, find an $\alpha$-approximate isometry $U \in \bbc^{n\times p}$ and diagonal $D \in \bbc^{p\times p}$ such that we have $\sq_\phi(U(\cdot,i))$ for all $i$, $\abs{D_{ii} - \lambda_i} \leq \eps\|B\|^2/\sigma^2$ for $\lambda_i$ the eigenvalues of $\fsqrt(S_B)\finv(S_W)\fsqrt(S_B)$, and
\[ \|\fsqrt(S_B)\finv(S_W)\fsqrt(S_B)U - UD\| \leq \eps\|\fsqrt(S_B)\|^2\|\finv(S_W)\| \leq \eps\|B\|^2/\sigma^2. \]
\end{prob}

The choice of error bound is natural, since $\|B\|^2/\sigma^2$ is essentially $\|\fsqrt(S_B)\|^2\|\finv(S_W)\|$: we aim for additive error.
The quantum algorithm achieves a runtime of $\bOt{\frac{\|B\|_\fr^7}{\eps^3\sigma^7} + \frac{\|W\|_\fr^7}{\eps^3\sigma^7}}$, up to $\polylog(m,n)$ factors~\cite[Theorem~2]{cong2016quantum}.\!\footnote{This is the runtime of Step 2 of Algorithm 1. The normalization factor of $\max(\|B\|_\fr,\|S\|_\fr)$ is implicit there, $\kappa_{eff}$ corresponds to $\frac{\max(\|B\|_\fr,\|S\|_\fr)}{\sigma^2}$, and the error bound the algorithm achieves is the one we describe here, since the authors must implicitly rescale the inverse and square root function by a cumulative factor of $\|B\|^2/\sigma^2$ to apply their Theorem~1.}

\begin{corollary} \label{cor:disc}
For $\eps < \sigma/\|B\|$, we can solve \cref{prob:disc} in $\bOt{(\frac{\|B\|_\fr^6\|B\|^4}{\eps^6\sigma^{10}} + \frac{\|W\|_\fr^6\|W\|^{10}}{\eps^6\sigma^{16}})\log^3\frac{1}{\delta}}$ time, with $\sqrun_\phi(U(\cdot,i)) = \bOt{\frac{\|B\|_\fr^4}{\eps^2\sigma^4}\log^2\frac{1}{\delta}}$.
\end{corollary}

We prove this by using \cref{thm:evenSing} to approximate $\fsqrt(W^\dagger W) \approx R_W^\dagger U_W R_W$ and $\finv(B^\dagger B) \approx R_B^\dagger U_B R_B$ by RUR decompositions.
Then, we use \cref{prop:appr-mms} to approximate $R_W R_B^\dagger$ by small submatrices $R_W' R_B'^\dagger$.
This yields an approximate RUR decomposition of the matrix whose eigenvalues and vectors we want to find, $R_W^\dagger U R_W$ for $U = U_W R_W'R_B'^\dagger U_B R_B'R_W'^\dagger U_W$.

Finding eigenvectors from an RUR decomposition follows from an observation (\cref{lem:app-orth-expression}): for a matrix $C_W$ formed by sampling columns from $R_W$ (using $\sq(W)$), and $[C_W]_k$ the rank-$k$ approximation to $C_W$ (which can be computed because $C_W$ has size independent of dimension), $(([C_W]_k)^+ R_W)^\dagger$ has singular values either close to zero or close to one.
This roughly formalizes the intuition of $C_W$ preserving the left singular vectors and singular values of $R_W$.
We can rewrite $R_W^\dagger U R_W = R_W^\dagger (C_k^+)^\dagger C_k^\dagger U C_k C_k^+ R_W$, which holds by choosing $k$ sufficiently large and choosing $C$ to be the same sketch used for $U$.
Then, we can compute the eigendecomposition of the center $C_k^\dagger U C_k = VDV^\dagger$, which gives us an approximate eigendecomposition for $R_W^\dagger U R_W$: $(C_k^+ R_W)^\dagger V$ is an approximate isometry, so we choose its columns to be our eigenvectors, and our eigenvalues are the diagonal entries of $D$.
We show that this has the approximation properties analogous to the quantum algorithm.

\begin{proof}
By \cref{thm:evenSing}, we can find $R_B, C_B, R_W, C_W$ such that
\begin{align*}
  \|\fsqrt(B^\dagger B) - R_B^\dagger \overline{\fsqrt}(C_BC_B^\dagger) R_B\| &\leq \eps\|B\| \\
  \|\finv(W^\dagger W) - R_W^\dagger \overline{\finv}(C_WC_W^\dagger) R_W\| &\leq \eps/\sigma^2
\end{align*}
with
\begin{align*}
  r_B &= \bOt{\frac{\|B\|_\fr^2}{\eps^2\sigma^2}\log\frac{1}{\delta}} & c_B &= \bOt{\frac{\|B\|^4\|B\|_\fr^2}{\eps^2\sigma^6}\log\frac{1}{\delta}} \\
  r_W &= \bOt{\frac{\|W\|^2\|W\|_\fr^2}{\eps^2\sigma^4}\log\frac{1}{\delta}} & c_W &= \bOt{\frac{\|W\|^6\|W\|_\fr^2}{\eps^2\sigma^8}\log\frac{1}{\delta}}.
\end{align*}
Let $Z_B \coloneqq \overline{\fsqrt}(C_BC_B^\dagger)$ and $Z_W \coloneqq \overline{\finv}(C_WC_W^\dagger)$.
These approximations suffice for us:
\begin{multline*}
  \|\fsqrt(S_B)\finv(S_W)\fsqrt(S_B) - R_B^\dagger Z_B R_B R_W^\dagger Z_W R_W R_B^\dagger Z_B R_B\| \\
  \leq \|\fsqrt(S_B) - R_B^\dagger Z_B R_B\|\|\finv(S_W)\fsqrt(S_B)\| \\ + \| R_B^\dagger Z_B R_B\|\|\finv(S_W) - R_W^\dagger Z_W R_W\| \|\fsqrt(S_B)\| \\ + \| R_B^\dagger Z_B R_B R_W^\dagger Z_W R_W\|\|\fsqrt(S_B) - R_B^\dagger Z_B R_B\|,
\end{multline*}
each of which is bounded by $\eps\|B\|^2/\sigma^2$.
Next, we approximate $\|R_BR_W^\dagger - R_B'R_W'^\dagger\|_\fr \leq \eps\sigma^{3/2}\sqrt{\|B\|}$, since then
\begin{align*}
  &\|\bar{\Sigma}_B^{\frac12}\bar{U}_B^\dagger R_B R_W^\dagger \bar{U}_W\bar{\Sigma}_W^{\frac12} - \bar{\Sigma}_B^{\frac12}\bar{U}_B^\dagger R_B' R_W'^\dagger \bar{U}_W\bar{\Sigma}_W^{\frac12}\| \\
  &\leq \|\bar{\Sigma}_B^{\frac12}\bar{U}_B^\dagger\| \|R_B R_W^\dagger - R_B' R_W'^\dagger\| \|\bar{U}_W\bar{\Sigma}_W^{\frac12}\| \\
  &\leq \sigma^{-\frac12} \|R_B R_W^\dagger - R_B' R_W'^\dagger\| \sigma^{-2} \\
  &\leq \eps\sqrt{\|B\|/\sigma^2},
\end{align*}
and so
\begin{align*}
  \|R_B^\dagger Z_B R_B R_W^\dagger Z_W R_W R_B^\dagger Z_B R_B - R_B^\dagger Z_B R_B' R_W'^\dagger Z_W R_W' R_B'^\dagger Z_B R_B\| \lesssim \eps\|B\|^2/\sigma^2.
\end{align*}
Now, we can compute $Z \coloneqq Z_B R_B' R_W'^\dagger Z_W R_W' R_B'^\dagger Z_B$ and, using that $Z_B = Z_B [C_B]_{\frac{\sigma}{\sqrt{2}}} [C_B]_{\frac{\sigma}{\sqrt{2}}}^+$, rewrite
\begin{align*}
  R_B^\dagger Z R_B
  = R_B^\dagger ([C_B]_{\frac{\sigma}{\sqrt{2}}}^+)^\dagger [C_B]_{\frac{\sigma}{\sqrt{2}}}^\dagger Z [C_B]_{\frac{\sigma}{\sqrt{2}}} [C_B]_{\frac{\sigma}{\sqrt{2}}}^+ R_B.
\end{align*}
By \cref{lem:app-orth-expression}, $([C_B]_{\frac{\sigma}{\sqrt{2}}}^+ R_B)^\dagger$ is an $\eps\sigma/\|B\|$-approximate projective isometry\footnote{We get more than we need here: an $\eps$-approximate projective isometry would suffice for the subsequent arguments.} onto the image of $[C_B]_{\frac{\sigma}{\sqrt{2}}}^+$ (where we use that $\eps < \sigma/\|B\|$).
To turn this approximate projective isometry into an isometry, we compute the eigendecomposition $[C_B]_{\frac{\sigma}{\sqrt{2}}}^\dagger Z [C_B]_{\frac{\sigma}{\sqrt{2}}} = V\Sigma V^\dagger$, where we truncate so that $V$ is full rank.
Consequently, $U \coloneqq R_B^\dagger ([C_B]_{\frac{\sigma}{\sqrt{2}}}^+)^\dagger V$ is full rank---the image of $V$ is contained in the image of $[C_B]_{\frac{\sigma}{\sqrt{2}}}^+$---and thus is an $\eps\sigma/\|B\|$-approximate isometry.
So, our eigenvectors are $U$ and our eigenvalues are $D \coloneqq \Sigma$.
This satisfies the desired bounds because
\begin{align*}
  \|\fsqrt(S_B)\finv(S_W)\fsqrt(S_B)U - UD\|
  &\leq \|\fsqrt(S_B)\finv(S_W)\fsqrt(S_B)U - UDU^\dagger U\| + \|UDU^\dagger U - UD\| \\
  &\leq \eps\frac{\|B\|^2}{\sigma^2}\|U\| + \|UD\|\|U^\dagger U - I\| \lesssim \eps\frac{\|B\|^2}{\sigma^2}.
\end{align*}
The eigenvalues are correct because, by the approximate isometry condition, $\|U - \tilde{U}\| \lesssim \eps\frac{\sigma}{\|B\|}$ for $\tilde{U}$ an isometry, and so we can use \cref{lem:apporth-facts} to conclude
\begin{align*}
  &\|\fsqrt(S_B)\finv(S_W)\fsqrt(S_B) - \tilde{U}D\tilde{U}^\dagger\| \\
  &\leq \|\fsqrt(S_B)\finv(S_W)\fsqrt(S_B) - UDU^\dagger\| + \|UDU^\dagger - \tilde{U}D\tilde{U}^\dagger\| \\
  &\lesssim \eps\frac{\|B\|^2}{\sigma^2} + \eps\frac{\sigma}{\|B\|}\|D\| \lesssim \eps\frac{\|B\|^2}{\sigma^2}.
\end{align*}
$\tilde{U}D\tilde{U}^\dagger$ is an eigendecomposition.
Furthermore, this is an approximation of a Hermitian PSD matrices, where singular value error bounds align with eigenvalue error bounds.
So, Weyl's inequality (\cref{lem:weylineq}) implies the desired bound $\abs{D_{ii} - \lambda_i} \lesssim \eps\frac{\|B\|^2}{\sigma^2}$ for $\lambda_i$ the true eigenvalues.

We have $\sq_\phi(U(\cdot,i))$ by \cref{lem:b-sq-approx,lemma:sample-Mv}, since $U(\cdot, i) = R_B^\dagger ([C_B]_{\frac{\sigma}{\sqrt{2}}}^+)^\dagger V(\cdot,i)$.
The runtime is $\sqrun_\phi(U(\cdot, i)) = r_B\phi\log\frac{1}{\delta}$, where
\begin{align*}
  \phi = r_B\frac{\sum_{j=1}^{r_B} \|R_B(j,\cdot)\|^2\abs{[([C_B]_{\frac{\sigma}{\sqrt{2}}}^+)^\dagger V(\cdot,i)](j)}^2}{\|U(\cdot,i)\|}
  \lesssim \|B\|_\fr^2\|([C_B]_{\frac{\sigma}{\sqrt{2}}}^+)^\dagger V(\cdot,i)\|^2
  \lesssim \frac{\|B\|_\fr^2}{\sigma^2}.
\end{align*}
This gives the stated runtime.
\end{proof} 


\section{More singular value transformation}
\label{sec:cur}

In this section, we present more general versions of our algorithm for even SVT to get results for generic SVT (\cref{thm:gen-svt}) and eigenvalue transformation (\cref{thm:eig-svt}).
In applications we mainly use even SVT to allow for more fine-tuned control over runtime, but we do use eigenvalue transformation in \cref{sec:SDP}.

For generic SVT: consider a matrix $A \in \bbc^{m\times n}$ and a function $f: [0,\infty) \to \bbc$ satisfying $f(0) = 0$ (so the singular value transformation $f^\svt(A)$ is well-defined as in \cref{def:svt}).
Given $\sq(A)$ and $\sq(A^\dagger)$, we give an algorithm to output a CUR decomposition approximating $f^\svt(A)$.

\begin{restatable}[Generic singular value transformation]{theorem}{gensvt}
\label{thm:gen-svt}
Let $A \in \bbc^{m\times n}$ be given with both $\sq_\phi(A)$ and $\sq_\phi(A^\dagger)$ and let $f: [0,\infty) \to \bbc$ be a function such that $f(0) = 0$, $g(x) \coloneqq f(\sqrt{x})/\sqrt{x}$ is $L$-Lipschitz, and $\bar{g}(x) \coloneqq (g(x)-g(0))/x$ is $\bar{L}$-Lipschitz.
Then, for $0 < \eps \leq \min(L\|A\|^3, \bar{L}\|A\|^5)$, we can output sketches $R \coloneqq SA \in \bbc^{r\times n}$ and $C \coloneqq AT \in \bbc^{m\times c}$, along with $M \in \bbc^{r\times c}$ such that
\begin{align*}
    \Pr\Big[\|CMR + g(0)A - f^\svt(A)\| > \eps \Big] < \delta,
\end{align*}
with $r = \bOt{\phi^2L^2\|A\|^2\|A\|_\fr^4\frac{1}{\eps^2}\log\frac{1}{\delta}}$ and $c = \bOt{\phi^2L^2\|A\|^4\|A\|_\fr^2\frac{1}{\eps^2}\log\frac{1}{\delta}}$.
Finding $S$, $M$, and $T$ takes time
\begin{multline*}
    \tilde{\mathcal{O}}\Big((\bar{L}^2\|A\|^8\|A\|_\fr^2 + L^2\|A\|^2\|A\|_\fr^4)\frac{\phi^2}{\eps^2}\log\frac1\delta(\scb_\phi(A)+\scb_\phi(A^\dagger)+\pcb_\phi(A)+\pcb_\phi(A^\dagger)) \\
    + (L^2\bar{L}^2\|A\|^{12}\|A\|_\fr^4 + L^4\|A\|^6\|A\|_\fr^6)\frac{\phi^4}{\eps^4}\log^2\frac{1}{\delta}\qcb(A) \\
    + (L^4\bar{L}^2\|A\|^{16}\|A\|_\fr^6 + L^6\|A\|^{10}\|A\|_\fr^8)\frac{\phi^6}{\eps^6}\log^3\frac{1}{\delta} + \ncb_\phi(A)\Big).
\end{multline*}
\end{restatable}

If we only wish to assume $\sq_\phi(A)$, we can do so by using \cref{lem:mmhp} instead of \cref{prop:appr-mms} in our proof, paying an additional factor of $\frac{1}{\delta}$.

Note that if $\sqcb_\phi(A), \sqcb_\phi(A^\dagger) = \bigO{1}$, then this runtime is dominated by the last term.
Moreover, if $A$ is strictly low-rank, with minimum singular value $\sigma$, or essentially equivalently, if $f(x) = 0$ for $x \leq \sigma$ and so $g(x) = 0$ for $x \leq \sigma^2$, then $L \leq \ell/\sigma^2$ and $\bar{L} = 2\ell/\sigma^4$ for $\ell$ the Lipschitz constant of $f$.
In this case the complexity is
\begin{equation} \label{eqn:gensvt-sigma}
    \bOt{\Big(\frac{\|A\|^{10}\|A\|_\fr^6}{\sigma^{16}} + \frac{\|A\|^{4}\|A\|_\fr^8}{\sigma^{12}}\Big)\Big(\frac{\ell\|A\|}{\eps}\Big)^6\phi^6\log^3\frac{1}{\delta}}.
\end{equation}
Importantly, when $\eps = \bigO{\ell\|A\|}$ (that is, if we want relative error), this runtime is independent of dimension.
If one desires greater generality, where we only need to depend on the Lipschitz constant of $f$, we can use a simple trick: as we aim for a spectral norm bound, we can essentially treat $A$ as if it had strictly low rank.
Consider the variant of $f$, $f_{\geq\sigma}$, which is zero below $\sigma/2$, $f$ above $\sigma$, and is a linear interpolation in between.
\begin{equation*}
    f_{\geq\sigma}(x) \coloneqq \begin{cases}
        0 & 0 \leq x < \sigma/2 \\
        (2x/\sigma - 1)f(\sigma) & \sigma/2 \leq x < \sigma \\
        f(x) & \sigma \leq x
    \end{cases}
\end{equation*}
Then $\|f^\svt(A) - f_{\geq\eps/\ell}^\svt(A)\| \leq \eps$, because $f(\eps/\ell) \leq \eps$.
Further, the Lipschitz constant of $f_{\geq\eps/\ell}$ is at most $2\ell$: the slope of the linear interpolation is $2f(\sigma)/\sigma \leq 2\ell\sigma/\sigma$.
So, we can run our algorithm for arbitrary $\ell$-Lipschitz $f$ in the time given by \cref{eqn:gensvt-sigma}, with $\sigma = \eps/\ell$.

Our proof strategy is to apply our main result \cref{thm:evenSing} to $g(A^\dagger A)$, for $g(x) \coloneqq f(\sqrt{x})/\sqrt{x}$, and subsequently approximate matrix products with \cref{prop:appr-mms} to get an approximation of the form $A'R'^\dagger U R + g(0)A$:
\begin{align*}
    f^\svt(A) = A g(A^\dagger A) \approx A R^\dagger U R + A (g(0)I) \approx A'R'^\dagger U R + g(0)A.
\end{align*}
Here, $A'R'^\dagger UR$ is a CUR decomposition as desired, since $A'$ is a normalized subset of columns of $A$.
One could further approximate $g(0)A$ by a CUR decomposition if necessary (e.g.\ by adapting the eigenvalue transformation result below).

We do not use this theorem in our applications.
Sometimes we implicitly use a similar strategy (e.g.\ in \cref{sec:matrix-inversion}), but because we apply our matrix to a vector ($f(A^\dagger) b$) we can use \cref{lemma:inner-prod} instead of \cref{prop:appr-mms} when approximating.
This allows for the algorithm to work with only $\sq_\phi(A)$ and still achieve a poly-logarithmic dependence on $\frac{1}{\delta}$.

\begin{proof}
  If we want to compute $\hat{f}^\svt(A)$, we can work with $f(x) \coloneqq \hat{f}(x) - g(0)x$, so that $g(0) = 0$ without loss of generality.
  Notice that $\hat{f}^\svt(A) = f^\svt(A) + g(0)A$, so if we get a CUR decomposition for $f^\svt(A)$ we can add $g(0)A$ after to get the decomposition in the theorem statement.

  Consider the SVT $g(x) \coloneqq f(\sqrt{x})/\sqrt{x}$, so that $f^\svt(A) = Ag(A^\dagger A)$.
  First, use \cref{thm:evenSing} to get $SA \in \bbc^{r\times n}$, $SAT \in \bbc^{r\times c}$ such that, with probability $\geq 1-\delta/4$,
  \begin{equation}\label{eq:genSVTErr1}
    \|(SA)^\dagger \bar{g}((SAT)(SAT)^\dagger)SA - g(A^\dagger A)\| \leq \frac{\eps}{2\|A\|}.
  \end{equation}
  Second, use \cref{prop:appr-mms} to get a sketch $T'^\dagger \in \bbc^{c'\times n}$ such that, with probability $\geq 1-\delta/4$,
  \begin{equation}\label{eq:genSVTErr2}
    \|A(SA)^\dagger - AT'(SAT')^\dagger\| \leq \eps(3L\|A\|)^{-1}.
  \end{equation}
  The choices of parameters necessary are as follows (using that $\|SA\|_\fr = O(\|A\|_\fr)$ by \cref{eq:evenSVTbdR} and we have a $2\phi$-oversampled distribution for $(SA)^\dagger$ by \cref{lem:sq-sketching}):
  \begin{align*}
    r &= \tilde{\Theta}\Big(\phi^2L^2\|A\|^4\|A\|_\fr^2\frac{1}{\eps^2}\log\frac{1}{\delta}\Big) \\
    c &= \tilde{\Theta}\Big(\phi^2\bar{L}^2\|A\|^8\|A\|_\fr^2\frac{1}{\eps^2}\log\frac{1}{\delta}\Big) \\
    c' &= \tilde{\Theta}\Big(\phi^2L^2\|A\|^2\|A\|_\fr^4\frac{1}{\eps^2}\log\frac{1}{\delta}\Big)
  \end{align*}
  This implies the desired bound through the following sequence of approximations:
  \begin{align*}
    f^\svt(A) &= A g(A^\dagger A) \\
    &\approx A(SA)^\dagger \bar{g}((SAT)(SAT)^\dagger)SA \\
    &\approx \underbrace{AT'}_{C}\underbrace{(SAT')^\dagger\bar{g}((SAT)(SAT)^\dagger)}_{M}\underbrace{SA}_R.
  \end{align*}
  This gives us a CUR decomposition of $f^\svt(A)$.
  These two approximations only incur $\bigO{\eps}$ error in spectral norm; for the first, notice that
  \begin{align*}
    &\|A g(A^\dagger A) - A(SA)^\dagger \bar{g}((SAT)(SAT)^\dagger)SA\| \\
    &\leq \|A\|\|(SA)^\dagger \bar{g}((SAT)(SAT)^\dagger)SA - g(A^\dagger A)\| \leq \frac{\eps}{2} \tag{by \eqref{eq:genSVTErr1}}.
  \end{align*}    
   For the second approximation observe that $\abs{g(x)} \leq L\abs{x}$ (and, by corollary, $\bar{g}(x) \leq L$) due to $g$ being $L$-Lipschitz and $g(0) = 0$, therefore
   \begin{align*}
    &\|(A(SA)^\dagger - AT'(SAT')^\dagger)\bar{g}((SAT)(SAT)^\dagger)SA\| \\
    &\leq \|AT'(SAT')^\dagger - A(SA)^\dagger\|\|\sqrt{\bar{g}((SAT)(SAT)^\dagger)}\|\|\sqrt{\bar{g}((SAT)(SAT)^\dagger)}SA\| \\
    &\leq \eps(3L\|A\|)^{-1}\|\sqrt{\bar{g}((SAT)(SAT)^\dagger)}\|\|\sqrt{\bar{g}((SAT)(SAT)^\dagger)}SA\| \tag{by \eqref{eq:genSVTErr2}}\\
    &\leq  \eps(3L\|A\|)^{-1}\sqrt{L}\sqrt{\|g(A^\dagger A)\|+\frac{\eps}{2\|A\|}} \tag{since $\bar{g}(x) \leq L$ and by \eqref{eq:genSVTErr1}} \\
    &\leq \eps(3L\|A\|)^{-1} \sqrt{\frac32} L\|A\| < \frac{\eps}{2}. \tag{since $\abs{g(x)} \leq L\abs{x}$ and $\eps \leq L\|A\|^3$}
  \end{align*}

  The time complexity of this procedure is
  \begin{align*}
    \bigO{\ncb_\phi(A) + (r+c+c')(\scb_\phi(A) + \qcb_\phi(A)) + c'(\scb_\phi(A^\dagger) + \qcb_\phi(A^\dagger)) + (rc + rc')\qcb(A) + r^2c + r^2c'},
  \end{align*}
  which comes from producing sketches, querying all the relevant entries of $SAT$ and $SAT'$, the singular value transformation of $SAT$, and the matrix multiplication in $M$.
  We get $r$ factors in the latter two terms because we can separate $\bar{g}((SAT)(SAT)^\dagger) = \sqrt{\bar{g}}(SAT)(\sqrt{\bar{g}}(SAT))^\dagger$ where $\sqrt{\bar{g}}(x) \coloneqq \sqrt{\bar{g}(x)}$.
\end{proof}

As for eigenvalue transformation, consider a function $f: \bbr \to \bbc$ and a Hermitian matrix $A \in \bbc^{n\times n}$, given $\sq(A)$.
If $f$ is even (so $f(x) = f(-x)$), then $f(A) = f(\sqrt{A^\dagger A})$, so we can use \cref{thm:evenSing} to compute the eigenvalue transform $f(A)$.
For non-even $f$, we cannot use this result, and present the following algorithm to compute it.

\begin{restatable}[Eigenvalue transformation]{theorem}{eigsvt} \label{thm:eig-svt}
Suppose we are given a Hermitian $\sq_\phi(A) \in \bbc^{n\times n}$ with eigenvalues $\lambda_1\geq \cdots \geq \lambda_n$, a function $f: \bbr \to \bbc$ that is $L$-Lipschitz on $\cup_{i=1}^n [\lambda_i - d, \lambda_i+d]$ for some $d > \frac{\eps}{L}$, and some $\eps \in (0,L\|A\|/2]$.
Then we can output matrices $S \in \bbc^{r\times n}$, $N \in \bbc^{s \times r}$, and $D \in \bbc^{s\times s}$, with $r = \bOt{\phi^2\|A\|^4\|A\|_\fr^2\frac{L^6}{\eps^6}\log\frac{1}{\delta}}$ and $s = \bOt{\|A\|_\fr^2\frac{L^2}{\eps^2}}$, such that
\begin{align*}
    \Pr\Big[\|(SA)^\dagger N^\dagger D N (SA) + f(0)I - f^\evt(A)\| > \eps\Big] < \delta,
\end{align*}
in time
\begin{multline*}
    \tilde{\mathcal{O}}\Big((L^{10}\eps^{-10}\|A\|^8\|A\|_\fr^2\phi^2\log\frac{1}{\delta} + L^6\eps^{-6}\|A\|_\fr^6\phi\log\frac{1}{\delta})(\scb_\phi(A)+\pcb_\phi(A)) \\
    + (L^{16}\eps^{-16}\|A\|^{12}\|A\|_\fr^4\phi^4\log^2\frac{1}{\delta}+L^{18}\eps^{-18}\|A\|^8\|A\|_\fr^{10}\phi^5\log^3\frac{1}{\delta})\qcb(A) \\
    + L^{22}\eps^{-22}\|A\|^{16}\|A\|_\fr^6\phi^6\log^3\frac{1}{\delta} + \ncb_\phi(A)\Big).
\end{multline*}
Moreover, this decomposition satisfies the following two properties.
First, $NSA$ is an approximate isometry: $\|(NSA)(NSA)^\dagger - I\| \leq (\frac{\eps}{L\|A\|})^3$.
Second, $D$ is a diagonal matrix and its diagonal entries satisfy $\abs{D(i,i)+f(0) - f(\lambda_i)} \leq \eps$ for all $i \in [n]$ (where $D(i,i) \coloneqq 0$ for $i > s$).
\end{restatable}
Under the reasonable assumptions\footnote{The correct way to think about $\eps$ is as some constant fraction of $L\|A\|$.
If $\eps > L\|A\|$ then $f(0)I$ is a satisfactory approximation.
The bound we give says that we want an at least $\|A\|_\fr/\|A\|$ improvement over trivial, which is modest in the close-to-low-rank regime that we care about.
Similar assumptions appear in \cref{sec:applications}.
} that $\sqcb(A)$ and $\phi$ are small ($\bigO{1}$, say) and $\eps \leq L\|A\|\frac{\|A\|}{\|A\|_\fr}$, the complexity of this theorem is $\bOt{L^{22}\eps^{-22}\|A\|^{16}\|A\|_\fr^6\log^3\frac{1}{\delta}}$.

We now outline our proof.
Our strategy is similar to the one used for quantum-inspired semidefinite programming~\cite{chia2019QInspiredSubLinLowRankSDPSolver}: first we find the eigenvectors and eigenvalues of $A$ and then apply $f$ to the eigenvalues.
Let $\pi(x)$ be a (smoothened) step function designed so it can zeroes out small eigenvalues of $A$ (in particular, eigenvalues smaller than $\eps/\sqrt{2}L$).
Then
\begin{align*}
    A &\approx \pi(A^\dagger A) A \pi(A^\dagger A) \tag*{by definition of $\pi$} \\
    &\approx R^\dagger \bar{\pi}(CC^\dagger) R A R^\dagger \bar{\pi}(CC^\dagger) R \tag*{by \cref{thm:evenSing}} \\
    &\approx R^\dagger \bar{\pi}(CC^\dagger) M \bar{\pi}(CC^\dagger) R \tag*{by sketching $M \approx RAR^\dagger$} \\
    &= R^\dagger (C_\sigma C_\sigma^+)^\dagger \bar{\pi}(CC^\dagger) M \bar{\pi}(CC^\dagger) C_\sigma C_\sigma^+ R. \tag*{where $\sigma = \eps/\sqrt{2}L$}
\end{align*}
Here, $C_\sigma$ is the low-rank approximation of $C$ formed by transforming $C$ according to the ``filter'' function on $x$ that is $0$ for $x < \sigma$ and $x$ otherwise.
$\hat{U} \coloneqq C_\sigma^+ R \in \bbc^{c \times n}$ is an approximate isometry by \cref{lem:app-orth-expression}.
We are nearly done now: since the rest of the matrix expression, $C_\sigma^\dagger \bar{\pi}(CC^\dagger) M \bar{\pi}(CC^\dagger) C_\sigma \in \bbc^{c\times c}$, consists of submatrices of $A$ of size independent of $n$, we can directly compute its unitary eigendecomposition $UDU^\dagger$.
This gives the approximate decomposition $A \approx (\hat{U}U)D(\hat{U}U)^\dagger$, with $\hat{U}U$ and $D$ acting as approximate eigenvectors and eigenvalues of $A$, respectively.
An application of \cref{lem:lipschitz-oper} shows that $f(A) \approx (\hat{U}U) f(D) (\hat{U}U)^\dagger$ in the desired sense.
Therefore, our output approximation of $f(A)$ comes in the form of an RUR decomposition that can be rewritten in the form of an approximate eigendecomposition.
The only major difference between this proof sketch and the proof below is that we perform our manipulations on the SVD of $C_\sigma$, to save on computation time: note that the SVD can be made small in dimension, using that the rank of $C_\sigma$ is bounded by $\|C\|_\fr^2/\sigma^2$.

\begin{proof}
  Throughout this proof $\eps$ is not dimensionless; if choices of parameters are confusing, try replacing $\eps$ with $\eps\|A\|$.
  We will take $f(0) = 0$ without loss of generality.
  First, consider the ``smooth projection'' singular value transformation
  \begin{align*}
    \pi(x) = \begin{cases}
      0 & x < \frac{\eps^2}{2L^2} \\
      \frac{2L^2}{\eps^2}x - 1 & \frac{\eps^2}{2L^2} \leq x < \frac{\eps^2}{L^2} \\
      1 & \frac{\eps^2}{L^2} \leq x
    \end{cases}
  \end{align*}
  Since $\pi$ is a projector onto the large eigenvectors of $A$, we can add these projectors to our expression without incurring too much spectral norm error.
  \begin{equation*}
    \|\pi(A^\dagger A) A \pi(A^\dagger A) - A\|
    = \max_{i \in [n]} \abs{\pi(\lambda_i^2)\lambda_i \pi(\lambda_i^2) - \lambda_i} \leq \eps/L
  \end{equation*}
  Second, use \cref{thm:evenSing} to get $SA \in \bbc^{r\times n}$, $SAT \in \bbc^{r\times c}$ such that, with probability $\geq 1-\delta$,
  \begin{equation*}
    \|(SA)^\dagger \bar{\pi}((SAT)(SAT)^\dagger)SA - \pi(A^\dagger A)\| \leq \frac{\eps}{L\|A\|}.
  \end{equation*}
  The necessary sizes for these bounds to hold are as follows (Lipschitz constants for $\pi$ are $2L^2/\eps^2$ and $4L^4/\eps^4$, $\|SA\|_\fr = O(\|A\|_\fr)$ by \cref{eq:evenSVTbdR}, and we have a $2\phi$-oversampled distribution for $(SA)^\dagger$ by \cref{lem:sq-sketching}):\footnote{The constraint on the size of $\eps$ from \cref{thm:evenSing} here is $\eps(L\|A\|)^{-1} \lesssim \min(L^2\|A\|^2/\eps^2, L^4\|A\|^4/\eps^4)$, which is true since $\eps \leq L\|A\|/2$.}
  \begin{align*}
    r &= \tilde{\Theta}\Big(\phi^2\frac{L^4}{\eps^4}\|A\|^2\|A\|_\fr^2\frac{L^2\|A\|^2}{\eps^2}\log\frac{1}{\delta}\Big)
    = \tilde{\Theta}\Big(\phi^2\|A\|^4\|A\|_\fr^2\frac{L^6}{\eps^6}\log\frac{1}{\delta}\Big), \\
    c &= \tilde{\Theta}\Big(\phi^2\frac{L^8}{\eps^8}\|A\|^6\|A\|_\fr^2\frac{L^2\|A\|^2}{\eps^2}\log\frac{1}{\delta}\Big)
    = \tilde{\Theta}\Big(\phi^2\|A\|^8\|A\|_\fr^2\frac{L^{10}}{\eps^{10}}\log\frac{1}{\delta}\Big).
  \end{align*}
  This approximation does not incur too much error:
  \begin{align*}
    &\|R^\dagger \bar{\pi}(CC^\dagger)RAR^\dagger \bar{\pi}(CC^\dagger)R - \pi(A^\dagger A )A\pi(A^\dagger A)\| \\
    &\leq \|\pi(A^\dagger A)A(\pi(A^\dagger A) - R^\dagger \bar{\pi}(CC^\dagger)R)\| + \|(\pi(A^\dagger A) - R^\dagger \bar{\pi}(CC^\dagger)R)AR^\dagger\bar{\pi}(CC^\dagger)R\| \\
    &\leq \frac{\eps}{L\|A\|}\Big(\|\pi(A^\dagger A)\|\|A\| + \|A\|\|R^\dagger \bar{\pi}(CC^\dagger)R\|\Big)\leq \frac{\eps}{L\|A\|}\Big(\|A\| + \|A\|\Big(1+\frac{\eps}{L\|A\|}\Big)\Big)
    \leq 3\frac{\eps}{L}.
  \end{align*}
  Third, apply \cref{lemma:xAy}(b) $r^2$ times to approximate each entry of $RAR^\dagger$: pull $t$ samples from $\sq_\phi(A)$ for $t \coloneqq \bigO{\phi\|A\|_\fr^6\frac{L^6}{\eps^6}\log\frac{r^2}{\delta}}$ such that, given some $\q(x), \q(y)$, with probability $\geq 1-\frac{\delta}{r^2}$, one can output an estimate of $x^\dagger A y$ up to $\frac{\eps^3\|x\|\|y\|}{L^3\|A\|_\fr^2}$ additive error with \emph{no additional queries to} $\sq_\phi(A)$.
  Then, by union bound, with probability $\geq 1-\delta$, using the same $t$ samples from $A$ each time, one can output an estimate of $R(i,\cdot) A R(j,\cdot)^\dagger$ up to $\frac{\eps^3\|R(i,\cdot)\|\|R(j,\cdot)\|}{L^3\|A\|_\fr^2}$ error for all $i,j \in [r]$ such that $i \leq j$.
  Let $M$ be the matrix of these estimates.
  Then, using that $\|R\|_\fr = \bigO{\|A\|_\fr}$ from \cref{eq:evenSVTbdR},
  \begin{align*}
    \|M - RAR^\dagger\|_\fr^2
    \leq \sum_{i=1}^r \sum_{j=1}^r \Big(\frac{\eps^3\|R(i,\cdot)\|\|R(j,\cdot)\|}{L^3\|A\|_\fr^2}\Big)^2
    = \frac{\eps^6\|R\|_\fr^4}{L^6\|A\|_\fr^4}
    \lesssim \frac{\eps^6}{L^6}.
  \end{align*}
  From \cref{eq:evenSVTbdRC,eq:evenSVTbdC},
  \begin{align*}
    \|R^\dagger \bar{\pi}(CC^\dagger)(RAR^\dagger-M) \bar{\pi}(CC^\dagger)R\|
    \lesssim \frac{\eps^3}{L^3}\|R^\dagger \bar{\pi}(CC^\dagger)\|^2
    \leq \frac{\eps^3}{L^3}\Big(1+\frac{\eps}{L\|A\|}\Big)\frac{L^2}{\eps^2}
    \lesssim \frac\eps L.
  \end{align*}
  So far, we have shown that we can find an RUR approximation to $A$, with
  \begin{align*}
    \|R^\dagger \bar{\pi}(CC^\dagger)M\bar{\pi}(CC^\dagger)R - A\| \lesssim \frac{\eps}{L}
  \end{align*}
  However, if we wish to apply an eigenvalue transformation to $A$, we need to access the eigenvalues of $A$ as well.
  To do this, we will express this decomposition as an approximate unitary eigendecomposition.
  Using that $\bar{\pi}$ zeroes out singular values that are smaller than $\frac{\eps^2}{2L^2}$, we can write our expression as $\hat{U}\hat{D}\hat{U}^\dagger$, for $\hat{U} \in \bbC^{n\times s}$ and $\hat{D} \in \bbC^{s\times s}$:
  \begin{align} 
    &R^\dagger \bar{\pi}(CC^\dagger)M\bar{\pi}(CC^\dagger) R \nonumber \\
    &= \big(C_{\frac{\eps}{\sqrt{2}L}}^+R\big)^\dagger\big(C_{\frac{\eps}{\sqrt{2}L}}^\dagger \bar{\pi}(CC^\dagger)M\bar{\pi}(CC^\dagger) C_{\frac{\eps}{\sqrt{2}L}}\big)\big(C_{\frac{\eps}{\sqrt{2}L}}^+R\big)\nonumber \\
    &= \underset{\hat{U}}{\underbrace{\big(
      R^\dagger U^{(C)}_{\frac{\eps}{\sqrt{2}L}}(D^{(C)}_{\frac{\eps}{\sqrt{2}L}})^{-1}
    \big)}}\underset{\hat{D}}{\underbrace{\big(
      D^{(C)}_{\frac{\eps}{\sqrt{2}L}}(U^{(C)}_{\frac{\eps}{\sqrt{2}L}})^\dagger \bar{\pi}(CC^\dagger)M\bar{\pi}(CC^\dagger) U^{(C)}_{\frac{\eps}{\sqrt{2}L}} D^{(C)}_{\frac{\eps}{\sqrt{2}L}}
    \big)}}\underset{\hat{U}^\dagger}{\underbrace{\big(
      (D^{(C)}_{\frac{\eps}{\sqrt{2}L}})^{-1} (U^{(C)}_{\frac{\eps}{\sqrt{2}L}})^\dagger R
    \big)}}. \label{eqn:evt-udu}
  \end{align}
  Here, we are using an SVD of $C$ truncated to ignore singular values smaller than $\frac{\eps}{\sqrt{2}L}$, where $U_{\frac{\eps}{\sqrt{2}L}}^{(C)} \in \bbc^{r\times s},\,D_{\frac{\eps}{\sqrt{2}L}}^{(C)} \in \bbc^{s\times s},\,V_{\frac{\eps}{\sqrt{2}L}}^{(C)} \in \bbc^{c\times s}$, where $s$ is the number of singular values of $C$ that are at least $\frac{\eps}{\sqrt{2}L}$.
  Note that, as a result, $s \leq \|C\|_\fr^2/(\eps/\sqrt{2}L)^2 \lesssim \|A\|_\fr^2L^2\eps^{-2}$ and $s \leq \min(r,c,n)$.
  By \cref{lem:app-orth-expression} with our values of $r$ and $c$, we get that $\hat{U} \coloneqq R^\dagger U^{(C)}_{\frac{\eps}{\sqrt{2}L}}(D^{(C)}_{\frac{\eps}{\sqrt{2}L}})^{-1}$ is an $\bigO{\frac{\eps^3}{L^3\|A\|^3}}$-approximate isometry; we rescale $\eps$ until this is at most $\frac12$.

  The rest of this error analysis will show that, since $\hat{U}$ is an approximate isometry, $f(A) \approx \hat{U} f(\hat{D})\hat{U}^\dagger$ in the senses required for the theorem statement.
  Though $\hat{D}$ is not diagonal, since it is $s\times s$, we can compute its unitary eigendecomposition $U^{(\hat{D})}D^{(\hat{D})}(U^{(\hat{D})})^\dagger$; so, we can take $D \coloneqq f(D^{(\hat{D})})$ and $N \coloneqq (U^{(\hat{D})})^\dagger (D^{(C)}_{\frac{\eps}{\sqrt{2}L}})^+ (U^{(C)}_{\frac{\eps}{\sqrt{2}L}})^\dagger$ to get the decomposition in the theorem statement.
  (Including the isometry $(U^{(\hat{D})})^\dagger$ in our expression for $\hat{U}$ does not change the value of $\alpha$).

  First, consider the eigenvalues of $\hat{D}$. Note that $\|\hat{U}^\dagger \hat{U} - I\| \leq \alpha$ since $\hat{U}$ is an $\alpha$-approximate isometry, and by \cref{lem:apporth-facts}, there exists an isometry $U$ such that $\|U - \hat{U}\| \leq \alpha$.
  We first observe that, using \cref{lem:apporth-facts} and our bound on $\alpha$,
  \begin{align*}
    \|A - U\hat{D}U^\dagger\|
    &\leq \|A - \hat{U}\hat{D}\hat{U}^\dagger\| + \|\hat{U}\hat{D}\hat{U}^\dagger - U\hat{D}U^\dagger\| \\
    &\leq \frac{\eps}{L} + \alpha\frac{2-\alpha}{(1-\alpha)^2}\|\hat{U}\hat{D}\hat{U}^\dagger\| \\
    &\leq \frac{\eps}{L} + \alpha\frac{2-\alpha}{(1-\alpha)^2}\Big(\|A\| + \frac{\eps}{L}\Big)
    \lesssim \frac{\eps}{L}.
  \end{align*}
  Consequently, by Weyl's inequality (\cref{lem:weylineq}), the eigenvalues of $U\hat{D}U^\dagger$, $\hat{\lambda}_1 \geq \cdots \geq \hat{\lambda}_n$ satisfy $\abs{\lambda_i - \hat{\lambda_i}} \lesssim \frac{\eps}{L}$ for all $i \in [n]$, and by assumption, $f$ is $L$-Lipschitz on the spectrums of $A$ and $U\hat{D}U$.
  From this, we can conclude that we can compute estimates for the eigenvalues of $f(A)$, since the eigenvalues of $U\hat{D}U^\dagger$ are the eigenvalues of $\hat{D}$ (padded with zero eigenvalues) which we know from our eigendecomposition of $\hat{D}$
  Further, our estimates $f(\hat{\lambda}_i)$ satisfy the desired bound $\abs{f(\hat{\lambda}_i) - f(\lambda_i)} \leq \eps$.
  Finally, since $f$ is Lipschitz on our spectrums of concern, the desired error bound for our approximation holds by the following computation (which uses \cref{lem:apporth-facts} extensively):
  \begin{align*}
    \|f(A) - \hat{U}f(\hat{D})\hat{U}^\dagger\|
    &\leq \|f(A) - Uf(\hat{D})U^\dagger\| + \|Uf(\hat{D})U^\dagger - \hat{U}f(\hat{D})\hat{U}^\dagger\| \\
    &\leq \|f(A) - Uf(\hat{D})U^\dagger\| + (2\alpha + \alpha^2)\|f(\hat{D})\| \\
    &\lesssim L(\|A - U\hat{D}U^\dagger\| + (2\alpha + \alpha^2)\|\hat{D}\|)\log s \tag*{by \cref{lem:lipschitz-oper}} \\
    &\leq L(\|A - \hat{U}\hat{D}\hat{U}^\dagger\| + 2(2\alpha + \alpha^2)\|\hat{D}\|)\log s \\
    &\leq L\Big(\frac{\eps}{L} + \frac{2(2\alpha + \alpha^2)}{(1-\alpha)^2}\|\hat{U}\hat{D}\hat{U}^\dagger\|\Big)\log s \\
    &\leq L\Big(\frac{\eps}{L} + \frac{2(2\alpha + \alpha^2)}{(1-\alpha)^2}\Big(\|A\| + \frac{\eps}{L}\Big)\Big)\log s\lesssim \eps\log s. \tag*{by $\alpha \leq \frac{\eps}{L\|A\|}$}
  \end{align*}
  Finally, we rescale $\eps$ down by $\log^2 s$ so that this final bound is $\bigO{\eps}$.
  This term is folded into the $\polylog$ terms of $r$, $c$, and $s$.
  (We need to scale by more than $\log s$ because $s$ has a dependence on $\frac{1}{\eps^2}$.)
  This completes the error analysis.

  The complexity analysis takes some care: we want to compute our matrix expressions in the correct order.
  First, we will sample to get $S$ and $T$, and then compute the \emph{truncated} singular value decomposition of $C \coloneqq SAT$, which we have denoted 
  $C_{\frac{\eps}{\sqrt{2}L}} = U_{\frac{\eps}{\sqrt{2}L}}^{(C)}D_{\frac{\eps}{\sqrt{2}L}}^{(C)}(V_{\frac{\eps}{\sqrt{2}L}}^{(C)})^\dagger$ for $U_{\frac{\eps}{\sqrt{2}L}}^{(C)} \in \bbc^{r\times s},\,D_{\frac{\eps}{\sqrt{2}L}}^{(C)} \in \bbc^{s\times s},\,V_{\frac{\eps}{\sqrt{2}L}}^{(C)} \in \bbc^{c\times s}$.
  Then, we will perform the inner product estimation protocol $r^2$ times to get our estimate $M \in \bbc^{r\times r}$, and compute the eigendecomposition of
  \begin{align*}
    \hat{D}
    &= D^{(C)}_{\frac{\eps}{\sqrt{2}L}}(U^{(C)}_{\frac{\eps}{\sqrt{2}L}})^\dagger \bar{\pi}(CC^\dagger)M\bar{\pi}(CC^\dagger) U^{(C)}_{\frac{\eps}{\sqrt{2}L}} D^{(C)}_{\frac{\eps}{\sqrt{2}L}} \\
    &= D^{(C)}_{\frac{\eps}{\sqrt{2}L}}\bar{\pi}((D_{\frac{\eps}{\sqrt{2}L}}^{(C)})^2)(U_{\frac{\eps}{\sqrt{2}L}}^{(C)})^\dagger M U_{\frac{\eps}{\sqrt{2}L}}^{(C)}\bar{\pi}((D_{\frac{\eps}{\sqrt{2}L}}^{(C)})^2)D^{(C)}_{\frac{\eps}{\sqrt{2}L}}
  \end{align*}
  via the final expression above, with the truncations propagated through the matrices, to get $\hat{D} = U^{(\hat{D})}D^{(\hat{D})}(U^{(\hat{D})})^\dagger$.
  Then, we compute and output $D = \hat{D}$ and $N = (U^{(\hat{D})})^\dagger (D^{(C)}_{\frac{\eps}{\sqrt{2}L}})^+ (U^{(C)}_{\frac{\eps}{\sqrt{2}L}})^\dagger$.
  By evaluating the expression for $\hat{D}$ from left-to-right, we only need to perform matrix multiplications that (naively) take $s^3$ or $sr^2$ time.
  The only cost of $c$ we incur is in computing the SVD of $C$.
  The runtime is
  \begin{equation*}
    \bigO{(r+c+t)(\scb_\phi(A) + \qcb_\phi(A)) + (rc+r^2t)\qcb(A) + s^3 + r^2s + r^2c + \ncb_\phi(A)}. \qedhere
  \end{equation*}
\end{proof}


\section*{Acknowledgments}

ET thanks Craig Gidney for the reference to alias sampling. AG is grateful to Saeed Mehraban for insightful suggestions about proving perturbation bounds on partition functions. Part of this work was done while visiting the Simons Institute for the Theory of Computing. We gratefully acknowledge the Institute's hospitality.

\bibliography{qc_gily,ref}

\newcommand{\etalchar}[1]{$^{#1}$}
\newcommand{\lName}{1}\newcommand{\arxiv}[1]{arXiv:
  \href{https://arxiv.org/abs/#1}{\ttfamily{#1}}\removefirstdot}\newcommand{\arXiv}[1]{arXiv:
  \href{https://arxiv.org/abs/#1}{\ttfamily{#1}}\removefirstdot}\def\removefirstdot#1{\if.#1{}\else#1\fi}\providecommand{\multiletter}[1]{#1}\renewcommand{\multiletter}[1]{#1}\DeclareRobustCommand{\dutchPrefix}[2]{#2}\providecommand{\dutchPrefix}[2]{#2}\renewcommand{\dutchPrefix}[2]{#2}\newcommand{\skp}[3]{#2}\newcommand{\focs
  }[1]{\if\lName1\skp{ }{Proceedings of the #1 {IEEE} Symposium on Foundations
  of Computer Science ({FOCS})}{ }\else{FOCS}\fi}\newcommand{\stoc
  }[1]{\if\lName1\skp{ }{Proceedings of the #1 {ACM} Symposium on the Theory of
  Computing ({STOC})}{ }\else{STOC}\fi}\newcommand{\soda }[1]{\if\lName1\skp{
  }{Proceedings of the #1 {ACM-SIAM} Symposium on Discrete Algorithms
  ({SODA})}{ }\else{SODA}\fi}\newcommand{\stacs }[1]{\if\lName1\skp{
  }{Proceedings of the #1 Symposium on Theoretical Aspects of Computer Science
  ({STACS})}{ }\else{STACS}\fi}\newcommand{\itcs }[1]{\if\lName1\skp{
  }{Proceedings of the #1 Innovations in Theoretical Computer Science
  Conference ({ITCS})}{ }\else{ITCS}\fi}\newcommand{\fsttcs
  }[1]{\if\lName1\skp{ }{Proceedings of the #1 International Conference on
  Foundations of Software Technology and Theoretical Computer Science
  ({FSTTCS})}{ }\else{FSTTCS}\fi}\newcommand{\mfcs }[1]{\if\lName1\skp{
  }{Proceedings of the #1 International Symposium on Mathematical Foundations
  of Computer Science ({MFCS})}{ }\else{MFCS}\fi}\newcommand{\ccc
  }[1]{\if\lName1\skp{ }{Proceedings of the #1 {IEEE} Conference on
  Computational Complexity ({CCC})}{ }\else{CCC}\fi}\newcommand{\isit
  }[1]{\if\lName1\skp{ }{Proceedings of the #1 {IEEE} International Symposium
  on Information Theory ({ISIT})}{ }\else{ISIT}\fi}\newcommand{\colt
  }[1]{\if\lName1\skp{ }{Proceedings of the #1 Conference On Learning Theory
  ({COLT})}{ }\else{COLT}\fi}\newcommand{\nips }[1]{\if\lName1\skp{ }{Advances
  in Neural Information Processing Systems #1 ({NIPS})}{
  }\else{NIPS}\fi}\newcommand{\aistats }[1]{\if\lName1\skp{ }{Proceedings of
  the #1 International Conference on Artificial Intelligence and Statistics
  ({AISTATS})}{ }\else{AISTATS}\fi}\newcommand{\icml }[1]{\if\lName1\skp{
  }{Proceedings of the #1 International Conference on Machine Learning
  ({ICML})}{ }\else{ICML}\fi}\newcommand{\icalp }[1]{\if\lName1\skp{
  }{Proceedings of the #1 International Colloquium on Automata, Languages, and
  Programming ({ICALP})}{ }\else{ICALP}\fi}\newcommand{\esa
  }[1]{\if\lName1\skp{ }{Proceedings of the #1 Annual European Symposium on
  Algorithms ({ESA})}{ }\else{ESA}\fi}\newcommand{\tqc }[1]{\if\lName1\skp{
  }{Proceedings of the #1 Conference on the Theory of Quantum Computation,
  Communication, and Cryptography ({TQC})}{}\else{TQC}\fi}\newcommand{\isaac
  }[1]{\if\lName1\skp{ }{Proceedings of the #1 International Symposium on
  Algorithms and Computation ({ISAAC})}{ }\else{ISAAC}\fi}\newcommand{\jacm
  }{\if\lName1\skp{ }{Journal of the ACM}{ }\else{J. ACM}\fi}\newcommand{\acmta
  }{\if\lName1\skp{ }{ACM Transactions on Algorithms}{ }\else{{ACM} Tr.
  Alg}\fi}\newcommand{\acmtct }{\if\lName1\skp{ }{ACM Transactions on
  Computation Theory}{ }\else{ACM Tr. Comp. Th.}\fi}\newcommand{\acmtqc
  }{\if\lName1\skp{ }{ACM Transactions on Quantum Computing}{ }\else{ACM Tr.
  Quant. Comp.}\fi}\newcommand{\jams }{\if\lName1\skp{ }{Journal of the AMS}{
  }\else{J. AMS}\fi}\newcommand{\pams }{\if\lName1\skp{ }{Proceedings of the
  AMS}{ }\else{Proc. AMS}\fi}\newcommand{\linalgappl }{\if\lName1\skp{ }{Linear
  Algebra and its Applications}{ }\else{Lin. Alg. \&
  App.}\fi}\newcommand{\jalgo }{\if\lName1\skp{ }{Journal of Algorithms}{
  }\else{J. Alg.}\fi}\newcommand{\jcss }{\if\lName1\skp{ }{Journal of Computer
  and System Sciences}{ }\else{J. Comp. Sys. Sci.}\fi}\newcommand{\cc
  }{\if\lName1\skp{ }{Computational Complexity}{ }\else{Comp.
  Comp.}\fi}\newcommand{\algor }{\if\lName1\skp{ }{Algorithmica}{
  }\else{Alg.}\fi}\newcommand{\comb }{\if\lName1\skp{ }{Combinatorica}{
  }\else{Comb.}\fi}\newcommand{\cacm }{\if\lName1\skp{ }{Communications of the
  ACM}{ }\else{Comm. ACM}\fi}\newcommand{\sigart }{\if\lName1\skp{ }{SIGART
  Bulletin}{ }\else{SIGART Bull.}\fi}\newcommand{\sigactn }{\if\lName1\skp{
  }{SIGACT News}{ }\else{SIGACT News}\fi}\newcommand{\eatcsbul
  }{\if\lName1\skp{ }{Bulletin of the {EATCS}}{ }\else{Bull.
  {EATCS}}\fi}\newcommand{\siamrev }{\if\lName1\skp{ }{SIAM Review}{
  }\else{SIAM Rev.}\fi}\newcommand{\siamjc }{\if\lName1\skp{ }{SIAM Journal on
  Computing}{ }\else{SIAM J. Comp.}\fi}\newcommand{\siamjo }{\if\lName1\skp{
  }{SIAM Journal on Optimization}{ }\else{SIAM J. Opt.}\fi}\newcommand{\siamjdm
  }{\if\lName1\skp{ }{SIAM Journal on Discrete Mathematics}{ }\else{SIAM J.
  Disc. Math.}\fi}\newcommand{\siamjnum }{\if\lName1\skp{ }{SIAM Journal on
  Numerical Analysis}{ }\else{SIAM J. Num. Anal.}\fi}\newcommand{\siamjmathanal
  }{\if\lName1\skp{ }{SIAM Journal on Mathematical Analysis}{ }\else{SIAM J.
  Math. Anal.}\fi}\newcommand{\discmath }{\if\lName1\skp{ }{Discrete
  Mathematics}{ }\else{Disc. Math.}\fi}\newcommand{\das }{\if\lName1\skp{
  }{Discrete Applied Mathematics}{ }\else{Disc. App.
  Math.}\fi}\newcommand{\amatstat }{\if\lName1\skp{ }{Annals of Mathematical
  Statistics}{ }\else{Ann. Math. Stat.}\fi}\newcommand{\rms }{\if\lName1\skp{
  }{Russian Mathematical Surveys}{ }\else{Russ. Math.
  Surv.}\fi}\newcommand{\invmath }{\if\lName1\skp{ }{Inventiones Mathematicae}{
  }\else{Inv. Math.}\fi}\newcommand{\jnumber }{\if\lName1\skp{ }{Journal of
  Number Theory}{ }\else{J. Num. Th.}\fi}\newcommand{\tcs }{\if\lName1\skp{
  }{Theoretical Computer Science}{ }\else{Theor. Comput.
  Sci.}\fi}\newcommand{\toc }{\if\lName1\skp{ }{Theory of Computing}{
  }\else{Th. Comp.}\fi}\newcommand{\cjtcs }{\if\lName1\skp{ }{Chicago Journal
  of Theoretical Computer Science}{}\else{Chic. J. Th. Comp.
  Sci.}\fi}\newcommand{\quantum }{\if\lName1\skp{ }{{Quantum}}{
  }\else{Quant.}\fi}\newcommand{\cmp }{\if\lName1\skp{ }{Communications in
  Mathematical Physics}{ }\else{Comm. Math. Phys.}\fi}\newcommand{\jmp
  }{\if\lName1\skp{ }{Journal of Mathematical Physics}{ }\else{J. Math.
  Phys.}\fi}\newcommand{\rspa }{\if\lName1\skp{ }{Proceedings of the Royal
  Society A}{ }\else{Proc. Roy. Soc. A}\fi}\newcommand{\qic }{\if\lName1\skp{
  }{Quantum Information and Computation}{ }\else{Quant. Inf. \&
  Comp.}\fi}\newcommand{\physrev }{\if\lName1\skp{ }{Physical Review}{
  }\else{Phys. Rev.}\fi}\newcommand{\pra }{\if\lName1\skp{ }{Physical Review
  A}{ }\else{Phys. Rev. A}\fi}\newcommand{\prb }{\if\lName1\skp{ }{Physical
  Review B}{ }\else{Phys. Rev. B}\fi}\newcommand{\pre }{\if\lName1\skp{
  }{Physical Review E}{ }\else{Phys. Rev. E}\fi}\newcommand{\prr
  }{\if\lName1\skp{ }{Physical Review Research}{ }\else{Phys. Rev.
  Research}\fi}\newcommand{\prx }{\if\lName1\skp{ }{Physical Review X}{
  }\else{Phys. Rev. X}\fi}\newcommand{\prl }{\if\lName1\skp{ }{Physical Review
  Letters}{ }\else{Phys. Rev. Lett.}\fi}\newcommand{\njp }{\if\lName1\skp{
  }{New Journal of Physics}{ }\else{New J. Phys.}\fi}\newcommand{\prapp
  }{\if\lName1\skp{ }{Physical Review Applied}{ }\else{Phys. Rev.
  Appl.}\fi}\newcommand{\physrep }{\if\lName1\skp{ }{Physics Reports}{
  }\else{Phys. Rep.}\fi}\newcommand{\rmp }{\if\lName1\skp{ }{Reviews of Modern
  Physics}{ }\else{Rev. Mod. Phys. }\fi}\newcommand{\phystoday
  }{\if\lName1\skp{ }{Physics Today}{ }\else{Phys.
  Today}\fi}\newcommand{\physics }{\if\lName1\skp{ }{Physics}{
  }\else{Phys.}\fi}\newcommand{\nature }{\if\lName1\skp{ }{Nature}{
  }\else{Nat.}\fi}\newcommand{\natcomm }{\if\lName1\skp{ }{Nature
  Communications}{ }\else{Nat. Comm.}\fi}\newcommand{\natphys }{\if\lName1\skp{
  }{Nature Physics}{ }\else{Nat. Phys.}\fi}\newcommand{\npjqi }{\if\lName1\skp{
  }{npj Quantum Information}{ }\else{npj Quant. Inf.}\fi}\newcommand{\scirep
  }{\if\lName1\skp{ }{Scientific Reports}{ }\else{Sci.
  Rep.}\fi}\newcommand{\science }{\if\lName1\skp{ }{Science}{
  }\else{Sci.}\fi}\newcommand{\jpa }{\if\lName1\skp{ }{Journal of Physics A:
  Mathematical and Theoretical}{ }\else{J. Phys. A}\fi}\newcommand{\ijtp
  }{\if\lName1\skp{ }{International Journal of Theoretical Physics}{
  }\else{Int. J. Th. Phys.}\fi}\newcommand{\jmo }{\if\lName1\skp{ }{Journal of
  Modern Optics}{ }\else{J. Mod. Opt.}\fi}\newcommand{\jstatph
  }{\if\lName1\skp{ }{Journal of Statistical Physics}{ }\else{J. Stat.
  Phys.}\fi}\newcommand{\pnas }{\if\lName1\skp{ }{Proceedings of the National
  Academy of Sciences}{ }\else{PNAS}\fi}\newcommand{\lncs }{\if\lName1\skp{
  }{Lecture Notes in Computer Science}{ }\else{L. Notes Comp.
  Sci.}\fi}\newcommand{\lnai }{\if\lName1\skp{ }{Lecture Notes in Artificial
  Intelligence}{ }\else{L. Notes Art. Int.}\fi}\newcommand{\lnm
  }{\if\lName1\skp{ }{Lecture Notes in Mathematics}{ }\else{L. Notes
  Math.}\fi}\newcommand{\tams }{\if\lName1\skp{ }{Transactions of the American
  Mathematical Society}{ }\else{Trans. AMS}\fi}\newcommand{\ieeetit
  }{\if\lName1\skp{ }{{IEEE} Transactions on Information Theory}{ }\else{{IEEE}
  Trans. Inf. Th.}\fi}\newcommand{\iscs }{\if\lName1\skp{ }{International
  Series in Computer Science}{ }\else{Int. Ser. Comp.
  Sci.}\fi}\newcommand{\tocl }{\if\lName1\skp{ }{Theory of Computing Library}{
  }\else{Th. Comp. Lib.}\fi}
\begin{thebibliography}{{\dutchPrefix{Apeldoorn}{v}}AGG{\dutchPrefix{Wolf}{d}}W20}

\bibitem[Aar15]{aaronson2015caveat}
Scott Aaronson.
\newblock \href{http://dx.doi.org/10.1038/nphys3272}{Read the fine print}.
\newblock {\em Nature Physics}, 11(4):291, 2015.

\bibitem[ACQ22]{aharononv2021QAlgorithmicMeasurement}
Dorit Aharonov, Jordan Cotler, and Xiao-Liang Qi.
\newblock \href{http://dx.doi.org/10.1038/s41467-021-27922-0}{Quantum
  algorithmic measurement}.
\newblock {\em \natcomm}, 13(1):1--9, 2022.
\newblock \arxiv{2101.04634}.

\bibitem[ADBL20]{arrazola2019QInspiredInPractice}
Juan~Miguel Arrazola, Alain Delgado, Bhaskar~Roy Bardhan, and Seth Lloyd.
\newblock \href{http://dx.doi.org/10.22331/q-2020-08-13-307}{Quantum-inspired
  algorithms in practice}.
\newblock {\em \quantum}, 4:307, 2020.
\newblock \arxiv{1905.10415}.

\bibitem[AK16]{arora2016CombPrimDualSDP}
Sanjeev Arora and Satyen Kale.
\newblock \href{http://dx.doi.org/10.1145/2837020}{A combinatorial, primal-dual
  approach to semidefinite programs}.
\newblock {\em \jacm}, 63(2):12:1--12:35, 2016.
\newblock Earlier version in STOC'07.

\bibitem[AP10]{aleksandrov2009OperatorHolderZygmundFunctions}
Alexei~B. Aleksandrov and Vladimir~V. Peller.
\newblock \href{http://dx.doi.org/10.1016/j.aim.2009.12.018}{Operator
  {H}ölder--{Z}ygmund functions}.
\newblock {\em Advances in Mathematics}, 224(3):910--966, 2010.
\newblock \arxiv{0907.3049}.

\bibitem[AP11]{aleksandrov2011EstimatesOfOperatorModuli}
Alexei~B. Aleksandrov and Vladimir~V. Peller.
\newblock \href{http://dx.doi.org/10.1016/j.jfa.2011.07.009}{Estimates of
  operator moduli of continuity}.
\newblock {\em Journal of Functional Analysis}, 261(10):2741--2796, 2011.
\newblock \arxiv{1104.3553}.

\bibitem[{\dutchPrefix{Apeldoorn}{v}}AG19]{apeldoorn2018ImprovedQSDPSolving}
Joran {\dutchPrefix{Apeldoorn}{v}}an~Apeldoorn and Andr{\'a}s Gily{\'e}n.
\newblock \href{http://dx.doi.org/10.4230/LIPIcs.ICALP.2019.99}{Improvements in
  quantum {SDP}-solving with applications}.
\newblock In {\em \icalp{46th}}, pages 99:1--99:15, 2019.
\newblock \arxiv{1804.05058}.

\bibitem[{\dutchPrefix{Apeldoorn}{v}}AGG{\dutchPrefix{Wolf}{d}}W20]{apeldoorn2017QSDPSolvers}
Joran {\dutchPrefix{Apeldoorn}{v}}an~Apeldoorn, Andr{\'a}s Gily{\'e}n, Sander
  Gribling, and Ronald {\dutchPrefix{Wolf}{d}}e~Wolf.
\newblock \href{http://dx.doi.org/10.22331/q-2020-02-14-230}{Quantum
  {SDP}-solvers: {B}etter upper and lower bounds}.
\newblock {\em \quantum}, 4:230, 2020.
\newblock Earlier version in FOCS'17. \arxiv{1705.01843}.

\bibitem[ATS03]{aharonov2003adiabatic}
Dorit Aharonov and Amnon Ta-Shma.
\newblock \href{http://dx.doi.org/10.1145/780542.780546}{Adiabatic quantum
  state generation and statistical zero knowledge}.
\newblock In {\em Proceedings of the 35th {ACM} Symposium on the Theory of
  Computing}, pages 20--29, New York, NY, USA, 2003. ACM, Association for
  Computing Machinery.
\newblock \arXiv{quant-ph/0301023}.

\bibitem[BCK15]{berry2015HamSimNearlyOpt}
Dominic~W. Berry, Andrew~M. Childs, and Robin Kothari.
\newblock \href{http://dx.doi.org/10.1109/FOCS.2015.54}{Hamiltonian simulation
  with nearly optimal dependence on all parameters}.
\newblock In {\em \focs{56th}}, pages 792--809, 2015.
\newblock \arxiv{1501.01715}.

\bibitem[BE95]{borwein1995polynomials}
Peter Borwein and Tam{\'a}s Erd{\'e}lyi.
\newblock \href{http://dx.doi.org/10.1007/978-1-4612-0793-1}{{\em Polynomials
  and polynomial inequalities}}, volume 161 of {\em Graduate Texts in
  Mathematics}.
\newblock Springer, New York, NY, USA, 1995.

\bibitem[Bel97]{BellmanIntroMatrixAnal}
R.~Bellman.
\newblock \href{http://dx.doi.org/10.1137/1.9781611971170}{{\em Introduction to
  matrix analysis}}.
\newblock Society for Industrial and Applied Mathematics, Philadelphia, PA,
  USA, second edition, 1997.

\bibitem[Bha97]{bhatia1997MatrixAnalysis}
Rajendra Bhatia.
\newblock \href{http://dx.doi.org/10.1007/978-1-4612-0653-8}{{\em Matrix
  Analysis}}, volume 169 of {\em Graduate Texts in Mathematics}.
\newblock Springer, 1997.

\bibitem[BHK97]{Belhumeur1997EigenFisherFaces}
Peter~N. {Belhumeur}, João~P. {Hespanha}, and David~J. {Kriegman}.
\newblock \href{http://dx.doi.org/10.1109/34.598228}{{E}igenfaces vs.
  {F}isherfaces: recognition using class specific linear projection}.
\newblock {\em IEEE Transactions on Pattern Analysis and Machine Intelligence},
  19(7):711--720, 1997.

\bibitem[BKL{\etalchar{+}}19]{brandao2017QSDPSpeedupsLearning}
Fernando G. S.~L. Brand{\~a}o, Amir Kalev, Tongyang Li, Cedric Yen-Yu Lin,
  Krysta~M. Svore, and Xiaodi Wu.
\newblock \href{http://dx.doi.org/10.4230/LIPIcs.ICALP.2019.27}{Quantum {SDP}
  solvers: Large speed-ups, optimality, and applications to quantum learning}.
\newblock In {\em \icalp{46th}}, pages 27:1--27:14, 2019.
\newblock \arxiv{1710.02581}.

\bibitem[BS17]{brandao2016QSDPSpeedup}
Fernando G. S.~L. Brand{\~a}o and Krysta~M. Svore.
\newblock \href{http://dx.doi.org/10.1109/FOCS.2017.45}{Quantum speed-ups for
  solving semidefinite programs}.
\newblock In {\em \focs{58th}}, pages 415--426, 2017.
\newblock \arxiv{1609.05537}.

\bibitem[CCH{\etalchar{+}}22]{chepurko2020quantum}
Nadiia Chepurko, Kenneth Clarkson, Lior Horesh, Honghao Lin, and David
  Woodruff.
\newblock Quantum-inspired algorithms from randomized numerical linear algebra.
\newblock In {\em International Conference on Machine Learning}, pages
  3879--3900. PMLR, 2022.
\newblock \arXiv{2011.04125}.

\bibitem[CD16]{cong2016quantum}
Iris Cong and Luming Duan.
\newblock \href{http://dx.doi.org/10.1088/1367-2630/18/7/073011}{Quantum
  discriminant analysis for dimensionality reduction and classification}.
\newblock {\em New Journal of Physics}, 18(7):073011, jul 2016.
\newblock \arXiv{1510.00113}.

\bibitem[CGJ19]{chakraborty2018BlockMatrixPowers}
Shantanav Chakraborty, András Gilyén, and Stacey Jeffery.
\newblock \href{http://dx.doi.org/10.4230/LIPIcs.ICALP.2019.33}{The power of
  block-encoded matrix powers: {I}mproved regression techniques via faster
  {H}amiltonian simulation}.
\newblock In {\em \icalp{46th}}, pages 33:1--33:14, 2019.
\newblock \arxiv{1804.01973}.

\bibitem[CGL{\etalchar{+}}20]{chia2020QuantInsLinEqSolving}
Nai-Hui Chia, András Gilyén, Han-Hsuan Lin, Seth Lloyd, Ewin Tang, and
  Chunhao Wang.
\newblock
  \href{http://dx.doi.org/10.4230/LIPIcs.ISAAC.2020.47}{Quantum-inspired
  algorithms for solving low-rank linear equation systems with logarithmic
  dependence on the dimension}.
\newblock In {\em \isaac{31st}}, pages 47:1--47:17, 2020.

\bibitem[CHI{\etalchar{+}}18]{Ciliberto2018}
Carlo Ciliberto, Mark Herbster, Alessandro~Davide Ialongo, Massimiliano Pontil,
  Andrea Rocchetto, Simone Severini, and Leonard Wossnig.
\newblock \href{http://dx.doi.org/10.1098/rspa.2017.0551}{Quantum machine
  learning: a classical perspective}.
\newblock {\em Proceedings of the Royal Society A: Mathematical, Physical and
  Engineering Sciences}, 474(2209):20170551, January 2018.

\bibitem[CLLW20]{chia2019QInspiredSubLinLowRankSDPSolver}
Nai-Hui Chia, Tongyang Li, Han-Hsuan Lin, and Chunhao Wang.
\newblock \href{http://dx.doi.org/10.4230/LIPIcs.MFCS.2020.23}{Quantum-inspired
  sublinear algorithm for solving low-rank semidefinite programming}.
\newblock In {\em \mfcs{45th}}, pages 23:1--23:15, 2020.
\newblock \arxiv{1901.03254}.

\bibitem[CLS{\etalchar{+}}19]{chen2019NMF}
Zhihuai Chen, Yinan Li, Xiaoming Sun, Pei Yuan, and Jialin Zhang.
\newblock \href{http://dx.doi.org/10.24963/ijcai.2019/627}{A quantum-inspired
  classical algorithm for separable non-negative matrix factorization}.
\newblock In {\em Proceedings of the 28th International Joint Conference on
  Artificial Intelligence}, pages 4511--4517, Palo Alto, CA, USA, 2019. AAAI
  Press, AAAI.
\newblock \arXiv{1907.05568}.

\bibitem[CLW18]{chia2018QInspiredSubLinLowRankLinEqSolver}
Nai-Hui Chia, Han-Hsuan Lin, and Chunhao Wang.
\newblock Quantum-inspired sublinear classical algorithms for solving low-rank
  linear systems.
\newblock \arxiv{1811.04852}, 2018.

\bibitem[CMM17]{cmm17}
Michael~B. Cohen, Cameron Musco, and Christopher Musco.
\newblock \href{http://dx.doi.org/10.1137/1.9781611974782.115}{Input sparsity
  time low-rank approximation via ridge leverage score sampling}.
\newblock In {\em Proceedings of the Twenty-Eighth Annual {ACM}-{SIAM}
  Symposium on Discrete Algorithms}, pages 1758--1777, Philadelphia, PA, USA,
  January 2017. Society for Industrial and Applied Mathematics.

\bibitem[CVB20]{camps2020ApxQCircSynthesisBlockEncoding}
Daan Camps and Roel Van~Beeumen.
\newblock \href{http://dx.doi.org/10.1103/PhysRevA.102.052411}{Approximate
  quantum circuit synthesis using block encodings}.
\newblock {\em \pra}, 102(5):052411, 2020.
\newblock \arxiv{2007.01417}.

\bibitem[CW17]{cw17}
Kenneth~L. Clarkson and David~P. Woodruff.
\newblock \href{http://dx.doi.org/10.1145/3019134}{Low-rank approximation and
  regression in input sparsity time}.
\newblock {\em J. ACM}, 63(6), January 2017.

\bibitem[DBH21]{ding2019SVM}
Chen Ding, Tian-Yi Bao, and He-Liang Huang.
\newblock \href{http://dx.doi.org/10.1109/TNNLS.2021.3084467}{Quantum-inspired
  support vector machine}.
\newblock {\em IEEE Transactions on Neural Networks and Learning Systems},
  pages 1--13, 2021.
\newblock \arXiv{1906.08902}.

\bibitem[DHLT20]{du2019conical}
Yuxuan Du, Min-Hsiu Hsieh, Tongliang Liu, and Dacheng Tao.
\newblock
  \href{http://dx.doi.org/10.1103/PhysRevResearch.2.033199}{Quantum-inspired
  algorithm for general minimum conical hull problems}.
\newblock {\em Physical Review Research}, 2(3):033199, 2020.
\newblock \arXiv{1907.06814}.

\bibitem[DKM06]{DKM06}
Petros Drineas, Ravi Kannan, and Michael~W. Mahoney.
\newblock \href{http://dx.doi.org/10.1137/S0097539704442684}{Fast {M}onte
  {C}arlo algorithms for matrices {I}: Approximating matrix multiplication}.
\newblock {\em SIAM Journal on Computing}, 36(1):132--157, 2006.

\bibitem[DKR02]{DKR02}
Petros Drineas, Iordanis Kerenidis, and Prabhakar Raghavan.
\newblock \href{http://dx.doi.org/10.1145/509907.509922}{Competitive
  recommendation systems}.
\newblock In {\em Proceedings of the 34th {ACM} Symposium on the Theory of
  Computing ({STOC})}, pages 82--90, New York, NY, USA, 2002. Association for
  Computing Machinery.

\bibitem[DKW18]{DKW18}
Yogesh Dahiya, Dimitris Konomis, and David~P Woodruff.
\newblock \href{http://dx.doi.org/10.1145/3219819.3220098}{An empirical
  evaluation of sketching for numerical linear algebra}.
\newblock In {\em Proceedings of the 24th ACM SIGKDD International Conference
  on Knowledge Discovery \& Data Mining}, pages 1292--1300, New York, NY, USA,
  2018. ACM, Association for Computing Machinery.

\bibitem[DM07]{drineas2007randomized}
Petros Drineas and Michael~W. Mahoney.
\newblock \href{http://dx.doi.org/10.1016/j.laa.2006.08.023}{A randomized
  algorithm for a tensor-based generalization of the singular value
  decomposition}.
\newblock {\em Linear Algebra and its Applications}, 420(2-3):553--571, 2007.

\bibitem[DMM08]{Drineas2008}
Petros Drineas, Michael~W. Mahoney, and S.~Muthukrishnan.
\newblock \href{http://dx.doi.org/10.1137/07070471x}{Relative-error {CUR}
  matrix decompositions}.
\newblock {\em {SIAM} Journal on Matrix Analysis and Applications},
  30(2):844--881, January 2008.

\bibitem[DW20]{Dunjko2020nonreviewofquantum}
Vedran Dunjko and Peter Wittek.
\newblock \href{http://dx.doi.org/10.22331/qv-2020-03-17-32}{A non-review of
  {Q}uantum {M}achine {L}earning: trends and explorations}.
\newblock {\em {Quantum Views}}, 4:32, March 2020.

\bibitem[Fey51]{FeynmanOpCalc}
Richard~P. Feynman.
\newblock \href{http://dx.doi.org/10.1103/PhysRev.84.108}{An operator calculus
  having applications in quantum electrodynamics}.
\newblock {\em Physical Review}, 84:108--128, 1951.

\bibitem[Fey82]{feynman1982SimQPhysWithComputers}
Richard~P. Feynman.
\newblock \href{http://dx.doi.org/10.1007/BF02650179}{Simulating physics with
  computers}.
\newblock {\em \ijtp}, 21(6-7):467--488, 1982.

\bibitem[FKV04]{frieze2004FastMonteCarloLowRankApx}
Alan Frieze, Ravi Kannan, and Santosh Vempala.
\newblock \href{http://dx.doi.org/10.1145/1039488.1039494}{Fast {M}onte-{C}arlo
  algorithms for finding low-rank approximations}.
\newblock {\em \jacm}, 51(6):1025--1041, 2004.

\bibitem[GCD20]{gcd20}
Casper Gyurik, Chris Cade, and Vedran Dunjko.
\newblock Towards quantum advantage via topological data analysis, 2020.
\newblock \arXiv{2005.02607}.

\bibitem[Gil10]{Gil10}
Michael~I. Gil.
\newblock \href{http://dx.doi.org/10.13001/1081-3810.1375}{Perturbations of
  functions of diagonalizable matrices}.
\newblock {\em Electronic Journal of Linear Algebra}, 20:303--313, 2010.

\bibitem[GLM08]{giovannetti2007QuantumRAM}
Vittorio Giovannetti, Seth Lloyd, and Lorenzo Maccone.
\newblock \href{http://dx.doi.org/10.1103/PhysRevLett.100.160501}{Quantum
  random access memory}.
\newblock {\em \prl}, 100(16):160501, 2008.
\newblock \arxiv{0708.1879}.

\bibitem[GLT18]{gilyen2018QInsLowRankHHL}
András Gilyén, Seth Lloyd, and Ewin Tang.
\newblock Quantum-inspired low-rank stochastic regression with logarithmic
  dependence on the dimension.
\newblock \arxiv{1811.04909}, 2018.

\bibitem[GR02]{grover2002SuperposEffIntegrProbDistr}
Lov Grover and Terry Rudolph.
\newblock Creating superpositions that correspond to efficiently integrable
  probability distributions.
\newblock \arxiv{quant-ph/0208112}, 2002.

\bibitem[GSLW19]{gilyen2018QSingValTransf}
András Gilyén, Yuan Su, Guang~Hao Low, and Nathan Wiebe.
\newblock \href{http://dx.doi.org/10.1145/3313276.3316366}{Quantum singular
  value transformation and beyond: {E}xponential improvements for quantum
  matrix arithmetics}.
\newblock In {\em \stoc{51st}}, pages 193--204, 2019.
\newblock \arxiv{1806.01838}.

\bibitem[HHL09]{harrow2009QLinSysSolver}
Aram~W. Harrow, Avinatan Hassidim, and Seth Lloyd.
\newblock \href{http://dx.doi.org/10.1103/PhysRevLett.103.150502}{Quantum
  algorithm for linear systems of equations}.
\newblock {\em \prl}, 103(15):150502, 2009.
\newblock \arxiv{0811.3171}.

\bibitem[HKP21]{huang2021InfThBoundsOnQAdvantageML}
Hsin-Yuan Huang, Richard Kueng, and John Preskill.
\newblock
  \href{http://dx.doi.org/10.1103/PhysRevLett.126.190505}{Information-theoretic
  bounds on quantum advantage in machine learning}.
\newblock {\em \prl}, 126(19):190505, 2021.
\newblock \arxiv{2101.02464}.

\bibitem[HKS11]{hks11}
Elad Hazan, Tomer Koren, and Nati Srebro.
\newblock
  \href{http://papers.nips.cc/paper/4359-beating-sgd-learning-svms-in-sublinear-time.pdf}{Beating
  {SGD}: Learning {SVM}s in sublinear time}.
\newblock In J.~Shawe-Taylor, R.~S. Zemel, P.~L. Bartlett, F.~Pereira, and
  K.~Q. Weinberger, editors, {\em Advances in Neural Information Processing
  Systems 24}, pages 1233--1241, Red Hook, NY, USA, 2011. Curran Associates,
  Inc.

\bibitem[JLGS20]{jethwani2019QInsClassAlgSVT}
Dhawal Jethwani, François Le~Gall, and Sanjay~K. Singh.
\newblock \href{http://dx.doi.org/10.4230/LIPIcs.MFCS.2020.53}{Quantum-inspired
  classical algorithms for singular value transformation}.
\newblock In {\em \mfcs{45th}}, pages 53:1--53:14, 2020.
\newblock \arxiv{1910.05699}.

\bibitem[Kal07]{kale2007efficient}
Satyen Kale.
\newblock \href{http://www.satyenkale.com/papers/thesis.pdf}{{\em Efficient
  algorithms using the multiplicative weights update method}}.
\newblock PhD thesis, Princeton University, 2007.

\bibitem[KP17]{kerenidis2016QRecSys}
Iordanis Kerenidis and Anupam Prakash.
\newblock \href{http://dx.doi.org/10.4230/LIPIcs.ITCS.2017.49}{Quantum
  recommendation systems}.
\newblock In {\em \itcs{8th}}, pages 49:1--49:21, 2017.
\newblock \arxiv{1603.08675}.

\bibitem[KP20]{kerenidis2017QGradDesc}
Iordanis Kerenidis and Anupam Prakash.
\newblock \href{http://dx.doi.org/10.1103/PhysRevA.101.022316}{Quantum gradient
  descent for linear systems and least squares}.
\newblock {\em \pra}, 101(2):022316, 2020.
\newblock \arxiv{1704.04992}.

\bibitem[KS48]{KarplusMicriWave}
Robert Karplus and Julian Schwinger.
\newblock \href{http://dx.doi.org/10.1103/PhysRev.73.1020}{A note on saturation
  in microwave spectroscopy}.
\newblock {\em Physical Review}, 73:1020--1026, 1948.

\bibitem[KV17]{kannan2017RandAlgNumLinAlg}
Ravindran Kannan and Santosh Vempala.
\newblock \href{http://dx.doi.org/10.1017/S0962492917000058}{Randomized
  algorithms in numerical linear algebra}.
\newblock {\em Acta Numerica}, 26:95--135, 2017.

\bibitem[LC17]{low2016HamSimQSignProc}
Guang~Hao Low and Isaac~L. Chuang.
\newblock \href{http://dx.doi.org/10.1103/PhysRevLett.118.010501}{Optimal
  {H}amiltonian simulation by quantum signal processing}.
\newblock {\em \prl}, 118(1):010501, 2017.
\newblock \arxiv{1606.02685}.

\bibitem[LGZ16]{lloyd2016topological}
Seth Lloyd, Silvano Garnerone, and Paolo Zanardi.
\newblock \href{http://dx.doi.org/10.1038/ncomms10138}{Quantum algorithms for
  topological and geometric analysis of data}.
\newblock {\em Nature Communications}, 7:10138, 2016.
\newblock \arXiv{1408.3106}.

\bibitem[Llo96]{lloyd1996UnivQSim}
Seth Lloyd.
\newblock \href{http://dx.doi.org/10.1126/science.273.5278.1073}{Universal
  quantum simulators}.
\newblock {\em \science}, 273(5278):1073--1078, 1996.

\bibitem[LMR13]{lloyd2013Clustering}
Seth Lloyd, Masoud Mohseni, and Patrick Rebentrost.
\newblock Quantum algorithms for supervised and unsupervised machine learning,
  2013.
\newblock \arXiv{1307.0411}.

\bibitem[LMR14]{lloyd2013QPrincipalCompAnal}
Seth Lloyd, Masoud Mohseni, and Patrick Rebentrost.
\newblock \href{http://dx.doi.org/10.1038/nphys3029}{Quantum principal
  component analysis}.
\newblock {\em \natphys}, 10:631--633, 2014.
\newblock \arxiv{1307.0401}.

\bibitem[LRS15]{lee2015LowerBoundSDPRelax}
James~R. Lee, Prasad Raghavendra, and David Steurer.
\newblock \href{http://dx.doi.org/10.1145/2746539.2746599}{Lower bounds on the
  size of semidefinite programming relaxations}.
\newblock In {\em \stoc{47th}}, pages 567--576, 2015.
\newblock \arxiv{1411.6317}.

\bibitem[Mah11]{mahoney2011randomized}
Michael~W. Mahoney.
\newblock \href{http://dx.doi.org/10.1561/2200000035}{Randomized algorithms for
  matrices and data}.
\newblock {\em Foundations and Trends{\textregistered} in Machine Learning},
  3(2):123--224, 2011.

\bibitem[McD89]{mcdiarmid89}
Colin McDiarmid.
\newblock \href{http://dx.doi.org/10.1017/CBO9781107359949.008}{{\em On the
  method of bounded differences}}, page 148–188.
\newblock London Mathematical Society Lecture Note Series. Cambridge University
  Press, Cambridge, England, 1989.

\bibitem[MMD08]{mahoney2008tensor}
Michael~W. Mahoney, Mauro Maggioni, and Petros Drineas.
\newblock \href{http://dx.doi.org/10.1137/060665336}{Tensor-{CUR}
  decompositions for tensor-based data}.
\newblock {\em SIAM Journal on Matrix Analysis and Applications},
  30(3):957--987, 2008.

\bibitem[MRTC21]{martyn2021GrandUnificationQAlgs}
John~M. Martyn, Zane~M. Rossi, Andrew~K. Tan, and Isaac~L. Chuang.
\newblock \href{http://dx.doi.org/10.1103/PRXQuantum.2.040203}{Grand
  unification of quantum algorithms}.
\newblock {\em \prx}, 2(4):040203, 2021.
\newblock \arXiv{2105.02859}.

\bibitem[PC99]{pan1999ComplexityMatEigenProb}
Victor~Y. Pan and Zhao~Q. Chen.
\newblock \href{http://dx.doi.org/10.1145/301250.301389}{The complexity of the
  matrix eigenproblem}.
\newblock In {\em \stoc{31st}}, page 507–516, 1999.

\bibitem[Pra14]{prakash2014QLinAlgAndMLThesis}
Anupam Prakash.
\newblock
  \href{https://www2.eecs.berkeley.edu/Pubs/TechRpts/2014/EECS-2014-211.pdf}{{\em
  Quantum Algorithms for Linear Algebra and Machine Learning}}.
\newblock PhD thesis, University of California at Berkeley, 2014.

\bibitem[Pre18]{preskill2018QuantCompNISQEra}
John Preskill.
\newblock \href{http://dx.doi.org/10.22331/q-2018-08-06-79}{Quantum {C}omputing
  in the {NISQ} era and beyond}.
\newblock {\em \quantum}, 2:79, 2018.
\newblock \arxiv{1801.00862}.

\bibitem[RL18]{rebentrost2018QuantumFinance}
Patrick Rebentrost and Seth Lloyd.
\newblock Quantum computational finance: quantum algorithm for portfolio
  optimization.
\newblock \arxiv{1811.03975}, 2018.

\bibitem[RML14]{rebentrost2014QSVM}
Patrick Rebentrost, Masoud Mohseni, and Seth Lloyd.
\newblock \href{http://dx.doi.org/10.1103/PhysRevLett.113.130503}{Quantum
  support vector machine for big data classification}.
\newblock {\em \prl}, 113(13):130503, 2014.
\newblock \arxiv{1307.0471}.

\bibitem[RSW{\etalchar{+}}19]{rebentrost2016QGradDesc}
Patrick Rebentrost, Maria Schuld, Leonard Wossnig, Francesco Petruccione, and
  Seth Lloyd.
\newblock \href{http://dx.doi.org/10.1088/1367-2630/ab2a9e}{Quantum gradient
  descent and {N}ewton's method for constrained polynomial optimization}.
\newblock {\em \njp}, 21(7):073023, 2019.
\newblock \arxiv{1612.01789}.

\bibitem[RV07]{rudelson2007sampling}
Mark Rudelson and Roman Vershynin.
\newblock \href{http://dx.doi.org/10.1145/1255443.1255449}{Sampling from large
  matrices: An approach through geometric functional analysis}.
\newblock {\em Journal of the {ACM} ({JACM})}, 54(4):21–es, July 2007.

\bibitem[RWC{\etalchar{+}}20]{rudi2018nystrom}
Alessandro Rudi, Leonard Wossnig, Carlo Ciliberto, Andrea Rocchetto,
  Massimiliano Pontil, and Simone Severini.
\newblock \href{http://dx.doi.org/10.22331/q-2020-02-20-234}{Approximating
  {H}amiltonian dynamics with the {N}ystr{\"o}m method}.
\newblock {\em {Quantum}}, 4:234, 2020.
\newblock \arXiv{1804.02484}.

\bibitem[Sho97]{shor1994Factoring}
Peter~W. Shor.
\newblock \href{http://dx.doi.org/10.1137/S0097539795293172}{Polynomial-time
  algorithms for prime factorization and discrete logarithms on a quantum
  computer}.
\newblock {\em \siamjc}, 26(5):1484--1509, 1997.
\newblock Earlier version in FOCS'94. \arxiv{quant-ph/9508027}.

\bibitem[SWZ16]{swz16}
Zhao Song, David Woodruff, and Huan Zhang.
\newblock
  \href{http://papers.nips.cc/paper/6496-sublinear-time-orthogonal-tensor-decomposition.pdf}{Sublinear
  time orthogonal tensor decomposition}.
\newblock In {\em Advances in Neural Information Processing Systems 29}, pages
  793--801. Curran Associates, Inc., Red Hook, NY, USA, 2016.

\bibitem[Tan19]{tang2018QuantumInspiredRecommSys}
Ewin Tang.
\newblock \href{http://dx.doi.org/10.1145/3313276.3316310}{A quantum-inspired
  classical algorithm for recommendation systems}.
\newblock In {\em \stoc{51st}}, pages 217--228, 2019.
\newblock \arxiv{1807.04271}.

\bibitem[Tan21]{tang2018QInspiredClassAlgPCA}
Ewin Tang.
\newblock \href{http://dx.doi.org/10.1103/PhysRevLett.127.060503}{Quantum
  principal component analysis only achieves an exponential speedup because of
  its state preparation assumptions}.
\newblock {\em \prl}, 127(6):060503, 2021.
\newblock \arxiv{1811.00414}.

\bibitem[Tao10]{tao2010notesOnHermitianEigenvalues}
Terence Tao.
\newblock 254a, {N}otes 3a: {E}igenvalues and sums of {H}ermitian matrices,
  2010.
\newblock
  \url{https://terrytao.wordpress.com/2010/01/12/254a-notes-3a-eigenvalues-and-sums-of-hermitian-matrices/}.

\bibitem[VdN11]{vanDenNest2011SimulatingQCompwProbMeth}
Maarten Van~den Nest.
\newblock \href{http://dx.doi.org/10.26421/QIC11.9-10}{Simulating quantum
  computers with probabilistic methods}.
\newblock {\em \qic}, 11(9\&10):784--812, 2011.
\newblock \arxiv{0911.1624}.

\bibitem[Vos91]{Vose1991}
Michael~D. Vose.
\newblock \href{http://dx.doi.org/10.1109/32.92917}{A linear algorithm for
  generating random numbers with a given distribution}.
\newblock {\em {IEEE} Transactions on Software Engineering}, 17(9):972--975,
  1991.

\bibitem[Wel09]{Welling2009FisherLDA}
Max Welling.
\newblock
  \href{https://www.ics.uci.edu/~welling/teaching/273ASpring09/Fisher-LDA.pdf}{Fisher
  linear discriminant analysis}.
\newblock
  \url{https://www.ics.uci.edu/~welling/teaching/273ASpring09/Fisher-LDA.pdf},
  2009.

\bibitem[Woo14]{w14}
David~P. Woodruff.
\newblock \href{http://dx.doi.org/10.1561/0400000060}{Sketching as a tool for
  numerical linear algebra}.
\newblock {\em Foundations and Trends® in Theoretical Computer Science},
  10(1–2):1--157, 2014.

\bibitem[WZP18]{wossnig2018QLinSysAlgForDensMat}
Leonard Wossnig, Zhikuan Zhao, and Anupam Prakash.
\newblock \href{http://dx.doi.org/10.1103/PhysRevLett.120.050502}{Quantum
  linear system algorithm for dense matrices}.
\newblock {\em \prl}, 120(5):050502, 2018.
\newblock \arxiv{1704.06174}.

\bibitem[YSSK20]{yamasaki2020OptRandFeat}
Hayata Yamasaki, Sathyawageeswar Subramanian, Sho Sonoda, and Masato Koashi.
\newblock Learning with optimized random features: Exponential speedup by
  quantum machine learning without sparsity and low-rank assumptions.
\newblock In H.~Larochelle, M.~Ranzato, R.~Hadsell, M.~F. Balcan, and H.~Lin,
  editors, {\em Advances in Neural Information Processing Systems}, volume~33,
  pages 13674--13687. Curran Associates, Inc., Red Hook, NY, USA, 2020.
\newblock \arXiv{2004.10756}.

\bibitem[ZFF19]{zhao2015QAssisstedGaussProcRegr}
Zhikuan Zhao, Jack~K. Fitzsimons, and Joseph~F. Fitzsimons.
\newblock \href{http://dx.doi.org/10.1103/PhysRevA.99.052331}{Quantum-assisted
  {G}aussian process regression}.
\newblock {\em \pra}, 99(5):052331, 2019.
\newblock \arxiv{1512.03929}.

\end{thebibliography}
\bibliographystyle{alphaUrlePrint}


\appendix
\section{Proof sketch for Remark~\ref{rmk:a-to-a-dagger}}\label{apx:sketch}

Recall that we wish to show that, given $\sq(A)$, we can simulate $\sq_\phi(B)$ for $B$ such that $\|B - A^\dagger\| \leq \eps\|A\|$ with probability $\geq 1-\delta$.

Following the argument from \cref{rmk:spectral-low-rank-approximation}, we can find a $B \coloneqq AR^\dagger \bar{t}(CC^\dagger)R$ satisfying the above property in $\bOt{\frac{\|A\|_\fr^{28}}{\|A\|^{28}\eps^{22}}\log^3\frac{1}{\delta}}$ (rescaling $\eps$ appropriately).
Here, $R$ and $\bar{t}(CC^\dagger)$ come from an application of \cref{thm:evenSing} with $t: \bbR \to \bbC$ a smooth step function that goes from zero to one around $(\eps\|A\|)^2$.
If we had sampling and query access to the columns of $AR^\dagger$, we would be done, since then $B = \sum_{i=1}^r\sum_{j=1}^r [\bar{t}(CC^\dagger)](i,j) [A'R'^\dagger](\cdot,i)R(j,\cdot)$, and we can express $B$ as a sum of $r^2$ outer products of vectors that we have sampling and query access to.
This gives us both $\sq_\phi(B)$ and $\sq_\phi(B^\dagger)$.

We won't get exactly this, but using that $\bar{t}(CC^\dagger) = (C_{\eps\|A\|/2}^+)^\dagger t(C^\dagger C) C_{\eps\|A\|/2}^+$, for $UDV^\dagger$ the SVD of $C$ and $U_{\eps\|A\|/2}D_{\eps\|A\|/2}V_{\eps\|A\|/2}^\dagger$ the SVD truncated to singular values at least $\eps\|A\|/2$, we can rewrite
\[
    B = A(R^\dagger U_{\eps\|A\|/2}D_{\eps\|A\|/2}^+) (t(D^2) D_{\eps\|A\|/2}^+ U_{\eps\|A\|/2}^\dagger) R.
\]
Now it suffices to get sampling and query access to the columns of $A(R^\dagger U_{\eps\|A\|/2}D_{\eps\|A\|/2}^+)$, and by \cref{lem:app-orth-expression}, $R^\dagger U_{\eps\|A\|/2}D_{\eps\|A\|/2}^+$ is an $\eps^3$-approximate isometry.
Further, we can lower bound the norms of these columns, using that $R^\dagger R \approx A^\dagger A$ and $CC^\dagger \approx RR^\dagger$.
\begin{align*}
    \|A(R^\dagger U_{\eps\|A\|/2}D_{\eps\|A\|/2}^+)\|^2
    &= \|(U_{\eps\|A\|/2}D_{\eps\|A\|/2}^+)^\dagger RA^\dagger AR^\dagger (U_{\eps\|A\|/2}D_{\eps\|A\|/2}^+)\| \\
    &\approx \|(U_{\eps\|A\|/2}D_{\eps\|A\|/2}^+)^\dagger RR^\dagger RR^\dagger (U_{\eps\|A\|/2}D_{\eps\|A\|/2}^+)\| \\
    &= \|RR^\dagger (U_{\eps\|A\|/2}D_{\eps\|A\|/2}^+)\|^2 \\
    &\approx \|CC^\dagger U_{\eps\|A\|/2}D_{\eps\|A\|/2}^+\|^2 \\
    &= \|UD^2U^\dagger U_{\eps\|A\|/2}D_{\eps\|A\|/2}^+\|^2 \\
    &\geq \eps^2\|A\|^2
\end{align*}
Consider one particular column $v \coloneqq [R^\dagger U_{\eps\|A\|/2}D_{\eps\|A\|/2}^+](\cdot,\ell)$; summarizing our prior arguments, we know $\|v\| \geq \frac12$ from approximate orthonormality and $\|Av\| \gtrsim \eps\|A\|$, which we just showed.
We can also query for entries of $v$ since it is a linear combination of rows of $R$.
We make one more approximation $Av \approx u$, using \cref{lemma:inner-prod} as we do in \cref{corollary:rec-systems}.
That is, if we want to know $[Av](i) = A(i,\cdot)v$, we use our inner product protocol to approximate it to $\gamma\|A(i,\cdot)\|\|v\|$ error, and declare it to be $u(i)$.
This implicitly defines $u$ via an algorithm to compute its entries from $\sq(A)$ and $\q(v)$.
Let $B'$ be the version of $B$, with the columns of $AR^\dagger U_{\eps\|A\|/2}D_{\eps\|A\|/2}^+$ replaced with their $u$ versions.
One can set $\gamma$ such that the correctness bound $\|B' - A^\dagger\| \lesssim \eps$ and our lower bound $u \gtrsim \eps\|A\|$ both still hold.
All we need now to get $\sq_\phi(u)$ (thereby completing our proof sketch) is a bound $\tilde{u}$ such that we have $\sq(\tilde{u})$.
We will take $\tilde{u}(i) \coloneqq 2\|A(i,\cdot)\|$.
We have $\sq(\tilde{u})$ immediately from $\sq(A)$, $\phi = \|\tilde{u}\|^2/\|u\|^2 \lesssim \eps^2\|A\|_\fr^2/\|A\|^2$ (from our lower bound on $\|u\|$), and $|\tilde{u}(i)| \geq \|A(i,\cdot)\| + \gamma\|A(i,\cdot)\|v\| \geq |u(i)|$ (from our correctness bound from \cref{lemma:inner-prod}).

\section{Deferred proofs}\label{apx:proofs}

\apporthfacts*

\begin{proof}
  Let $\hat{X} = UD V^\dagger$ be a singular value decomposition of $\hat{X}$, with singular values $\sigma_1,\ldots,\sigma_n$ and $U \in \bbC^{m\times n},\,D \in \bbR^{n\times n},\,V \in \bbC^{n\times n}$.
  We set $X\coloneqq UV^\dagger$ which is an isometry since $(UV^\dagger)^\dagger UV^\dagger = I$, and has the same columnspace as $\hat{X}$, and
  \begin{multline*}
    \|UV^\dagger - \hat{X}\|
    = \|UV^\dagger - UD V^\dagger\|
    = \|I - D\|
    = \max_{i \in [n]}\abs{1 - \sigma_i}
    \leq \max_{i \in [n]}\abs{1 - \sigma_i}\abs{1 + \sigma_i} \\
    = \max_{i \in [n]}\abs{1 - \sigma_i^2}
    = \|I - D^2\|
    = \|I - V D U^\dagger U D V^\dagger\|
    = \|I - \hat{X}^\dagger \hat{X}\|
    \leq \alpha.
  \end{multline*}
  Consequently,
  \begin{align*}
    \|\hat{X}Y\hat{X}^\dagger - XYX^\dagger\|
    &\leq \|\hat{X}Y(\hat{X} - X)^\dagger\| + \|(\hat{X} - X)YX\| \\
    &\leq \alpha(\|\hat{X}Y\| + \|YX\|) \\
    &\leq \alpha(\|XY\| + \alpha\|Y\| + \|YX\|) \\
    &= (2\alpha + \alpha^2)\|Y\|
  \intertext{Suppose $\alpha < 1$, ruling out the possibility that $\hat{X}$ is the zero matrix.
  Then by \cref{lem:weylineq} we have}
    \|\hat{X}^+\|
    &= \max_{i \in [n]} \frac{1}{\sigma_i}
    \leq \frac{1}{1 - \alpha}, \text{ and consequently } \\
    \|\hat{X}Y\hat{X}^\dagger - XYX^\dagger\|
    &\leq \alpha(\|\hat{X}Y\| + \|Y\|) \\
    &\leq \alpha(\|\hat{X}Y\hat{X}^\dagger\|\|\hat{X}^+\| + \|\hat{X}Y\hat{X}^\dagger\|\|\hat{X}^+\|^2) \\
    &\leq \alpha\frac{1 - \alpha + 1}{(1-\alpha)^2}\|\hat{X}Y\hat{X}^\dagger\|. \qedhere
  \end{align*}
\end{proof}

\lemmm*

\begin{proof}
Using that the rows of $S$ are selected independently, we can conclude the following:
\begin{align*}
  \E[(SX)^\dagger (SY)] &= r\cdot \E[[SX](1,\cdot)^\dagger [SY](1,\cdot)] = r \sum_{i=1}^np(i)\frac{X(i,\cdot)^\dagger Y(i,\cdot)}{r p(i)} = X^\dagger Y \\
  \E[\|X^\dagger S^\dagger SY - X^\dagger Y\|_\fr^2]
  &= \sum_{i=1}^m\sum_{j=1}^p\E\big[\big|[X^\dagger S^\dagger SY - X^\dagger Y](i,j)\big|^2\big] \\
  &= r\sum_{i=1}^m\sum_{j=1}^p\E\big[\big|[SX](1,i)^\dagger [SY](1,j) - [X^\dagger Y](i,j)\big|^2\big] \\
  &\leq r\sum_{i=1}^m\sum_{j=1}^p\E\big[\big|[SX](1,i)^\dagger [SY](1,j)\big|^2\big] \\
  &= r\E\big[\|[SX](1,\cdot)\|^2\|[SY](1,\cdot)\|^2\big] \\
  &= r \sum_{k=1}^n p(k)\frac{\|X(k,\cdot)\|^2}{r\cdot p(k)}\frac{\|Y(k,\cdot)\|^2}{r\cdot p(k)} \\
  &= \frac1r \sum_{k=1}^n \frac{1}{p(k)}\|X(k,\cdot)\|^2\|Y(k,\cdot)\|^2 \\
  &\leq \frac1r \sum_{k=1}^n \frac{\phi\sum_\ell \|X(\ell,\cdot)\|\|Y(\ell,\cdot)\|}{\|X(k,\cdot)\|\|Y(k,\cdot)\|}\|X(k,\cdot)\|^2\|Y(k,\cdot)\|^2 \\
  &= \frac{\phi}{r} \Big(\sum_k \|X(k,\cdot)\|\|Y(k,\cdot)\|\Big)^{\!2}.
\end{align*}
To prove concentration, we use McDiarmid's ``independent bounded difference inequality''~\cite{mcdiarmid89}.
\begin{lemma}[{\cite[Lemma (1.2)]{mcdiarmid89}}]\label{lem:ibdi}
  Let $X_1,\ldots, X_c$ be independent random variables with $X_s$ taking values in a set $A_s$ for all $s\in[c]$. Suppose that $f$ is a real valued measurable function on the product set $\Pi_s A_s$ such that $\abs{f(x)-f(x')}\leq b_s$ whenever the vectors $x$ and $x'$ differ only in the $s$-th coordinate. Let $Y$ be the random variable $f[X_1,\ldots,X_c]$. Then for any $\gamma>0$:
  \begin{align*}
  \Pr[\abs{Y-\E[Y]}\geq\gamma]\leq 2\exp\Big(-\frac{2\gamma^2}{\sum_s b_s^2}\Big).
  \end{align*}
\end{lemma}
To use \cref{lem:ibdi}, we think about this expression as a function of the indices that are randomly chosen from $p$.
That is, let $f$ be the function $[n]^r\to \mathbb{R}$ defined to be
\begin{align*}
f(i_1,i_2, \ldots, i_r)\coloneqq \Big\|X^\dagger Y-\sum_{s=1}^r\frac{1}{r\cdot p(i_s)}X(i_s,\cdot)^\dagger Y(i_s,\cdot)\Big\|_\fr,
\end{align*}
Then, by Jensen's inequality, we have
\begin{align*}
 \E[f] = \E[\|X^\dagger S^\dagger SY - X^\dagger Y\|_\fr] \leq \sqrt{\E[\|X^\dagger S^\dagger SY - XY\|_\fr^2]} \leq \sqrt{\frac{\phi}{r}}\sum_k \|X(k,\cdot)\|\|Y(k,\cdot)\|.
\end{align*}
Now suppose that the index sequences $\vec{i}$ and $\vec{i}'$ only differ at the $s$-th position. Then by the triangle inequality,
\begin{align*}
\abs{f(\vec{i})-f(\vec{i}')}&\leq \frac1r\Big\|\frac1{p(i_s)}X(i_s,\cdot)^\dagger Y(i_s,\cdot)-\frac1{p(i'_s)}X(i_s',\cdot)^\dagger Y(i_s',\cdot)\Big\|_\fr\\
&\leq \frac2r\max_{k \in [n]}\Big\|\frac1{p(k)}X(k,\cdot)^\dagger Y(k,\cdot)\Big\|_\fr \leq \frac{2\phi}{r}\sum_{k=1}^n\|X(k,\cdot)\|\|Y(k,\cdot)\|.
\end{align*}
Now, by \cref{lem:ibdi}, we conclude that
\begin{align*}
\Pr\left[\abs{f-E[f]}\geq \sqrt{\frac{2\phi^2\ln(2/\delta)}{r}}\sum_k\|X(\cdot,k)\|\|Y(k,\cdot)\|\right]\leq \delta.
\end{align*}
So, with probability $\geq 1-\delta$,
\begin{align*}
  \|X^\dagger S^\dagger SY - X^\dagger Y\|_\fr
  &\leq \E[\|X^\dagger S^\dagger SY - X^\dagger Y\|_\fr] + \sqrt{\frac{2\phi^2\ln(2/\delta)}{r}}\sum_k\|X(\cdot,k)\|\|Y(k,\cdot)\| \\
  &\leq \Big(\sqrt{\frac{\phi}{r}} + \sqrt{\frac{2\phi^2\ln(2/\delta)}{r}}\Big)\sum_k\|X(\cdot,k)\|\|Y(k,\cdot)\| \\
  &\leq \sqrt{\frac{8\phi^2\ln(2/\delta)}{r}}\sum_k\|X(\cdot,k)\|\|Y(k,\cdot)\|. \qedhere
\end{align*}
\end{proof}

\ipest*

\begin{proof}
Define a random variable $Z$ by sampling an index from the distribution $p$ given by $\sq_\phi(u)$, and setting $Z \coloneqq u(i)v(i)/p(i)$.
Then
\begin{align*}
    \E[Z] = \langle u,v\rangle \quad\mbox{and}\quad \E[\abs{Z}^2] = \sum_{i=1}^n p(i)\frac{\abs{u(i)v(i)}^2}{p(i)^2} \leq \sum_{i=1}^n \abs{u(i)v(i)}^2\frac{\phi\|u\|^2}{\abs{u(i)}^2} = \phi\|u\|^2\|v\|^2.
\end{align*}
So, we just need to boost the quality of this random variable.
Consider taking $\bar{Z}$ to be the mean of $x \coloneqq 8\phi \|u\|^2\|v\|^2\frac{1}{\eps^2}$ independent copies of $Z$.
Then, by Chebyshev's inequality (stated here for complex-valued random variables),
\begin{align*}
    \Pr[\abs{\bar{Z} - \E[\bar{Z}]} \geq \eps/\sqrt{2} ]\leq \frac{2\Var[Z]}{x\eps^2} \leq \frac14.
\end{align*}
Next, we take the (component-wise) median of $y \coloneqq 8\log\frac{1}{\delta}$ independent copies of $\bar{Z}$, which we call $\tilde{Z}$, to decrease failure probability.
Consider the median of the real parts of $\bar{Z}$.
The key observation is that if $\Re(\tilde{Z}-\E[Z]) \geq \eps/\sqrt{2}$, then at least half of the $\bar{Z}$'s satisfy $\Re(\bar{Z} - \E[Z]) \geq \eps/\sqrt{2}$.
Let $E_i = \chi(\Re(\bar{Z}_i - \E[Z]) \geq \eps/\sqrt{2})$ be the characteristic function for this event for a particular mean.
The above argument implies that $\Pr[E_i] \leq \frac14$.
So, by Hoeffding's inequality,
\begin{align*}
  \Pr\left[\frac1q\sum_{i=1}^q E_i \geq \frac12\right] \leq \Pr\left[\frac1q\sum_{i=1}^q E_i \geq \frac14 + \Pr[E_i]\right] \leq \exp(-q/8) \leq \frac\delta2.
\end{align*}
With this combined with our key observation, we can conclude that $\Pr[\Re(\tilde{Z} - \langle u,v\rangle) \geq \eps/\sqrt{2}] \leq \delta/2$.
From a union bound together with the analogous argument for the imaginary component, we have $\Pr[\abs{\tilde{Z} - \langle u,v\rangle} \geq \eps] \leq \delta$ as desired.
The time complexity is the number of samples multiplied by the time to create one instance of the random variable $Z$, which is $\bigO{\sqcb(u)+\qcb(v)}$.
\end{proof}

\lowdeg*

\begin{proof}
  We use the following Markov-Bernstein inequality {\cite[5.1.E.17.f]{borwein1995polynomials}}. For every $p\in\bbC[x]$ of degree at most $d$
  \begin{equation}\label{eqn:markov-bernstein}
    \max_{x \in [-1,1]}\abs{p^{(k)}(x)} \lesssim \Big(\min\Big(d^2,\frac{d}{\sqrt{1-x^2}}\Big)\Big)^k \max_{x \in [-1,1]} \abs{p(x)},
  \end{equation}
  where $\lesssim$ hides a constant depending on $k$.
  Note that by replacing $x$ in the above equation with $2y - 1$, we get that $\max_{y \in [0,1]} \abs{p^{(k)}(y)} \lesssim d^{2k} \max_{y \in [0,1]}\abs{p(y)}$ (paying an additional $2^k$ constant factor).

  We make a couple observations about $\bar{r}(x)$ using Taylor expansions, where $r(x)$ is any degree-$d$ polynomial.
  First,
  \[
    \bar{r}(x)
    = \frac{r(x) - r(0)}{x}
    = \frac{r(x) - (r(x) - r'(y)x)}{x}
    = r'(y),
  \]
  where $y \in [0,x]$ comes from the remainder term of the Taylor expansion of $r(x)$ at $x$.
  Similarly,
  \begin{multline*}
    \bar{r}'(x)
    = \Big(\frac{r(x) - r(0)}{x}\Big)'
    = \frac{1}{x^2}\Big(xr'(x) - r(x) + r(0)\Big) \\
    = \frac{1}{x^2}\Big(xr'(x) - r(x) + r(x) - r'(x)x + r''(y)\frac{x^2}{2}\Big)
    = \frac12r''(y)
  \end{multline*}
  for some $y \in [0,x]$.
  Then, for $p$ even, $\max_{x \in [0,1]} |q(x)|\leq 1$ by definition.
  We also have
  \begin{align*}
    \max_{x \in [0,1]} |q'(x)| &\lesssim d^2\max_{x \in [0,1]}|q(x)| \leq d^2 \\
    \max_{x \in [0,1]} |\bar{q}(x)| &\leq \max_{y \in [0,1]} |q'(y)| \lesssim d^2 \\
    \max_{x \in [0,1]} |\bar{q}'(x)| &\leq \max_{y \in [0,1]} \frac12|q''(y)| \lesssim d^4
  \end{align*}
  For $p$ odd, the same argument applies provided we can show that $\max_{x \in [0,1]}|q(x)| \lesssim d$, which we do by splitting into two cases: $x \leq \frac12$ and $x > \frac12$.
  \begin{align*}
    \max_{x \in [0,\frac12]} |q(x)|
    &= \max_{x \in [0,\frac12]} \Big|\frac{p(x)}{x}\Big|
    = \max_{y \in [0,\frac12]}|p'(y)|
    \lesssim \max_{x \in [0,\frac12]}\frac{d}{\sqrt{1-x^2}} \max_{x \in [-1,1]} |p(x)|
    \lesssim d; \\
    \max_{x \in (\frac12,1]} |q(x)|
    &= \max_{x \in (\frac12,1]} \Big|\frac{p(x)}{x}\Big|
    \leq \max_{x \in (\frac12, 1]} |2p(x)| \leq 2. \qedhere
  \end{align*}
\end{proof}

\svmrev*

\begin{proof}
Denote $\sigma^2 \coloneqq \gamma^{-1}$, and redefine $X \leftarrow X^T$ (so we have $\sq(X)$ instead of $\sq(X^T)$).
By the block matrix inversion formula\footnote{In a more general setting, we would use the Sherman-Morrison inversion formula, or the analogous formula for functions of matrices subject to rank-one perturbations.} we know that
\begin{align*}
  \begin{bsmallmatrix}
    0 & \vec{1}^T \\
    \vec{1} & M
  \end{bsmallmatrix}^{-1}
  = \begin{bsmallmatrix}
    -\frac{1}{\vec{1}^TM^{-1}\vec{1}} & \frac{\vec{1}^TM^{-1}}{\vec{1}^TM^{-1}\vec{1}} \\
    \frac{M^{-1}\vec{1}}{\vec{1}^TM^{-1}\vec{1}} & M^{-1} - \frac{M^{-1}\vec{1}\vec{1}^TM^{-1}}{\vec{1}^TM^{-1}\vec{1}}
  \end{bsmallmatrix}\quad \Rightarrow\quad
  \begin{bsmallmatrix}
    0 & \vec{1}^T \\
    \vec{1} & M
  \end{bsmallmatrix}^{-1}
  \begin{bsmallmatrix} 0 \\ y \end{bsmallmatrix}
  = \begin{bsmallmatrix}
    \frac{\vec{1}^TM^{-1}y}{\vec{1}^T M^{-1} \vec{1}} \\
    M^{-1}\Big(y - \frac{\vec{1}^TM^{-1}y}{\vec{1}^TM^{-1}\vec{1}}\vec{1}\Big)
  \end{bsmallmatrix}.
\end{align*}
So, we have reduced the problem of inverting the modified matrix to just inverting $M^{-1}$ where $M = X^TX + \sigma^{-2}I$.
$M$ is invertible because $M \succeq \sigma^2 I$.
Note that $M^{-1} = f(X^TX)$, where
\begin{align*}
  f(\lambda) \coloneqq \frac{1}{\lambda + \sigma^2}.
\end{align*}
So, by \cref{thm:evenSing}, we can find $R^\dagger \bar{f}(CC^\dagger) R$ such that $\|R^\dagger \bar{f}(CC^\dagger) R + \frac{1}{\sigma^2} I - f(X^TX)\| \leq \eps\sigma^{-2}$, where (because $L = \sigma^{-4}$, $\bar{L} = \sigma^{-6}$)
\begin{align*}
  r &= \bOt{\frac{L^2\|A\|^2\|A\|_\fr^2\sigma^4}{\eps^2}\log\frac{1}{\delta}}
  = \bOt{\frac{K\kappa}{\eps^2}\log\frac{1}{\delta}}, \\
  c &= \bOt{\frac{\bar{L}^2\|A\|^6\|A\|_\fr^2\sigma^4}{\eps^2}\log\frac{1}{\delta}}
  = \bOt{\frac{K\kappa^3}{\eps^2}\log\frac{1}{\delta}}.
\end{align*}
So, the runtime for estimating this is $\bOt{\frac{K^3\kappa^5}{\eps^6}\log^3\frac1\delta}$.
We further approximate using \cref{prop:appr-mms}: we find $r_1 \approx R^\dagger \vec{1}$, $r_y \approx R^\dagger \vec{y}$, and $\gamma \approx \vec{1}^\dagger y$ in $O(r\frac{K}{\eps^2}\log\frac{1}{\delta})$ time (for the first two) and $O(\frac{1}{\eps^2}\log\frac{1}{\delta})$ time (for the last one) such that the following bounds hold:
\begin{equation} \label{eqn:svm-approx-bounds}
  \|R^\dagger \vec{1} - r_1\| \leq \eps\sqrt{m}\sigma
  \qquad \|R^\dagger y - r_y\| \leq \eps\sqrt{m}\sigma
  \qquad \abs{\vec{1}^\dagger y - \gamma} \leq \eps m
\end{equation}
Via \cref{lem:evenSingBounds}, we observe the following additional bounds:
\begin{equation} \label{eqn:svm-svt-bounds}
  \|M^{-1}\| \leq \sigma^{-2}
  \qquad \|R^\dagger (\bar{f}(CC^\dagger))^{1/2}\| \leq (1+\eps)\sigma^{-1}
  \qquad \|(\bar{f}(CC^\dagger))^{1/2}\| \leq \sigma^{-2}
\end{equation}

Now, we compute what the subsequent errors are for replacing $M^{-1}$ with $N \coloneqq R^\dagger Z R + \frac{1}{\sigma^2}I$, where $Z \coloneqq \bar{f}(CC^\dagger)$.
\begin{align*}
\frac{\vec{1}^\dagger M^{-1}y}{\vec{1}^\dagger M^{-1}\vec{1}}&= \frac{\vec{1}^\dagger(R^\dagger Z R + \sigma^{-2}I)y \pm \|\vec{1}\|\|y\|\|R^\dagger Z R + \sigma^{-2}I - M^{-1}\|}{\vec{1}^\dagger(R^\dagger Z R + \sigma^{-2}I)\vec{1} \pm \|\vec{1}\|^2\|R^\dagger Z R + \sigma^{-2}I - M^{-1}\|} \\
  &= \frac{\vec{1}^\dagger R^\dagger Z Ry + \sigma^{-2}\vec{1}^\dagger y \pm \eps\sigma^{-2}m}{\vec{1}^\dagger R^\dagger Z R\vec{1} + \sigma^{-2}\vec{1}^\dagger \vec{1} \pm \eps\sigma^{-2}m} \tag*{by SVT bound} \\
  &= \frac{\vec{1}^\dagger R^\dagger Z r_y \pm \|\vec{1}^\dagger R^\dagger Z\|\|Ry - r_y\| + \sigma^{-2}\gamma \pm \sigma^{-2}\abs{\gamma - \vec{1}^\dagger y} \pm \eps\sigma^{-2}m}{\vec{1}^\dagger R^\dagger Z r_1 \pm \|\vec{1}^\dagger R^\dagger Z\|\|R\vec{1} - r_1\| + (1\pm \eps)\sigma^{-2}m} \\
  &= \frac{\vec{1}^\dagger R^\dagger Z r_y \pm (\sqrt{m}(1+\eps)\sigma^{-3})(\eps\sigma\sqrt{m}) + \sigma^{-2}\gamma \pm 2\eps\sigma^{-2}m}{\vec{1}^\dagger R^\dagger Z r_1 \pm (\sqrt{m}(1+\eps)\sigma^{-3})(\eps\sigma\sqrt{m}) + \sigma^{-2}m \pm \eps\sigma^{-2}m} \tag*{by \cref{eqn:svm-approx-bounds,eqn:svm-svt-bounds}} \\
  &= \frac{r_1^\dagger Z r_y \pm \|R\vec{1} - r_1\|\|Zr_y\| + \sigma^{-2}\gamma \pm \bigO{\eps\sigma^{-2}m}}{r_1^\dagger Z r_1 \pm \|R\vec{1} - r_1\|\|Zr_1\| + \sigma^{-2}m \pm \bigO{\eps\sigma^{-2}m}} \\
  &= \frac{r_1^\dagger Z r_y \pm \eps\sigma\sqrt{m}(\|ZRy\| + \|Z\|\|Ry - r_y\|) + \sigma^{-2}\gamma \pm \bigO{\eps\sigma^{-2}m}}{r_1^\dagger Z r_1 \pm \eps\sigma\sqrt{m}(\|ZR\vec{1}\| + \|Z\|\|R\vec{1} - r_1\|) + \sigma^{-2}m \pm \bigO{\eps\sigma^{-2}m}} \tag*{by \cref{eqn:svm-approx-bounds}} \\
  &= \frac{r_1^\dagger Z r_y + \sigma^{-2}\gamma \pm \bigO{\eps\sigma^{-2}m}}{r_1^\dagger Z r_1 + \sigma^{-2}m \pm \bigO{\eps\sigma^{-2}m}} \tag*{by \cref{eqn:svm-approx-bounds,eqn:svm-svt-bounds}} \\
  &= \frac{r_1^\dagger Z r_y + \sigma^{-2}\gamma}{r_1^\dagger Z r_1 + \sigma^{-2}m}(1 \pm \bigO{\eps}) \pm \bigO{\eps} \tag*{by $r_1^\dagger Z r_1 \geq 0$}.
\end{align*}
We will approximate the output vector as
\[
  M^{-1}y - \frac{\vec{1}^\dagger M^{-1}y}{\vec{1}^\dagger M^{-1}\vec{1}}M^{-1}\vec{1} \approx R^\dagger Zr_y + \sigma^{-2}y - \frac{r_1^\dagger Z r_y + \sigma^{-2}\gamma}{r_1^\dagger Z r_1 + \sigma^{-2}m}(R^\dagger Z r_1 + \sigma^{-2}\vec{1}).
\]
To analyze this, we first note that
\begin{align*}
  \|M^{-1}y - R^\dagger Zr_y + \sigma^{-2}y\|
  &\leq \|M^{-1} - R^\dagger Z R - \sigma^{-2}I\|\|y\| + \|R^\dagger Z\|\|Ry - r_y\| \\
  &\leq \eps\sigma^{-2}\sqrt{m} + (1+\eps)\sigma^{-3}\eps\sigma\sqrt{m}\lesssim \eps\sigma^{-2}\sqrt{m}
\end{align*}
and analogously, $\|M^{-1}\vec{1} - R^\dagger Zr_1 + \sigma^{-2}\vec{1}\| \lesssim \eps\sigma^{-2}\sqrt{m}$.
We also use that
\begin{align} \label{eqn:svm-scalar-bound}
  \frac{\vec{1}^\dagger M^{-1}y}{\vec{1}^\dagger M^{-1}\vec{1}}
  \leq \frac{\|\vec{1}M^{-1/2}\|\|M^{-1/2}y\|}{\|M^{-1/2}\vec{1}\|^2}
  = \frac{\|M^{-1/2}y\|}{\|M^{-1/2}\vec{1}\|}
  \leq \frac{\|X\|}{\sigma}.
\end{align}
With these bounds, we can conclude that (continuing to use \cref{eqn:svm-svt-bounds,eqn:svm-approx-bounds})
\begin{align*}
  &\left\|M^{-1}\Big(y - \frac{\vec{1}^\dagger M^{-1}y}{\vec{1}^\dagger M^{-1}\vec{1}}\vec{1}\Big) -
  \Big(R^\dagger Zr_y + \sigma^{-2}y - \frac{r_1^\dagger Z r_y + \sigma^{-2}\gamma}{r_1^\dagger Z r_1 + \sigma^{-2}m}(R^\dagger Z r_1 + \sigma^{-2}\vec{1})\Big)\right\| \\
  &\leq \|M^{-1}y - R^\dagger Zr_y + \sigma^{-2}y\| + \left\|\frac{\vec{1}^\dagger M^{-1}y}{\vec{1}^\dagger M^{-1}\vec{1}}M^{-1}\vec{1} - \frac{r_1^\dagger Z r_y + \sigma^{-2}\gamma}{r_1^\dagger Z r_1 + \sigma^{-2}m}(R^\dagger Z r_1 + \sigma^{-2}\vec{1})\Big)\right\| \\
  &\leq \eps\sigma^{-2}\sqrt{m} + \frac{\vec{1}^\dagger M^{-1}y}{\vec{1}^\dagger M^{-1}\vec{1}}\|M^{-1}\vec{1} - R^\dagger Z r_1 - \sigma^{-2}\vec{1}\| + \Big|\frac{\vec{1}^\dagger M^{-1}y}{\vec{1}^\dagger M^{-1}\vec{1}} - \frac{r_1^\dagger Z r_y + \sigma^{-2}\gamma}{r_1^\dagger Z r_1 + \sigma^{-2}m}\Big|\|R^\dagger Z r_1 + \sigma^{-2}\vec{1}\| \\
  &\lesssim \Big(1+\frac{\vec{1}^\dagger M^{-1}y}{\vec{1}^\dagger M^{-1}\vec{1}}\Big)\eps\sigma^{-2}\sqrt{m} + \eps\Big(1+\frac{\vec{1}^\dagger M^{-1}y}{\vec{1}^\dagger M^{-1}\vec{1}}\Big)\|R^\dagger Z r_1 + \sigma^{-2}\vec{1}\| \\
  &= \eps\Big(1+\frac{\vec{1}^\dagger M^{-1}y}{\vec{1}^\dagger M^{-1}\vec{1}}\Big)\Big(\sigma^{-2}\sqrt{m} + \|R^\dagger Z r_1 + \sigma^{-2}\vec{1}\|\Big) \\
  &\lesssim \eps\frac{\|X\|}{\sigma}\Big(\sigma^{-2}\sqrt{m} + \|M^{-1}\vec{1}\| + \|(R^\dagger Z R + \sigma^{-2}I - M^{-1})\vec{1}\| + \|R^\dagger Z r_1 - R^\dagger Z R\vec{1}\|\Big) \tag*{by \cref{eqn:svm-scalar-bound}} \\
  &\lesssim \eps\frac{\|X\|}{\sigma}\Big(\sigma^{-2}\sqrt{m} + \sigma^{-2}\sqrt{m} + \eps\sigma^{-2}\sqrt{m} + \eps\sigma^{-2}\sqrt{m}\Big) \\
  &\lesssim \eps\frac{\|X\|}{\sigma}\sigma^{-2}\sqrt{m}.
\end{align*}
So, by rescaling $\eps$ down by $\frac{\|X\|}{\sigma}$, it suffices to sample from
\[
  \hat{\alpha} \coloneqq R^\dagger Z\Big(r_y - \frac{r_1^\dagger Z r_y + \sigma^{-2}\gamma}{r_1^\dagger Z r_1 + \sigma^{-2}m}r_1\Big) - \sigma^{-2}\Big(y - \frac{r_1^\dagger Z r_y + \sigma^{-2}\gamma}{r_1^\dagger Z r_1 + \sigma^{-2}m}\vec{1}\Big).
\]
To gain sampling and query access to the output, we consider this as a matrix-vector product, where the matrix is $(R^\dagger\mid y\mid \vec{1})$ and the vector is the corresponding coefficients in the linear combination.
Then, by \cref{lem:b-sq-approx,lemma:sample-Mv}, we can get $\sq_\phi(\hat{\alpha})$ for
\begin{align*}
  \phi &= (r+2)\left(\frac{\|X\|_\fr^2}{r}\Big\|Z\Big(r_y - \frac{r_1^\dagger Z r_y + \sigma^{-2}\gamma}{r_1^\dagger Z r_1 + \sigma^{-2}m}r_1\Big)\Big\|^2 + \sigma^{-4}\Big(\|y\|^2 + \Big(\frac{r_1^\dagger Z r_y + \sigma^{-2}\gamma}{r_1^\dagger Z r_1 + \sigma^{-2}m}\Big)^2\|\vec{1}\|^2\Big)\right)\|\hat{\alpha}\|^{-2} \\
  &\lesssim \left(\|X\|_\fr^2\frac{\|X\|^2}{\sigma^2}\sigma^{-6}m + r\sigma^{-4}\frac{\|X\|^2}{\sigma^2}m\right)\|\hat{\alpha}\|^{-2}\lesssim \Big(\frac{\|X\|_\fr^2}{\sigma^2} + r\Big)\frac{\|X\|^2}{\sigma^2}\frac{\sigma^{-4}m}{\|\hat{\alpha}\|^2}
\end{align*}
so $\sqrun_\phi(\hat{\alpha}) = \phi\sqcb_\phi(\hat{\alpha})\log\frac{1}{\delta} = \bigO{r(\frac{\|X\|_\fr^2}{\sigma^2} + r)\frac{\|X\|^2}{\sigma^2}\frac{\sigma^{-4}m}{\|\hat{\alpha}\|^2}\log\frac{1}{\delta}}$.
\end{proof}

\partpert*

\begin{proof}
  We will use the following formula introduced by \cite{KarplusMicriWave,FeynmanOpCalc} (see also \cite[Page 181]{BellmanIntroMatrixAnal}):
  \begin{equation}\label{eq:expDerIntegral}
  \frac{d}{dt}e^{M(t)}=\int_{0}^{1}e^{y M(t)}\frac{d M(t)}{dt}e^{(1-y) M(t)} dy.
  \end{equation}
  Let $A\in\bbC^{n\times n}$ with $\nrm{A}\leq 1$, we define the function $g_A(t)\coloneqq \Tr\big(A e^{H+t(\tilde{H}-H)}\big)$, and observe that
  \begin{align}
  g_A'(t)=& \frac{d}{dt}\Tr\left(A e^{H+t(\tilde{H}-H)}\right)\tag*{by definiton}\\
  =& \Tr\left(A\frac{d}{dt} e^{H+t(\tilde{H}-H)}\right)\tag*{by linearity of trace}\\
  =& \Tr\left(A\int_{0}^{1}e^{y [H+t(\tilde{H}-H)]}(\tilde{H}-H)e^{(1-y) [H+t(\tilde{H}-H)]} dy\right)\tag*{by \cref{eq:expDerIntegral}}\\
  =& \int_{0}^{1}\Tr\left(A e^{y [H+t(\tilde{H}-H)]}(\tilde{H}-H)e^{(1-y) [H+t(\tilde{H}-H)]} \right)dy\tag*{by linearity of trace\footnotemark}\\
  \leq& \int_{0}^{1}\nrm{A e^{y [H+t(\tilde{H}-H)]}(\tilde{H}-H)e^{(1-y) [H+t(\tilde{H}-H)]} }_1dy\tag*{by trace-norm inequality}\\
  \leq& \int_{0}^{1}\nrm{A e^{y [H+t(\tilde{H}-H)]}}_{\frac1y}\nrm{(\tilde{H}-H)e^{(1-y) [H+t(\tilde{H}-H)]} }_\frac{1}{1-y}dy\tag*{by Hölder's inequality}\\
  \leq& \int_{0}^{1}\nrm{A} \nrm{e^{y [H+t(\tilde{H}-H)]}}_{\frac1y}\nrm{\tilde{H}-H}\nrm{e^{(1-y) [H+t(\tilde{H}-H)]} }_\frac{1}{1-y}dy\tag*{by Hölder's inequality}\\
  \leq& \nrm{\tilde{H}-H}\int_{0}^{1} \nrm{e^{y [H+t(\tilde{H}-H)]}}_{\frac1y}\nrm{e^{(1-y) [H+t(\tilde{H}-H)]} }_\frac{1}{1-y}dy\tag*{since $\nrm{A}\leq 1$}\\
  =& \nrm{\tilde{H}-H}\nrm{e^{H+t(\tilde{H}-H)}}_1.\label{eq:derivBound}
  \end{align} \footnotetext{Note that in case $A=I$, by the cyclicity of trace, this equation implies that $\frac{d}{dt}\Tr(e^{H(t)})=\Tr(e^{H(t)}\frac{d}{dt}H(t))$.}

  Now we consider $z(t)\coloneqq g_I(t)=\Tr\big(e^{H+t(\tilde{H}-H)}\big)$.
  From \cref{eq:derivBound} we have $z'(t)\leq \|\tilde{H}-H\|z(t)$.
  Using Grönwall's differential inequality, we can conclude that $z(t)\leq z(0)e^{t\|\tilde{H}-H\|}$ for every $t\in [0,\infty)$.

  Finally, we use the fact that there exists a matrix $A$ of operator norm at most $1$ such that $\big\|e^{\tilde{H}}-e^H\big\|_1=\Tr(A(e^{\tilde{H}}-e^H))$ (take, e.g., $\operatorname{sgn}(e^{\tilde{H}}-e^H)$).
  We finish the proof by observing that for such an $A$, $\big\|e^{\tilde{H}}-e^H\big\|_1=\Tr(Ae^{\tilde{H}})-\Tr(A e^H)=g_A(1)-g_A(0)=\int_{0}^{1}g_A'(t) dt$ and
  \begin{align*}
  \int_{0}^{1}g_A'(t) dt
  \overset{\eqref{eq:derivBound}}{\leq} \int_{0}^{1}\|\tilde{H}-H\|z(t)dt\leq z(0)\int_{0}^{1}\|\tilde{H}-H\|e^{t\|\tilde{H}-H\|} dt=\Tr(e^H)\left(e^{\|\tilde{H}-H\|}-1\right). &\qedhere
  \end{align*}
\end{proof} 

\end{document}